\documentclass[11pt,twoside,a4paper]{book}

\usepackage[T1]{fontenc}
\usepackage[utf8]{inputenc}
\usepackage[english]{babel}
\usepackage[pdftex]{graphicx}           
\usepackage{setspace}                   
\usepackage{indentfirst}                
\usepackage{makeidx}                    
\usepackage[nottoc]{tocbibind}          
\usepackage{courier}                    
\usepackage{type1cm}                    
\usepackage{listings}                   
\usepackage{titletoc}
\usepackage[fixlanguage]{babelbib}
\usepackage[font=small,format=plain,labelfont=bf,up,textfont=it,up]{caption}
\usepackage[usenames,svgnames,dvipsnames]{xcolor}
\usepackage[a4paper,top=2.54cm,bottom=2.0cm,left=2.0cm,right=2.54cm]{geometry} 
\usepackage[pdftex,plainpages=false,pdfpagelabels,pagebackref,colorlinks=true,citecolor=DarkGreen,linkcolor=NavyBlue,urlcolor=DarkRed,filecolor=green,bookmarksopen=true]{hyperref} 
\usepackage[all]{hypcap}                
\usepackage[square,sort,nonamebreak,comma]{natbib}  
\fontsize{60}{62}\usefont{OT1}{cmr}{m}{n}{\selectfont}

\usepackage{ragged2e}

\usepackage{epigraph}
\usepackage{framed}
\usepackage{tikz, tkz-euclide}
\usetikzlibrary{patterns}
\usetikzlibrary{calc}
\usetikzlibrary{decorations.markings}
\usepackage{fancybox}
\usepackage{mdframed}
\usepackage{float}

\usepackage{amsmath,amssymb,amsthm,epsfig,color,graphicx, mathtools, bbm, dsfont, geometry, float, mathrsfs, upgreek, bm, enumitem, comment, hyperref,hyphenat}

\usepackage{graphicx}
\graphicspath{ {images/} }
\newcommand{\Z}{\mathbb{Z}}
\newcommand{\bbr}{\mathbb{R}}
\newcommand{\bbn}{\mathbb{N}}
\newcommand{\din}{\partial_{\mathrm{in}}}

\newcommand{\diam}{\mathrm{diam}}
\newcommand{\dis}{\mathrm{dist}}

\newcommand{\lab}{\mathrm{lab}}     
\newcommand{\Conv}{\sideset{}{^{(\ast)}}\prod}%

\newcommand{\D}{\mathcal{D}}
\newcommand{\T}{\mathcal{T}}
\newcommand{\I}{\mathrm{I}}
\newcommand{\Sp}{\mathrm{sp}}
\newcommand{\ssum}[1]{
	\sum_{\mathclap{\substack{#1}}}
}

\newcommand{\vertiii}[1]{{\left\vert\kern-0.25ex\left\vert\kern-0.25ex\left\vert #1 
    \right\vert\kern-0.25ex\right\vert\kern-0.25ex\right\vert}}
\def\supp{\mathop{\textrm{\rm supp}}\nolimits}            

\newcommand{\be}{\begin{equation}}
	\newcommand{\ee}{\end{equation}}

\numberwithin{equation}{section}

\usepackage{bbm}
\newtheorem*{theorem*}        {Theorem}
\newtheorem*{conjecture*}   {Conjecture}
\newtheorem{theorem}           {Theorem}[section]
\newtheorem{lemma}[theorem]{Lemma}
\newtheorem*{lemma*}          {Lemma}

\newtheorem{definition}[theorem]{Definition}
\newtheorem{example}[theorem]{Example}
\newtheorem{corollary}[theorem]{Corollary}
\newtheorem{proposition}[theorem]{Proposition}

\newtheorem{remark}[theorem]{Remark}

\usepackage{fancyhdr}
\renewcommand{\chaptermark}[1]{\markboth{\MakeUppercase{#1}}{}}

\graphicspath{{./figuras/}}             
\makeindex                              
\fontsize{60}{62}\usefont{OT1}{cmr}{m}{n}{\selectfont}
\cleardoublepage
\normalsize

\begin{document}
\frontmatter 
\fancyhead[RO]{{\footnotesize\rightmark}\hspace{2em}\thepage}
\setcounter{tocdepth}{2}
\fancyhead[LE]{\thepage\hspace{2em}\footnotesize{\leftmark}}
\fancyhead[RE,LO]{}
\fancyhead[RO]{{\footnotesize\rightmark}\hspace{2em}\thepage}

\onehalfspacing  

\thispagestyle{empty}
\begin{center}
    \vspace*{2.3cm}
    \textbf{\Large{Multidimensional Contours \`a la Fr\"{o}hlich-Spencer and Boundary Conditions for Quantum Spin Systems}}\\
    
    \vspace*{1.2cm}
    \Large{Lucas Affonso}
    
    \vskip 2cm
    \textsc{
    Thesis presented\\[-0.25cm] 
    to the \\[-0.25cm]
    Institute of Mathematics and Statistics\\[-0.25cm]
    of the\\[-0.25cm]
    University of S\~ao Paulo\\[-0.25cm]
    in partial fulfillment of the requirements\\[-0.25cm]
    for the degree\\[-0.25cm]
    of\\[-0.25cm]
    Doctor of Science}
    
    \vskip 1.5cm
    Program: Applied Mathematics\\
    Advisor: Prof. Dr. Rodrigo Bissacot\\ 

   	\vskip 1cm
    \normalsize{During the development of this work the author was supported by FAPESP grants 2017/18152-2 and 2020/14563-0.}
    
    \vskip 0.5cm
    \normalsize{S\~ao Paulo, May 2023}
\end{center}

%
%
%




%
%
%
%
\newpage
\thispagestyle{empty}
    \begin{center}
        \vspace*{2.3 cm}
        \textbf{\Large{Multidimensional Contours \`a la Fr\"{o}hlich-Spencer and Boundary Conditions for Quantum Spin Systems}}\\
        \vspace*{2 cm}
    \end{center}

    \vskip 2cm

    \begin{FlushRight}
     This is the final version of this thesis and it contains corrections \\
     and changes suggested by the examiner committee during the\\
     defense of the original work realized on July 26th, 2023.\\
     A copy of the original version of this text is available at the \\
     Institute of Mathematics and Statistics of the University of S\~ao Paulo.
    
    \vskip 2cm

   \end{FlushRight}
    \vskip 4.2cm

    \begin{quote}
    \noindent Examination Board:
    
    \begin{itemize}
		\item Prof. Dr. Walter de Siqueira Pedra - ICMC-USP (President)
		\item Prof. Dr. In\'{e}s Armend\'ariz  - Buenos Aires University
		\item Prof. Dr. Roberto Fern\'andez - NYU-Shanghai
        \item Prof. Dr. Abel Klein - University of California - Irvine
        \item Prof. Dr. Pieter Naaijkens - Cardiff University
    \end{itemize}
      
    \end{quote}
\pagebreak

\pagenumbering{roman}     

\chapter*{Acknowledgements}

\epigraph{The gatekeeper has to bend way down to him, for the great difference has changed things to the disadvantage of the man. “What do you still want to know, then?” asks the gatekeeper. “You are insatiable.” “Everyone strives after the law,” says the man, “so how is that in these many years, no one except me has requested entry?” The gatekeeper sees that the man is already dying and, in order to reach his diminishing sense of hearing, he shouts at him, “Here no one else can gain entry, since this entrance was assigned only to you. I’m going now to close it."
}{\textbf{Franz Kafka}\\ \textit{Before the Law}}

The quote opening this section is the final paragraph of the text \textit{Before the Law} from Franz Kafka. In that story, our hero encounters the opportunity of his life: to cross the gate and find what he thinks he wants. The object of his desire, in his belief, will then make him happy. But the story is not one with a happy ending. Although the gate is open and our hero can glimpse what is inside, he encounters an obstacle, a gatekeeper, who states that he is powerful and that he cannot grant our hero entrance at that moment. Our hero accepts passively his misfortune, and just sits and waits until the gatekeeper grants him permission to enter. The reader can verify by the quote, this never happens. But the passage was made for him! Sometimes in life, we glimpse an opportunity to go to the other side, we just need to be brave enough to do it. There will always be gatekeepers, but this does not mean the path we chose is not made for us. Difficulties abound, and somehow they are necessary for the proper transformation into what you want to be. The Ph.D. years were, in many ways, pretty much like the process described by the gatekeeper in Kafka's story: Too many gatekeepers, each one more powerful than the other. Sometimes you don't know what to expect and, as my therapist Marco Chiusano said to me many times, sometimes one needs to lose control in order to gain some. To forget the old ways of a simple undergrad student and become a researcher is challenging and a process that none of us can pass unharmed. Usually one spends many of their youth years in this pursuit, entangling personal life and career, working hard towards a goal that is not clear they will be able to achieve. The gatekeeper is indeed very powerful. I can only say that I feel very happy to have had this experience, ending this part of my life in a very positive way: Among friends and family, during a huge Mathematical Physics conference and having the opportunity to show them what I found during these years. I would like to start by thanking my mother and father, Valéria Affonso Silva, and Carlos Eduardo Pereira for being so supportive during these years. My family does not come from a privileged background, we are from the Zona Leste of São Paulo, and they decided to bet on a good education for me and my sister. Speaking of her, I would also like to thank Karoline Affonso Silva Pereira for being supportive and taking good care of me. She is, in a sense, my oldest friend (almost 30 years and still counting!). I would also like to thank my grandmother and grandfather Sueli Affonso Silva and Paulo Antônio Silva, and my father-in-law and mother-in-law Luciene de Morais Rebechi and Antônio Donizetti Rebechi for their love and support. 
During my undergrad years, I made many good friends that are still with me today. They are Denis Assis Pinto Garcia, Leonardo Barbosa, Geovane Grossi, Henrique Corsini, Nickolas Kokron, Matheus Prado, Jean Lazarotto, Pedro Mendes, and Christian Táfula. I am very happy to have shared a period of my life with you. Pedro helped me a lot by inviting me to stay at his home during the summer of 2016 in Rio, where I could take the Functional Analysis course at IMPA. This was very important to me and I would like to thank his kindness (and also the kindness of his father!). I would like to thank the new friends I made during my Ph.D. years, Lucas Garcia, João Fernando Nariyoshi, Thiago Raszeja, who kindly made many of the figures you will see at this thesis, Rodrigo Frausino, Eric Endo, Rodrigo Cabral, Kelvyn Welsch, João Maia, João Rodrigues. I had many interesting discussions with all of you during these years.

I would like to give special thanks to my love Vitória de Morais Rebechi. We have been together since my early undergrad years and we changed a lot during this period. I am very happy to share my life with you. I strongly believe that without your support, love, care, and faith in me, I would not be able to be here. We've experienced moments of joy and sadness together, but through it all, we've stood strong. I'm grateful for our journey, and I eagerly look forward to many more years of being by your side. 

Special thanks also to the members of the Kings of August group: Christian Táfula, João Nariyoshi, and Lucas Garcia. I had many good laughs and thoughtful discussions with you. Christian was in Canada during the year I spent there, under the supervision of Marcelo Laca which helped me a lot, and we spoke a lot during this period, sharing laughs and having deep conversations, I can safely say that you had a deep impact on what I am now.  

I consider myself to be very lucky to have had many mentors during this period. The first one is Walter Pedra. We met when I was a naive undergrad, and through his guidance and patience, I took my first steps in mathematical physics, more specifically on applications of the theory of C*-algebras to quantum statistical mechanics. I would also like to thank Eric Endo, one of my coauthors and a very good friend, from whom I learned a lot. I thank Aernout van Enter for reading a preliminary version of this thesis and for many discussions we had through the years, I learned a lot of classical statistical mechanics and many references on relevant problems. I think this work would be very impoverished without your help. 

A special and very important place in this history is reserved for my advisor, Rodrigo Bissacot. During the first year, when I visited UFMG, I took a look at his doctoral thesis where I could find the words 
\begin{quote}
In the absence of a quote from a famous writer, I end with a sentence that summarizes well the current situation in the country and my journey to this point. This was possible thanks to those I mentioned above and many other colleagues and friends who, even when they saw that the situation was complicated, didn't try to destroy a dream: "In a country where a metallurgist can be President, a welder's son can have a Ph.D."    
\end{quote}

I would say that many outside of academia may not understand the power this quote has. It gave me the certainty that I could continue on this path. Rodrigo has a huge heart and works for at least 5 people (just a low estimate, some experts believe that number must be higher!) and is relentless in his love for science and the university.
I will never forget the late nights we spent working at IME along with Joao Vitor Maia and many others. I also learned from you how to be a scientist during this great period we spent together and I was able to witness you create a great mathematical physics group. I feel honored to have participated in this period of construction and I hope to be present when many more victories and successes happen, in the many years to come (Although maybe, as a plan B, we should call Thiago Raszeja and open that store...)

Research... is hard! 

I would also like to thank the Examination Board for accepting to come to Brazil during this very special moment and for the improvements suggested to this thesis. Finally, I want to thank the Conselho Nacional de Desenvolvimento Científico e Tecnológico (CNPq) processo 132966/2017-4 and the Fundação ao Amparo à Pesquisa do Estado de São Paulo (FAPESP) for the financial support during the development of the work presented in this thesis through the grants 2017/18152-2 and 2020/14310-5. 

I think I still haven't reached the other side, so there are many more gatekeepers to meet. This will not be a problem, for as long as I can share the journey with the people I mentioned here (and the many that I will meet along the way). 

\chapter*{Resumo}

    \noindent AFFONSO, L. \textbf{Contornos multidimensionais \`a la Fr\"ohlich-Spencer e Condi\c{c}\~{o}es de Fronteira em Sistemas de Spin Qu\^{a}ntico}. 
2023. 119 f.
Tese (Doutorado) - Instituto de Matem\'atica e Estat\'istica,
Universidade de S\~ao Paulo, S\~ao Paulo, 2023.
\\

Nesta tese, apresentamos resultados advindos da investiga\c{c}\~ao de dois problemas: um deles, relacionado a transi\c{c}\~ao de fase de modelos de Ising de longo-alcance e o outro, está relacionado com a caracterização de estados de equilíbrio em sistemas de spin quântico. 

Devido ao caráter longo-alcance das intera\c{c}\~oes do tipo $J|x-y|^{-\alpha}$, estimativas usando contornos usualmente encontrados na literatura apresentam restri\c{c}\~oes no alcance das intera\c{c}\~oes ($\alpha>d+1$ em Ginibre, Grossmann e Ruelle em 1966 e subsequentemente Park em 1988 para sistemas com spin discreto possivelmente não simétricos porém $\alpha>3d+1$). Conseguimos estender o argumento de transi\c{c}\~ao de fase para modelos tipo Ising de longo-alcance ferromagnéticos para toda regi\~ao $\alpha>d$ utilizando os argumentos multiescala apresentados nos artigos de Fr\"ohlich e Spencer. 

 Em mec\^anica estat\'istica qu\^antica, a condi\c{c}\~ao KMS \'e utilizada como caracteriza\c{c}\~ao dos estados de equilibrio do sistema. Amplamente estudada hoje em dia, sabe-se que esta condição é equivalente a outras no\c{c}\~oes de equil\'ibrio tal como a de satisfazer o princ\'ipio variacional para sistemas invariantes por transla\c{c}\~ao. Apresentamos uma outra poss\'ivel caracteriza\c{c}\~ao de estados de equilíbrio para sistemas de spin quântico atrav\'es de uma generaliza\c{c}\~ao das equações DLR para o contexto qu\^antico utilizando representa\c{c}\~oes com processos de Poisson. Tamb\'em discutimos a rela\c{c}\~ao destas equa\c{c}\~oes DLR qu\^anticas com os estados KMS de uma subclasse de intera\c{c}\~oes que cont\'em o modelo de Ising quântico com campo transversal.
\\
\noindent \textbf{Palavras-chave:} transi\c{c}\~ao de fase, modelo de Ising longo-alcance, contornos, análise multiescala, Fr\"ohlich-Spencer, mec\^anica estat\'istica cl\'assica, estados KMS, algebras-C* de grup\'oides, processos de Poisson, mecânica estatística quântica.

\chapter*{Abstract}
\noindent AFFONSO, L. \textbf{Multidimensional Contours \`a la Fr\"{o}hlich-Spencer and Boundary Conditions for Quantum Spin Systems}. 
2023. 119 pages.
PhD Thesis - Institute of Mathematics and Statistics,
University of S\~ao Paulo, S\~ao Paulo, 2023.
\\
In this thesis, we present results from the investigation of two problems, one related to the phase transition of long-range Ising models and the other one associated with the characterization of equilibrium states in quantum spin systems. 
Due to the long-range nature of the interactions, $J|x-y|^{-\alpha}$, estimates using contours usually found in the literature have restrictions on the range of interactions ($\alpha>d+1$ in Ginibre, Grossmann, and Ruelle in 1966 and Park in 1988 for discrete spin systems and possibly non-symmetric situations but with the restrictions $\alpha>3d+1 $). We were able to extend the phase transition argument for long-range Ising-type models to the entire region $\alpha>d$ using the multi-scale arguments presented in the articles by Fr\"ohlich and Spencer.

In quantum statistical mechanics, the KMS condition is used as a characterization for the equilibrium states of the system. Widely studied today, it is known to be equivalent to other equilibrium notions such as satisfying the variational principles. We present another possible characterization of equilibrium states in quantum spin systems by generalizing the DLR equations to the quantum context using Poisson point process representations. We also discuss the relationship of these quantum DLR equations with the KMS states of a subclass of interactions that contains the Ising model with a transverse field. 
\\
\noindent \textbf{Keywords:} phase transition, long-range Ising model, contours, multiscale analysis, Fr\"ohlich-Spencer, classical statistical mechanics, KMS states, groupoid C*-algebras, Poisson point process, quantum statistical mechanics.

\tableofcontents    



\listoffigures            

\mainmatter

\fancyhead[RE,LO]{\thesection}

\singlespacing              

\chapter*{Introduction}
\addcontentsline{toc}{chapter}{Introduction}
\markboth{INTRODUCTION}{}
\epigraph{Slowly, slowly to become hard like a precious stone - and at last to lie there, silent and a joy to eternity.}{\textbf{Friedrich Nietzsche}\\ \textit{The Dawn of the Day}}

Statistical mechanics consists of the study of thermodynamic properties of materials through the statistical analysis of their microscopic behavior. Since its early days, many successes have been achieved in the study of general aspects of systems, with a rigorous theory of variational principles for quite general interactions \cite{Is}, stability properties for equilibrium states \cite{Bra2}. Nonetheless, To gain a deeper understanding of statistical mechanics and make significant contributions to the theory, mathematicians and physicists often focus on studying specific models. These models serve as effective representations of real-world materials or phenomena, allowing researchers to investigate the underlying principles and behaviors of them. One model that stands out is the \emph{Ising model}.

Introduced by Wilhelm Lenz in 1920, he gave the task to his Ph.D. student Ernst Ising of investigating if the model could present a feature known as \emph{phase transition}, giving a potential model for explaining the phenomena known as \emph{ferromagnetism}. At the time, phase transition was characterized by a lack of analyticity on the free energy of the model and Ising, by calculating explicitly the free energy for the model on the integer lattice $\Z$, found that the model did not present a phase transition, i.e., the free energy was an analytic function of its parameters. The negative result obtained by Ising was an inspiration to Heisenberg to put forward his own model for explaining ferromagnetism\footnote{See \cite{BRUSH1967, Niss2004} for more information on the early history of the model.}.

In 1934, Rudolf Peierls in \cite{Pei} gave an argument suggesting that the behavior of the Ising model would be drastically different already in dimension $d=2$. However, the argument was later found to be incomplete. During the 1940s renewed attention was given to the problem mainly due to the new theory of duality by Kramers and Wannier, where they could predict a phase transition and calculate explicitly the temperature where it would occur, and later the rigorous computations of the free energy by Lars Onsager, confirming the picture predicted by Peierls. Many years later, Peierls argument was made rigorous independently by Griffiths \cite{Griffiths1964}\footnote{In this paper, the author uses the Appendix to explain the problem with the original argument.} and Dobrushin \cite{Dobrushin1965}.

Nowadays, the Ising model is the most widely studied model in statistical mechanics, not only through rigorous methods and includes also extensive numerical and theoretical investigations. Although other models are used to explain ferromagnetism in different materials, the Ising model can still be applied as an effective model for many different collective phenomena in applied sciences\footnote{One can read more about other applications \href{https://mathoverflow.net/questions/413767/interesting-and-surprising-applications-of-the-ising-model}{here}.}. The model can be formally described by its Hamiltonian function
\[
H = -\sum_{x,y \in \Z^d}J_{xy}\sigma_x\sigma_y - \sum_{x \in \Z^d}h_x \sigma_x ,
\] 
where $J_{xy}$ are the coupling constants. The model is called \emph{ferromagnetic} when $J_{x,y}\geq 0$. For the nearest neighbor, one takes $J_{x,y}=J$  if $x$ and $y$ as nearest neighbors and $J_{x,y}=0$ otherwise. The terms $h_x$ represent the magnetic field acting in each site of the lattice and in the previous paragraph we were discussing the case where $h_x=0$ for every $x \in \Z^d$. We are considering in the following paragraphs and this thesis ferromagnetic Ising models where $J_{xy} = J|x-y|^{-\alpha}$, i.e. with \emph{polynomial decay}, and $h_x = h^*|x|^{-\delta}, \delta >0$.

The results of Griffiths and Dobrushin were promptly generalized in 1966 by Ginibre, Grossmann, and Ruelle \cite{GGR}. They extended the Peierls argument from the nearest neighbor case, with the usual Peierls contours, to an arbitrary two-body long-range perturbation, as long as it decays at least polynomially $\alpha > d+1$. It was between 1967-1969 that the correlation inequalities by Griffiths \cite{Griffiths1967,Griffiths1969} and Kelly and Sherman \cite{Kelly1968} were proved, known nowadays as Griffiths inequalities or GKS inequalities (see chapter 3 of \cite{Vel}), implying the existence of a phase transition for ferromagnetic long-range interactions whenever the nearest-neighbor Ising model has a phase transition for dimension $d$, including the case for $\alpha > d$, an important region for the decay of long-range deterministic models since it will imply that each point of the lattice can feel at most a finite amount of energy. This extends the results of Ginibre, Grossmann, and Ruelle, but only for the case where the long-range perturbation is \emph{ferromagnetic}.

Negative results were established when one has the nearest-neighbor ferromagnetic Ising model together with a long-range antiferromagnetic interaction by van Enter \cite{vanEnter1981}, which showed that an arbitrarily small, in an appropriate sense, antiferromagnetic long-range perturbation can set the magnetization of the system to 0. The nearest-neighbor ferromagnetic Ising model perturbed by an antiferromagnetic long-range interaction is a model that has seen intense research in the past decades and even its ground state picture is not fully understood yet (see \cite{Bisk} for the positive temperature results and \cite{Fermi2022} and references therein for more information on what is known for the ground states).

For more general discrete state spaces, a generalization of the Peierls argument for possible nonsymmetric systems is available in the scope of the \emph{Pirogov-Sinai theory} for phase transitions \cite{Vel, Pirogov, Sinaibook}. One of the shortcomings of the theory, as presented in \cite{Vel}, is that the interactions must be short-range for it to apply. It was Park in 1988 \cite{Park1, Park2} who extended the Pirogov-Sinai theory to long-range interactions. In his arguments, Park considers the two-body long-range interaction as a perturbation of a short-range interaction, yielding a restriction on the decay of the interaction (at least $\alpha > 3d+1$). The careful reader will discern a recurring pattern in all these findings: the long-range term is consistently handled as a perturbation of an auxiliary short-range model. While this approach yields robust results, it is essential to acknowledge that such a procedure may inherently possess certain limitations.

The situation described above is different in one-dimensional long-range systems. The conclusion of Ising in his Ph.D. thesis can be further generalized and one-dimensional systems can be shown to have no phase transition when the interaction decays faster than $1/r^{2+\varepsilon}$ (see \cite{10.1007/BFb0013371}). Kac and Thompson conjectured in \cite{KT} that a one-dimensional long-range model exhibits a phase transition at low temperatures when  $\alpha \in (1,2]$. The conjecture was proved in 1969 by Freeman Dyson in \cite{Dyson} when $\alpha \in (1,2)$. Dyson used the Griffiths inequalities to compare the long-range one-dimensional Ising model with another one that he introduced, known nowadays as \emph{hierarchical model}, letting just the case $\alpha=2$. In \cite{Dyson}, Dyson even reports a private communication with Thompson, saying that the latter believed that there was no phase transition for $\alpha=2$. The model $\alpha=2$ is special because of its connection with the \emph{Kondo effect}, described Yuval and Anderson in \cite{Anderson1969} and also by the presence of the \emph{Thouless effect}, predicted by Thouless in \cite{Thouless} and proved rigorously by Aizenman, Chayes, Chayes, and Newman in \cite{ACCN}. In 1982, 13 years after Dyson's results, Fr\"{o}hlich and Spencer \cite{Fro1} introduced a notion of one-dimensional contour and proved the phase transition using a Peierls-type argument.  
Their idea to construct such contours came from techniques introduced by the same authors in \cite{Fro3} called \emph{multiscale analysis}. The method for the one-dimensional system consists in organizing the spin flips into contours, which may or may not be connected while ensuring that a specific condition related to their distance from each other is satisfied. In the referenced work \cite{Fro2}, this condition is referred to as \emph{Condition D}.

The introduction of the contours allowed many other questions to be subsequently investigated. For instance, Imbrie \cite{Imbr1} and Imbrie and Newman \cite{Imbr2}  using cluster expansions were able to study the decay of correlations, showing a varying decay exponent of the correlations with respect to the temperature. Regarding the Peierls argument, Cassandro, Ferrari, Merola, and Presutti \cite{Cass} extended the contour argument to different exponents $\alpha\in (1,2)$. They manage to show the phase transition assuming the interaction decay to satisfy $\alpha \in (2-\alpha^+,2]$, where ${\alpha^+ =\log(3)/\log(2) - 1}$. They introduced a more geometric approach to the problem of the phase transition. Unfortunately, their results do not extend to the whole region $(1,2]$ since as shown by Littin and Picco \cite{LP} the quasi-additive bound for the energy of the subtraction of a contour does not hold below some value $2-\alpha^+$. Their argument also needs that the coupling for nearest-neighbors $J(1)$ to be large. This condition seems not optimal since, at least for ferromagnetic systems, the couplings $J_{x,y}$ should favor the alignment of the spins. Therefore, asking for strong nearest-neighbor interactions seems to treat it as a small perturbation of a short-range model, not fully exploiting the ferromagnetic nature of the model. Actually, the condition on the nearest neighbor's coupling $J(1)\gg 1$ was removed by Bissacot, Endo, van Enter, Kimura, and Ruszel \cite{Eric2} but with further restrictions on the decay, now having to satisfy $\alpha >2-\alpha^*$, where $\sum_{n\geq 1}1/n^{\alpha^*}=2$. Nonetheless, their contour argument allowed them to investigate important questions regarding the model such as phase separation \cite{Cass3} and phase transition for the one-dimensional long-range model with a random field \cite{Cass1}.

Our investigation of the phase transition problem by a contour argument in the multidimensional long-range Ising models was motivated by a result in 2018 by Bissacot, Endo, van Enter, Kimura, and Ruszel \cite{Eric2}. The authors, based on the contour argument in \cite{Cass}, considered the model with the presence of the decaying magnetic field $(h_x)_{x\in \Z}$ given by $h_x = h^*|x|^{-\delta}, \delta >0$, for $x\neq 0$. Our first result is related to a generalization of this phase transition result to multidimensional models ($d\geq 2$). Our main classical result is the following 
\begin{theorem*}
		For a fixed $d\ge 2$, suppose that $\alpha>d$ and $\delta>0$.
		There exists $\beta_c\coloneqq \beta_c(\alpha,d)>0$ such that, for every $\beta>\beta_c$, the long-range Ising model with coupling constant  (\ref{long}) and magnetic field (\ref{magfield}) undergoes a phase transition at inverse of temperature $\beta$ when
		\begin{itemize}
			\item $d<\alpha<d+1$ and $\delta >\alpha -d$; $\delta=\alpha-d$ if $h^*$ is small enough;
			\item $\alpha \geq d+1$ and $\delta>1$; $\delta=1$ if $h^*$ is small enough.
		\end{itemize}.
	\end{theorem*}
The Theorem above can be summarized by the following picture

\begin{figure}[ht]
	\centering
	
	\tikzset{every picture/.style={line width=0.75pt}} 
	
	\begin{tikzpicture}[scale=1]
		
		
		\draw[-{Stealth}, black] (-1,0)--(5,0);
		\draw[-{Stealth}, black] (0,-1)--(0,4);
		\draw[-{Stealth}, black] (6,2.6)to[out=180, in=0](5,2);
		
		\draw[-, dashed] (0,0)--(2.5,2)--(5,2);
		\fill[color= black, opacity= 0.2] (0,0)--(2.5,2)--(5,2)--(5,4)--(0,4)-- cycle;	
		
		\fill[white, rounded corners, thick] (1.5,3) rectangle (3.7,3.4);
		\draw[fill=white, draw=black, rounded corners, thick] (1.5,3) rectangle (3.7,3.4);
		\draw[black, rounded corners, thick] (3,1) rectangle (4.6,1.5);
		\draw[black, rounded corners, thick] (6,2.3) rectangle (8.1,2.9);
		
		
		\draw (0,0) node[anchor=north east, scale=1] {$0$};
		\draw (5,0) node[anchor=north east, scale=1] {$\alpha -d$};
		\draw (0,4) node[anchor=north east, scale=1] {$\delta$};
		\draw (0.1,2)--(-0.1,2);
		\draw (0,2) node[anchor=north east, scale=1] {$1$};
		\draw (2.5,0.1)--(2.5,-0.1);
		\draw (2.5,0) node[anchor=north west, scale=1] {$1$};
		\draw(1.45,3.05) node[anchor=south west, scale=0.7] {\;\footnotesize\textbf{Phase Transition}};
		\draw(2.9,1.05) node[anchor=south west, scale=0.7] {\;\footnotesize\textbf{Uniqueness?}};
		\draw(5.9,2.6) node[anchor=south west, scale=0.7] {\;\footnotesize\textbf{Phase Transition}};
		\draw(6.3,2.3) node[anchor=south west, scale=0.7] {\;\footnotesize\textbf{for small $h^*$}};

	\end{tikzpicture}
	\caption{The phase diagram for the long-range Ising model at low temperatures depends on $\alpha$ and $\delta$.}
\end{figure}
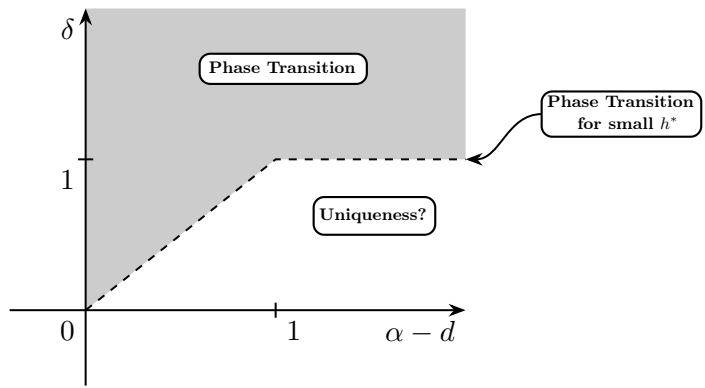

In quantum statistical mechanics, the {\it KMS condition} was used by Kubo \cite{Kubo} and Martin and Schwinger \cite{MS} as a technique to study multiparticle systems and it is related to a certain analytic continuation property of the thermodynamic Green functions for equilibrium systems. The KMS condition was proposed as a rigorous characterization of equilibrium in a seminal paper by Haag, Hugenholtz and Winnink \cite{HHW} that noticed that if one could make sense of the thermodynamic limit of the dynamics, although the state could not be of the usual Gibbs form, they would still satisfy the KMS condition. Restricting to classical interactions, Ruelle-Lanford \cite{Lan} and Dobrushin \cite{Dob1} proposed a set of equations, nowadays called the {\it DLR equations}, as a characterization of equilibrium for states in classical statistical mechanics, and it consists of describing the conditional expectations of the state concerning the $\sigma$-algebra of the events localized outside a finite region $\Lambda$ of $\Z^d$. An intriguing historical fact is that the KMS condition actually came first even as a proposal for a characterization of equilibrium, in 1967, while Dobrushin and Lanford and Ruelle's papers are from 1968 and 1969 respectively.

The DLR equations and the KMS condition are different in nature. For the formulation of the KMS condition one needs a well-defined dynamics on the observable algebra, while the DLR equations need a notion of conditional expectations for the states and are basically static. Some critical no-go theorems proved after their proposal put further constraints on the understanding of the relationship between these two characterizations for equilibrium and we proceed to explain them briefly.

 One could try to use a generalization to C*-algebras of conditionals expectations as in \cite{Kad}. In this case, a conditional expectation between two unital C*-algebras $\mathfrak{B}\subset\mathfrak{A}$ is a map 
 $E:\mathfrak{A}\rightarrow \mathfrak{B}$ that is completely positive linear, $E(\mathbbm{1})=\mathbbm{1}$ and for each $b_1,b_2 \in \mathfrak{B}$ it holds $E(b_1 ab_2) = b_1 E(a)b_2$ for all $a \in \mathfrak{A}$. But a theorem proved by Takesaki (see \cite{Acc1} for proof and a generalization) has as one of its consequences that if $\mathfrak{A}$ is the quasilocal algebra of observables and $\mathfrak{A}_{\Lambda^c}$ is the subalgebra of observables localized outside $\Lambda$, then if for state $\mu:\mathfrak{A}\rightarrow \mathbb{C}$ one has a conditional expectation $E:\mathfrak{A}\rightarrow \mathfrak{A}_{\Lambda^c}$ such that $\mu \circ E = \mu$, then the state is actually a product state. This means that the subsystem can have no interaction with the outside system. There is a corresponding generalized conditional expectation \cite{Acc1}, where the usual non-commutative conditional expectations in \cite{Kad} are a subclass. Accardi then used them to introduce a notion of a Quantum Markov Field (see \cite{Acc2} and references therein). 
 
 Besides the restriction on the existence of conditional expectations, there is also a restriction on the dynamics for classical systems. Let a C*-algebra $\mathfrak{A}$ be equipped with a strongly continuous one-parameter group of *-automorphisms $\tau_t:\mathfrak{A}\rightarrow \mathfrak{A}$, and a KMS state $\mu$ for the corresponding dynamics. Let also $\mathfrak{A}_{\tau}$ the subalgebra of $\mathfrak{A}$ of invariant observables by dynamics $\tau_t$. A well-known result in operator algebras (Proposition 5.3.28 in \cite{Bra2}) says that $\mathfrak{A}_\tau$ is equal to the centralizer subalgebra, i.e., the subalgebra of operators $a \in \mathfrak{A}$ such that $\mu([a,b])=0$, where $[a,b]=ab-ba$, for every $b \in \mathfrak{A}$. This theorem implies that if your algebra of observables is commutative, then every element must be invariant by the dynamics\footnote{In classical \emph{continuum} statistical mechanics or in lattice systems where the state space is a symplectic manifold, 
a notion of classical KMS can be given using the Poisson bracket (see, for instance, \cite{Aiz4, Drago} and references therein). The notion of a Poisson bracket structure is absent in finite-spin lattice systems}. One way to circumvent this problem is to embed the algebra of continuous functions into a larger noncommutative C*-algebra. This is one of the main ideas of Brascamp in \cite{Bras}, where the author shows that for Ising spin systems, the states that satisfy the DLR equations are exactly the KMS states for classical interactions. For general finite spin systems, the relationship between the DLR equations and KMS conditions was clarified by Araki and Ion \cite{Ara1}, which studied the problem for one-dimensional systems and the high-temperature case for all dimensions, and Araki \cite{Ara2}, which solved the problem completely for quantum spin systems. To this end, Araki and Ion introduced what we call here the \emph{Gibbs-Araki-Ion condition}, which is equivalent to the KMS condition (see Theorem 6.2.18 from \cite{Bra2}) and reduces to the DLR equations if the interaction is classical.\footnote{See Chapter \ref{ch:quantstatmech} for more details and Theorem \ref{t2} for the proof.}

 The definition of the Gibbs-Araki-Ion condition relies on the perturbation theory for bounded operators developed by Araki (see Chapter 5.4 of \cite{Bra2} for a detailed account), posing some difficulties that are absent in the classical statistical mechanics setting. These difficulties were best explained by Matsui in \cite{Matsui}, and we quote here

\begin{quotation}
	"One mathematically interesting question is whether any KMS state is obtained in this procedure; namely, one may ask whether any KMS state is a thermodynamic limit of finite-volume Gibbs states with suitable boundary conditions for Hamiltonians as is described here. Theorem 3.3 may be taken as an answer to this question; however, this is not what we want. We are asking the effect of \emph{the boundary condition of our Hamiltonian} in a large system while the Gibbs condition is \emph{the boundary condition imposed on the states}. From a practical point of view, the [Gibbs-Araki-Ion condition] is cumbersome to handle because in it the modular automorphism group is used which is state-dependent and the imaginary time evolution is also difficult to compute."
\end{quotation}

Above, Matsui was asking if there was a family of boundary conditions in the form of operators $B_{\partial\Lambda}$, localized near the boundary $\partial \Lambda$ such that the thermodynamic limits of the perturbed Hamiltonians $H_\Lambda+B_{\partial\Lambda}$ would be able to describe all the KMS states. Indeed, if the interaction is classical this is already true for a very specific type of $B_{\partial\Lambda}$ (see Theorem 7.12 in \cite{Geo}), but only for pure phases since it is already known that there are examples of \emph{non-extremal} DLR measures which cannot be obtained via a thermodynamic limit procedure, see \cite{Co, Mi}. There are even proposals by Israel \cite{Is} and Simon \cite{Simon}, where they condition the Hamiltonian to a state of the C*-algebra outside the box (more on this in Chapter \ref{ch:groupoids}).
In \cite{Wer}, M. Fannes and R. F. Werner raised concerns that maybe the use of boundary conditions as proposed earlier in this paragraph would not generate all possible extremal KMS states for spin systems, a phenomenon that they called the \emph{failure of the DLR inclusion}. 

Despite the seemingly negative results of Fannes and Werner, adding boundary condition terms is an important procedure to generate KMS states. For instance, Datta, Fernand\'{e}z, and Fr\"{o}hlich \cite{Datta1}, and also Borgs, Koteck\'{y}, and Ueltschi \cite{BKU} extended the Pirogov-Sinai theory for quantum lattice systems. In their work, the models are treated as a small perturbation of a classical Hamiltonian. The authors even construct infinite volume states depending on the classical boundary conditions, coming from ground states of the classical Hamiltonian. Even for ground states, a recent result by Cha, Naaijkens, and Nachtergaele \cite{Pieter} have characterized the ground states of a class of Kitaev's quantum double models using a suitable notion of boundary operators. Nevertheless, the objection made by Fannes and Werner cannot be completely rejected by these results just cited; they \emph{do not say that boundary conditions would not produce new states}, what they mean is that this procedure \emph{should not generate all the extremal KMS states for all interactions}.

Compelling evidence supporting the development of a DLR theory for quantum spin systems comes from a number of results that we describe now. First, there are the results from Fichtner and Freudenberg for bosonic systems \cite{Fich1, Fich2}, which are closely related to ours in spirit since they use random representations. Fichtner and Freudenberg managed to construct a family of conditional expectations using point processes, which they called reduced density matrices, in order to describe locally normal states of bosonic states. Since the construction of infinite-volume dynamics for general bosonic interactions is still not settled, they could not relate their results to KMS states.

There is also the book by Albeverio, Kondratiev, Kozitsky, and R\"{o}ckner (see \cite{AKKR} and references therein) for a DLR approach to anharmonic crystals. These systems also do not have, as far as we know, a corresponding KMS theory, so the approach consists of studying the thermodynamic Green functions directly. Their basic strategy is to use Feynman-Kac representations, a rigorous version of the path integral, to construct the Euclidean Gibbs measures. 

For general systems, Klein and Landau in \cite{Klein} unveiled some deep connections between some stochastic processes and a specific class of KMS states. Klein and Landau also have a probabilistic interpretation of the Tomita-Takesaki theory, perturbation theory, and also study systems where the equilibrium states can be described using density matrices. For quantum spin systems, there are random representations using Poisson point processes \cite{Aiz1, Aiz2}. Aizenman and Nachtergaele proposed the notion of \emph{quasi-states}, a linear functional that has a positive restriction to an abelian sub-C*-algebra. To our knowledge, no relationship with KMS states has been further investigated using the quasi-state notion, but many important results were derived using these random representations, such as the continuity of the magnetization and the sharpness of the phase transition \cite{Bjorn1, Bjorn2}. The interested reader can check \cite{Ueltschi} (and references therein) for further details.

At last, recent successes in characterizing KMS states in groupoid C*-algebras \cite{Neshveyev2013, Thomsen2016, BEFR} arising from dynamical systems theory raised the question if quantum statistical mechanics could not profit from groupoid methods. When studying the monograph by Gruber, Hintermann, and Merlini \cite{Gruber}, I noticed that the transformation group studied there had the usual spin algebra as its C*-algebra, giving us a candidate for "quantum space" to act as a substitute for the configuration space. Combining random representations and the groupoid model for the C*-algebra of the quantum spin system, we could find a suitable generalization of the DLR equations for the quantum setting. The theory can be developed into a subclass of interactions that we called \emph{admissible}. An important subclass of admissible interactions is the ones called \emph{stoquastic} \cite{Klassen}. These interactions have the important property that the exponentials of their Hamiltonians, with the convolution product, are pointwise positive functions in the groupoid. Since quantum spin systems have a well-developed KMS theory, we could relate it to our DLR theory for a specific class of models including the transverse field Ising model, XY model, and the Toric Code. The results are summarized in Figure \ref{fig:admissible}. 

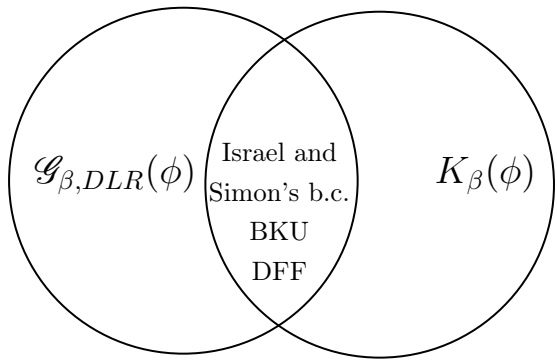
\begin{figure}[ht]
    \centering

\tikzset{every picture/.style={line width=0.75pt}} 

\begin{tikzpicture}[x=0.75pt,y=0.75pt,yscale=-1,xscale=1]

\draw   (140,157) .. controls (140,108.95) and (178.95,70) .. (227,70) .. controls (275.05,70) and (314,108.95) .. (314,157) .. controls (314,205.05) and (275.05,244) .. (227,244) .. controls (178.95,244) and (140,205.05) .. (140,157) -- cycle ;
\draw   (238,159) .. controls (238,110.95) and (276.95,72) .. (325,72) .. controls (373.05,72) and (412,110.95) .. (412,159) .. controls (412,207.05) and (373.05,246) .. (325,246) .. controls (276.95,246) and (238,207.05) .. (238,159) -- cycle ;

\draw (235,142.4) node [anchor=north east][inner sep=0.75pt]    {\Large $\mathscr{G}_{\beta,DLR}(\phi) $};
\draw (351,142.4) node [anchor=north west][inner sep=0.75pt]    {\Large$K_\beta(\phi)$};
\draw (275,142.4) node [inner sep=0.75pt]    {\small Israel and };
\draw (275,162.4) node [inner sep=0.75pt]    {\small Simon's b.c.};
\draw (275,182.4) node [inner sep=0.75pt]    {\small BKU};
\draw (275,202.4) node [inner sep=0.75pt]    {\small DFF};

\end{tikzpicture}
    \caption{The set of quantum DLR states and KMS states have a non-empty intersection for admissible interactions. At the intersection, there are proposals for boundary conditions from Israel and Simon \cite{Is,Simon}. The acronyms BKU and DFF stand for, respectively, the boundary conditions present in the papers by Borgs, Koteck\'y and Ueltschi \cite{BKU} and Datta, Fernand\'ez and Fr\"ohlich \cite{Datta1} for quantum spin systems.}
    \label{fig:admissible}
\end{figure}

\pagebreak

\subsubsection{Chapter \ref{ch:classtatmech}}

We recall some basic definitions and facts about Gibbs measures as found in \cite{Vel, Geo}, with emphasis on the DLR approach. We also discuss the constructions of the contours, closely related to the ones appearing in Pirogov-Sinai theory, and rigorously prove phase-transition for the nearest-neighbor Ising model with a decaying field $h_{x} = h^*|x|^{-\delta}$ when $\delta>1$ and $h^*> 0$. This result appeared first in \cite{Bis2}, and the main ideas for the generalization to the long-range case can be seen in this proof.

\subsubsection{Chapter \ref{ch:longrange}}
 We explain the results in \cite{Aff}. We start by giving a heuristic argument à la Imry-Ma for the relation between the parameters of the model. We will introduce the concept of $(M,a,r)$-partition, a construction that uses ideas from the multiscale analysis, as in the papers of Fr\"ohlich and Spencer \cite{Fro1, Fro2, Fro3}. Finally, we prove 
a quasi-additive lower bound for the energy of the system when we erase a contour. We stress that this bound holds for all the region $\alpha>d$. We end the chapter with a proof by a contour argument of the phase transition for the long-range Ising model with a decaying field.

\subsubsection{Chapter \ref{ch:quantstatmech}}
We present the basics of quantum statistical mechanics which is relevant to the study of statistical mechanics of finite quantum spin systems à la Bratteli-Robinson using the language of transformation groupoids and their C* algebras. We also introduce in the groupoid setting the $d$-dimensional Jordan-Wigner transform, inspired by \cite{kochmanski1998jordanwigner}. We give some details about the construction of the dynamics, introduce the KMS condition, the Gibbs-Araki-Ion condition, and show that they are equivalent. In the end, we show that the KMS states of classical interactions satisfy the DLR equations of Chapter \ref{ch:classtatmech}. We follow closely the expositions of \cite{Bra2, putnam, Renault1980, SimsSzaboWilliams2020}.
\subsubsection{Chapter \ref{ch:groupoids}}
We derive a random representation for the Gibbs density operator using a Poisson Point process on the groupoid. After that, we use this random representation to introduce the finite volume functionals with boundary conditions and the quantum Gibbsian specification. We prove that the closed convex hull of these functionals is exactly the set of all quantum DLR states. 

\subsubsection{Chapters \ref{ch:conclusion2} and \ref{ch:conclusion}}
We make some comments about the results and also discuss some directions for further research.

Chapters 1 and 4 are essentially reviews of the basic theory of statistical mechanics. The new results are contained in Chapters 2 and 5.

\chapter{Classical Statistical Mechanics}
\label{ch:classtatmech}
\epigraph{For the sake of persons of different types, scientific truth should be presented in different forms and should be regarded as equally scientific, whether it appears in the robust form and the vivid coloring of a physical illustration, or in the tenuity and paleness of a symbolical expression}{\textbf{J. C. Maxwell} \\ \textit{As quoted by R. B. Lindsay in "On the Relation of Mathematics and Physics"}}

In classical statistical mechanics, one way to characterize equilibrium states was introduced by Dobrushin \cite{Dob} and Lanford and Ruelle \cite{Lan} and nowadays is called the \emph{DLR equations}. In this chapter, we will introduce the basic formalism related to discrete spin systems. Further details can be found in \cite{Geo, Vel}.

\section{The DLR equations and the Ising model}

    Consider a system defined on the lattice $\Z^d$ with a discrete set of possible values for its spins. The distance between two points $x,y\in\Z^d$ will be given by the $\ell^1$-norm and we will write it as $|x-y|$. Given this, the diameter of a set $\Lambda \subset \Z^d$ is defined as being 
    \[
        \diam(\Lambda) = \sup_{x,y\in \Lambda}|x-y|.
    \]
    We will also write $|\Lambda|$ for the number of points of $\Lambda$. We will write $\Lambda \Subset \Z^d$ to denote the fact that $\Lambda$ is a subset of the lattice and is finite, and the set of all finite subsets of $\Z^d$ is $\mathcal{F}(\Z^d)$. For each subset $\Lambda \Subset \Z^d$, we call $\Lambda^{(0)}$ the unique unbounded connected component of $\Lambda^c$ and use $\Lambda^{(k)}$ for connected components of $\Lambda$, $k\geq 1$. Then, we define the \emph{volume} by $V(\Lambda)= (\Lambda^{(0)})^c$. Note that the set $V(\Lambda)$ is a union of simply connected sets that contains $\Lambda$, and it is the smallest one in the partial order given by the inclusion. The \emph{interior} is defined by $\I(\Lambda)= \Lambda^c \setminus \Lambda^{(0)}$. 
    
     For each $\Lambda \subset \Z^d$, local configuration spaces $\Omega_\Lambda \coloneqq E^{\Lambda}$,  where $E$ is a discrete set, with the product topology. Let $\Omega = E^{\Z^d}$ the {\it configuration space}. This choice of topology makes the configuration space $\Omega$ a compact metrizable space. We can define for each $\Lambda \subset \Z^d$, the projections $\pi_\Lambda:\Omega \rightarrow \Omega_\Lambda$, defined as, given $\sigma\in \Omega$, being the unique configuration in $\Omega_\Lambda$ such that the values coincide with the values of $\sigma$ inside $\Lambda$. We will write this as $\pi_\Lambda(\sigma)\coloneqq \sigma_\Lambda$. Let also $C(\Omega_\Lambda)$ be the set of all continuous functions $f:\Omega_\Lambda\rightarrow \mathbb{C}$. An important definition will be that of an \emph{interaction}.
    \begin{definition}
        An interaction $\phi$ is a function $\phi:\mathcal{F}(\Z^d)\rightarrow C(\Omega)$ that satisfies
        \begin{itemize}
            \item [(i)] For each $\Lambda \Subset \Z^d$, the function $\phi(\Lambda)\coloneqq \phi_\Lambda$ is real-valued.
            \item [(ii)] The function is local, i.e., for each $\sigma,\omega \in \Omega$ it holds
            \[
            \text{ if } \omega_{\Lambda} = \sigma_{\Lambda} \Rightarrow \phi_\Lambda(\sigma) = \phi_\Lambda(\omega).
            \]
            Moreover, an interaction is said to be \textbf{short-range} if there is a $R>0$ such that whenever $X$ has $\diam(X)>R$ then $\phi_X = 0$. Otherwise, the interaction will be called \emph{long-range}. 
        \end{itemize}
    \end{definition}
     Given an interaction $\phi$, a configuration $\omega \in \Omega$, and $\Lambda \in \Subset \Z^d$, the local Hamiltonian with is a function $H_\Lambda^\omega(\phi):\Omega_\Lambda \rightarrow \mathbb{R}$ is
\[
H_\Lambda^\omega(\phi)(\sigma_\Lambda) = \sum_{X\cap \Lambda\neq \emptyset} \phi_X(\sigma_\Lambda\omega_{\Lambda^c}),
\]
where $\sigma_\Lambda\omega_{\Lambda^c}$ is the \emph{spin concatenation} of two configurations defined by
\[
(\sigma_\Lambda\omega_{\Lambda^c})_x = \begin{cases}
    \sigma_x & x \in \Lambda \\
    \omega_x & x \in \Lambda^c.
\end{cases}
\]
The configuration $\omega$ is referred to as the \emph{boundary condition}. We will also write $\Omega_\Lambda^\omega$ to denote the subset of $\Omega$ where every configuration is of the form $\sigma_\Lambda\omega_{\Lambda^c}$ for some $\sigma_\Lambda \in \Omega_\Lambda$. The Hamiltonian $H_\Lambda^\omega(\phi)$ is a local function that depends only on the spins located at the set $\Lambda$, with the difference that the system can be influenced by the configuration $\omega_{\Lambda^c}$ on the outside. 
An important class of interactions with two-body interactions\footnote{$|X|>2$ implies $\phi_X=0$.} is when the state space $E=\{-1,+1\}$ is the \emph{Ising models}. They can be defined directly by their Hamiltonians $H^\omega_{\Lambda,\textbf{h}}:\Omega^{\omega}_\Lambda \rightarrow \mathbb{R}$, for each $\Lambda \in \mathcal{F}(\Z^d)$ and $\omega \in \Omega$, by
	\begin{equation}\label{Isingsys}
		H_{\Lambda,\textbf{h}}^\omega(\sigma_\Lambda) = -\sum_{\{x,y\} \subset \Lambda} J_{xy} \sigma_x\sigma_y - \sum_{\substack{x \in \Lambda \\ y \in \Lambda^c}}J_{xy} \sigma_x \omega_y - \sum_{x \in \Lambda} h_x \sigma_x.
	\end{equation}
	Important examples are the nearest neighbor Ising model and the long-range Ising model, defined through its coupling constants respectively
	\begin{equation}\label{long}
J_{xy} = \begin{cases}
			J &\text{ if } |x -y|=1 \\
			0 & \text{otherwise}
		\end{cases}
	\quad \text{ and } \quad		 
  J_{xy} = \begin{cases}
			\frac{J}{|x-y|^\alpha} &\text{ if } x\neq y \\
			0 & \text{otherwise}.
		\end{cases}
	\end{equation}
 For the long-range Ising model, we assume that $\alpha>d$(see Remark \ref{remark_abssum}). The case $J>0$ is called the ferromagnetic Ising model, and $J<0$ is the antiferromagnetic Ising model. In \cite{Bis1}, Bissacot and Cioletti introduced a modification of the nearest neighbor Ising model by a spatially dependent magnetic field,
	\begin{equation}\label{magfield}
		h_x = \begin{cases}
			\frac{h^*}{|x|^\delta}& x \neq 0 \\
			h^*& x=0,
		\end{cases}
	\end{equation}
	where $h^*,\delta$ are positive constants. This model is not translation invariant and yet presents very interesting properties. The first of them is that the pressure of the model is equal to the pressure of the translation invariant model, i.e.,
	\[
	\lim_{n\rightarrow \infty}\frac{\log Z^\omega_{\beta,\Lambda_n}}{|\Lambda_n|} = p_{\beta},
	\]
	where $p_\beta$ is the pressure of the model with $h^*=0$ and $\Lambda_n$ is a sequence of finite subsets invading the lattice $\Z^d$. This result is a strong argument in favor of the belief that the Ising model with a decaying magnetic field should present the same thermodynamic behavior as the Ising model with a zero magnetic field. However, Bissacot, Cassandro, Cioletti, and Presutti in \cite{Bis2} showed it to be false, for surface terms are relevant in the analysis of the Gibbs states. These states are probability measures on the configuration space $\Omega$ that captures the properties of the state of the macroscopic system. Let us proceed with their definition. For each continuous function $f\in C(\Omega)$ we define
\be\label{finiteclassical}
\mu_{\beta,\phi,\Lambda}^\omega(f) = \frac{1}{Z_{\beta,\phi,\Lambda}^{\omega}}\sum_{\sigma_\Lambda \in \Omega_\Lambda}f(\sigma_\Lambda\omega_{\Lambda^c})e^{-\beta H_\Lambda^\omega(\phi)(\sigma_\Lambda)},
\ee
where the normalization is given by the partition function 
\[
Z_{\beta,\phi,\Lambda}^\omega = \sum_{\sigma_\Lambda \in \Omega_\Lambda} e^{-\beta H_\Lambda^\omega(\phi)(\sigma_\Lambda)}.
\]
It is not hard to see that Equation \eqref{finiteclassical} defines positive linear functionals in $C(\Omega)$ and, due to the Riesz-Markov theorem, are in one-to-one correspondence with probability measures in $\Omega$, with the Borel $\sigma$-algebra. For this reason, we will refer to $\mu_{\beta,\phi,\Lambda}^\omega$ as \emph{finite volume Gibbs measure}. Since we can have any $\Lambda \Subset \Z^d$ and $\omega \in \Omega$ in Equation \eqref{finiteclassical}, we actually have a family of finite volume Gibbs measures that satisfies the following properties (see \cite{Geo, Vel})
\begin{itemize}
    \item[(i)] \textbf{(Consistency Condition)} For any $\Delta \subset \Lambda$ and continuous function $f\in C(\Omega)$ we have
    \[
    \mu_{\beta,\phi,\Lambda}^{\omega}(\mu_{\beta,\phi,\Delta}^{(\cdot)}(f))=\mu_{\beta,\phi,\Lambda}^{\omega}(f).
    \]
    \item[(ii)] \textbf{(Proper)} For any continuous functions $f_1,f_2\in C(\Omega)$ such that $f_2 \in C(\Omega_{\Lambda^c})$ then
    \[
    \mu_{\beta,\phi,\Lambda}^\omega(f_1 f_2) = \mu_{\beta,\phi,\Lambda}^\omega(f_1) f_2(\omega).
    \]
    \item[(iii)] \textbf{(Feller Continuity)} For each $\Lambda \Subset \Z^d$ and continuous function $f\in C(\Omega)$, the function $\omega \mapsto \mu_{\beta,\phi,\Lambda}^\omega(f)$ is continuous.
\end{itemize}
\vspace{0.1cm}

The family $\{\mu_{\beta,\phi,\Lambda}^{(\cdot)}\}_{\Lambda\in\mathcal{F}(\Z^d)}$ of all finite volume Gibbs measures satisfying the conditions above is called a \emph{Gibbsian specification}. 
\begin{remark}\label{remark_abssum}
  The definition of finite volume Gibbs states and the Gibbsian specification theory holds with much more generality. Indeed, one just needs the interaction $\phi$ to satisfy
    \[
    \|\phi\| \coloneqq \sup_{x\in \Z^d} \sum_{X \ni x}\|\phi_X\| <\infty,
    \]
    where $\|\phi_X\|$ is the supremum norm on $C(\Omega)$. The interactions satisfying this condition are called \textbf{absolutely summable}. Notice that for the long-range Ising model, the finiteness of the norm $\|\phi\|$ imposes the restriction $\alpha>d$.
\end{remark}
\begin{remark}
    A theory of general specifications, defined only using measurable spaces, is also available and is presented, for instance, in \cite{Geo, LeNy}, and not everything is Gibbsian \cite{vanEnter1993}. In these more general situations, such as in the one treated on \cite{Geo}, one does not require the Feller continuity property to hold on the definition of a Gibbsian specification. We do it here for two reasons. The first one is that, for finite state spaces, the Feller property equivalent to the more general notion of \emph{quasilocality} that always holds for absolutely summable interactions (see \cite{vanEnter1993} and lemma 6.28 in \cite{Vel} for a proof). The second reason is that in the C*-algebra setting, which will be developed later in Chapters \ref{ch:quantstatmech} and \ref{ch:groupoids}, it is natural to consider continuous functions, so we chose to add it to the definition of Gibbsian specification. 
\end{remark}
\begin{definition}
A state $\mu$ on $C(\Omega)$ is said to satisfy the \textbf{DLR equations} if, for any continuous function $f\in C(\Omega)$ if, for any $\Lambda \subset \Z^d$ finite we have
\[
\mu(f) = \mu(\mu_{\beta,\phi,\Lambda}^{(\cdot)}(f)).
\]
The set of all the $DLR$ states $\mu$ for an interaction $\phi$ is denoted by $\mathscr{G}_\beta^{DLR}(\phi)$.
\end{definition}
A Gibbsian specification can be understood as a way to prescribe conditional expectations and the DLR equations are just the usual invariance property that conditional expectations satisfy by construction. The DLR theory is much more general and works for systems with not only discrete spins but also, for example, $\mathbb{S}^1,\bbn$ or $\bbr$, or even when the lattice $\Z^d$ is replaced by more general graphs or $\bbr^d$ (For more details see \cite{Geo, LeNy, Jan, Vel}). 

Boundary conditions can also be used to define equilibrium states using the notion of the \emph{thermodynamic limit}. We say that a sequence of boxes $\Lambda_n$ is invading $\Z^d$ if for every $\Delta \Subset \Z^d$ there is $N\coloneqq N(\Delta)$ such that for any $n>N$ we have $\Delta \subset \Lambda_n$. For a sequence $\Lambda_n$ invading $\Z^d$ and $\omega$ a boundary condition, we can define 
\[
\mu_{\beta,\phi}^\omega = w^*-\lim_{n\rightarrow \infty} \mu_{\beta,\phi,\Lambda_n}^\omega.
\]
whenever the $w^*$-limit above exists, i.e., the limit $\mu_{\beta,\phi,\Lambda_n}^\omega(f)$ exists and converge to $\mu_{\beta,\phi}^\omega(f)$ for every continuous function $f$. The weak* limit will exist at least for one sequence since the space of states on $C(\Omega)$ is compact in the weak* topology, by the Banach-Alaoglu theorem. The set of \emph{Gibbs measures} is defined as 
\[
\mathscr{G}_{\beta}(\phi) = \overline{\text{co}}\{\mu_{\beta,\phi}^\omega: \mu_{\beta,\phi}^\omega = w^*-\lim_{n\rightarrow \infty} \mu_{\beta,\phi,\Lambda_n}^\omega\},
\]

where $\overline{\text{co}}$ is the closed convex hull. We say that the model has \emph{uniqueness at $\beta$} if $|\mathscr{G}_\beta|= 1$ and it undergoes to a \emph{phase transition at $\beta$} if $|\mathscr{G}_\beta|> 1$. An important result is that the set of all Gibbs measures is equivalent to the DLR equations. 

\begin{proposition}\label{classical_gibbs=dlr}
    It holds that $\mathscr{G}_\beta(\phi) = \mathscr{G}_\beta^{DLR}(\phi)$
\end{proposition}
\begin{proof}
    We follow Theorem III.2.6 of \cite{Simon}. Let $X$ be a Banach space and $\{C_n\}_{n\geq 1}$ be a sequence of sets of continuous linear functionals $\mu:X\rightarrow \mathbb{C}$ contained in some weak* compact set $B$. Define
    \[
    L(\{C_n\}_{n\geq 1}) = \{\mu: \exists \mu_{n_k} \in C_{n_k}, k\geq 1 \;\; \text{s.t. }\;\; \mu = w^*-\lim_{k\geq\infty}\mu_{n_k}\},
    \]
    the set of all weak* limit points of the sequence $\{C_n\}_{n\geq 1}$. We claim that 
    \be\label{gibbs=dlr_classical_eq1}
    L(\{\overline{\text{co}}(C_n)\}_{n\geq 1}) \subset \overline{\text{co}}(L(\{C_n\}_{n\geq 1})).
    \ee
     Suppose that there is a linear functional  $$\mu \in  L(\{\overline{\text{co}}(C_n)\}_{n\geq 1}) \setminus \overline{\text{co}}(L(\{C_n\}_{n\geq 1})).$$ Since the sequence $\{C_n\}_{n\geq 1}$ is contained in a weak* compact set, the set $\overline{\text{co}}(L(\{C_n\}_{n\geq 1}))$ is weak-* compact also, by the Krein-Smulian theorem (see Theorem 13.4 in \cite{Conway2007}). Hence we can apply the geometric form of the Hahn-Banach theorem and the fact that all the continuous linear functionals in the weak*-topology are the evaluation functionals \footnote{See Proposition 3.14 in \cite{Brezis2011}.}, there must exist $x \in X$ and constants $a$ and $b$ such that $$\mathrm{Re}(\mu(x)) \leq a < b \leq \mathrm{Re}(\nu(x)),$$ for all $\nu \in \overline{\text{co}}(L(\{C_n\}_{n\geq 1}))$. Let $b_n = \inf_{\nu \in C_n}\mathrm{Re}(\nu(x))$. The real numbers $b_n$ are finite since $C_n$ is contained in a weak*-compact set and $\nu \mapsto \mathrm{Re}(\nu(x))$ is a continuous function. Since the linear functional $\mu \in L(\{\overline{\mathrm{co}}(C_n)\}_{n\geq 1})$, there is a sequence $\mu_{n_k} \in \overline{\mathrm{co}}(C_{n_k})$ converging in the weak* topology to $\mu$. For each $n_k$, there is a sequence of positive real numbers $\lambda_{k,m}$, $m=1,\dots,\ell_k$, whose sum is unity and linear functionals $\nu_{k,m} \in C_{n_k}$ such that
    \[
        w^*-\lim \sum_{m=1}^{\ell_k}\nu_{k,m}=\mu_{n_k}.
    \]
    Since each $\nu_{k,m}$ satisfies $\mathrm{Re}(\nu_{k,m}(x))\geq b_{n_k}$ by the definition of infimum and the real part is an $\mathbb{R}$-linear functional, we get that $\mathrm{Re}(\mu_{n_k}(x))\geq b_{n_k}$, therefore $\mathrm{Re}(\mu(x))\geq \limsup b_n \geq b$. But we had that $\mathrm{Re}(\mu(x))\leq a < b$, yielding us a contradiction. Returning to the Gibbs measures, notice that the consistency condition 
    \[
    \mu_{\beta,\phi,\Lambda}^{\omega}(\mu_{\beta,\phi,\Delta}^{(\cdot)}(f))=\mu_{\beta,\phi,\Lambda}^{\omega}(f),  
    \]
    imply that $\mathscr{G}_{\beta}(\phi)\subset\mathscr{G}_{\beta,DLR}(\phi)$.
    Consider $\Lambda_n$ a sequence of finite sets invading the lattice, and define the sets
    \[
    C_n \coloneqq \{\mu_{\beta,\Lambda_n}^\omega: \omega \in \Omega\}.
    \]
    These are contained in a weak* compact set, namely, the unit ball. The DLR equations imply that each $\mu \in \mathscr{G}_{\beta,DLR}(\phi)$ is in $\overline{\mathrm{co}}(C_n)$, for every $n$. Then by \eqref{gibbs=dlr_classical_eq1} we get that $\mathscr{G}_{\beta,DLR}(\phi) \subset \mathscr{G}_{\beta}(\phi)$. 
\end{proof}

Important Gibbs measures for ferromagnetic Ising models are the ones with $+$ (resp. $-$) boundary condition, i.e., $\omega$ satisfies $\omega_x=+1$ (resp. $-1$) for every $x \in \Z^d$. The weak* limits of $\mu_{\beta,\Lambda_n}^+$ and $\mu_{\beta,\Lambda_n}^-$ can be shown to exist using the FKG inequality (see Chapter 3 of \cite{Vel} for details). For the Gibbs measures for the system with a decaying field we add a subscript $\textbf{h}$ to the measure, so for instance $\mu_{\beta,\textbf{h},\Lambda}^+$ is the finite volume Gibbs measure with plus boundary condition and decaying field. Notice that the addition of a decaying field makes the system lose spin-flip symmetry, so we will explain what is the strategy to show phase transition. It consists in showing that $\mu^+_{\beta,\textbf{h}}\neq \mu^-_{\beta,\textbf{h}}$, so we just need to find a special observable $f\in C(\Omega)$ such that $\mu^+_{\beta,\textbf{h}}(f)\neq \mu^-_{\beta,\textbf{h}}(f)$. Again, since we know that the finite volume Gibbs measures with the plus and minus boundary conditions converge, it is only necessary to find an observable $f$ and show that the sequences $\mu_{\beta,\widehat{h},\Lambda}^+(f)$ and $\mu_{\beta,\widehat{h},\Lambda}^-(f)$ converge to different values. For ferromagnetic systems, the most natural candidate is the observable $f(\sigma) = \sigma_0$, which we simply denote by $\sigma_0$. First, by the FKG inequality, we know that
\[
\mu_{\beta,\textbf{h},\Lambda}^+(\sigma_0)\geq\mu_{\beta,\textbf{h},\Lambda}^+(\sigma_0).
\]
The inequality above will hold also for the limit, therefore if the measure with zero field has $\mu_{\beta}^+(\sigma_0)>0$, so it will be the sign of the magnetization for the system with a decaying field. The same argument does not hold for the minus boundary condition, so our strategy will be to show that $\mu_{\beta,\textbf{h}}^-(\sigma_0)$ will still be negative. This will follow once we show that $$\mu_{\beta,\textbf{h},\Lambda}^-(\sigma_0=+1)<1/2,$$ since $\mu_{\beta,\textbf{h}}^-(\sigma_0) = 2\mu_{\beta,\textbf{h}}(\sigma_0=+1)-1$. In order to show this, we will follow a Peierls-type argument using objects called \emph{contours}, that we elaborate on in the next section.

	\section{Contours}
	
	 Contours are geometric objects arising from the deviations of a given configuration from the most likely configurations to occur, the ground states. Many extensions of the Peierls argument are available to other systems, (see for example \cite{Datta1, Park1, Park2, Pirogov, Zaradnik}; this is by no means an exhaustive list). One important extension of the Peierls argument was made by S. Pirogov and Y. Sinai in \cite{Pirogov}, and later improved by Zarahdnik \cite{Zaradnik}. Their work is known as \emph{Pirogov-Sinai} theory and can be applied to systems where the interactions are always short-range and have no symmetry. Pirogov-Sinai theory not only has a more robust version of the Peierls argument but also can be used to show the stability of the phase diagram for the models where it applies.
	In the Pirogov-Sinai theory, the contours are deviations from the ground states of the system under consideration, and in this section, we will follow the definition of contours, with subtle changes, usually encountered in presentations of Pirogov-Sinai theory, as in \cite{Vel, Park1, Park2}. For $s \in [1,\infty)$ and $x\in \Z^d$, we define
	\[
	B_s(x)=\{ y\in \Z^d: |x-y|\leq  s \}
	\]
	be the ball in the $\ell_1$-norm centered in $x$ with radius $s$. 
	
	\begin{definition}\label{def1}
		Given $\sigma \in \Omega$, a point $x \in \Z^d$ is called \emph{+ (or - resp.)} correct if $\sigma_y = +1$, (or $-1$, respectively) for all points $y$ in $B_s(x)$. The \emph{boundary} of $\sigma$, denoted by $\partial \sigma$, is defined as the set of all points in $\Z^d$ that are neither $+$ nor $-$ correct.
	\end{definition}
	
	For what follows, we will always consider $s=1$. The boundary can be an infinite subset of $\Z^d$. Indeed, if we take $\sigma \in \Omega$ defined by
	\[
	\sigma_x = \begin{cases}
		+1& |x|\text{ is even} \\
		-1& \text{otherwise},
	\end{cases}
	\]
	and $s=1$, it is easy to see that every point in $\Z^d$ is incorrect with respect to $\sigma$,and thus $\partial\sigma = \Z^d$. To avoid this situation, we will deal only with configurations such that $\partial \sigma$ is a finite set of $\Z^d$. This happens, for example, when for a configuration $\sigma$ there exists $\Lambda \Subset \Z^d$ such that $\sigma_{\Lambda^c}= +1$ (or $-1$).
	
	\begin{figure}[ht]
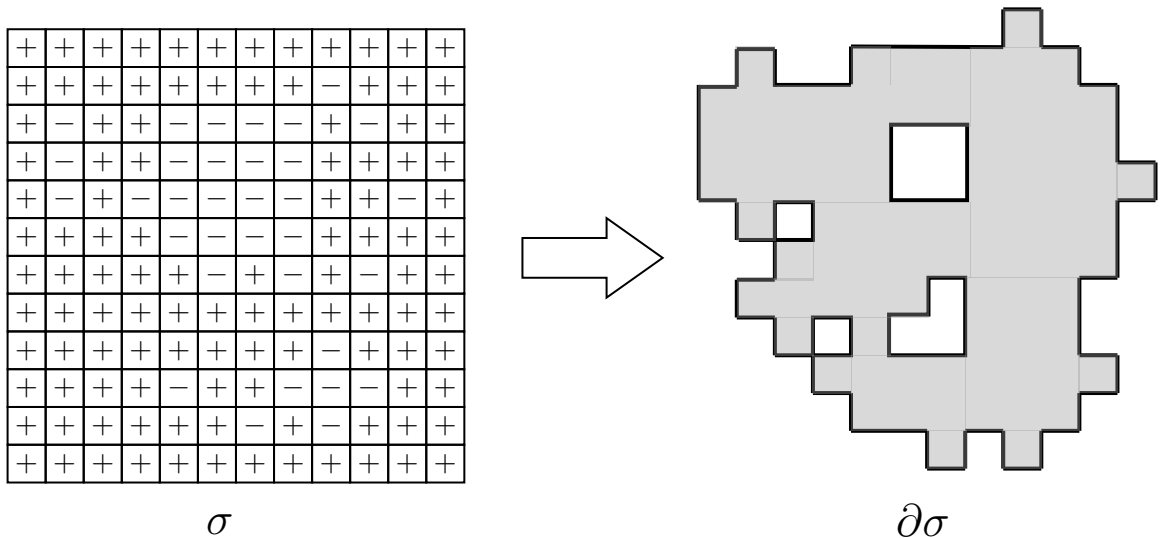

		\centering
		
		\tikzset{every picture/.style={line width=0.75pt}} 
		

		\caption{A configuration $\sigma$ with $+$-boundary condition and its respective $\partial\sigma$ set.}
		
	\end{figure}

    Fix a configuration $\sigma \in \Omega$ with boundary $\partial\sigma$ finite. We can decompose $\partial \sigma$ into a finite number of maximally connected components\footnote{a set $\Lambda \Subset \Z^d$ is maximally connected if is not included in any other connected finite set.}. Thus, we can write 
	\[
	\partial \sigma = \bigcup_{k=1}^n \overline{\gamma}_k.
	\]
We will denote by $\Gamma(\sigma)=\{\overline{\gamma}_1,\dots, \overline{\gamma}_n\}$.
\begin{definition}
		For $\Lambda \subset \Z^d$, we define the \emph{inner boundary} $\din \Lambda = \{x \in \Lambda : \inf_{y \in \Lambda^c} |x-y| =1\}$,
		and the \emph{edge boundary} as $\partial\Lambda = \{\{x,y\} \subset \Z^d: |x-y|=1, x \in \Lambda, y \in \Lambda^c\}$.
	\end{definition}
	\begin{remark}\label{isoperimetric}
		The usual isoperimetric inequality says $2d|\Lambda|^{1- \frac{1}{d}} \leq |\partial \Lambda|$\footnote{the reader can check Chapter 6 of \cite{Peres}.}. The inner boundary and the edge boundary are related by  $|\din \Lambda| \leq |\partial \Lambda| \leq 2d |\din \Lambda|$, yielding us the inequality $|\Lambda|^{1-\frac{1}{d}} \leq |\din\Lambda|$, which we will use in the rest of the chapter. 
	\end{remark}

	Another important concept for our analysis of phase transition is the interior of a contour. The following sets will be useful 
	\[
	\I_\pm(\gamma) = \hspace{-1cm}\bigcup_{\substack{k \geq 1, \\ \lab_{\overline{\gamma}}(\I(\Sp(\gamma))^{(k)})=\pm 1}}\hspace{-1cm}\I(\Sp(\gamma))^{(k)} , \;\;\;
	\I(\gamma) = \I_+(\gamma) \cup \I_-(\gamma), \;\;\;
	V(\gamma) = \Sp(\gamma) \cup \I(\gamma),
	\]
	where $\I(\Sp(\gamma))^{(k)}$ are the connected components of $\I(\Sp(\gamma))$. 	The \textit{label} of $\overline{\gamma}$ is defined as the function $\lab_{\overline{\gamma}} :\{(\overline{\gamma})^{(0)}, \I(\overline{\gamma})^{(1)}\dots, \I(\overline{\gamma})^{(n)}\} \rightarrow \{-1,+1\}$ defined as: $\lab_{\overline{\gamma}}(\I(\overline{\gamma})^{(k)})$ is the sign of the configuration $\sigma$ in $\din V(\I(\overline{\gamma})^{(k)})$, for $k\geq 1$, and $\lab_{\overline{\gamma}}((\overline{\gamma})^{(0)})$ is the sign of $\sigma$ in $\din V(\overline{\gamma})$.
	
	\begin{definition}\label{definition_short_range}
		Given a configuration $\sigma$ with finite boundary, its \textbf{contours} $\gamma$ are pairs $(\overline{\gamma},\lab_{\overline{\gamma}})$,  where $\overline{\gamma} \in \Gamma(\sigma)$ and $\lab_{\overline{\gamma}}$ is the label function defined previously. The \textbf{support of the contour}  $\gamma$ is defined as $\Sp(\gamma)\coloneqq \overline{\gamma}$ and its \emph{size} is given by $|\gamma| \coloneqq |\Sp(\gamma)|$.
	\end{definition}
	
    A contour is called a \emph{$+$- contour} (resp. \emph{- contour}) if the label of $\lab_{\overline{\gamma}}(\overline{\gamma}^{(0)})=+1$ (respectively $-1$).
	
	\begin{figure}[ht]
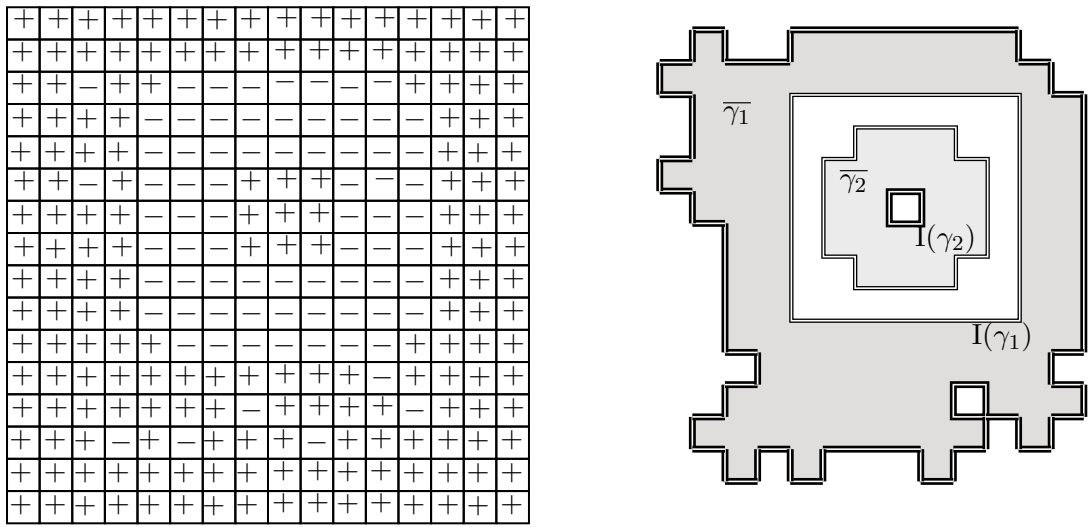

		
		\centering 
		
		\tikzset{every picture/.style={line width=0.75pt}} 
		

		\caption{A configuration and the contours associated with it. The darker line is used for the label $+$ and the other to the label $-$.}
		\label{incorrect}
	\end{figure}
	
	The contours in Definition \ref{definition_short_range} have more than only the geometrical information; they carry the labels of its interior and exterior. This is in high contrast with the usual contours found in the standard Peierls argument (see \cite{Vel}), mostly because there is \emph{no bijection} between the configurations and the contours. Given a family of contours $\gamma_1,...,\gamma_n$, they may not have a configuration associated with it. We say that a family of contours $\gamma_1, \gamma_2, ..., \gamma_n$ is \emph{compatible} if there exists $\sigma$ a configuration $\Gamma(\sigma)=\{\gamma_1,\dots,\gamma_n\}$. Compatibility is a global condition, meaning it can hold for the family $\gamma_1, \gamma_2,..., \gamma_n$ does not imply compatibility holds for all of its subfamilies. We say that a contour $\gamma$ in a family $\Gamma$ is \emph{external} if its external connected components are not contained in any other $V(\gamma')$, for $\gamma' \in \Gamma\setminus\{\gamma\}$. For each $\Lambda \Subset \Z^d$, let us define the set of all external compatible families of contours $\Gamma$ with label $\pm$ contained in $\Lambda$ by
	\[
	\mathcal{E}^\pm_\Lambda \coloneqq\{\Gamma= \{\gamma_1, \ldots, \gamma_n\}: \Gamma \text{ is compatible,} \gamma_i \text{ is external}, \lab(\gamma_i)=\pm1, V(\Gamma) \subset \Lambda\},
	\]
	where $V(\Gamma)=\bigcup_{1 \leq i\leq n}V(\gamma_i)$. When we write $\gamma \in \mathcal{E}^\pm_\Lambda$ we mean $\{\gamma\} \in \mathcal{E}^\pm_\Lambda$. We finish this section with a Lemma stating an upper bound to the set of all contours with the same size.

\begin{lemma}\label{pirogov-sinai-entropy}
Let $m\geq 1$, $d\ge 2$, and $\Lambda \Subset \mathbb{Z}^d$. Consider the set $\mathcal{C}_0(m)$ given by
		\[
		\mathcal{C}_0(m) = \{\overline{\gamma} \Subset \Z^d:\exists\gamma \in \mathcal{E}_\Lambda^- \text{ s.t. } \Sp(\gamma)=\overline{\gamma} , 0 \in V(\gamma), |\gamma|=m\}.
		\]
		There exists an $c'\coloneqq c'(d)>0$ such that 
		\[
		|\mathcal{C}_0(m)| \leq e^{c' m}.
		\] 
\end{lemma}

\begin{proof}
The proof is in Lemma 2.7  \cite{Sinaibook}. For a given contour $\gamma$, define the set $\mathcal{C}_\gamma$ by
		\be
		\mathcal{C}_\gamma \coloneqq \{\Sp(\gamma') \in \mathcal{C}_0(m): \exists \; x \in \Z^d \text{ s.t. } \Sp(\gamma') = \Sp(\gamma)+x\}.
		\ee
		Thus, we can partition the set $\mathcal{C}_0(m)$ into
		\[
		\mathcal{C}_0(m) = \bigcup_{\substack{0 \in \Sp(\gamma) \\ |\gamma|=m}}C_\gamma.
		\]	
		Given a contour $\gamma \in \mathcal{E}_\Lambda$, there are at most $|V(\gamma)|$ possibilities for the position of the point $0$.
		Then,
		\be\label{eq_cont_1}
		|\mathcal{C}_0(m)| \leq \sum_{\substack{0 \in \Sp(\gamma) \\ |\gamma|=m}}|\mathcal{C}_\gamma| \leq \sum_{\substack{0 \in \Sp(\gamma) \\ |\gamma|=m}}|V(\gamma)|.
		\ee
		Using the isoperimetric inequality and the fact $\din V(\gamma) \subset \Sp(\gamma)$ we obtain,
		\be\label{eq_cont_2}
		\sum_{\substack{0 \in \Sp(\gamma) \\ |\gamma|=m}}|V(\gamma)| \leq m^{1+ \frac{1}{d-1}}|\{\overline{\gamma} \Subset \Z^d:\exists\gamma \in \mathcal{E}_\Lambda^- \text{ s.t. } \Sp(\gamma)=\overline{\gamma}, 0 \in \Sp(\gamma), |\gamma|=m \}|.
		\ee
  The set above is less than the number of connected subsets of $\Z^d$ that contain 0. For a given connected subset $\Lambda \Subset \Lambda$ that contains zero, let $L_k(\Lambda) = |\Lambda \cap S_k(0)|$, where $S_k(0) = \{x\in \Z^d:|x|=k\}$. Since the size of $\Lambda$ is $m$, we have
  \[
  \sum_{k=1}^m L_k(\Lambda)=m-1.
  \]
  Then, summing over all solutions of the equation above, we have
  \begin{equation*}
  \begin{split}
  \{\Lambda\Subset\Z^d:\Lambda \text{ is connected}, &|\Lambda|=m, 0\in \Lambda\}\leq \\ &\sum_{\sum_{k=1}^m L_k=m-1}\{\Lambda\Subset\Z^d:\Lambda \text{ is connected}, L_k(\Lambda)=L_k, k=1,\dots,m, 0\in \Lambda\}.
  \end{split}
  \end{equation*}
  Since the set $\Lambda$ is connected, the number, given $L_k$ and a set $\Lambda\cap S_k(0)$, the number of possibles sets $\Lambda\cap S_{k+1}(0)$ is bounded by $d^{L_k}$, for $k\geq 2$. For $k=1$, since we are fixing the number $0$, there will be at most $\binom{2d}{L_1}$ ways of choosing the set $\Lambda\cap S_1(0)$. Hence,
  \[
  \{\Lambda\Subset\Z^d:\Lambda \text{ is connected}, L_k(\Lambda)=L_k, k=1,\dots,m, 0\in \Lambda\}\leq (2d)^{L_1}\prod_{k=2}^{m} d^{L_k} = 2^{L_1}d^{m-1-L_1}.
  \]
  Since the number of solutions for $\sum_{k=1}^m L_k=m-1$ is less than $2^{m-1}$, we get $|\mathcal{C}_0(m)|\leq e^{c'(d)m}$ where $c'(d) = \max\{\log(2),\log(d)\}+1 + (d-1)^{-d+1}$.
\end{proof}

\section{Phase Transition}
 
 For $x \in \Z^d$, let  $\Theta_x : \Omega \rightarrow \mathbb{R}$ be defined as
	\[
	\Theta_x(\sigma) = \prod_{\substack{y \in \Z^d \\ |x-y| \leq 1}}\mathbbm{1}_{\{\sigma_y = +1\}}- \prod_{\substack{y \in \Z^d \\ |x-y| \leq 1}}\mathbbm{1}_{\{\sigma_y = -1\}}. 
	\]
	
	The function $\Theta_x$ yields $+1$ if the point $x$ is $+$ correct, $-1$ if the point is $-$ correct, and $0$ if $x$ is incorrect for $\sigma$. By the definition of contours, given a finite $\Lambda \Subset \Z^d$ and a configuration $\sigma \in \Omega^-_\Lambda$ it may happen that a contour $\gamma$ associated with it has volume outside $\Lambda$. To avoid this problem consider the probability measure 
	\be
	\nu^-_{\beta,\textbf{h},\Lambda}(A) \coloneqq \mu^-_{\beta,\textbf{h},\Lambda}(A|\Theta_x =  - 1, x \in \din \Lambda),
	\ee
	for every measurable set $A$. The spatial Markov property\footnote{Equation (3.26) in Chapter 3 of \cite{Vel}.} implies that the probability measures $\nu^-_{\beta,\textbf{h},\Lambda}$ are the finite volume Gibbs measure in a subset of $\Lambda$ and to work with them is advantageous since we can study important quantities in terms of contours. We also will assume that $\Lambda$ is simply connected, in this way guaranteeing that for every contour $\gamma$ whose support is inside $\Lambda$ we will automatically get  $V(\gamma) \subset \Lambda$. Fixed $\Lambda \Subset \Z^d$, for each $\Lambda' \subset \Lambda$, the restricted partition functions is
	\be
	Z_{\beta,\textbf{h}}^-(\Lambda') = \sum_{\substack{\Gamma \in \mathcal{E}_\Lambda^-\\ V(\Gamma)\subset \Lambda'}}\sum_{\sigma \in \Omega(\Gamma)}e^{-\beta H_{\Lambda, \textbf{h}}^-(\sigma)},
	\ee
	
where the space of configurations compatible with a given family $\Gamma \in \mathcal{E}^-_\Lambda$ is $$\Omega(\Gamma)\coloneqq \{\sigma \in \Omega^-_\Lambda: \Gamma \subset \Gamma(\sigma)\}.$$ The following two Lemmas are necessary to control the term corresponding to the magnetic field. The first one is just a combinatorial lemma bounding the number of integer points in the $\ell^1$-sphere and the second one will give us control over the sum of the truncated magnetic field in finite regions of $\Z^d$.
	
	\begin{lemma}\label{comblema}
		Let $s_d(n)$ be the cardinality of integer points in the $\ell^1$ sphere, centered at the origin and with radius $n$. Then, for any $n\geq d$, we have,
		\[
		s_d(n) = \sum_{k=0}^{d-1}2^{d-k}\binom{d}{k}\binom{n-1}{d-k-1}.
		\]
		If $n<d$, the sum above starts in $k=d-n$. 
	\end{lemma}
	\begin{proof}
		We need to count the number of integer solutions to the equation
		\[
		|x_1|+|x_2|+...+|x_d| = n.
		\]
		Suppose, first, that $n \geq d$ and assume that $|x_k| \neq 0$, for $k=1,2,...,d$. Then, the number of solutions is $\binom{n-1}{d-1}$. Since we are summing the absolute value, we need to count the parity, yielding us $2^d$ different solutions. Assume that there is $k$ such that $x_k=0$. Then, this is the same as counting the number of solutions to the equation
		\[
		|x_1|+|x_2|+...+|x_{k-1}|+|x_{k+1}|+...+|x_d| = n.
		\]
		This number is $\binom{n-1}{d-2}$. The parity, again, gives us $2^{d-1}$ solutions. Since we have $d$ possibilities where $x_k =0$, we have $d= \binom{d}{1}$ possibilities. It is easy to see that for the general case with $k$ zero entries the number of solutions is
		\[
		2^{d-k}\binom{d}{k}\binom{n-1}{d-k-1}.
		\] 
		Summing over $k$ yields the desired result. Consider the case $n< d$. Then necessarily we have $|x_k|=0$ for at least $d-n$ indices. Nonetheless, the same reasoning as before applies to this case and we get the desired result. 
	\end{proof}
\begin{corollary}
For every $d\geq 2$ and $n\geq 1$ it holds that 
     \be\label{eq:desi_sd}
	c_d n^{d-1}\leq s_d(n) \leq e^{-1}(2e+1)^d n^{d-1}, 
	\ee
 where $c_d = 2(d-1)^{d-1}$.
\end{corollary}
\begin{proof}
We start by the upper bound. Previous Lemma together with the fact that $\binom{n-1}{d-k-1}\leq ((n-1)e/(d-k-1))^{d-k-1}$, for $k<d-1$ implies that 
\begin{equation*}
\begin{split}
s_d(n) &\leq \sum_{k=\min\{0,d-n\}}^{d-1}2^{d-k}\binom{d}{k}\frac{(n-1)^{d-1-k}}{(d-k-1)^{d-1-k}}e^{d-1-k} \\
&\leq n^{d-1}e^{-1} \sum_{k=0}^{d}(2e)^{d-k}\binom{d}{k} = e^{-1}(2e+1)^d n^{d-1}
\end{split}
\end{equation*}
 For the lower bound, we need to split it into two cases, depending on $n$. Regardless, one will need to use the Chu-Vandermonde identity for the binomial coefficients that we write here for completeness
 \[
 \binom{n+m}{r}=\sum_{k=0}^r\binom{m}{k}\binom{n}{r-k},
 \]
 for $m,n,r$ positive integers. For the first case, consider $n\geq d$. Then, we have 
 \[
 s_d(n) \geq 2\sum_{k=0}^{d-1}\binom{d}{k}\binom{n-1}{d-1-k} = 2\binom{n+d-1}{d-1} \geq \frac{2(n+d-1)^{d-1}}{(d-1)^{d-1}}\geq 2(d-1)^{d-1}n^{d-1}.
 \]
 For the case $1<n<d$, we need to make a change of variables and write $j = k-d+n$, thus
 \begin{equation*}
     \begin{split}
         s_d(n)&=\sum_{k=d-n}^{d-1}2^{d-k}\binom{d}{k}\binom{n-1}{d-1-k}=\sum_{j=0}^{n-1}2^{n-j}\binom{d}{d+j-n}\binom{n-1}{n-1-j} \\
         &=\sum_{j=0}^{n-1}2^{n-j}\binom{d}{n-j}\binom{n-1}{j} = \sum_{j=0}^{n-1}2^{n-j}\Bigg(\binom{d}{n-1-j}+\binom{d-1}{n-1-j}\Bigg)\binom{n}{j} \\
         &\geq 2\Bigg(\binom{n+d-1}{n-1}+\binom{n+d}{n-1}\Bigg) = 2\Bigg(\frac{n+2d+1}{d+1}\Bigg)\binom{n+d-1}{d} \geq 4 \frac{(n+d-1)^d}{d^d}.
         \end{split}
 \end{equation*}
 where in the inequality above we used $2^{n-j}\geq 2$ in the range of summation and the Chu-Vandermonde identity. To get the desired bound, just divide and multiply by $c_d n^{d-1}$ and use the fact that $1<n<d$ to get the desired bound. Note that $s_d(1) = 2d$, satisfying the inequalities in \eqref{eq:desi_sd}. 
 \end{proof}
 \pagebreak
	\begin{lemma}\label{lema1}
		Let $(h_x)_{x \in\Z^d}$ be the magnetic field as in (\ref{magfield}) with $\delta<d$. Then, there exists $c_5\coloneqq c_5(d,\delta,h^*)>0$ such that for any $\Lambda \Subset \Z^d$ 
		\begin{equation}\label{eq4}
			\sum_{x \in \Lambda} h_x \leq c_5 |\Lambda|^{1-\frac{\delta}{d}}.
		\end{equation} 
	\end{lemma}
        \begin{proof}
		Fix a set $\Lambda \Subset \Z^d$. In order to prove the inequality (\ref{eq4}), we will show that the sum in the l.h.s is always upper bounded by the sum of the magnetic field $h_x$ in some ball $B_R(0)$ with $R$ large enough. In fact, the magnetic field satisfies $h_x\geq h^*/R^\delta$ for $x \in B_R(0)$, and $h_x< h^*/R^\delta$ for $x \in \Lambda\setminus B_R(0)$. Then, we have
		\be\label{eq:diff_field}
		\sum_{x\in B_R(0)}h_x -\sum_{x \in \Lambda} h_x \geq  \frac{h^*}{R^\delta}\left( |B_R(0)|-|\Lambda|\right).
		\ee
  Using the lower bound in \eqref{eq:desi_sd}, we deduce that the ball satisfies, for every $R \in \mathbb{N}$,
		\begin{align*}
			|B_R(0)| = \sum_{n=0}^R s_d(n) &\geq c_d\sum_{n=1}^R n^{d-1}  \geq \frac{c_d}{d} R^d,
		\end{align*}
		where the last inequality follows by lower bounding the sum by the integral. Thus, if we choose $R \coloneqq \lceil (dc_d^{-1}|\Lambda| )^{\frac{1}{d}}\rceil$, the right-hand side of Inequality \eqref{eq:diff_field} is nonnegative. We can bound the sum of the magnetic field in a ball $B_R(0)$ in the following way,
		\be
        \begin{split}
		\ssum{x \in B_R(0)} h_x \leq h^*e^{-1}(2e+1)^d\sum_{n=1}^{R} n^{d-1-\delta} 
  \end{split}
        \ee
		The result is obtained taking $c_5 = h^*e^{-1}(2e+1)^d(d-\delta)^{-1}((dc_d^{-1})^{1/d}+2)^{d-\delta}$ after bounding the sum above by an integral.
		
	\end{proof}
	
We are ready to prove the main result of this chapter.

 \begin{proposition}
		For $\beta$ large enough, it holds that 
		\be
		\nu^-_{\beta,\widehat{\textbf{h}},\Lambda}(\sigma_0= + 1) < \frac{1}{2},
		\ee
		for every $\Lambda \Subset \Z^d$.
		\end{proposition}
	\begin{proof}
		
	Let $R>0$ and $(\widehat{h}_x)_{x \in\Z^d}$ be the truncated magnetic field 
	\begin{equation}\label{magfield2}
		\widehat{h}_x=\begin{cases}
			0 & |x|<R, \\
			h_x& |x| \geq R.
		\end{cases}
	\end{equation}
         The constant $R$ will be chosen later. The existence of phase transition under the presence of the truncated field implies phase transition for the model with the decaying field (see Theorem 7.33 of \cite{Geo} for a more general statement). If $\sigma_0=+1$ there must exist a contour $\gamma$ such that $0 \in V(\gamma)$. Hence
		\[
		\nu^-_{\beta,\widehat{\textbf{h}},\Lambda}(\sigma_0= + 1) \leq \sum_{\substack{\gamma \in \mathcal{E}_\Lambda^- \\ 0 \in V(\gamma)}} \nu^-_{\beta,\widehat{\textbf{h}},\Lambda}(\Omega(\gamma)). 
		\]
		
	Consider the Hamiltonian function (\ref{Isingsys}) for the set $\Lambda = V(\Gamma)$, where $\Gamma$ is a family of external $-$-contours. If $\sigma \in \Omega(\gamma)$ is such that $\Gamma(\sigma)=\Gamma$, we can split the Hamiltonian as
	\begin{equation}
    \begin{split}
		H_{\Lambda,\widehat{\textbf{h}}}^-(\sigma) =& H_{\Lambda\setminus V(\Gamma),\widehat{\textbf{h}}}^-(\sigma^-)+\sum_{\gamma'\subset \Gamma}\Bigg[H_{\I_+(\gamma'),\widehat{\textbf{h}}}^+(\sigma) + H_{\I_-(\gamma'),\widehat{\textbf{h}}}^-(\sigma
		) + H_{\Sp(\gamma'),\widehat{\textbf{h}}}^\eta(\sigma)\Bigg] ,
  \end{split}
	\end{equation}
	where the boundary condition $\eta$ is defined as
	\[
	\eta_x = \begin{cases}
		+1 & x \in \I_+(\Gamma), \\
		-1 & x \in o.w.
	\end{cases}
	\]
    For each $\sigma \in \Omega(\gamma)$, it holds
	\begin{equation}\label{nnIsing_eq1}
 \begin{split}
	&\nu^-_{\beta,\widehat{\textbf{h}},\Lambda}(\Omega(\gamma)) = \sum_{\Gamma \supset \gamma}\nu^-_{\beta,\widehat{\textbf{h}},\Lambda}(\sigma\in:\Omega(\gamma):\Gamma(\sigma)=\Gamma) \\
 &= \frac{1}{Z^-_{\beta,\widehat{\textbf{h}}}(\Lambda)}\sum_{\Gamma \supset \gamma}e^{-\beta H_{\Lambda\setminus V(\Gamma),\widehat{\textbf{h}}}^-(\sigma^-)}\prod_{\gamma' \in \Gamma}Z^+_{\beta,\widehat{\textbf{h}}}(\I_+(\gamma')) Z^-_{\beta,\widehat{\textbf{h}}}(\I_- (\gamma'))e^{-\beta H_{\Sp(\gamma'),\widehat{\textbf{h}}}^{\eta}(\sigma)}.
 \end{split}
	\end{equation}
    Since every family of contours $\Gamma$ contains the fixed contour $\gamma$ we can factorize this term and we can rewrite the difference between the Hamiltonians as
	\[
	H^\eta_{\Sp(\gamma),\widehat{\textbf{h}}}(\sigma) - H^-_{\Sp(\gamma),,\widehat{\textbf{h}}}(\sigma^-)= 2\sum_{\mathclap{\{x,y\} \subset \Sp(\gamma)}} J_{xy} \mathbbm{1}_{\{\sigma_x \neq \sigma_y\}} + 2\sum_{\mathclap{\substack{x \in \Gamma \\ y \in \Gamma^c}}}J_{xy}\mathbbm{1}_{\{\sigma_y = +1\}}-\sum_{x\in \Sp(\gamma)}\widehat{h}_x. 
	\]
	 The definition of an incorrect point implies that for every $x \in \Sp(\gamma)$, there is a nearest neighbor point $y$ such that $\sigma_x \neq \sigma_y$. Thus,
	\[
	2\sum_{\mathclap{\{x,y\} \subset \Sp(\gamma)}} J_{xy} \mathbbm{1}_{\{\sigma_x \neq \sigma_y\}} \geq J|\gamma|.
	\]
  Notice that $\sum_{x \in \Sp(\gamma)}\hat{h}_x \leq \frac{h^* |\gamma|}{R^\delta}\leq J|\gamma|/2$
 if $R> (2h^*)^{1/\delta}$. The Hamiltonian function for the plus and minus boundary conditions satisfies
	\begin{align*}
		H^+_{\Lambda,,\widehat{\textbf{h}}}(\sigma) &= -\sum_{\{x,y\} \subset \Lambda} J_{xy} \sigma_x\sigma_y - \sum_{\substack{x \in \Lambda \\ y \in \Lambda^c}}J_{xy} \sigma_x -  \sum_{x \in \Lambda}\widehat{h}_x \sigma_x \\
		&=  -\sum_{\{x,y\} \subset \Lambda} J_{xy} (-\sigma_x)(-\sigma_y) - \sum_{\substack{x \in \Lambda \\ y \in \Lambda^c}}J_{xy} (-\sigma_x)(-1) - \sum_{x \in \Lambda}h_x (-\sigma_x) - 2\sum_{x \in \Lambda}\widehat{h}_x\sigma_x \\
		&= H^-_{\Lambda,,\widehat{\textbf{h}}}(-\sigma) - 2\sum_{x\in \Lambda} \widehat{h}_x\sigma_x \geq H^-_{ \Lambda,\widehat{\textbf{h}}}(-\sigma) - 2\sum_{x \in \Lambda} h_x .
	\end{align*}
	Hence the restricted partition function satisfies the inequality
	\begin{equation}\label{partition2}
		Z^+_{\beta,\widehat{\textbf{h}}}(\I_+(\gamma)) \leq \exp\left(2\beta  \sum_{x \in \I_+(\gamma)} \widehat{h}_x\right) Z^-_{\beta,\widehat{\textbf{h}}}(\I_+(\gamma)).
	\end{equation}
	   Equations \eqref{nnIsing_eq1} and \eqref{partition2} yield,
    \be
    \nu_{\beta,\widehat{\textbf{h}},\Lambda}(\Omega(\gamma)) \leq \exp\Bigg(\left(\log(2)-\beta \frac{J}{2}\right)|\gamma|+2\beta\sum_{x \in \I_+(\gamma)} \widehat{h}_x\Bigg)\frac{Z^-_{\beta,\widehat{h}}(\Lambda\setminus\Sp(\gamma))}{Z^-_{\beta,\widehat{h}}(\Lambda)}.
    \ee
    where the $\log(2)$ factor comes from the fact that there are at most $2^{|\gamma|}$ incorrect points. By Lemma \ref{lema1}, there exists a constant $c_5>0$ such that
	\[
	\sum_{x \in \I_+(\gamma)}\widehat{h}_x \leq c_5|\I_+(\gamma)|^{1-\frac{\delta}{d}}.
	\]
	Also, if $|\I_+(\gamma)|\leq C$, the truncated magnetic field satisfies $\sum_{x \in \I_-\Gamma} \widehat{h}_x \leq \frac{h^* C}{R^\delta}$. The Isoperimetric inequality stated in Remark \ref{isoperimetric} implies that for $d \geq 2$ and $\Lambda \in \mathcal{F}(\Z^d)$, we have $|\Lambda|^{1 - \frac{1}{d}} \leq |\partial\Lambda|$. Therefore, the following holds
	\[
	|\I_-(\gamma)|^{1-\frac{\delta}{d}}\leq |\partial\I_-(\gamma)|^{\frac{d-\gamma}{d-1}} \leq |\gamma|^{\frac{d-\delta}{d-1}}.
	\]
	If $R$ is large enough, the truncated field can be bounded above by $J|\gamma|/4$. The case $\delta = 1$ can also be carried out in a similar fashion. However, since the ratio $\frac{d-\delta}{d-1}=1$, we must assume that the field $h^*$ is small. 
		Summing over all contours yields, together with Proposition \ref{pirogov-sinai-entropy},
		\begin{align*}
			\nu_{\beta,\bm{\widehat{h}},\Lambda}^-(\sigma_0 = +1) &\leq \sum_{\substack{\gamma \in \mathcal{E}_\Lambda^- \\ 0 \in V(\gamma)}}e^{(\log(2)-\beta \frac{J}{4})|\gamma|}\frac{Z^-_{\beta, \bm{\widehat{h}}}(\Lambda\setminus\Sp(\gamma))}{Z^-_{\beta, \bm{\widehat{h}}}(\Lambda)} \nonumber \\
			&\leq \sum_{m \geq 1}|\mathcal{C}_0(m)|e^{(\log(2)-\beta \frac{c_2}{2})m} \nonumber\\
			&\leq \sum_{m \geq 1}e^{(c_1 +\log(2) - \beta \frac{c_2}{2})m}< \frac{1}{2},
		\end{align*}
		for $\beta$ large enough. 
	\end{proof}

 Together with the discussion at the end of Section 1.1, the proposition above implies phase transition for the model. The uniqueness for $\delta<1$ can be found in \cite{Bis2}.

\chapter{Multidimensional Fr\"{o}hlich-Spencer Contours}
\label{ch:longrange}
	
	In this chapter, we extend the analysis in \cite{Bis2}, also presented in the previous chapter, from the nearest-neighbor to the long-range Ising model considering decaying external fields in the Hamiltonian \eqref{hamiltonian}. The contents are edited from the paper \cite{Aff}. By a similar approach as Fr\"{o}hlich and Spencer, we define a notion of contour for the model to show the phase transition at low temperature when $d<\alpha<d+1$ and $\delta>\alpha-d$, and when $\alpha\ge d+1$  and $\delta>1$. We begin with an heuristics for this result. Consider the configurations
	\be
	\sigma_x = \begin{cases}
		+1& \text{if } x \in B_R(z), \\
		-1& \text{otherwise},
	\end{cases}
	\ee
	where $B_R(z)\subset \Z^d$ is the closed ball in the $\ell_1$-norm centered in $z$ with radius $R\geq 0$. Let $\Omega_c$ be the collection of all such configurations and, for fixed $\Lambda \Subset \Z^d$, let $\Omega_{c,\Lambda}\subset \Omega_c$ be the subset of configurations where $B_R(z)\subset \Lambda$. Then, we have
	\be\label{contaheuristica}
	\mu_{\beta,\Lambda}^-( \sigma_0 = +1|\Omega_{c,\Lambda}) \leq \sum_{R\geq 1}R^d \exp\left(-\beta (cR^{d-1}+ F_{B_R}- \sum_{x \in B_R(0)}h_x)\right),
	\ee
	where the quantity $F_{B_R}$ is defined as
	\be
	F_{B_R} \coloneqq \sum_{\substack{x \in B_R(0)\\ y \in B_R(0)^c}}J_{xy}.
	\ee
	One can understand this quantity as a surface energy term, and it has different asymptotics depending on the parameters $\alpha$ and $d$.  Denoting by $f\approx g$  the fact that given two functions $f,g:\mathbb{R}^{+} \rightarrow \mathbb{R}^{+}$, there exist positive constants $A'\coloneqq A'(\alpha, d),  A\coloneqq A(\alpha, d)$ such that $A'f(x) \leq g(x) \leq Af(x)$ for every $x>0$ large enough, we have
	\be
	F_{B_R} \approx \begin{cases}
		R^{2d-\alpha}&\text{if } d<\alpha<d+1, \\
		R^{d-1}\log(R)&\text{if } \alpha = d+1, \\
		R^{d-1}&\text{if } \alpha>d+1.
	\end{cases}
	\ee
	For the proof, see Propositions 3.1 and 3.4 in  \cite{BBCK}.  For our purposes, we will need estimates for more general subsets than balls, see Lemma \ref{lema3}. The inequality \eqref{contaheuristica} together with Lemma \ref{lema1} shows us that phase transition occurs when $\delta> \alpha - d$, by comparing the exponents, since the field influence in the system is given by
	\be
	\sum_{x \in B_R(0)}h_x = O(R^{d-\delta}).
	\ee

 In the following sections, we elaborate on the construction of the contours.
 
	\section{The \texorpdfstring{$(M,a,r)$}{TEXT} -partition}
	
	Park, in \cite{Park1, Park2} extended the theory of Pirogov-Sinai to systems with two-body long-range interactions that satisfy a condition equivalent to \eqref{long} having decay $\alpha > 3d+1$. Inspired by \cite{Fro1}, in this section, we will introduce new contours more suitable for studying long-range two-body interactions. 
	
	In Pirogov-Sinai theory, the construction of the contours starts by considering first the connected subsets of the boundary $\partial\sigma$, as we did in the previous chapter.  However, this approach presents certain challenges when applied to long-range models. In these particular systems, every lattice point interacts with all other lattice points, resulting in the emergence of non-negligible interactions within any proposed contour definition. To avoid this problem, we will divide the boundary of a configuration in a way where the interaction between them will be manageable as we will elaborate on later. 
	\begin{definition}\label{def:d_condition}
		Fix real numbers $M,a,r>0$. For each configuration $\sigma \in \Omega$ with finite boundary $\partial\sigma$, a set $\Gamma(\sigma) \coloneqq \{\overline{\gamma} : \overline{\gamma} \subset \partial\sigma\}$ is called an $(M,a,r)$-\emph{partition} when the following conditions are satisfied:
		\begin{enumerate}[label=\textbf{(\Alph*)}, series=l_after] 
			\item They form a partition of $\partial\sigma$, i.e.,  $\bigcup_{\overline{\gamma} \in \Gamma(\sigma)}\overline{\gamma}=\partial\sigma$ and $\overline{\gamma} \cap \overline{\gamma}' = \emptyset$ for distinct elements of $\Gamma(\sigma)$. Moreover, each $\overline{\gamma}'$ is contained in only one connected component of $(\overline{\gamma})^c$. 
			
			\item For all $\overline{\gamma} \in \Gamma(\sigma)$ there exist $1\leq n \leq 2^r-1$ and a family of subsets $(\overline{\gamma}_{k})_{1\leq k \leq n}$ satisfying 
			
			\begin{enumerate}[label=\textbf{(B.\arabic*)}]
				\item  $\overline{\gamma} = \bigcup_{1\leq k \leq n}\overline{\gamma}_{k}$,
				\item For all distinct $\overline{\gamma},\overline{\gamma}' \in \Gamma(\sigma)$,
				\be\label{B_distance}
				\dis(\overline{\gamma},\overline{\gamma}') > M \min\left \{\underset{1\leq k \leq n}{\max}\diam(\overline{\gamma}_k),\underset{1\leq j\leq n'}{\max}\diam(\overline{\gamma}'_{j})\right\}^a,
				\ee
				where $(\overline{\gamma}'_j)_{1\leq j \leq n'}$ is the family given by item $\textbf{(B1)}$ for $\overline{\gamma}'$.
			\end{enumerate} 
		\end{enumerate}
	\end{definition}
	Note that the sets $\overline{\gamma} \in \Gamma(\sigma)$ may be disconnected. In Condition \textbf{(A)}, $\overline{\gamma}'$ is contained in the unbounded component of $\overline{\gamma}^c$ if and only if $V(\overline{\gamma})\cap V(\overline{\gamma}')= \emptyset$. Some results are true for any $M,a,r>0$, as the existence of $(M,a,r)$-partition for any configuration $\sigma$ with finite boundary $\partial\sigma$,  see Proposition \ref{prop1}.  However, for the main purposes of this chapter, which is the proof of the phase transition, the constant $a$ is chosen as $a \coloneqq a(\alpha, d,\varepsilon) = \max\left\{\frac{d+1+\varepsilon}{\alpha-d}, d+1+\varepsilon\right\}$, for some $\varepsilon>0$ fixed and $r$ given by $r= \lceil \log_2(a+1) \rceil + d +1$, where $\lceil x \rceil$ is the smallest integer greater than or equal to $x$. The motivation for these choices will be clear in the proofs. The constant $M\coloneqq M(\alpha, d)$ will be chosen later. 	
 
 For a fixed configuration $\sigma$ with finite boundary $\partial \sigma$, the $(M,a,r)$-partitions will be the \textit{support of the contours},  subsets of $ \Z^d$ where every point is incorrect. Although we are strongly inspired by the papers of Fr\"{o}hlich and Spencer \cite{Fro3, Fro1, Fro2} (see also \cite{Imbr1, Imbr2}), we implemented modifications that allow us to cover all the region $\alpha > d$. 
\\
	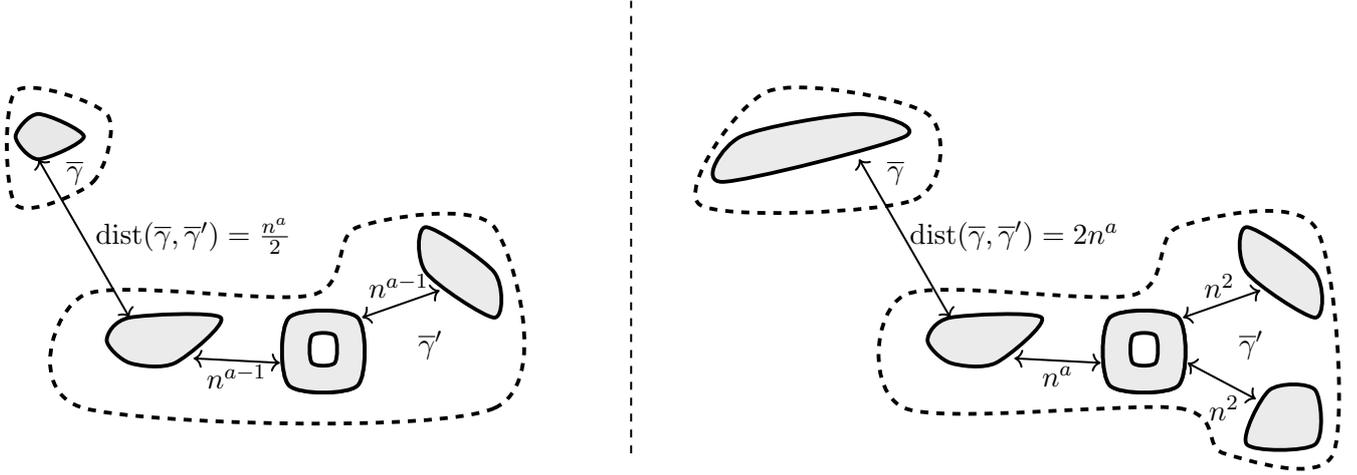
\begin{figure}[H]
		\centering
		
		\tikzset{every picture/.style={line width=0.75pt}}

		\begin{tikzpicture}[scale=1.2]
			\begin{scope}[shift={(-5,0)}]
				\draw [ line width=0.5mm, draw=black, dashed] plot [smooth cycle] coordinates {(0.25,1.75) (1.1,2) (1.25,2.75) (0.25,3) }--cycle;
				\draw [line width=0.5mm, fill=black!9!white, draw=black] plot [smooth cycle] coordinates {(0.5,2.25) (1,2.5) (0.5,2.75) (0.25,2.5)}--cycle;
				
				\draw [line width=0.5mm, fill=white, draw=black, dashed] plot [smooth cycle] coordinates {(1,-0.5) (5.5,-0.5) (5.5,1.5) (4,1.5) (3.5,0.75) (1,0.75) }--cycle;
				\draw [line width=0.5mm, fill=black!8!white, draw=black] plot [smooth cycle] coordinates {(1.5,0) (2,0) (2.5,0.5) (1.5,0.5) (1.25,0.25) }--cycle;
				\draw [line width=0.5mm, fill=black!8!white, draw=black] plot [smooth cycle] coordinates {(3.25,0.5) (3.25,-0.25) (4,-0.25) (4,0.5) }--cycle;
				\draw [line width=0.5mm, fill=black!8!white, draw=black] plot [smooth cycle] coordinates {(3.25,0.5) (3.25,-0.25) (4,-0.25) (4,0.5) }--cycle;
				\draw [line width=0.5mm, fill=white, draw=black] plot [smooth cycle] coordinates {(3.75,0.3) (3.75,0) (3.5,0)  (3.5,0.3) }--cycle;
				\draw [line width=0.5mm, fill=black!8!white, draw=black] plot [smooth cycle] coordinates {(4.75,1.5) (4.75,1) (5.5,0.5) (5.5,1) }--cycle;
				\draw (0.9,2.1) node[scale=1] {$\overline{\gamma}$};
				\draw (4.8,0.2) node[scale=1] {$\overline{\gamma}'$};
				\draw[<->] (0.5,2.25)--(1.5,0.5);
				\draw (2.2,1.4) node[scale=1] {$\dis(\overline{\gamma},\overline{\gamma}') = \frac{n^a}{2}$};
				\draw[<->] (2.2,0.05)--(3.15,0);
				\draw (2.675,-0.15) node[scale=1] {$n^{a-1}$};
				\draw[<->] (4.05,0.5)--(4.9,0.8);
				\draw (4.45,0.85) node[scale=1] { $n^{a-1}$};
			\end{scope}
			
			\draw [-, dashed] (2,-1)--(2,4);
			
			\begin{scope}[shift={(4,0)}]
				\draw [ line width=0.5mm, draw=black, dashed] plot [smooth cycle] coordinates {(-1.25,1.75) (1.1,1.8) (1.25,2.75) (-0.5,3) }--cycle;
				\draw [line width=0.5mm, fill=black!8!white, draw=black] plot [smooth cycle] coordinates {(-1,2) (1,2.5) (0.5,2.75) (-0.8,2.5)}--cycle;
				
				\draw [line width=0.5mm, fill=white, draw=black, dashed] plot [smooth cycle] coordinates {(1,-0.5) (4,-0.5) (4.5,-1) (5.7,-1) (5.5,1.5) (4,1.5) (3.5,0.75) (1,0.75) }--cycle;
				\draw [line width=0.5mm, fill=black!8!white, draw=black] plot [smooth cycle] coordinates {(1.5,0) (2,0) (2.5,0.5) (1.5,0.5) (1.25,0.25) }--cycle;
				\draw [line width=0.5mm, fill=black!8!white, draw=black] plot [smooth cycle] coordinates {(3.25,0.5) (3.25,-0.25) (4,-0.25) (4,0.5) }--cycle;
				\draw [line width=0.5mm, fill=black!8!white, draw=black] plot [smooth cycle] coordinates {(3.25,0.5) (3.25,-0.25) (4,-0.25) (4,0.5) }--cycle;
				\draw [line width=0.5mm, fill=white, draw=black] plot [smooth cycle] coordinates {(3.75,0.3) (3.75,0) (3.5,0)  (3.5,0.3) }--cycle;
				\draw [line width=0.5mm, fill=black!8!white, draw=black] plot [smooth cycle] coordinates {(4.75,1.5) (4.75,1) (5.5,0.5) (5.5,1) }--cycle;
				\draw [line width=0.5mm, fill=black!8!white, draw=black] plot [smooth cycle] coordinates {(5,-0.3) (4.75,-0.9) (5.5,-0.9) (5.5,-0.3) }--cycle;
				
				\draw (0.9,2.1) node[scale=1] {$\overline{\gamma}$};
				\draw (4.8,0.2) node[scale=1] {$\overline{\gamma}'$};
				\draw[<->] (0.5,2.25)--(1.5,0.5);
				\draw (2.2,1.4) node[scale=1] {$\dis(\overline{\gamma},\overline{\gamma}') = 2n^a$};
				\draw[<->] (2.2,0.05)--(3.15,0);
				\draw (2.675,-0.15) node[scale=1] { $n^a$};
				\draw[<->] (4.05,0.5)--(4.9,0.8);
				\draw (4.45,0.85) node[scale=1] { $n^2$};
				\draw[<->] (4.1,0)--(4.85,-0.4);
				\draw (4.5,-0.5) node[scale=1] { $n^2$};
			\end{scope}	
		\end{tikzpicture}
		\caption{Consider $M=1$, $r=2$. For the image on the left, consider that all the connected components (grey regions) have a diameter equal to $n$. In this case, there is no partition of $\overline{\gamma}'$ satisfying condition \textbf{(B)}. The correct $(M,a,r)$-partition for this case is $\Gamma(\sigma)=\{\overline{\gamma}\cup \overline{\gamma}'\}$. For the figure on the right, consider that all the connected components of $\overline{\gamma}'$ have diameter $n$ and $\diam(\overline{\gamma})=n^2$. Notice that, in this case, the families of subsets of $\overline{\gamma'}$ satisfying Inequality \eqref{B_distance} must have $n' >2^r-1$.}
	\end{figure}
	
	\begin{figure}[H]
		\centering
		\tikzset{every picture/.style={line width=0.75pt}} 
		\begin{tikzpicture}[scale=1.5]
			\draw [ line width=0.5mm, fill=black!8!white, draw=black] plot [smooth cycle] coordinates {(-2,1) (-2,-1) (0,-2) (2,-2) (6,0) (4,2) (0,2) }--cycle;
			\draw [ line width=0.5mm, fill=white, draw=black] plot [smooth cycle] coordinates {(-1.5,1) (-1,-1.2) (1,-1.5) (5.5,0) (3.5, 1.8) }--cycle;
			\draw [ line width=0.5mm, fill=black!8!white, draw=black] plot [smooth cycle] coordinates {(-0.75,-1) (1,-1) (2.8,0) (1.5, 1.2) (-1,1) }--cycle;
			\draw [ line width=0.5mm, fill=white, draw=black] plot [smooth cycle] coordinates {(-0.55,-0.7) (0.75,-0.75) (1,0)  (-0.5,0.5) }--cycle;
			\draw [ line width=0.5mm, fill=white, draw=black] plot [smooth cycle] coordinates {(1.5,-0.4)  (2.5, 0.2) (1.7,0.8) (0,0.7) (1, 0.4) }--cycle;
			\begin{scope}[shift={(0,-0.3)}]
				\draw [ line width=0.5mm, pattern=dots, draw=black] plot [smooth cycle] coordinates {(0,-0.2) (0.25,-0.25) (0.5,0) (0.3,0.3) (-0.1,0.2) }--cycle;
				\draw [ line width=0.5mm, pattern=dots, draw=black] plot [smooth cycle] coordinates {(3.5,0) (3.75,0.2) (4,0.8) (3.5, 1)  }--cycle;
				\draw (3.7,0.6) node[scale=1.5] {\large $\overline{\gamma}'$};
				\draw (0.2,0) node[scale=1.5] {\large $\overline{\gamma}'$};
			\end{scope}
			
			\begin{scope}[shift={(1.5,0.2)}]
				\draw [ line width=0.5mm, pattern=dots, draw=black] plot [smooth cycle] coordinates {(0,-0.2) (0.25,-0.25) (0.5,0) (0.3,0.3) (-0.1,0.2) }--cycle;
				\draw (0.2,0) node[scale=1.5] {\large $\overline{\gamma}'$};
			\end{scope}
			\draw (0.9,1) node[scale=1.5] {\large $\overline{\gamma}$};
			\draw (1.1,-1.8) node[scale=1.5] {\large $\overline{\gamma}$};
		\end{tikzpicture}
		
		\caption{To illustrate how Condition \textbf{(A)} works, consider the figure above. In this case, the connected components of $\overline{\gamma}'$ are dotted and the connected components of $\overline{\gamma}$ are grey. One can readily see that there is a connected component   of $\overline{\gamma}$ that has a nonempty intersection with $V(\overline{\gamma})$ but does not fully contain it. In order to fix such problem, one should separate $\overline{\gamma}'$ in three different sets.}
	\end{figure}
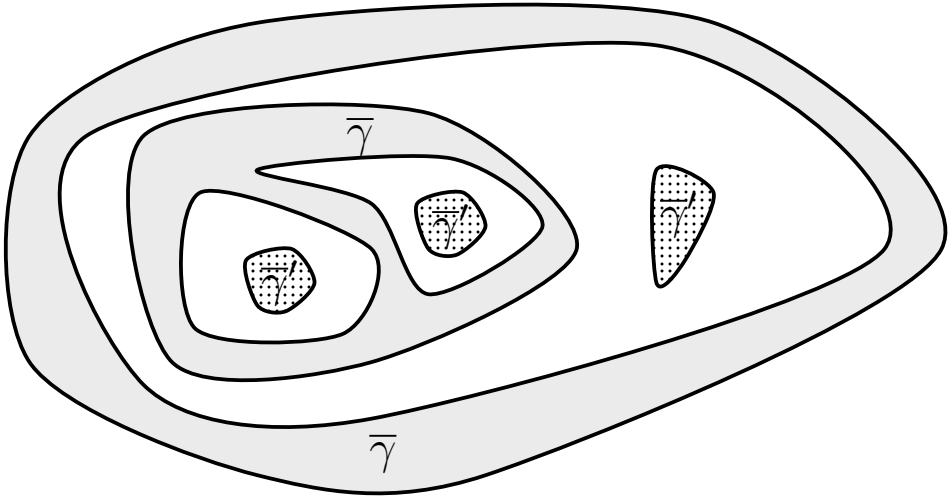
	
	In general, there are many possible $(M,a,r)$-partitions given fixed $M,a,r>0$. If we have two $(M,a,r)$-partitions $\Gamma$ and $\Gamma'$, we say that $\Gamma$ \emph{is finer than} $\Gamma'$ if for every $\overline{\gamma} \in \Gamma$ there is $\overline{\gamma}' \in \Gamma'$ with $\overline{\gamma} \subseteq \overline{\gamma}'$.
	\begin{proposition}
		For each finite $\partial \sigma$ there is a finest $(M,a,r)$-partition. 
	\end{proposition}
	\begin{proof}
		The existence of non-trivial $(M,a,r)$-partitions will be given in the next section. Consider two $(M,a,r)$-partitions for $\partial\sigma$, namely $\Gamma$ and $\Gamma'$. We will show that we can build a $(M,a,r)$-partition that is finer than $\Gamma$ and $\Gamma'$. Define the set
		\[
		\Gamma\cap\Gamma' \coloneqq \{\gamma \cap \gamma': \gamma \in \Gamma, \gamma' \in \Gamma', \gamma \cap \gamma' \neq \emptyset\}.
		\]
		We will prove that $\Gamma \cap \Gamma'$ is a $(M,a,r)$-partition. It is easy to see that $\Gamma\cap \Gamma'$ forms a partition of $\partial \sigma$. Let $\overline{\gamma}_i \in \Gamma$ and $\overline{\gamma}'_i \in \Gamma'$ such that $\overline{\gamma}_i\cap \overline{\gamma}_i' \in \Gamma \cap \Gamma'$ for $i=1,2$. Then, by Condition \textbf{(A)}, there is a connected component $A$ of $\overline{\gamma}_2^c$ that contains $\overline{\gamma}_1$. Hence,
		\[
		\overline{\gamma}_1 \cap \overline{\gamma}'_1\subset A \subset (\overline{\gamma}_2 \cap \overline{\gamma}'_2)^c.
		\]  
		For condition \textbf{(B)}, consider $\overline{\gamma}_i\cap \overline{\gamma}'_i \in \Gamma \cap \Gamma'$, $i=1,2$. For the family $(\overline{\gamma}'_{i,k})_{1 \leq k \leq n'}$ given by condition \textbf{(B)} take $(\overline{\gamma}_i\cap\overline{\gamma}_i')_k \coloneqq \overline{\gamma}_i\cap \overline{\gamma}'_{i,k}$ when the intersection is not empty. We have,
		\begin{align*}
			d(\overline{\gamma}_1\cap\overline{\gamma}'_1,\overline{\gamma}_2\cap\overline{\gamma}'_2) \geq d(\overline{\gamma}'_1, \overline{\gamma}'_2) &> M \min\left \{\underset{1\leq k \leq n'}{\max}\diam(\overline{\gamma}'_{1,k}),\underset{1\leq j\leq m'}{\max}\diam(\overline{\gamma}'_{2,j})\right\}^a \\
			&> M \min\left \{\underset{1\leq k \leq n'}{\max}\diam((\overline{\gamma}_1\cap\overline{\gamma}_1')_k),\underset{1\leq j\leq m'}{\max}\diam((\overline{\gamma}_2\cap\overline{\gamma}_2')_j)\right\}^a.
		\end{align*}
		
		Since there are a finite number of $(M,a,r)$-partitions, we can construct the finest one by intersecting all of them, following the above construction. If we had two $(M,a,r)$-partitions that have the property of being the finest, we could produce another $(M,a,r)$-partition finer than both, yielding a contradiction. 
	\end{proof}
	\begin{remark}
		In their one-dimensional paper \cite{Fro1}, Fr\"{o}hlich and Spencer assumed that they could choose the finest partition of spin flips satisfying their Condition D, but the uniqueness of this partition would not be a problem, since you just need to fix a partition for the phase transition argument to hold. It was in \cite{Imbr1} where Imbrie settled this question and our proof is inspired by his argument.
	\end{remark}
 
	\subsection{Discussion about contours on long-range Ising models}
		
	Differently from the original papers of Fr\"{o}hlich and Spencer, we have no arithmetic condition over the $(M,a,r)$-partitions, which means we do not ask the sum of the spins should be zero over the support of the contours. In particular, there are no constraints over the size of the contours they could have an even or odd number of points of $\Z^d$. The definition is stated only in terms of the distances among subsets of $\Z^d$. In order to control the entropy of contours, we needed to introduce a parameter $r$. It is worthwhile to stress that in the original works of Fr\"{o}hlich and Spencer \cite{Fro3, Fro1, Fro2} the parameter $r=1$.  
	
	The exponent $a$ plays an important role in our arguments.  When $\alpha$ is close to $d$. Consequently, it becomes necessary for the contours to maintain a greater distance from one another to ensure that their mutual influence remains negligible. In the original papers of Fr\"{o}hlich and Spencer, $a$  is chosen as a fixed number or to belong in a finite interval as follows.  In the first paper \cite{Fro3}, for continuous spin bidimensional models, they choose $\frac{3}{2} < a <2$, while in the one dimensional paper \cite{Fro2} for the long-range Ising model with $\alpha =2$, they fix $a=\frac{3}{2}$.  In section 4 of the paper on multidimensional random Schr\"{o}dinger operators \cite{Fro2},  they assumed $1 \leq a <2$. These papers use the idea of multiscale analysis to study different problems, which strongly inspired us in the definition and construction of our contours. A brief summary mentioning the multiscale methods and their power is in \cite{Simon} and at the references therein.
	
	 Cassandro, Ferrari, Merola, Presutti \cite{Cass} defined the contours without any arithmetic condition that depends only on the distance among the subsets of  $\Z^d$.  They choose $a = 3$ for unidimensional long-range Ising models with $\alpha \in (2-\alpha^+, 2]$, where $\alpha^+=\log(3)/\log(2) - 1$.  This is an important point; while our arguments work for any $2\leq d<\alpha$,  Littin and Picco proved that in the unidimensional case, it is impossible to produce a direct proof of the phase transition using Peierls' argument and definition of contour in \cite{Cass}, although they prove the phase transition for the entire region with $\alpha \in (1,2]$. The papers  \cite{Cass} and \cite{LP} also make the extra assumption $J(1) \gg 1$, which means that the nearest neighbor interaction should be large. Recently it was proved that this extra assumption can be relaxed \cite{Eric2, Kim}.

	The following proposition guarantees the existence of a $(M,a,r)$-partition for each configuration $\sigma$ with a finite boundary.
	\begin{proposition}\label{prop1}
		Fix real numbers $M,a,r>0$. For every $\sigma \in \Omega$ with finite boundary there is a $(M,a,r)$-partition $\Gamma(\sigma)$. 
	\end{proposition}
	\begin{proof}
		For each $x \in \Z^d$ and $n\geq 0$ we define a $n$-cube $C_n(x) \subset \Z^d$ as
		\begin{equation}\label{cubes}
			C_n(x) \coloneqq \left(\prod_{i=1}^d[2^{n-1}x_i-2^{n-1},2^{n-1}x_i+2^{n-1}]\right)\cap \Z^d.
		\end{equation}
		These cubes have side length $2^n$ and center at the point $2^{n-1} x$. For $n=0$, we adopt the convention that $C_{0}(x) = x$, for any point $x \in \Z^d$. For each $\Lambda \Subset \Z^d$ and $n\geq 0$, we define $\mathscr{C}_n(\Lambda)$ as a minimal cover of $\Lambda$ by $n$-cubes. For each cover $n\geq 0$, we define the graph $G_n(\Lambda) = (v(G_n(\Lambda)),e(G_n(\Lambda)))$ by $v(G_n(\Lambda))\coloneqq\mathscr{C}_n(\Lambda)$ and $e(G_n(\Lambda))\coloneqq \{(C_n(x),C_n(y)): \dis(C_n(x),C_n(y))\leq Md^a2^{an}\}$.
		
		Note that $d2^n$ is the diameter in the $\ell^1$-norm of a $n$-cube. Let $\mathscr{G}_n(\Lambda)$ be the set of all connected components of the graph $G_n(\Lambda)$ and, for each $G \in \mathscr{G}_n(\Lambda)$, define
		\[
		\gamma_G = \bigcup_{C_n(x) \in v(G)} (\Lambda \cap C_n(x)).
		\]
		We are ready to establish the existence of an $(M,a,r)$-partition for the boundary of a configuration $\partial \sigma$. Set $\partial \sigma_0:=\partial\sigma$ and 
		\[
		\mathscr{P}_0:=\{G \in \mathscr{G}_0(\partial \sigma_0) : |v(G)|\leq 2^r-1\}.
		\]
		Notice that this set separates all points that are distant by at least $M d^a$. Define inductively, for $n\geq 1$,  
		\[
		\mathscr{P}_n:=\{G \in \mathscr{G}_n(\partial \sigma_n): |v(G)|\leq 2^r-1\},
		\]
		where $\partial\sigma_n:=\partial\sigma_{n-1}\setminus \underset{G \in \mathscr{P}_{n-1}}{\bigcup}\gamma_G$, for $n\geq1$. Since the $n$-cubes invade the lattice, when we continue increasing $n$, there exists $N\geq 0$ such that $\partial\sigma_n = \emptyset$ for every $n \geq N$. In this case, we define $\mathscr{P}_n = \emptyset$. Let $\mathscr{P} = \underset{n\geq 0}{\bigcup}\mathscr{P}_n$. We are going to show that the family $\Gamma(\sigma)\coloneqq \{\gamma_G: G \in \mathscr{P}\}$ is a  $(M,a,r)$-partition. 
		
		In order to show that Condition \textbf{(B)} is satisfied, we will construct families of subsets, with at most $2^r-1$ elements, where Inequality \eqref{B_distance} is verified. We will write $G_n \coloneqq G_n(\partial \sigma_n)$ to simplify our notation. Take distinct $\gamma_G, \gamma_{G'} \in \Gamma(\sigma)$. There are positive integers $n,m$, with $n \leq m$, such that $G \in \mathscr{P}_n$ and $G' \in \mathscr{P}_m$. Let $G''$ be the subgraph of $G_n$ such that the cubes of $v(G'')$ covers $\gamma_{G'}$ and it is minimal in the sense that all other subgraphs $G'''$ of $G_n$ satisfying this property have $v(G''')\geq v(G'')$. Thus, defining $\gamma_k = \gamma_G\cap C_n(x_k)$, for each $C_n(x_k) \in v(G)$, we have,
		\be\label{eq_distance}
		\dis(\gamma_G, C_n(z)) = \min_{1 \leq k \leq |v(G)|}\dis(\gamma_k, C_n(z)) \geq \min_{1 \leq k \leq |v(G)|} \dis(C_n(x_k),C_n(z)),
		\ee
		for each $C_n(z)\in v(G'')$. There is no edge between the subgraph $G''$ and the connected component $G$, by construction. 
		Thus, 
		\[
		\dis(C_n(x_k),C_n(z)) > Md^a 2^{an}.
		\]
		Consider, also, the sets $\gamma'_j = \gamma_{G'}\cap C_m(y_j)$, where $C_m(y_j) \in v(G')$. Then
		\be
		\dis(\gamma_G, \gamma_{G'})= \min_{1 \leq j \leq |v(G')|} \min_{\substack{C_n(z) \in v(G'') \\ C_n(z) \cap \gamma'_j \neq \emptyset}}\dis(\gamma_G,\gamma'_j \cap C_n(z)) \geq \min_{ C_n(z) \in v(G'')} \dis(\gamma_G, C_n(z)).
		\ee
		Using Inequality \eqref{eq_distance}, we arrive at the inequality $\dis(\gamma_G,\gamma_{G'})> Md^a 2^{an}$. Note that, by construction, both $\{\gamma_k\}_{1\leq k \leq |v(G)|}$ and $\{\gamma'_j\}_{1\leq j \leq |v(G')|}$ satisfy, respectively, 
		
		\[
		\max_{1\leq k \leq |v(G)|}\diam(\gamma_k) \leq d2^n \quad \text{and} \quad \max_{1\leq j \leq |v(G')|}\diam(\gamma'_j) \leq d2^m.
		\] 
		Since we assumed that $n \leq m$, we get
		\[
		\min\left\{\max_{1\leq k \leq |v(G)|}\diam(\gamma_k) ,\max_{1\leq j \leq |v(G')|}\diam(\gamma'_j)\right\}^a \leq d^a 2^{an},
		\]
		and we proved that the family $\Gamma(\sigma)$ satisfy Condition \textbf{(B)}. 
		
		In order to establish Condition \textbf{(A)}, we first note that the equality $\partial \sigma = \bigcup_{G \in \mathscr{P}}\gamma_G$ follows by construction. The elements of $\Gamma(\sigma)$ are pairwise disjoint since Inequality \eqref{B_distance} is satisfied.
		
		Let $\gamma_G, \gamma_{G'} \in \Gamma(\sigma)$, with $V(\gamma_G)\cap V(\gamma_{G'})\neq \emptyset$. There are positive integers $n,m$ satisfying $n\leq m$ such that $G \in \mathscr{P}_n$ and $G' \in \mathscr{P}_m$. Consider $G''$, as before, the minimal subgraph of $G_n$ that covers $\gamma_{G'}$. Since $\gamma_G \cap \gamma_{G'}=\emptyset$ it holds $\gamma_G \subset (\gamma_{G'})^c$. We will show that $\gamma_G$ must be contained in only one connected component of $(\gamma_{G'})^c$. Every $n$-cube $C_n(x) \in v(G)$ cannot have a nonempty intersection with $\gamma_{G'}$, since the latter is covered by the $n$-cubes in $v(G'')$ and there is no edge between $G$ and $G''$. This is sufficient to conclude that each $n$-cube in $v(G)$ is in only one connected component of $(\gamma_{G'})^c$.
		
		
		If $(\gamma_{G'})^c$ have only one connected component or $|v(G)|=1$, there is nothing to prove. Suppose, by contradiction, that there exist two $n$-cubes $C_n(x), C_n(x') \in v(G)$ in different connected components of $(\gamma_{G'})^c$. We claim 
		\be\label{eq1111}
		\dis(C_n(x), C_n(x')) \geq 2M d^a 2^{an}.
		\ee
		Indeed, take two points $z \in C_n(x)$ and $z' \in C_n(x')$ such that $\dis(C_n(x), C_n(x'))=|z-z'|$. Let $\lambda_{z,z'}$ be a minimal path in $\Z^d$ starting at $z$ and ending at $z'$. Note that $|\lambda_{z,z'}| = |z-z'|$. Since $C_n(x)$ and $C_n(x')$ are in different connected components of $(\gamma_{G'})^c$ there must exist $y \in \lambda_{z,z'} \cap \gamma_{G'}$. We can break $\lambda_{z,z'}$ as the union of minimal paths $\lambda_{z,y}$ and $\lambda_{y,z'}$. This fact implies
		\begin{align*}
			\dis(C_n(x), C_n(x')) &= |z-y| + |y-z'|\nonumber \\
			&\geq \min_{y' \in \gamma_{G'}}[|z-y'| + |y'-z'|] \nonumber\\
			&\geq \min_{y' \in \gamma_{G'}}[\dis(C_n(x),y') + \dis(C_n(x'),y')] \nonumber\\
			&\geq \min_{C_n(z) \in v(G'')}[\dis(C_n(x),C_n(z)) + \dis(C_n(x'),C_n(z))]\\
			& \geq 2M d^a 2^{an},
		\end{align*}
		where the last inequality is due to the fact that the subgraphs $G$ and $G''$ have no edge between them. Inequality \eqref{eq1111} is valid for any pair of $n$-cubes in different connected components of $(\gamma_{G'})^c$, thus our discussion implies that $C_n(x)$ and $C_n(x')$ are vertices of two different connected components. This cannot happen since $G$ is connected, arriving at a contradiction. 
	\end{proof}
	In order to define the label of a contour, we must be careful since the inner boundary of a set $\overline{\gamma}$ may have different signs; see the figure below. 
	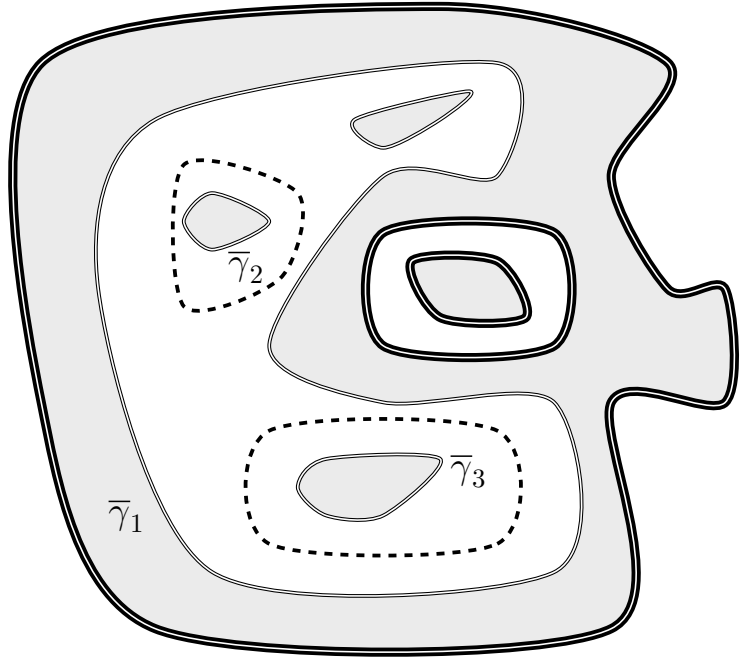
\begin{figure}[H]
		\centering
		
		\tikzset{every picture/.style={line width=0.75pt}} 
		
		\begin{tikzpicture}[scale=1.5]
			\draw [line width=0.5mm,fill=black!8!white,double] plot [smooth cycle]  coordinates {(0,-1) (4,-1) (4,1) (5,1) (5,2) (4.5,2) (4,3) (4.5,4) (2.5,4.5) (-1,4) (-1,1)}--cycle;
			\draw [line width=0.5mm,fill=white, thin, double] plot [smooth cycle] coordinates {(0.5,-0.5) (3.5,-0.5) (3.5,1) (2,1) (1, 1.5) (2,3) (3,3) (3,4) (0,3.5) (-0.5,2)}--cycle;
			\draw [line width=0.5mm,fill=white, double] plot [smooth cycle] coordinates {(2,1.5) (3.5,1.5) (3.5,2.5) (2,2.5)}--cycle;
			\draw [line width=0.5mm,fill=black!8!white, double] plot [smooth cycle] coordinates {(2.5,1.75) (3.25,1.75) (3,2.25) (2.25,2.25)}--cycle;
			\draw [line width=0.5mm,fill=black!8!white, thin, double] plot [smooth cycle] coordinates {(2,3.25) (2.5,3.5) (2.75,3.75) (1.75,3.5)}--cycle;
			
			\begin{scope}[shift={(0,0.1)}]
				\draw [ line width=0.5mm, draw=black, dashed] plot [smooth cycle] coordinates {(0.25,1.75) (1.1,2) (1.25,2.75) (0.25,3) }--cycle;
				\draw [line width=0.5mm, fill=black!8!white, thin, double] plot [smooth cycle] coordinates {(0.5,2.25) (1,2.5) (0.5,2.75) (0.25,2.5)}--cycle;
			\end{scope}
			\draw [line width=0.5mm, draw=black, dashed] plot [smooth cycle] coordinates {(1,-0.25) (3,-0.25) (3,0.75) (1,0.75) }--cycle;
			\draw [line width=0.5mm, fill=black!8!white, thin, double] plot [smooth cycle] coordinates {(1.5,0) (2,0) (2.5,0.5) (1.5,0.5) (1.25,0.25) }--cycle;
			
			\draw (-0.25,0) node[scale=1] {\Large $\overline{\gamma}_1$};
			\draw (0.8,2.2) node[scale=1] {\Large $\overline{\gamma}_2$};
			\draw (2.75,0.4) node[scale=1] {\Large $\overline{\gamma}_3$};
			
		\end{tikzpicture}
		\caption{An example of $\Gamma(\sigma)=\{\overline{\gamma}_1,\overline{\gamma}_2, \overline{\gamma}_3\}$, with $\overline{\gamma}_1$ having regions in the  inner boundary with different signs. In the figure, the grey region represents the incorrect points, and the thin and thick border corresponds to, respectively, $+1$ and $-1$ labels.}
	\end{figure}
	
 A connected component $\overline{\gamma}^{(k)}$ of $\overline{\gamma}\in \Gamma(\sigma)$ is called \emph{external} if for any other connected component $\overline{\gamma}^{(k')}$ with $V(\overline{\gamma}^{(k')})\cap V(\overline{\gamma}^{(k)}) \neq \emptyset$ we have $V(\overline{\gamma}^{(k')})\subset V(\overline{\gamma}^{(k)})$.
	\begin{lemma}
		Any configuration $\sigma$ is constant on $\din V(\overline{\gamma})$, for each $\overline{\gamma} \in \Gamma(\sigma)$. 
	\end{lemma}
	\begin{proof}
		Note that $V(\overline{\gamma}) $ is the union of $V(\overline{\gamma}^{(k)})$ for all its external connected components $ \overline{\gamma}^{(k)}$. Suppose there are $\overline{\gamma}^{(k)},\overline{\gamma}^{(j)}$ external connected components of $\overline{\gamma}$ such that the sign of $\sigma$ on  $\din V(\overline{\gamma}^{(k)})$ is different from the sign on $\din V(\overline{\gamma}^{(j)})$. Then, the configuration $\sigma$ must change its sign inside the set $V(\overline{\gamma}^{(k)})^c\cap V(\overline{\gamma}^{(j)})^c$. Since $\sigma$ is constant outside some finite set $\Lambda$, either $\overline{\gamma}^{(k)}$ or $\overline{\gamma}^{(j)}$ must be surrounded by a different region of incorrect points, let us call it $\overline{\gamma}^{(l)}$. We can assume that the connected component surrounded by $\overline{\gamma}^{(l)}$ is $\overline{\gamma}^{(k)}$. The set $\overline{\gamma}^{(l)}$ cannot be a connected component  of $\overline{\gamma}$, otherwise the set $\overline{\gamma}^{(k)}$ would not be external. If $\overline{\gamma}^{(l)}$ is a connected component of another element $\overline{\gamma}' \in \Gamma(\sigma)$, then $\overline{\gamma}$ has a nonempty intersection with at least two connected components of $(\overline{\gamma}')^c$, contradicting Condition \textbf{(A)}. 
	\end{proof}
	
	The label of $\overline{\gamma}$ is defined similarly as in Chapter 1.
 
	\begin{definition}
		Given a configuration $\sigma$ with finite boundary, its \emph{contours} $\gamma$ are pairs $(\overline{\gamma},\lab_{\overline{\gamma}})$,  where $\overline{\gamma} \in \Gamma(\sigma)$ and $\lab_{\overline{\gamma}}$ is the label function defined previously. The \emph{support of the contour}  $\gamma$ is defined as $\Sp(\gamma)\coloneqq \overline{\gamma}$ and its \emph{size} is given by $|\gamma| \coloneqq |\Sp(\gamma)|$.
	\end{definition}

	Another important concept for our analysis of phase transition is the interior of a contour. The following sets will be useful 
	\[
	\I_\pm(\gamma) = \hspace{-1cm}\bigcup_{\substack{k \geq 1, \\ \lab_{\overline{\gamma}}(\I(\Sp(\gamma))^{(k)})=\pm 1}}\hspace{-1cm}\I(\Sp(\gamma))^{(k)} , \;\;\;
	\I(\gamma) = \I_+(\gamma) \cup \I_-(\gamma), \;\;\;
	V(\gamma) = \Sp(\gamma) \cup \I(\gamma),
	\]
	where $\I(\Sp(\gamma))^{(k)}$ are the connected components of $\I(\Sp(\gamma))$. Notice that the interior of contours in Pirogov-Sinai theory are at most unions of simple connected sets. In our case, they are only connected, i.e., they may have holes. There is no bijective correspondence between our contours and configurations. Usually, there is more than one configuration giving the same boundary set. Also, it is not true that for all families of contours $\Gamma \coloneqq \{\gamma_1, \gamma_2, \dots, \gamma_n\}$ there is a configuration $\sigma$ whose contours are exactly $\Gamma$. This happens because they may not form a $(M,a,r)$-partition, and, even if this is the case, their labels may not be compatible. When such a configuration exists, we say that the family of contours $\Gamma$ is \emph{compatible}. 
	
	\begin{figure}[H]
		\centering
		
		\tikzset{every picture/.style={line width=0.75pt}} 
		
		\begin{tikzpicture}[scale=0.6]
			\draw [line width=0.5mm,fill=black!8!white,  thin, double] plot [smooth cycle] coordinates {(0,-1) (4,-1) (4,1) (2.5,2) (2,3.5) (-1,4) (-1,1)}--cycle;
			\draw [line width=0.5mm,fill=white, double] plot [smooth cycle] coordinates {(1.5,0) (2,2) (0.5,2.5) (0,1)}--cycle;
			\draw [line width=0.5mm,fill=black!8!white, double] plot [smooth cycle] coordinates {(1,1) (1.5,1.5) (1,2) (0.5,1.5)}--cycle;
			
			\draw [line width=0.5mm,fill=black!8!white,  thin, double] plot [smooth cycle] coordinates {(5,4) (7,4) (7,8) (8,8) (8,10) (6,10) (6,11) (3,11) (3, 7)}--cycle;
			
			\draw [line width=0.5mm,fill=black!8!white,  thin, double] plot [smooth cycle] coordinates {(22,4) (24,4) (24,8) (25,8) (25,10) (23,10) (23,13) (19,13) (18,10) (16,6) (16,5)}--cycle;
			\draw [line width=0.5mm,fill=white, double] plot [smooth cycle] coordinates {(17,6) (23,5) (23,7) (20,8) }--cycle;
			\draw [line width=0.5mm,fill=black!8!white,  thin, double] plot [smooth cycle] coordinates {(20,6.25) (21,6.25) (21,6.9) (20,6.9) }--cycle;
			\draw [line width=0.5mm,fill=white, thin, double] plot [smooth cycle] coordinates {(19,10) 
				(22,10) (22,12) (20,12) }--cycle;
			\draw [line width=0.5mm,fill=black!8!white, double] plot [smooth cycle] coordinates {(20.25,10.75) (21.25,10.75) (21.25,11.4) (20.25,11.4) }--cycle;
			
			\draw (0,0) node[scale=0.75] {\Huge$\gamma_1$};
			\draw (1,1.5) node[scale=0.75] {\large$\gamma_2$};
			\draw (5,7) node[scale=0.75] {\Huge$\gamma_3$};
			\draw (22,8.5) node[scale=0.75] {\Huge$\gamma_4$};
			\draw (20.5,6.5) node[scale=0.75] {\Large$\gamma_5$};
			\draw (20.8,11) node[scale=0.75] {\Large$\gamma_6$};
		\end{tikzpicture}
		\caption{Above we have two situations where incompatibility happens. In the first case, we have $\gamma_1$ and $\gamma_2$ two contours that are close, thus they should not be separated. In the case $\gamma_4,\gamma_5,\gamma_6$ we have the usual problem of labels not matching.}
	\end{figure}
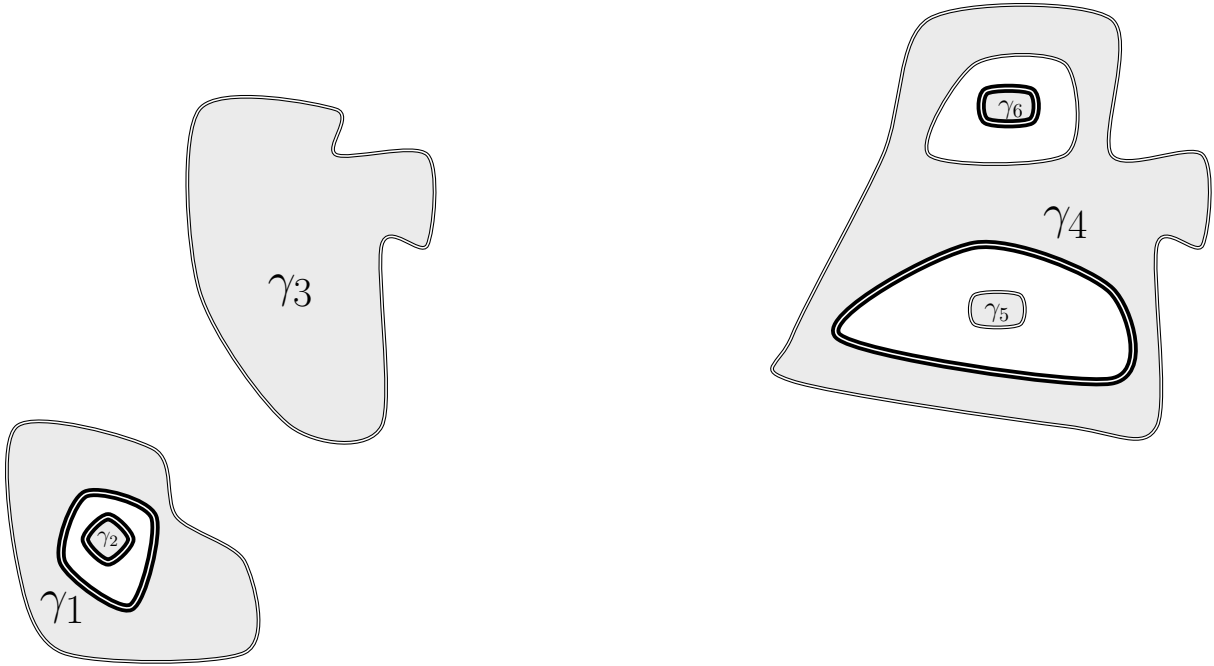
	
	\subsection{Entropy Bounds}
	
	The proofs in this section are highly inspired by section 4 of Fr\"{o}hlich-Spencer \cite{Fro2}, the one-dimensional case studied in Fr\"{o}hlich-Spencer \cite{Fro1} and by Cassandro, Ferrari, Merola, Presutti \cite{Cass}. We say that a contour $\gamma$ in a family $\Gamma$ is \emph{external} if its external connected components are not contained in any other $V(\gamma')$, for $\gamma' \in \Gamma\setminus\{\gamma\}$. As in the previous Chapter, we say that contour is a \emph{$+$- contour} (resp. \emph{- contour}) if the label of $\lab_{\overline{\gamma}}(\overline{\gamma}^{(0)})=+1$ (respectively $-1$). For each $\Lambda \Subset \Z^d$, define the set of all external compatible families of contours $\Gamma$ with label $\pm$ contained in $\Lambda$ by
	\[
	\mathcal{E}^\pm_\Lambda \coloneqq\{\Gamma= \{\gamma_1, \ldots, \gamma_n\}: \Gamma \text{ is compatible,} \gamma_i \text{ is external}, \lab(\gamma_i)=\pm1, V(\Gamma) \subset \Lambda\},
	\]
	where $V(\Gamma)=\bigcup_{1 \leq i\leq n}V(\gamma_i)$. When we write $\gamma \in \mathcal{E}^\pm_\Lambda$ we mean $\{\gamma\} \in \mathcal{E}^\pm_\Lambda$. To hold a Peierls-type argument, we need to find an upper bound of the number of contours with fixed size $|\gamma|$ containing a given point, meaning $0\in V(\gamma)$. To do so, we will need some auxiliary results. 
	
	To better estimate the entropy of the contours coming from a $(M,a,r)$-partition we must change the combinatorial arguments accordingly, since now they may be disconnected. Let us denote by $\mathscr{C}_n$ an arbitrary collection of $n$-cubes. For $n,m\ge 0$ with $n\le m$, we say that  
	$\mathscr{C}_{n}$ is \emph{subordinated} to $\mathscr{C}_{m}$, denoted by $\mathscr{C}_n \preceq \mathscr{C}_m$, if $\mathscr{C}_m$ is a minimal cover of $\cup_{C_n(x)\in \mathscr{C}_n}C_n(x)$.
	For each $n,m\geq 1$, with $n\leq m$, define
	\[
	N(\mathscr{C}_m, V_n) \coloneqq |\{\mathscr{C}_{n}: \mathscr{C}_n \preceq \mathscr{C}_m, |\mathscr{C}_n|=V_n, \mathscr{C}_m \}|
	\]
	to be the number of collections of $n$-cubes $\mathscr{C}_n$ subordinated to a given $\mathscr{C}_m$ such that $|\mathscr{C}_n| = V_n$. 
	The following two propositions are straightforward generalizations of Theorem 4.2 and Proposition 4.3 of \cite{Fro3}. The original proofs correspond to our case when $r=1$.
	
	\begin{proposition}\label{appB:prop1}
		Let $r,n\geq 1$ be integers, and $\mathscr{C}_{rn}$ be a fixed collection of $rn$-cubes. Then there exists a constant $c\coloneqq c(d,r)>0$ such that
		\be\label{eq_appB}
		N(\mathscr{C}_{rn}, V_{r(n-1)}) \leq e^{c V_{r(n-1)}}.
		\ee
	\end{proposition}
	\begin{proof}
		For each $rn$-cube $C_{rn}(x) \in \mathscr{C}_{rn}$, let $N_{C_{rn}(x)}$ be the number of cubes in a collection of $r(n-1)$-cubes $\mathscr{C}_{r(n-1)}$ that are covered by $C_{rn}(x)$. Fix $(n_{C_{rn}(x)})_{C_{rn}(x) \in \mathscr{C}_{rn}}$, with $n_{C_{rn}(x)}\geq 1$, an integer solution to the inequality
		\be\label{100}
		V_{r(n-1)}\leq \sum_{C_{rn}(x) \in \mathscr{C}_{rn}}n_{C_{rn}(x)}
		\leq 2dV_{r(n-1)},
		\ee
		and define $N(\mathscr{C}_{rn}, V_{r(n-1)}, (n_{C_{rn}(x)})_{C_{rn}(x) \in \mathscr{C}_{rn}})$ be the number of collections $\mathscr{C}_{r(n-1)}$ of $r(n-1)$-cubes subordinated to $\mathscr{C}_{rn}$ such that $|\mathscr{C}_{r(n-1)}|=V_{r(n-1)}$ and $N_{C_{rn}(x)}=n_{C_{rn}(x)}$ for each $C_{rn}(x) \in \mathscr{C}_{rn}$. We get
		\[
		N(\mathscr{C}_{rn}, V_{r(n-1)}) = \sum_{(n_{C_{rn}(x)})_{C_{rn}(x) \in \mathscr{C}_{rn}}}N(\mathscr{C}_{rn}, V_{r(n-1)}, (n_{C_{rn}(x)})_{C_{rn}(x) \in \mathscr{C}_{rn}}).
		\]
		The number of positions that a $r(n-1)$-cube can sit inside a $rn$-cube is at most $(2^{r+1}-1)^d$, and the number of subordinated minimal coverings $\mathscr{C}_{r(n-1)}$ with a given $N_{C_{rn}(x)}=n_{C_{rn}(x)}$ is at most $\binom{(2^{r+1}-1)^d}{n_{C_{rn}(x)}}$. Hence 
		\[
		N(\mathscr{C}_{rn}, V_{r(n-1)}, (n_{C_{rn}(x)})_{C_{rn}(x) \in \mathscr{C}_{rn}}) \leq \prod_{C_{rn}(x) \in \mathscr{C}_{rn}}\binom{(2^{r+1}-1)^d}{n_{C_{rn}(x)}}.
		\]
		The number of solutions to \eqref{100} is bounded by $2^{2dV_{r(n-1)}+1}$, concluding that Inequality (\ref{eq_appB}) holds for $c = (2d+1)\log(2)+ d\log(2^{r+1}-1)$.
	\end{proof}
	Given a subset $\Lambda \Subset \Z^d$ and integers $r\geq 1$ and $n\geq 0$, define the \emph{total volume} by
	\be
	V_r(\Lambda)= \sum_{n=0}^{n_r(\Lambda)} |\mathscr{C}_{rn}(\Lambda)|,
	\ee  
	where $n_r(\Lambda) = \lceil \log_{2^r}(\diam(\Lambda))\rceil$ and $\mathscr{C}_{rn}(\Lambda)$ is a \emph{minimal cover} of $\Lambda$ with $rn$-cubes $C_{rn}(x)$ defined in (\ref{cubes}). Observe that $|\mathscr{C}_{0}(\Lambda)|=|\Lambda|$. 
	Let $V\geq 1$ be a positive integer and $\mathscr{F}_V$ be the set defined by
	\be
	\mathscr{F}_V \coloneqq \{\Lambda \Subset \Z^d : V_r(\Lambda) =V, 0 \in \Lambda\}.
	\ee
	By using Proposition \ref{appB:prop1}, let us show that the number of elements in $\mathscr{F}_V$ is exponentially bounded by $V$.
	\begin{proposition}\label{appB:prop2}
		There exists $b\coloneqq b(d,r) >0$ such that,
		\be\label{bound_F_V}
		|\mathscr{F}_V| \leq e^{bV}.
		\ee
	\end{proposition}
	\begin{proof}
		For each $\Lambda \in \mathscr{F}_V$, there is a unique family of minimal covers $(\mathscr{C}_{rn}(\Lambda))_{0 \leq n \leq n_r(\Lambda)}$, since the minimal cover $\mathscr{C}_0$ characterizes the set $\Lambda$. Moreover, the minimal covers $\mathscr{C}_{rn}(\Lambda)$ can always be chosen in a way that $\mathscr{C}_{rn_1}$ is subordinated to $\mathscr{C}_{rn_2}$ whenever $n_1 \leq n_2$, since, in order to compute the total volume $V_r(\Lambda)$, we only need to know the size of each minimal cover $\mathscr{C}_{rn}(\Lambda)$. 
		Fix $(V_{rn})_{0 \leq n \leq n_r(\Lambda)-1}$ a solution to the equation
		\be\label{constraint}
		\sum_{n=0}^{n_r(\Lambda)-1} V_{rn}= V-1. 
		\ee
		We can estimate $|\mathscr{F}_V|$ by counting the number of families $(\mathscr{C}_{rn}(\Lambda))_{0 \leq n \leq n_r(\Lambda)}$ where the last cover $\mathscr{C}_{rn_r(\Lambda)}$ consists of a unique cube $C_{rn_r(\Lambda)}(x)$ containing $0$. Let $\mathscr{F}_{V,m}=\{\Lambda \in \mathscr{F}_V:  n_r(\Lambda)=m\}.$ Then,
		\[
		|\mathscr{F}_V| \leq \sum_{m=1}^{V}|\mathscr{F}_{V,m}|,
		\]
		since $n_r(\Lambda)\leq V_r(\Lambda)=V$. Now,
		\begin{align*}
			|\mathscr{F}_{V,m}| &= \sum_{(V_{rn})_{0 \leq n \leq m-1}}|\{(\mathscr{C}_{rn}(\Lambda))_{0 \leq n \leq m}: \mathscr{C}_{rn}\preceq \mathscr{C}_{r(n+1)}, |\mathscr{C}_{rn}|=V_{rn}, 0 \in C_{rm}(x)\}| \\
			&=\sum_{(V_{rn})_{0 \leq n \leq m-1}}\sum_{C_{rm}(x) \ni 0}|\{\mathscr{C}_{rn}(\Lambda)\}_{0 \leq n \leq m-1}: \mathscr{C}_{rn}\preceq \mathscr{C}_{r(n+1)}, |\mathscr{C}_{rn}|=V_{rn}| \\
			&= \sum_{(V_{rn})_{0 \leq n \leq m-1}m}\sum_{C_{rm}(x) \ni 0}\sum_{\substack{\mathscr{C}_{r(m-1)} \\ |\mathscr{C}_{r(m-1)}|=V_{r(m-1)}\\ \mathscr{C}_{r(m-1)} \preceq \mathscr{C}_{rm}}}|\{\mathscr{C}_{rn}(\Lambda)\}_{0 \leq n \leq m-2}: \mathscr{C}_{rn}\preceq \mathscr{C}_{r(n+1)}, |\mathscr{C}_{rn}|=V_{rn}| \\
			&=\sum_{(V_{rn})_{0 \leq n \leq m-1}}\sum_{C_{rm}(x) \ni 0}\sum_{\substack{\mathscr{C}_{r(m-1)} \\ |\mathscr{C}_{r(m-1)}|=V_{r(m-1)}\\ \mathscr{C}_{r(m-1)} \preceq \mathscr{C}_{rm}}}\dots \sum_{\substack{\mathscr{C}_{r} \\ |\mathscr{C}_r|=V_r\\ \mathscr{C}_r \preceq \mathscr{C}_{2r}}} N(\mathscr{C}_r, V_0).
		\end{align*}
		Iterating Inequality \eqref{eq_appB} we get 
		\begin{equation}\label{101}
        \begin{split}	
   |\mathscr{F}_{V,m}|&\leq
			e^{c V_0}\sum_{(V_{rn})_{0 \leq n \leq n_r(\Lambda)-1}}\sum_{C_{rm}(x) \ni 0}\sum_{\substack{\mathscr{C}_{r(m-1)} \\ |\mathscr{C}_{r(m-1)}|=V_{r(m-1)}\\ \mathscr{C}_{r(m-1)} \preceq \mathscr{C}_{rm}}}\dots \sum_{\substack{\mathscr{C}_{2r} \\ |\mathscr{C}_{2r}|=V_{2r}\\ \mathscr{C}_{2r} \preceq \mathscr{C}_{3r}}} N(C_{2r}, V_r) \\
			&\leq  |C_{rm}(x): C_{rm}(x) \ni 0| \sum_{(V_{rn})_{0 \leq n \leq m-1}} \prod_{n=0}^{m-1}e^{c V_{rn}}.
    \end{split}		
  \end{equation}
		Since the $rm$-cubes must contain the point $0$, we just need to count how many centers $2^{rm-1}x \in \Z^d$ are possible. This is equivalent to count how many $x_i$ satisfy
		$0 \in [2^{rm-1}x_i -2^{rm-1}, 2^{rm-1}x_i+2^{rm-1}]$ for all $1\leq i \leq d$. It is easy to see that $x_i \in \{-1,0,1\}$, for all $1 \leq i \leq d$, hence
		\be \label{eq102}
		|C_{rm}(x): C_{rm}(x) \ni 0| \leq 3^d.
		\ee
		We have at most $2^V$ solutions for Equation \eqref{constraint}, thus Inequality \eqref{101} together with the fact that the number of $m$-cubes containing $0$ is bounded by $3^d$ yield us
		\[
		|\mathscr{F}_V| \leq 3^d V 2^V e^{c V}.
		\]
		Therefore, Inequality (\ref{bound_F_V}) holds for $b = d\log(3) + \log(2)+ c+1$. 
	\end{proof}
	
	We are able to prove Proposition \ref{importprop} once we show that a fixed configuration $\sigma$ with $\Gamma(\sigma)=\{\gamma\}$ and a fixed volume $|\gamma|=m$ implies that the total volume $V_r(\Sp(\gamma))$ is finite.  We need the following auxiliary result about graphs, which is a generalization of Claim 4.2 in \cite{KP}.   
	\begin{proposition}\label{proptree}
		Let $k\ge 1$ and $G$ be a finite, non-empty, connected simple graph. Then, $G$ can be covered by $\lceil |v(G)|/k \rceil$ connected subgraphs of size at most $2k$.
	\end{proposition}
	
	\begin{proof}
		Since we can always consider a spanning tree from a connected graph $G$, it is sufficient to prove the proposition when $G$ is a tree. If either $k=1$ or $|v(G)|\leq 2k$, our statement is trivially true, so we suppose $k\geq 2$ and $|v(G)|\geq 2k+1$, and proceed by induction on $\lceil |v(G)|/k \rceil$.
		
		Choose a vertex $r \in v(G)$ to be our root. For every vertex $u$ of $G$ let $\text{dep}(u)$ be \emph{the depth} of the vertex $u$, i.e., the distance in the graph between $r$ and $u$. We say that a vertex $w$ is a \emph{descendant} of $u$ if there is a path $u_1=w, u_2, \dots, u_{n-1}, u_n=u$ in $G$ with $\text{dep}(u_i)> \text{dep}(u)$, for all $1\leq i \leq n-1$ and let $\text{desc}(u)$ be the number of descendants of $u$. Take a vertex $u^*$ from $\{u \in v(G): \text{desc}(u)\geq k\}$, that is not empty since $\text{desc}(r)\geq 2k$, with highest depth, i.e., such that for any other $u \in v(G)$ with $\text{desc}(u)\geq k$ we have $\text{dep}(u^*)\geq \text{dep}(u)$. Let $u_1,\dots,u_t$ be the children of the vertex $u^*$, and define $a_i = \text{desc}(u_i)+1$. By definition of $u^*$, we have that $a_i\leq k$, for $1\leq i \leq t$, and $a_1+ \dots +a_t \geq k$. Hence, there must be an $1\leq s \leq t$ for which $k \leq a_1+ \dots +a_s < 2k$. Therefore, we can consider the subtree $T$ whose vertex set is composed by $u^*, u_1, \dots, u_s$ and their descendants. By construction, it holds $k+1 \leq |v(T)|\leq 2k$. 
		The induced subgraph $H \coloneqq G[(v(G)\setminus v(T))\cup \{u^*\}]$ is connected and satisfies $|v(H)|\leq |v(G)|-k$. Using the induction hypothesis, $H$ can be covered by $\lceil|v(H)|/k\rceil\leq \lceil|v(G)|/k\rceil-1$ connected subgraphs. Adding $T$ to this cover completes the proof. 
	\end{proof}
		The importance of the choice $r=\lceil \log_2(a+1)\rceil +d+1$ will be seen in the following proposition, where we will show that the total volume can be bounded by the size of the contour.

	\begin{proposition}\label{prop4}
		There exists a constant $\kappa\coloneqq \kappa(d,M,r) >0$ such that, for any contour $\gamma \in \mathcal{E}_\Lambda^\pm$,
		\be\label{eq100}
		V_r(\Sp(\gamma)) \leq \kappa|\gamma|.
		\ee 
	\end{proposition}
	\begin{proof}
		Define $g:\mathbb{N} \rightarrow \Z$ by
		\be\label{prop4:eq1}
		g(n) = \left\lfloor \frac{n-2-\log_{2^r}(2Md^a)}{a}\right\rfloor.
		\ee
		We are going to prove 
		\be\label{prop4:eq2}
		|\mathscr{C}_{rn}(\Sp(\gamma))| \leq \frac{1}{2^{r-d-1}}|\mathscr{C}_{rg(n)}(\Sp(\gamma))|,
		\ee
		whenever $g(n)>0$, and either the graph $G_{rg(n)}(\Sp(\gamma))$ defined in Proposition \ref{prop1} has more than two connected components or it has only one connected component with at least $2^r$ vertices. Remember that $\mathscr{G}_{rg(n)}(\Sp(\gamma))$ is the set of all connected components of the graph $G_{rg(n)}(\Sp(\gamma))$. Note that
		\be\label{eq9}
		|\mathscr{C}_{rg(n)}(\Sp(\gamma))| = 2^r\sum_{G \in \mathscr{G}_{rg(n)}(\Sp(\gamma))} \frac{|v(G)|}{2^r}.
		\ee
		
		Proposition \ref{proptree} states that we can cover the vertex set $v(G)$ with a family of connected graphs $G_i$ with $1 \leq i \leq \lceil |v(G)|/2^r \rceil$ and $ |v(G_i)|\leq 2^{r+1}$. Using the inequality 
		\[
		\diam(\Lambda \cup \Lambda') \leq \diam(\Lambda) + \diam(\Lambda') + d(\Lambda,\Lambda'), \quad \text{for all } \Lambda, \Lambda' \Subset \Z^d,
		\]
		and the fact that we can always extract a vertex of a connected graph in a way that the induced subgraph is still connected, by removing a leaf of a spanning tree we can bound the diameter of $B_{G_i} = \bigcup_{C_{rg(n)}(x)\in v(G_i)}C_{rg(n)}(x)$ by
		\begin{equation}\label{eq10}
        \begin{split}
			\diam(B_{G_i}) &\leq \sum_{C_{rg(n)}(x) \in v(G_i)}\diam(C_{rg(n)}(x)) + |v(G_i)|Md^a 2^{arg(n)} \\
			&\leq d2^{r(g(n)+1)+1}+ Md^a2^{r(ag(n)+1)+1}  \\
			& \leq 2^{rn}.
		\end{split}
        \end{equation}
		The third inequality holds since $M,a,r \geq  1$. Therefore, each graph $G_i$ can be covered by one cube of side length $2^{rn}$ with center in $\Z^d$. We claim that every cube of side length $2^{rn}$ with an arbitrary center in $\Z^d$ can be covered by at most $2^d$ $rn$-cubes $C_{rn}(x)$. 
		Note that it is enough to consider the case where the cube has the form 
		\[
		\prod_{i=1}^d([q_i- 2^{rn-1}, q_i+ 2^{rn-1}])\cap \Z^d,
		\] 
		where $q_i \in \{0,1,\dots, 2^{rn-1}-1\}$, for $1 \leq i \leq d$.
		It is easy to see that
		\[
		[q_i- 2^{rn-1}, q_i+ 2^{rn-1}] \subset [-2^{rn-1}, 2^{rn-1}]\cup [0, 2^{rn}].
		\] 
		Taking products for all $1\leq i \leq d$, it concludes our claim. This reasoning allow us to conclude that the maximum number of $rn$-cubes required to cover each connected component $G$ of $G_{rg(n)}(\Sp(\gamma))$ is at most $2^d\lceil |v(G)|/2^r \rceil$, yielding us  
		\be\label{eq11}
		|\mathscr{C}_{rn}(\Sp(\gamma))| \leq \sum_{G \in \mathscr{G}_{rg(n)}(\Sp(\gamma))}\left|\mathscr{C}_{rn}\left(\bigcup_{1 \leq i \leq \lceil |v(G)|/2^r\rceil}B_{G_i}\right)\right| \leq \sum_{G \in \mathscr{G}_{rg(n)}(\Sp(\gamma))} 2^d\left\lceil\frac{|v(G)|}{2^r}\right\rceil.
		\ee
		
		If $|\mathscr{G}_{rg(n)}(\Sp(\gamma))|\geq 2$, since $\gamma \in \mathcal{E}^\pm_\Lambda$, each connected component $G \in \mathscr{G}_{rg(n)}(\Sp(\gamma))$ satisfies $|v(G)|\geq 2^r$. Indeed, if $|v(G)|\leq 2^r-1$, by our construction in Proposition \ref{prop1} the set $v(G)$ would be separated into another element of the $(M,a,r)$-partition. By hypothesis, if $|\mathscr{G}_{rg(n)}(\Sp(\gamma))|=1$ it already satisfies $|v(G_{rg(n)})|\geq 2^r$. 
		Together with 
		\[
		\frac{1}{2}\left\lceil\frac{|v(G)|}{2^r}\right\rceil  \leq \frac{|v(G)|}{2^r},
		\]
		Inequalities \eqref{eq9} and \eqref{eq11} yield 
		\[
		|\mathscr{C}_{rn}(\Sp(\gamma))| \leq 2^{d+1}\sum_{G \in \mathscr{G}_{rg(n)}(\Sp(\gamma))}\frac{|v(G)|}{2^r} = \frac{2^{d+1}}{2^r}|\mathscr{C}_{rg(n)}(\Sp(\gamma))|.
		\]
		
		So, Inequality \eqref{prop4:eq2} is proved. Let us define two auxiliary quantities
		\begin{equation*}
        \begin{split}
		&l_1(n) \coloneqq \max\{m \geq 0: g^m(n) \geq 0\} \quad \text{and} \\ &l_2(n) \coloneqq \max\{m \geq 0: |\mathscr{G}_{rg^m(n)}(\Sp(\gamma))|=1, |v(G_{rg^m(n)})|\leq 2^r-1\}.
        \end{split}		
        \end{equation*}
		For the set $\mathscr{G}_{rg^m(n)}(\Sp(\gamma))$ to be well defined we must have $g^m(n) \geq 0$, thus $l_2(n) \leq l_1(n)$. Moreover, knowing that $|\mathscr{C}_n(\Lambda)| \leq |\Lambda|$ for any $n \geq 0$,
		\be\label{prop4:eq4}
		|\mathscr{C}_{rn}(\Sp(\gamma))| \leq  |\mathscr{C}_{rg^{l_2(n)}(n)}(\Sp(\gamma))| \leq \frac{1}{2^{(r-d-1)(l_1(n)-l_2(n))}}|\gamma|.
		\ee
		
		We claim 
		\be\label{eq13}
		l_1(n) \geq \begin{cases}
			0 & \text{ if }0 \leq n \leq n_0, \\
			\left\lfloor \frac{\log_2 (n)-\log_2(n_0)}{\log_2(a)} \right\rfloor & \text{ if } n> n_0,
		\end{cases}
		\ee
		where 
		$n_0= (a+2+\log_{2^r}(2Md^a))(a-1)^{-1}$.
		The first bound is trivial. Let $n > n_0$ and consider the function 
		\[
		\tilde{g}(n) = \frac{n-a-2-\log_{2^r}(2Md^a)}{a}.
		\]
		From the fact that $g(n) \geq \tilde{g}(n)$ and both functions are increasing, we have
		\be\label{ineq:gtilde}
		g^m(n) \geq \tilde{g}^m(n), \quad \text{ for all } m \geq 1,
		\ee
		which implies $\max_{\tilde{g}^m(n) \geq 0} m \leq l_1(n)$. Thus, we need to compute a lower bound for $m$ such that $\tilde{g}^m(n)>0$. Since
		\[
		\tilde{g}^m(n) = \frac{n}{a^m} - b\left(\frac{a^m-1}{a^{m-1}(a - 1)}\right),
		\] 
		is sufficient to have
		\be\label{prop4:eq5}
		\frac{n}{a^m}  > \frac{ab}{a- 1},
		\ee
		where $b=(a+2+\log_{2^r}(2Md^a))a^{-1}$.
		We get the desired bound after taking the 
		logarithm with respect to base two in both sides of Inequality \eqref{prop4:eq5}. To finish our calculation, we will analyze two cases depending if $l_2(n)$ is zero or not. First, let us consider the case where $l_2(n)=0$. Using Inequality \eqref{prop4:eq4} we get 
		\be
		V_r(\Sp(\gamma)) \leq |\gamma|\sum_{n=0}^\infty\frac{1}{2^{(r-d-1)l_1(n)}}.
		\ee
		To finish, notice that the equation above can be bounded in the following way
		\[
		\sum_{n=0}^\infty\frac{1}{2^{(r-d-1)l_1(n)}} \leq n_0+ 1+2^{r-d-1}n_0^{\frac{r-d-1}{\log_2(a)}} \zeta\left(\frac{r-d-1}{\log_2(a)}\right).
		\] 	
		Now consider the case $l_2(n)\neq 0$. A similar bound as \eqref{eq10} and the fact that $\mathscr{C}_{rg^m(n)}(\Sp(\gamma))$ is a cover for the set $\Sp(\gamma)$ implies
		\begin{align*}
			\diam(\Sp(\gamma)) &\leq \diam(B_{G_{rg^m(n)}(\Sp(\gamma))})\leq (d2^{rg^m(n)}+ Md^a 2^{arg^m(n)})|v(G_{rg^m(n)}(\Sp(\gamma)))|
			\leq 2Md^a 2^{arg^m(n)+r}.
		\end{align*}
		The inequality above yields
		\[
		\log_{2^r}(\diam(\Sp(\gamma))) \leq \log_{2^r}(2Md^a)+ ag^m(n)+1 \leq \log_{2^r}(2Md^a)+ \frac{n}{a^{m-1}}+1.
		\]
		Let us assume that $\diam(\Sp(\gamma))> 2^{2r+1}Md^a$. Isolating the term that is a function of $m$ and taking the logarithm with respect to base two in both sides of the equation above, it gives us
		\[
		m \leq 1+ \frac{\log_2(n) - \log_2(\log_{2^r}(\diam(\Sp(\gamma)))-\log_{2^r}(2Md^a)-1)}{\log_2(a)}.
		\] 
		The inequality above is valid for all $m \in \{k\geq 0: |\mathscr{G}_{rg^k(n)}|=1, |v(G_{rg^k(n)})|\leq 2^r-1\}$. Thus, together with the lower bound \eqref{eq13}, we get for $n>n_0$
		\[
		l_1(n)-l_2(n) \geq \frac{ \log_2(\log_{2^r}(\diam(\Sp(\gamma)))-\log_{2^r}(2Md^a)-1)-\log_2(n_0)}{\log_2(a)}-2.
		\]
		Inequality \eqref{prop4:eq4} together with the inequality above yields  
		\begin{equation}
        \begin{split}
			V_r(\gamma) &\leq (n_0 +1) |\gamma|+ |\gamma|\frac{2^{2(r-d-1)}n_0^{\frac{r-d-1}{\log_2(a)}}n_r(\Sp(\gamma))}{(\log_{2^r}(\diam(\Sp(\gamma)))-\log_{2^r}(2Md^a)-1)^{\frac{r-d-1}{\log_2(a)}}} \\
			&\leq |\gamma|\left(n_0 +1+ 2^{2(d+1-r)}n_0^{\frac{d+1-r}{\log_2(a)}}+ \frac{2^{2(d+1-r)}n_0^{\frac{d+1-r}{\log_2(a)}}\log_{2^r}(\diam(\Sp(\gamma)))}{\log_{2^r}(\diam(\Sp(\gamma)))-\log_{2^r}(2Md^a)-1}\right) \\
			&\leq (n_0 + 1+ 2^{2(r-d-1)}n_0^{\frac{r-d-1}{\log_2(a)}}(2+\log_{2^r}(2Md^a)))|\gamma|,
		\end{split}
        \end{equation}
		where the last inequality is due to the fact that 
		$x/(x-w)\leq 1+w$ for any constant $x \geq w+1$.
		
		If $\diam(\Sp(\gamma))\leq 2^{2r+1}Md^a$, we have
		\be
		V_r(\Sp(\gamma))\leq (n_r(\Sp(\gamma))+1)|\gamma| \leq (3 + \log_{2^r}(2Md^a))|\gamma|.
		\ee
		Taking $$\kappa = \max\left\{n_0 + 1+ 2^{2(r-d-1)}n_0^{\frac{r-d-1}{\log_2(a)}}(2+\log_{2^r}(2Md^a)), n_0+ 1+2^{r-d-1}n_0^{\frac{r-d-1}{\log_2(a)}} \zeta\left(\frac{r-d-1}{\log_2(a)}\right)\right\},$$ concludes the desired result.
	\end{proof}
	
	
	The Proposition \ref{importprop} will show that although the contours may be disconnected, there are at most an exponential number of them, depending on their size.
	
	\begin{proposition}\label{importprop}
		Let $m\geq 1$, $d\ge 2$, and $\Lambda \Subset \mathbb{Z}^d$. Consider the set $\mathcal{C}_0(m)$ given by
		\[
		\mathcal{C}_0(m) = \{\overline{\gamma}\in \mathcal{F}(\Z^d):\exists\gamma \in \mathcal{E}_\Lambda^- \text{ s.t.
} \Sp(\gamma)=\overline{\gamma}, 0 \in V(\gamma), |\gamma|=m\}.
		\]
		There exists $c_1\coloneqq c_1(d,M,r)>0$ such that
		\[
		|\mathcal{C}_0(m)| \leq e^{c_1 m}.
		\]
	\end{proposition}
	
	\begin{proof}
		For a given contour $\gamma$, define the set $\mathcal{C}_\gamma$ by
		\be
		\mathcal{C}_\gamma \coloneqq \{\Sp(\gamma') \in \mathcal{C}_0(m): \exists \; x \in \Z^d \text{ s.t. } \Sp(\gamma') = \Sp(\gamma)+x\}.
		\ee
		Thus, we can partition the set $\mathcal{C}_0(m)$ into
		\[
		\mathcal{C}_0(m) = \bigcup_{\substack{0 \in \Sp(\gamma) \\ |\gamma|=m}}C_\gamma.
		\]	
		Given a contour $\gamma \in \mathcal{E}_\Lambda$, there are at most $|V(\gamma)|$ possibilities for the position of the point $0$.
		Then,
		\be\label{prop5:eq1}
		|\mathcal{C}_0(m)| \leq \sum_{\substack{0 \in \Sp(\gamma) \\ |\gamma|=m}}|\mathcal{C}_\gamma| \leq \sum_{\substack{0 \in \Sp(\gamma) \\ |\gamma|=m}}|V(\gamma)|.
		\ee
		Using the isoperimetric inequality and the fact $\din V(\gamma) \subset \Sp(\gamma)$ we obtain,
		\be\label{eq14}
		\sum_{\substack{0 \in \Sp(\gamma) \\ |\gamma|=m}}|V(\gamma)| \leq m^{1+ \frac{1}{d-1}}|\{\overline{\gamma}\in \mathcal{F}(\Z^d):\exists\gamma \in \mathcal{E}_\Lambda^- \text{ s.t.
} \Sp(\gamma)=\overline{\gamma}, 0 \in \Sp(\gamma), |\gamma|=m \}|.
		\ee
  By Proposition \ref{prop4}, and since not all the finite sets with bounded total volume are contours, we have
		\be\label{eq15}
		\{\overline{\gamma}\in \mathcal{F}(\Z^d):\exists\gamma \in \mathcal{E}_\Lambda^- \text{ s.t.
} \Sp(\gamma)=\overline{\gamma}, 0 \in \Sp(\gamma), |\gamma|=m \} \subset \{\Lambda \Subset \Z^d: 0 \in \Lambda, V_r(\Lambda)\leq \kappa m\}.
		\ee
		Proposition \ref{appB:prop2} yields
		\be\label{eq16}
		|\{\Lambda \Subset \Z^d: 0 \in \Lambda, V_r(\Lambda)\leq \kappa m \}|= \sum_{V= 1}^{\lceil \kappa m \rceil }|\mathscr{F}_V| \leq \frac{e^{2b\kappa m+1}}{e^{b}-1}.
		\ee
		Substituting Inequalities \eqref{eq14}, \eqref{eq15} and \eqref{eq16} into Inequality \eqref{prop5:eq1}, we conclude
		\be
		|\mathcal{C}_0(m)| \leq m^{1+ \frac{1}{d-1}}\frac{e^{2b\kappa m+1}}{e^{b}-1} \leq e^{c_1m},
		\ee
		for 
		$c_1 = 2b\kappa+1+ (d-1)^{-1}$.
	\end{proof}
	
	\section{Phase Transition}

	
	In this section, we prove that the long-range Ising model with a decaying field undergoes a phase transition at low temperature when $\min\{\alpha-d,1\}<\delta<d$. The strategy will follow closely the one from Chapter 1, and is described in the end of Section 1.1. When the magnetic field decays with power $\delta \geq d$, the result is straightforward. In fact, for $\delta>d$, the magnetic field is summable and, by a general result of Georgii (see Example 7.32 and Theorem 7.33 in \cite{Geo}), there is an affine bijection between the Gibbs measures of the Ising model with $h=0$. Then, the phase transition is already known in this case. For $\delta=d$ the sum $\sum_{x\in \Lambda}h_x$ can be bounded by $\log |\Lambda|$. This implies that $\sum_{x \in \Lambda} h_x = o(|\Lambda|^\varepsilon)$ for any $\varepsilon>0$. Thus, if we prove the phase transition for $\delta < d$, it is easy to extend to this case.
	
	\begin{theorem}\label{thm1}
		For a fixed $d\ge 2$, suppose that $\alpha>d$ and $\delta>0$.
		There exists $\beta_c\coloneqq \beta_c(\alpha,d)>0$ such that, for every $\beta>\beta_c$, the long range Ising model with coupling constant  (\ref{long}) and magnetic field (\ref{magfield}) undergoes a phase transition at inverse of temperature $\beta$ when
		\begin{itemize}
			\item $d<\alpha<d+1$ and $\delta >\alpha -d$; 
			\item $d<\alpha<d+1$ and $\delta=\alpha-d$ if $h^*$ is small enough;
			\item $\alpha \geq d+1$ and $\delta>1$;
			\item  $\alpha \geq d+1$ and $\delta=1$ if $h^*$ is small enough.
		\end{itemize}.
	\end{theorem}

	Throughout this section, we will also use $\Gamma(\sigma)$ to denote the set of contours associated with a configuration $\sigma$ instead of the $(M,a,r)$-partition. Define the map $\tau_\Gamma:\Omega(\Gamma) \rightarrow \Omega_{\Lambda}^-$ as
	\be
	\tau_\Gamma(\sigma)_x = 
	\begin{cases}
		\;\;\;\sigma_x &\text{ if } x \in \I_-(\Gamma)\cup V(\Gamma)^c, \\
		-\sigma_x &\text{ if } x \in \I_+(\Gamma),\\
		-1 &\text{ if } x \in \Sp(\Gamma).
	\end{cases}
	\ee
	The map $\tau_\Gamma$ erases a family of compatible contours since the spin-flip preserves incorrect points but transforms $+$-correct points into $-$-correct points. Given $\Gamma \in \mathcal{E}^-_\Lambda$ and a configuration $\sigma \in \Omega(\Gamma)$, we will calculate the energy cost to extract one of its elements. We start with a lemma giving a lower bound for the diameter of a finite subset of $\Z^d$.
	\begin{lemma}\label{diamlema}
		There exists $k_d>0$ such that for every $\Lambda \Subset \Z^d$ it holds,
		\be\label{diamlema:eq1}
		\diam(\Lambda)\geq k_d |\Lambda|^{\frac{1}{d}}
		\ee
	\end{lemma}
	\begin{proof}
		Consider a closed ball with a positive integer radius $n$. Lemma \ref{comblema} implies that the diameter satisfies
		\[
		\diam(B_n(x))= 2n \geq C_d |B_n(x)|^{\frac{1}{d}},
		\]
		where $C_d = \frac{(ed)^{\frac{1}{d}}}{2e+1}$. If we take $x^*, y^* \in \Lambda\Subset \Z^d$ such that $\diam(\Lambda)=|x^*-y^*|$ we have
		\[
		2\diam(\Lambda)= \diam( B_{|x^*-y^*|}(x^*)) \geq C_d|\Lambda|^{\frac{1}{d}},
		\]
		where the last inequality is due to the fact that $\Lambda\subset B_{|x^*-y^*|}(x^*)$. Inequality \eqref{diamlema:eq1} follows by choosing the constant $k_d = C_d/2$. 
	\end{proof}
	
	In the next proposition, we will give a lower bound for the cost of extracting a contour from a given configuration. The main difference is that one has a \emph{surface order term}, defined as
	\[
	F_\Lambda \coloneqq \sum_{\substack{x \in \Lambda \\ y \in \Lambda^c}}J_{xy},
	\]
	for every set $\Lambda \Subset \Z^d$. First, let us give a lower bound to the surface energy term, that will be useful to the proof of phase transition.
	\begin{lemma}\label{lema3}
		Given $\alpha>d$, there exists $K_\alpha\coloneqq K_\alpha(J,\alpha,d)>0$ such that for every $\Lambda \in \mathcal{F}(\Z^d)$ it holds
		\be
		F_\Lambda \geq K_{\alpha}\max\{|\Lambda|^{2-\frac{\alpha}{d}}, |\partial\Lambda|\}.
		\ee
	\end{lemma}
	\begin{proof}
		Since all the edges of $\partial\Lambda$ are present in the surface energy term $F_\Lambda$, we have the bound $F_\Lambda \geq J |\partial \Lambda|$. Fix $x \in \Lambda$.
		If we set $R = \lceil (dc_d^{-1}|\Lambda|)^{\frac{1}{d}} \rceil$ and using that $\sum_{\substack{y \in \Z^d}}J_{xy} < \infty$ we have
		\[
		\ssum{y \in \Lambda^c}J_{xy} - \ssum{y \in B_{R}(x)^c}J_{xy}=\ssum{y \in B_{R}(x)}J_{xy} - \ssum{y \in \Lambda}J_{xy} \geq \frac{J(|B_{R}(x)| - |\Lambda|)}{R^\alpha} \geq 0.
		\]
		Lemma \ref{comblema} yields us
		\[
		\sum_{y \in B_R(x)^c}J_{xy}  = J\sum_{n\geq R+1}\frac{s_d(n)}{n^\alpha} \geq J c_d\sum_{n \geq R+1} \frac{1}{n^{\alpha -d +1}}\geq \frac{Jc_d^{1+\alpha-d}}{(\alpha-d)d^{\alpha-d}} |\Lambda|^{1-\frac{\alpha}{d}},
		\]
		where the last inequality on the right-hand side above was bounded below by an integral. Summing over $x\in \Lambda$ and taking $K_\alpha = J\max\Bigg\{1,\frac{c_d^{1+\alpha-d}}{(\alpha-d)d^{\alpha-d}}\Bigg\}$ finish the proof.
	\end{proof}
	
	In the following proposition, we will denote $H_\Lambda^-$ the Hamiltonian function in \eqref{Isingsys} when the field $h_x = 0$, for all $x \in \Z^d$.
	\begin{proposition}\label{prop2}
		For $M$ large enough, there are constants $c_i\coloneqq c_i(\alpha,d)>0$, $i=2,3,4$, such that for $\Lambda \in \mathcal{F}(\Z^d)$, any fixed contour $\gamma \in \mathcal{E}^-_\Lambda$, and $\sigma \in \Omega(\gamma)$ it holds 
		\be
		H_\Lambda^-(\sigma)- H_\Lambda^-(\tau(\sigma))\geq c_2|\gamma|+ c_3 F_{\I_+(\gamma)} + c_4 F_{\Sp(\gamma)}.
		\ee
	\end{proposition}
	\begin{proof}
		Fix some $\sigma \in \Omega(\gamma)$. We will denote $\tau_\gamma(\sigma)\coloneqq \tau$ and $\Gamma(\sigma) \coloneqq \Gamma$ throughout this proposition. The difference between the Hamiltonians is
  \begin{equation}\label{prop2:eq1}
  \begin{split}
  H_{\Lambda}^-(\sigma)-H_{\Lambda}^-(\tau) &= \sum_{\{x,y\}\subset \Lambda}\hspace{-0.3cm}J_{x,y}(\tau_x\tau_y-\sigma_x\sigma_y) +\sum_{\substack{x\in \Lambda \\ y \in \Lambda^c}}J_{x,y}(\sigma_x-\tau_x) \\
  &= \sum_{\{x,y\}\subset V(\Gamma)}\hspace{-0.3cm}J_{x,y}(\tau_x\tau_y-\sigma_x\sigma_y) +\sum_{\substack{x\in V(\Gamma) \\ y \in \Lambda\setminus V(\Gamma)}}\hspace{-0.3cm}J_{x,y}(\sigma_x-\tau_x) +\sum_{\substack{x\in V(\Gamma) \\ y \in \Lambda^c}}\hspace{-0.3cm}J_{x,y}(\sigma_x-\tau_x) \\
  &= \sum_{\{x,y\}\subset V(\Gamma)}\hspace{-0.3cm}J_{xy}(\tau_x\tau_y - \sigma_x\sigma_y) + \sum_{\substack{x \in V(\Gamma)\\ y \in V(\Gamma)^c}}J_{xy}(\sigma_x - \tau_x),
  \end{split}
  \end{equation}
  where the second equality is due the fact that outside $V(\Gamma)$ the configurations $\sigma$ and $\tau$ are equal to $-1$.  Since $\tau_x=\sigma_x$ for $x\in V(\Gamma\setminus\{\gamma\})$ we have
  \begin{equation*}
  \begin{split}
  \sum_{\{x,y\}\subset V(\Gamma)}\hspace{-0.3cm}J_{xy}(\tau_x\tau_y - \sigma_x\sigma_y) + \hspace{-0.2cm} \sum_{\substack{x \in V(\Gamma)\\ y \in V(\Gamma)^c}}J_{xy}(\sigma_x - \tau_x)= \hspace{-0.5cm}\sum_{\{x,y\}\subset V(\gamma)}\hspace{-0.3cm}J_{x,y}(\tau_x\tau_y-\sigma_x\sigma_y)+\hspace{-0.3cm}\sum_{\substack{x\in V(\gamma) \\ y \in V(\gamma)^c}}\hspace{-0.3cm} J_{x,y}(\tau_x\sigma_y-\sigma_x\sigma_y).
  \end{split}
  \end{equation*}
  We focus now on the sum involving the terms $\{x,y\} \subset V(\gamma)$. We can split it accordingly with $V(\gamma)= \Sp(\gamma) \cup \I(\gamma)$. Then,
  \begin{equation*}
  \begin{split}
  \sum_{\{x,y\}\subset V(\gamma)}J_{x,y}(\tau_x\tau_y-\sigma_x\sigma_y) &= \hspace{-0.3cm}\sum_{\{x,y\}\subset \Sp(\gamma)}\hspace{-0.3cm}J_{x,y}(\tau_x\tau_y-\sigma_x\sigma_y)+\hspace{-0.3cm}\sum_{\{x,y\}\subset \I(\gamma)}\hspace{-0.3cm}J_{x,y}(\tau_x\tau_y-\sigma_x\sigma_y)+\sum_{\substack{x\in \Sp(\gamma) \\ y \in \I(\gamma)}}\hspace{-0.1cm}J_{x,y}(\tau_x\tau_y-\sigma_x\sigma_y) \\
  &=\hspace{-0.3cm}\sum_{\{x,y\}\subset \Sp(\gamma)}\hspace{-0.3cm}J_{x,y}(1-\sigma_x\sigma_y)-2\sum_{\substack{x\in \I_+(\gamma) \\ y \in \I_-(\gamma)}}J_{x,y}\sigma_x\sigma_y +\sum_{\substack{x\in \Sp(\gamma) \\ y \in \I(\gamma)}}J_{x,y}(\text{sign}(y)\sigma_y-\sigma_x\sigma_y) ,
  \end{split}
  \end{equation*}
  where for the second equality we used the definition of the map $\tau_\gamma$ and $\text{sign}(y)=+1$ if $y \in \I_+(\gamma)$ and $-1$ if $y \in \I_-(\gamma)$. For the same reason, we have
  \[
  \sum_{\substack{x\in V(\gamma) \\ y \in V(\gamma)^c}}J_{x,y}(\tau_x\sigma_y-\sigma_x\sigma_y) = \sum_{\substack{x\in \Sp(\gamma) \\ y \in V(\gamma)^c}}J_{x,y}(-\sigma_y-\sigma_x\sigma_y)-2\sum_{\substack{x\in \I_+(\gamma) \\ y \in V(\gamma)^c}}J_{x,y}\sigma_x\sigma_y
  \]
  Putting everything together and using that $\pm\sigma_y - \sigma_x\sigma_y = 2\mathbbm{1}_{\{\sigma_x\neq \sigma_y\}}-2\mathbbm{1}_{\{\sigma_y=\mp 1\}}$ we get
		\begin{align}\label{prop2:eq2}
			H_{\Lambda}^-(\sigma)-H_{\Lambda}^-(\tau) &= \sum_{\substack{x \in \Sp(\gamma)\\ y \in \Z^d}}J_{xy}\mathbbm{1}_{\{\sigma_x \neq \sigma_y\}}+\sum_{\substack{x \in \Sp(\gamma)\\ y \in \Sp(\gamma)^c}}J_{xy}\mathbbm{1}_{\{\sigma_x \neq \sigma_y\}}-2\sum_{\substack{x \in I_+(\gamma) \\ y \in B(\gamma)}}J_{xy}\sigma_x\sigma_y \nonumber\\
			&- 2\sum_{\substack{x \in \Sp(\gamma) \\ y \in B(\gamma)}}J_{xy}\mathbbm{1}_{\{\sigma_y = +1\}}-2\sum_{\substack{x \in \Sp(\gamma) \\ y \in \I_+(\gamma)}}J_{xy}\mathbbm{1}_{\{\sigma_y = -1\}},
		\end{align}
		where $B(\gamma)=\I_-(\gamma)\cup V(\gamma)^c$. We need to carefully analyze each negative term of Equation \eqref{prop2:eq2}. We start with the terms depending on $\Sp(\gamma)$. Notice that the characteristic functions on $B(\gamma)$ and $\I_+(\gamma)$ can only be different from zero at the other contours volumes. Thus,
		\be\label{prop2:eq3}
		\sum_{\substack{x \in \Sp(\gamma) \\ y \in B(\gamma)}}J_{xy}\mathbbm{1}_{\{\sigma_y = +1\}}+\sum_{\substack{x \in \Sp(\gamma) \\ y \in \I_+(\gamma)}}J_{xy}\mathbbm{1}_{\{\sigma_y = -1\}} \leq \sum_{\substack{x \in \Sp(\gamma) \\ y \in V(\Gamma')}}J_{xy},
		\ee
		where $\Gamma'\coloneqq \Gamma(\tau)$. Let $\gamma = \bigcup_{1 \leq k\leq n} \gamma_k$ and $\gamma'= \cup_{1\leq j\leq n'}\gamma'_j$ for each $\gamma' \in \Gamma'$ be the subsets given to us by condition \textbf{(B)}. We will divide the r.h.s of Equation \eqref{prop2:eq3} into two terms depending on the sets
		$\Upsilon_1= \{\gamma' \in \Gamma': \underset{1\leq k \leq n}{\max}\diam(\gamma_k)\leq \underset{1\leq j \leq n'}{\max}\diam(\gamma'_j)\}$  and $\Upsilon_2= \Gamma' \setminus \Upsilon_1$. On the first sum, Condition \textbf{(B2)} implies
		\[\label{prop2:eq4}
		\sum_{\substack{x \in \Sp(\gamma) \\ y \in V(\Upsilon_1)}}J_{xy} \leq \sum_{\substack{x \in \Sp(\gamma) \\ y \in B_R(x)^c}}J_{xy},
		\]
		where $R= M\underset{1\leq k \leq n}{\max}\diam(\gamma_k)^a$. Using Condition \textbf{(B1)} and Lemma \ref{diamlema} it holds,
		\be\label{prop2:eq5}
		\sum_{\substack{x \in \Sp(\gamma)\\ y \in B_R(x)^c}}J_{xy} \leq \frac{Je^{-1}(2e+1)^{d-1}|\gamma|}{(\alpha-d)M^{\alpha-d}}\underset{1\leq k \leq n}{\max}\diam(\gamma_k)^{a(d-\alpha)}\leq  \frac{Je^{-1}(2e+1)^{d-1}(2^r-1)}{(\alpha-d)k_d^{a(\alpha-d)}M^{\alpha -d}}.
		\ee
		We turn our attention to the sum depending on $\Upsilon_2$. We divide the set $\Upsilon_2$ into sets $\Upsilon_{2,m}$ consisting of contours of $\Gamma'$ where the maximum diameter of its partition is $m$. Thus, for each $x$ in $\Sp(\gamma)$ and $\gamma' \in \Gamma'$, there is at least one point $y_{\gamma',x}\in V(\gamma')$ such that $|x-y_{\gamma',x}| = d(x,\gamma')$. Then, it holds for each $1\leq m < \max_{1\leq k\leq n} \diam(\gamma_k)$,
		\[
		\sum_{\substack{x \in \Sp(\gamma) \\ y \in V(\Upsilon_{2,m})}}J_{xy} \leq \sum_{\substack{x \in \Sp(\gamma) \\ \gamma' \in \Upsilon_{2,m}}}|V(\gamma')|J_{xy_{\gamma',x}}.
		\]
		For each $\gamma' \in V(\Upsilon_{2,m})$, define the graph $G_{\gamma'}$ with vertex set $v(G_{\gamma'})=\{\gamma'_j\}_{1\leq j \leq n'}$ and an edge is placed when $d(\gamma'_j,\gamma'_i)\leq 1$. Let $G_j$ be the maximal connected component of $G_{\gamma'}$ such that $\gamma'_j$ is an element of its vertex set. Also, let $V(G_j)\subset V(\gamma')$ be the subset of all connected components of $V(\gamma')$ that have a non-empty intersection with the vertices of $G_j$. Using Lemma \ref{diamlema} we have
		\[
		|V(\gamma')| \leq \sum_{j=1}^{n'}|V(G_j)|\leq \frac{1}{k_d^d}\sum_{j=1}^{n'}\diam(V(G_j))^d
		\]
		The diameter of $V(G_j)$ is realized by the distance between two points, namely $x^*$ and $y^*$, that must be into $\gamma'_i, \gamma'_l \in v(G_j)$. We can make a path $\lambda_1$ in the graph $G_j$ between $\gamma'_i$ and $\gamma'_l$ since it is connected. Thus, using the path $\lambda_1$ we can construct a path $\lambda_2$ in $\Z^d$ connecting $x^*$ and $y^*$ that passes through every vertex that is visited by the path $\lambda_1$. Let $\lambda_3$ be a minimal path realizing the distance between $x^*$ and $y^*$. Since the path is minimal, we have
		\[
		\diam(V(G_j))=|\lambda_3|\leq |\lambda_2|\leq \sum_{\gamma'_i \in v(G_j)}\diam(\gamma'_i)+1.
		\]
		Hence,
		\be\label{eqvolume}
		\frac{1}{k_d^d}\sum_{j=1}^{n'}\diam(V(G_j))^d\leq \frac{1}{k_d^d}\sum_{j=1}^{n'} \left(\sum_{\gamma'_i \in v(G_j)}\diam(\gamma'_i)+1\right)^d.
		\ee
		The number of elements in $v(G_j)$ is at most $2^r-1$ by condition \textbf{(B)}, thus
		\be\label{prop2:eq6}
		\sum_{\substack{x \in \Sp(\gamma) \\ y \in V(\Upsilon_{2,m})}}J_{xy} \leq \frac{(2^r-1)^{d+1}}{k_d^d}(m+1)^d\sum_{\substack{x \in \Sp(\gamma) \\ \gamma' \in \Upsilon_{2,m}}}J_{xy_{\gamma',x}}.
		\ee
		We know that there is no other point $y_{\gamma'',x}$ at least in a ball of radius $Mm^a/3$ centered at $y_{\gamma',x}$. These balls with different centers are disjoint by Condition \textbf{(B2)}. Also, if $\lambda$ is the minimal path realizing the distance between $x$ and $y_{\gamma',x}$, we know that $\lambda$ must contain at least $Mm^a/3$ points (see Figure \ref{figure_minimalpath}). Thus,
		\be\label{prop2:eq7}
		\sum_{\substack{x \in \Sp(\gamma) \\ \gamma' \in \Upsilon_{2,m}}}J_{xy_{\gamma',x}} \leq \frac{3}{Mm^a}F_{\Sp(\gamma)}.
		\ee
		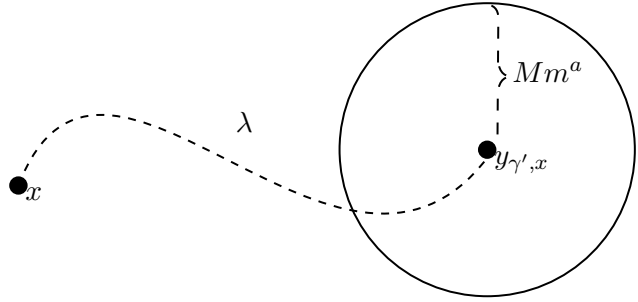
\begin{figure}[H]	
  \centering
			
			\tikzset{every picture/.style={line width=0.75pt}} 
			
			\begin{tikzpicture}[x=0.75pt,y=0.75pt,yscale=-1,xscale=1]
				
				\draw  [fill={rgb, 255:red, 0; green, 0; blue, 0 }  ,fill opacity=1 ] (4.8,99.1) .. controls (4.8,96.84) and (6.64,95) .. (8.9,95) .. controls (11.16,95) and (13,96.84) .. (13,99.1) .. controls (13,101.36) and (11.16,103.2) .. (8.9,103.2) .. controls (6.64,103.2) and (4.8,101.36) .. (4.8,99.1) -- cycle ;
				\draw  [fill={rgb, 255:red, 0; green, 0; blue, 0 }  ,fill opacity=1 ] (238.8,81.1) .. controls (238.8,78.84) and (240.64,77) .. (242.9,77) .. controls (245.16,77) and (247,78.84) .. (247,81.1) .. controls (247,83.36) and (245.16,85.2) .. (242.9,85.2) .. controls (240.64,85.2) and (238.8,83.36) .. (238.8,81.1) -- cycle ;
				\draw [dashed  ,draw opacity=1 ]   (8.9,99.1) .. controls (51.4,-12) and (169.5,182) .. (242.9,85.2) ;
				\draw   (169.2,81.1) .. controls (169.2,40.4) and (202.2,7.4) .. (242.9,7.4) .. controls (283.6,7.4) and (316.6,40.4) .. (316.6,81.1) .. controls (316.6,121.8) and (283.6,154.8) .. (242.9,154.8) .. controls (202.2,154.8) and (169.2,121.8) .. (169.2,81.1) -- cycle ;
				\draw  [dash pattern={on 4.5pt off 4.5pt}] (241,81) .. controls (245.67,81.03) and (248.02,78.72) .. (248.05,74.05) -- (248.18,54.25) .. controls (248.23,47.58) and (250.58,44.27) .. (255.25,44.3) .. controls (250.58,44.27) and (248.27,40.92) .. (248.32,34.25)(248.3,37.25) -- (248.45,15.05) .. controls (248.48,10.38) and (246.17,8.03) .. (241.5,8) ;
				
				\draw (244.9,80.4) node [anchor=north west][inner sep=0.75pt]    {$y_{\gamma',x}$};
				\draw (10.9,98.4) node [anchor=north west][inner sep=0.75pt]    {${\displaystyle x}$};
				\draw (116,60.4) node [anchor=north west][inner sep=0.75pt]    {$\lambda $};
				\draw (254,35.4) node [anchor=north west][inner sep=0.75pt]    {$Mm^{a}$};
				
			\end{tikzpicture}
			\caption{Minimal path $\lambda$ between $x$ and $y_{\gamma',x}$.}
   \label{figure_minimalpath}
		\end{figure}
		Inequalities \eqref{prop2:eq7}, \eqref{prop2:eq6}, \eqref{prop2:eq5} plugged into Inequality \eqref{prop2:eq3} yields,
		\be\label{prop2:eq8}
		\sum_{\substack{x \in \Sp(\gamma) \\ y \in B(\gamma)}}J_{xy}\mathbbm{1}_{\{\sigma_y = +1\}}+\sum_{\substack{x \in \Sp(\gamma) \\ y \in \I_+(\gamma)}}J_{xy}\mathbbm{1}_{\{\sigma_y = -1\}} \leq \frac{k^{(1)}_\alpha}{M^{(\alpha-d)\wedge 1}} F_{\Sp(\gamma)},
		\ee
		where, $k^{(1)}_\alpha = \max\left\{\frac{Je^{-1}(2e+1)^{d-1}(2^r-1)}{(\alpha-d)k_d^{a(\alpha-d)}},\frac{3(2^r-1)^{d+1}2^d\zeta(a-d)}{k_d^d}\right\}$ and $(\alpha-d)\wedge 1 = \min\{\alpha-d,1\}$. 
		
		The remaining term in our analysis is the one involving the interaction between $\I_+(\gamma)$ and $B(\gamma)$. Recall that $\Gamma(\tau )=\Gamma'$ is the set of external contours of $\Gamma$ after $\gamma$ is removed and define $\Gamma_1\subset \Gamma'$ as the set of external contours that are contained in $\I_+(\gamma)$ and $\Gamma_2 = \Gamma' \setminus \Gamma_1$. We have,
		\begin{align}\label{prop2:eq9}
			\sum_{\substack{x \in \I_+(\gamma) \\ y \in B(\gamma)}}J_{xy}\sigma_x\sigma_y= \sum_{\substack{x \in V(\Gamma_1) \\ y \in V(\Gamma_2)}}J_{xy} +\sum_{\substack{x \in \I_+(\gamma)\setminus V(\Gamma_1) \\ y \in V(\Gamma_2)}}2J_{xy}\mathbbm{1}_{\{\sigma_y=+1\}} + \sum_{\substack{x \in V(\Gamma_1) \\ y \in B(\gamma)\setminus V(\Gamma_2)}}2J_{xy}\mathbbm{1}_{\{\sigma_x=-1\}} \nonumber \\
			-\sum_{\substack{x \in \I_+(\gamma)\setminus V(\Gamma_1) \\ y \in V(\Gamma_2)}}J_{xy}-\sum_{\substack{x \in V(\Gamma_1) \\ y \in V(\Gamma_2)}}2J_{xy}\mathbbm{1}_{\{\sigma_x \neq \sigma_y\}}-\sum_{\substack{x \in \I_+(\gamma) \\ y \in B(\gamma)\setminus V(\Gamma_2)}}J_{xy},
		\end{align}
		We start our analysis with the first two terms on r.h.s of \eqref{prop2:eq9}. Note that,
		\be\label{again}
		\sum_{\substack{x \in V(\Gamma_1) \\ y \in V(\Gamma_2)}}J_{xy} + \sum_{\substack{x \in \I_+(\gamma)\setminus V(\Gamma_1) \\ y \in V(\Gamma_2)}}2J_{xy}\mathbbm{1}_{\{\sigma_y=+1\}}\leq 2 \sum_{\substack{x \in \I_+(\gamma) \\ y \in V(\Gamma_2)}}J_{xy}.
		\ee
		Consider the two sets $\Upsilon_3=\{\gamma' \in \Gamma_2: \underset{1\leq k \leq n}{\max}\diam(\gamma_k) \leq \underset{1\leq j \leq n'}{\max}\diam(\gamma'_j)\}$ and $\Upsilon_4 = \Gamma_2 \setminus \Upsilon_3$. Now we proceed as in the previous case for $\Upsilon_1$,
		\be\label{prop2:eq10}
		\sum_{\substack{x \in \I_+(\gamma)\\y \in V(\Upsilon_3)}}J_{xy} \leq \frac{Je^{-1} (2e+1)^{d-1} |\I_+(\gamma)|}{(\alpha-d)M^{\alpha-d}}\underset{1\leq k \leq n}{\max}\diam(\gamma_k)^{a(d-\alpha)}\leq \frac{ (2d)^\frac{d}{d-1}k_\alpha^{(1)}}{M^{\alpha-d}}|\din\gamma|^{\frac{1}{d-1}}\leq \frac{ (2d)^\frac{d}{d-1}k_\alpha^{(1)}}{M^{\alpha-d}}|\gamma|,
		\ee
		where the last inequality is due to the isoperimetric inequality. For the sum depending on the contours in $\Upsilon_4$, we will need to break, as before, into sets $\Upsilon_{4,m}$ whose contours have a maximum diameter equal to $m$. An argument similar to the one employed in \eqref{prop2:eq8} holds, hence
		\be\label{prop2:eq11}
		\sum_{\substack{x \in \I_+(\gamma)\\y \in V(\Upsilon_4)}}J_{xy} \leq \frac{(2^r-1)^{d+1}\zeta(a-d)}{k_d^d M}\sum_{\substack{x \in \I_+(\gamma) \\ y \in \I_-(\gamma)\cup V(\gamma)^c}}J_{xy}\leq \frac{k_\alpha^{(1)}}{M}F_{\I_+(\gamma)}.
		\ee
		For the next term, since we have $B(\gamma)\setminus V(\Gamma_2) \subset \I_+(\gamma)^c$, we get
		\[
		\sum_{\substack{x \in V(\Gamma_1) \\ y \in B(\gamma)\setminus V(\Gamma_2)}}J_{xy} \leq \sum_{\substack{x \in V(\Gamma_1) \\ y \in \I_+(\gamma)^c}}J_{xy}.
		\] 
		We claim that, for any $\gamma' \in \Gamma_1$, $\max_{1\leq j \leq n'}\diam(\gamma'_j) < \max_{1\leq k \leq n}\diam(\gamma_k)$ for $M>(2^r-1)^{d+1}/k_d^d$. Indeed, by condition \textbf{(A)}, $\Sp(\gamma')$ is contained in only one connected component of $\I(\gamma)$, let us call it $\I_+(\gamma)^{(1)}$. By similar reasonings as the one that gave us Inequality \eqref{eqvolume}, we have
		\[
		|\I_+(\gamma)^{(1)}|\leq \frac{(2^r-1)^{d+1}}{k_d^d}\max_{1\leq k \leq n}\diam(\gamma_k)^d.
		\] 
		Assume by contradiction $\max_{1\leq j \leq n'}\diam(\gamma'_j) \geq \max_{1\leq k \leq n}\diam(\gamma_k)$, then Condition \textbf{(B2)} implies that $d(\gamma',\gamma)\geq M \max_{1 \leq k \leq n}\diam(\gamma_k)^a$. Therefore, $|I_+(\gamma)|$ must have at least $M \max_{1 \leq k \leq n}\diam(\gamma_k)^a$ points inside it, which is a contradiction with our choice of $M$.
		
		Thus, let us break $\Gamma_1$ into layers $\Gamma_{1,m}$ where $\max_{1\leq j \leq n'}\diam(\gamma'_j)=m$. For each $y \in \I_+(\gamma)^c$ and $\gamma' \in \Gamma_1$ there is $x_{\gamma',y} \in V(\gamma')$ that realizes the distance between $V(\gamma')$ and $y$. Hence,
		\be\label{eqfinal}
		\sum_{\substack{x \in V(\Gamma_1) \\ y \in B(\gamma)\setminus V(\Gamma_2)}}J_{xy}\leq \sum_{m = 1}^{N-1}\sum_{\substack{x \in V(\Gamma_{1,m}) \\ y \in \I_+(\gamma)^c}}J_{xy} \leq \sum_{m = 1}^{N-1}\frac{(2^r-1)^{d+1}}{k_d^d}m^d \sum_{\substack{\gamma' \in \Gamma_{1,m} \\ y \in \I_+(\gamma)^c}}J_{x_{\gamma',y}y} \leq \frac{k_\alpha^{(1)}}{M} F_{\I_+(\gamma)},
		\ee
		where $N\coloneqq \max_{1\leq k\leq n}\diam(\gamma_k)$. We turn our attention to the term containing $\mathbbm{1}_{\{\sigma_x \neq \sigma_y\}}$ in the r.h.s of Inequality \eqref{prop2:eq9}. The triangle inequality implies that the following inequality holds
		\be\label{prop2:eq12}
		J_{xy} \geq \frac{1}{(2d+1)2^\alpha}\sum_{|x-x'| \leq 1}J_{x'y},
		\ee
		for every distinct pair of points $x,y \in \Z^d$. Thus, we have that
		\be\label{prop2:eq13}
		\sum_{\substack{x \in V(\Gamma_1)\\ y \in V(\Gamma_2)}}J_{xy}\mathbbm{1}_{\{\sigma_x \neq \sigma_y\}}\geq \frac{1}{(2d+1)2^\alpha}\sum_{\substack{x \in V(\Gamma_1)_0 \\ y \in V(\Gamma_2)}}J_{xy},
		\ee
		where $V(\Gamma_1)_0 = \{x \in V(\Gamma_1): \Theta_x(\sigma)=0\}$. Plugging Inequalities \eqref{prop2:eq10}, \eqref{prop2:eq11}, \eqref{eqfinal}, \eqref{prop2:eq13} into Equation \eqref{prop2:eq9}, we get
		\begin{align}\label{prop2:eq14}
			\sum_{\substack{x \in \I_-(\gamma)\\ y \in B(\gamma)}}J_{xy}\sigma_x\sigma_y &\leq \frac{3k^{(1)}_\alpha}{M}F_{\I_+(\gamma)} +\frac{2(2d)^{\frac{d}{d-1}}k_\alpha^{(1)}}{M^{\alpha-d}}|\gamma|-\frac{1}{(2d+1)2^{\alpha-1}}\sum_{\substack{x \in V(\Gamma_1)_0 \\ y \in V(\Gamma_2)}}J_{xy} \nonumber \\
			&- \sum_{\substack{x \in \I_+(\gamma)\\ y \in B(\gamma)\setminus V(\Gamma_2)}}J_{xy}- \sum_{\substack{x \in \I_+(\gamma)\setminus V(\Gamma_1) \\ y \in  V(\Gamma_2)}}J_{xy}.
		\end{align}
		We must add the regions with correct points into the sum depending on $V(\Gamma_1)_0$. But this is a simple task since we have,
		\be
		\sum_{\substack{x \in V(\Gamma_1)\setminus V(\Gamma_1)_0\\ y \in V(\Gamma_2)}}J_{xy} \leq \sum_{\substack{x \in V(\Gamma_1)\\ y \in V(\Gamma_2)}}J_{xy},
		\ee
		and proceeding as we did in \eqref{again} we arrive at the following inequality
		\be\label{prop2:eq15}
		\sum_{\substack{x \in \I_-(\gamma) \\ y \in B(\gamma)}}J_{xy}\sigma_x\sigma_y \leq \frac{4k^{(1)}_\alpha}{M}F_{\I_+(\gamma)} +\frac{3(2d)^{\frac{d}{d-1}}k_\alpha^{(1)}}{M^{\alpha-d}}|\gamma|-\frac{1}{(2d+1)2^{\alpha-1}}\sum_{\substack{x\in \I_+(\gamma) \\ y \in B(\gamma)}}J_{xy}.
		\ee
		Also, Inequality \eqref{prop2:eq12} implies that
		\be\label{prop2:eq16}
	   \begin{split}	
        \sum_{\substack{x \in \Sp(\gamma)\\ y \in \Z^d}}J_{xy}\mathbbm{1}_{\{\sigma_x \neq \sigma_y\}}+&\sum_{\substack{x \in \Sp(\gamma)\\ y \in \Sp(\gamma)^c}}J_{xy}\mathbbm{1}_{\{\sigma_x \neq \sigma_y\}} \geq \frac{1}{(2d+1)2^\alpha}\left(Jc_\alpha|\gamma|+ F_{\Sp(\gamma)}\right) \\
        &\geq \frac{1}{(2d+1)2^\alpha}\left(Jc_\alpha|\gamma|+ \frac{F_{\Sp(\gamma)}}{2}+\frac{1}{2}\sum_{\substack{x\in \Sp(\gamma) \\ y \in \I(\gamma)}}J_{xy}\right) 
        \end{split}
		\ee
		where $c_\alpha = \sum_{y \neq 0 \in \Z^d}\frac{1}{|y|^\alpha}$. Joining Inequalities \eqref{prop2:eq8},\eqref{prop2:eq15} and \eqref{prop2:eq16} into \eqref{prop2:eq2} yields
		\begin{align}
			H_\Lambda^-(\sigma)-H_\Lambda^-(\tau) &\geq \left( \frac{Jc_\alpha}{(2d+1)2^\alpha}- \frac{6(2d)^{\frac{d}{d-1}}k_\alpha^{(1)}}{M^{\alpha-d}}\right)|\gamma| + 2\left(\frac{1}{(2d+1)2^{\alpha+1}} -\frac{4k^{(1)}_\alpha}{M}\right)F_{\I_+(\gamma)} \nonumber \\
			&\left(\frac{1}{(2d+1)2^{\alpha+1}}-\frac{2k_\alpha^{(1)}}{M^{(\alpha-d)\wedge1}}\right)F_{\Sp(\gamma)}.
		\end{align}
		Letting $M> \max\{\frac{(2^r-1)^{d+1}}{k_d^d}, M_1, M_2\}$, where $$M_1^{\alpha-d}\coloneqq \frac{12(2d+1)(2d)^{\frac{d}{d-1}}k_\alpha^{(1)}2^{\alpha+1}}{Jc_\alpha},\text{ and } {M_2^{(\alpha-d)\wedge 1} \coloneqq (2d+1)k_\alpha^{(1)}2^{\alpha+4}},$$  we arrive at the desired result.
	\end{proof}
	
		As in the usual Peierls argument, Theorem \ref{thm1} will follow once we prove the following proposition.
	\begin{proposition}\label{prop3}
		Let $\alpha>d$ and $\delta>0$. For $\beta$ large enough, it holds that 
		\be
		\nu^-_{\beta,\textbf{h},\Lambda}(\sigma_0= + 1) < \frac{1}{2},
		\ee
		for every $\Lambda \in \mathcal{F}(\Z^d)$ when
		\begin{itemize}
			\item $d<\alpha<d+1$ and $\delta>\alpha -d$;
			\item $d<\alpha<d+1$ and $\delta=\alpha-d$ if $h^*$ is small enough;
			\item $\alpha \geq d+1$ and $\delta>1$;
			\item $\alpha \geq d+1$ and $\delta=1$ if $h^*$ is small enough.
		\end{itemize}
	\end{proposition}
	\begin{proof}
		Let $R>0$ and $(\widehat{h}_x)_{x \in\Z^d}$ be the truncated magnetic field defined in Equation \eqref{magfield2}. The constant $R$ will be chosen later. The existence of phase transition under the presence of the truncated field implies phase transition for the model with the decaying field (see Theorem 7.33 of \cite{Geo} for a more general statement). If $\sigma_0=+1$ there must exist a contour $\gamma$ such that $0 \in V(\gamma)$. Hence
		\[\label{prop3:eq1}
		\nu^-_{\beta,\bm{\widehat{h}},\Lambda}(\sigma_0= + 1) \leq \sum_{\substack{\gamma \in \mathcal{E}_\Lambda^- \\ 0 \in V(\gamma)}} \nu^-_{\beta,\bm{\widehat{h}},\Lambda}(\Omega(\gamma)). 
		\]
		Using Proposition \ref{prop2}, we know that the Hamiltonian $H^-_{\Lambda,\bm{\widehat{h}}}$ satisfies,
		\be\label{prop3:eq2}
		H^-_{\Lambda,\bm{\widehat{h}}}(\sigma)-H^-_{\Lambda,\bm{\widehat{h}}}(\tau(\sigma)) \geq c_2 |\gamma| + c_3 F_{\I_+(\gamma)} - 2\sum_{x \in \I_+(\gamma)\cup \Sp(\gamma)}\hat{h}_x,
		\ee
		where $\tau_\gamma(\sigma)= \tau(\sigma)$. Notice that 
		\[
		\sum_{x \in \Sp(\gamma)}\hat{h}_x \leq \frac{h^* |\gamma|}{R^\delta}.
		\]
		If $R^\delta > \frac{4h^*}{c_2}$ is sufficient to guarantee that the term $c_2 |\gamma|$ is larger than the field contribution in Inequality \eqref{prop3:eq2}. We want to prove that
		$
		c_3 F_{\I_+(\gamma)} - 2\sum_{x \in \I_+(\gamma)}\hat{h}_x\geq 0.
		$
		If $\I_+(\gamma) = \emptyset$ there is nothing to do, since the bound is trivial. Otherwise, we must analyse the competition of the decaying field with the different regimes of decay for the couplings constants $J_{xy}$.
		\begin{enumerate}[label=\textbf{(\roman*)}]
			\item \emph{Case} $d< \alpha < d+1$.
			By Lemmas \ref{lema3} and \ref{lema1}, we have 
			\be\label{prop3:eq3}
			c_3 F_{\I_+(\gamma)}- 2\sum_{x \in \I_+(\gamma)}\hat{h}_x \geq c_3K_\alpha |\I_+(\gamma)|^{2 - \frac{\alpha}{d}}- 2c_5|\I_+(\gamma)|^{1-\frac{\delta}{d}}.
			\ee
			Thus, if $\delta > \alpha - d$ and $|\I_+(\gamma)| \geq c'_\alpha \coloneqq \left(\frac{2c_5}{c_3 K_\alpha}\right)^{\frac{d}{\delta -(\alpha-d)}}$, we have that the r.h.s of Inequality \eqref{prop3:eq3} is nonnegative. In order to get a positive difference for all sizes of $\I_+(\gamma)$, we need to consider $R^\delta>R_1^\delta \coloneqq\frac{c'_\alpha}{c_3 K_\alpha}$. For the case $\delta=\alpha-d$, we must take $h^*$ small enough since the exponents in \eqref{prop3:eq3} will be equal.
			\item \emph{Case} $\alpha \geq d+1$.
			By Lemmas \ref{lema3} and \ref{lema1}, we have 
			\be\label{prop3:eq4}
			c_3 F_{\I_+(\gamma)}- 2\sum_{x \in \I_+(\gamma)}\hat{h}_x \geq c_3K_\alpha |\partial\I_+(\gamma)|- 2c_5|\I_+(\gamma)|^{1-\frac{\delta}{d}}.
			\ee
			Thus, if $\delta > 1$ and $|\I_+(\gamma)| \geq b_\alpha \coloneqq \left(\frac{c_5}{dc_3 K_\alpha}\right)^{\frac{d}{\delta-1}}$, we have that the r.h.s of Inequality \eqref{prop3:eq4} is nonnegative. In order to get a positive difference for all sizes of $\I_+(\gamma)$, we need to consider $R^\delta>R_2^\delta \coloneqq \frac{h^* b_\alpha}{dc_3 K_\alpha}$. The case where $\delta=1$, we must take $h^*$ small enough and use the isoperimetric inequality in Inequality \eqref{prop3:eq4}.
		\end{enumerate}
		It is clear that, by taking $R = \max\{\left(\frac{4h^*}{c_2}\right)^\frac{1}{\delta}, R_1, R_2\}$ together with \eqref{prop3:eq2} we get
		\[
		H^-_{\Lambda,\bm{\widehat{h}}}(\sigma)-H^-_{\Lambda,\bm{\widehat{h}}}(\tau(\sigma)) \geq \frac{c_2}{2} |\gamma|,
		\]
		which implies
		\be\label{eq112}
		\nu_{\beta,\bm{\widehat{h}},\Lambda}^-(\Omega(\gamma)) \leq \frac{e^{-\beta \frac{c_2}{2} |\gamma|}}{Z_{\beta, \bm{\widehat{h}}}^-(\Lambda)}\sum_{\sigma \in \Omega(\gamma)}e^{-\beta H^-_{\Lambda,\bm{\widehat{h}}}(\tau(\sigma))}.
		\ee
		Using the decomposition
		\[
		\Omega(\gamma) = \bigcup_{\Gamma: \Gamma \cup \{\gamma\} \in \mathcal{E}_\Lambda^-} \{ \sigma \in \Omega_\Lambda^-: \Gamma(\sigma)= \Gamma \cup \{\gamma\}\},
		\]
		together with the fact that, when we erase the contour $\gamma$, we may create new external contours but it always holds that $V(\Gamma(\tau(\sigma))) \subset \Lambda \setminus \Sp(\gamma)$. Hence, the r.h.s of Inequality \eqref{eq112} can be bounded as follows
		\begin{align*}
			\sum_{\sigma \in \Omega(\gamma)}e^{-\beta H^-_{\Lambda,\bm{\widehat{h}}}(\tau(\sigma))} &\leq \sum_{\substack{\Gamma \in \mathcal{E}^-_\Lambda \\ V(\Gamma)\subset \Lambda \setminus V(\gamma)}} \sum_{\substack{\Gamma' \in \mathcal{E}^-_\Lambda \\ V(\Gamma')\subset \I(\gamma)}}\sum_{\substack{\tau(\sigma) \\ \Gamma(\tau(\sigma))=\Gamma \cup \Gamma'}} \sum_{\omega: \tau(\omega)=\tau(\sigma)} e^{-\beta H^-_{\Lambda,\bm{\widehat{h}}}(\tau(\sigma))}\\
			&\leq |\{\sigma \in \Omega_{\Sp(\gamma)}: \Theta_x(\sigma)=0, \text{ for each } x \in \Sp(\gamma)\}| \sum_{\substack{\Gamma \in \mathcal{E}_{\Lambda}^- \\ V(\Gamma)\subset \Lambda \setminus \Sp(\gamma)}} \sum_{\sigma \in \Omega(\Gamma)}e^{-\beta H^-_{\Lambda,\bm{\widehat{h}}}(\sigma)}.
		\end{align*}
		Since the number of configurations that are incorrect in $\Sp(\gamma)$ is bounded by $2^{|\gamma|}$, we get
		\be\label{prop3:eq5}
		\nu_{\beta,\bm{\widehat{h}},\Lambda}^-(\Omega(\gamma)) \leq \frac{Z^-_{\beta, \bm{\widehat{h}}}(\Lambda\setminus\Sp(\gamma))e^{(\log(2)-\beta \frac{c_2}{2})|\gamma|}}{Z^-_{\beta, \bm{\widehat{h}}}(\Lambda)}.
		\ee
		Summing over all contours yields, together with Proposition \ref{importprop},
		\begin{align}
			\nu_{\beta,\bm{\widehat{h}},\Lambda}^-(\sigma_0 = +1) &\leq \sum_{\substack{\gamma \in \mathcal{E}_\Lambda^- \\ 0 \in V(\gamma)}}\frac{Z^-_{\beta, \bm{\widehat{h}}}(\Lambda\setminus\Sp(\gamma))e^{(\log(2)-\beta \frac{c_2}{2})|\gamma|}}{Z^-_{\beta, \bm{\widehat{h}}}(\Lambda)} \nonumber \\
			&\leq \sum_{m \geq 1}|\mathcal{C}_0(m)|e^{(\log(2)-\beta \frac{c_2}{2})m} \nonumber\\
			&\leq \sum_{m \geq 1}e^{(c_1 +\log(2) - \beta \frac{c_2}{2})m}< \frac{1}{2},
		\end{align}
		for $\beta$ large enough. 
	\end{proof}

\chapter{Conclusion and Further Research in Long-range spin systems}
\label{ch:conclusion2}
\epigraph{All our lives we postpone everything that can be postponed; perhaps we all have the certainty, deep inside, that we are immortal and that sooner or later every man will do everything, know all there is to know.}{\textbf{Jorge Luis Borges} \\ \textit{Funes el Memorioso}}

In Chapter \ref{ch:classtatmech} we briefly introduced the rigorous theory of classical equilibrium statistical mechanics through the Gibbsian specification formalism. We also introduced the concept of contours based on Pirogov-Sinai's theory. In Chapter \ref{ch:longrange}, we fully developed the multiscaled version of these contours through the use of $(M,a,r)$-partitions. Since long-range systems may have strong interactions between close regions, this concept allowed us to separate sets of incorrect points in regions that can be deemed weakly interacting. In this way, we were able to prove phase transition for long-range ferromagnetic Ising models when $d\geq 2$. As an application, we showed that the ferromagnetic Ising model with a decaying field of the form $h_x=h^*|x|^{-\delta}$ presents a phase transition at low temperatures when $\delta=\alpha-d$ we have to consider a small $h^*$. This is an indication that phase transition should not hold further into the region of the exponents. A similar phenomenon already happens in the short-range case, as studied by Bissacot, Cassandro, Cioletti, and Pressuti \cite{Bis2}, where the uniqueness of the Gibbs measure is true whenever $\delta<1$. 

Still on ferromagnetic models, our argument works for $d\geq 2$, but the notion of a $(M,a,r)$ partition only uses a notion of a metric to make sense. This brings up the question if the Peierls argument presented in Chapter \ref{ch:longrange} can be extended to the one-dimensional setting. It was shown by Littin and Picco \cite{LP} that the one-dimensional contours introduced in \cite{Cass} have the downside of not having a subadditive estimate for its energies when $\alpha \in (1,3-\log(3)/\log(2)]$. An interesting problem could be the extension of the argument to this setting. Also, the decay of correlation functions for polynomially decaying long-range Ising systems is a subtle matter and seems to be only studied for the one-dimensional long-range Ising model by Imbrie \cite{Imbr1} and Imbrie and Newman \cite{Imbr2}, trying to extend their results to higher dimensions seems another possible direction of further development. 

Another natural question is to investigate if we can extend the Peierls argument to more general interactions. For instance, the ferromagnet nearest-neighbor Ising model with a competing long-range antiferromagnet interaction, as considered in \cite{Bisk}. As stated in their paper, zero magnetization does not imply the absence of phase transition. Since the notion of incorrect points can be adapted to other systems, maybe some of the techniques developed here could be helpful to investigate the problem of phase transition in other models.

\chapter{Quantum Statistical Mechanics}
\label{ch:quantstatmech}

\epigraph{Going from Newton's mechanics to Einstein's must be, for the mathematician,  a bit like jumping from the good old Provençal dialect to the latest Parisian slang. On the other hand, going to quantum mechanics, I imagine, means going from French to Chinese.}{\textbf{Alexandre Grothendieck} \\ \textit{Récoltes et Semailles}}

	In this chapter, we will describe the basics of quantum statistical mechanics in the groupoid language. This encompasses basically the construction of the algebra of observables as a groupoid C*-algebra, as well as some general results on the existence of the infinite-volume dynamics, as well as make a quick review of KMS states, and the Gibbs-Araki-Ion condition for equilibrium, finishing the chapter with a proof of the equivalence between the DLR states introduce in Chapter 1 and KMS states for classical interactions. This result was first proved by Brascamp \cite{Bras} for lattice gases. All these constructions are quite standard and we cite the relevant literature as we delve into these matters. 

   The idea of using groupoids in quantum mechanics is not new. In \cite{Connes}, Connes argues that the first groupoid in physics was defined by Heisenberg's investigations on the foundations for a mathematical theory of quantum mechanics. He was not satisfied with the at the time ad-hoc methods of quantization such as the Bohr-Sommerfeld rule and his work was inspired by the Ritz-Ridberg combination principle for the transitions observed in the spectral lines of the hydrogen atom (see the book by Emch \cite{emch1984mathematical} for a wonderful account of the history of the concepts as well as the mathematics developed during this period). The idea of a groupoid was not recognized by Heisenberg, or by Born and Jordan, who instead looked to the algebra of observables and recognized the product of two of them as being the usual matrix product (although the matrices would need to be infinite). The idea of looking to the allowed transitions of the underlying quantum theory to be fundamental was later extended by Schwinger \cite{Schw} with what he coined to be the \emph{algebra of selected measurements}. This approach is being revisited and expanded by Ciaglia, Ibort and Marmo \cite{Ciaglia1} (many more works are being produced on this subject; we suggest the interested reader to look into the references therein but also on the papers that cite them). We, in this chapter, will not make justice to the rich subject of groupoid C*-algebras from either the mathematical and physical side of the subject. Most of the definitions and results stated in this section are present in \cite{putnam, Raszeja, Renault1980,  SimsSzaboWilliams2020}. The interested reader may check these references for further details, as well as the relevant literature cited before.

\section{Transformation Groupoids}
 
    A groupoid $\mathcal{G}$ is a set endowed with a partially\footnote{It can, of course, be defined in the whole domain $\mathcal{G}\times \mathcal{G}$. In this case, the groupoid will be a group.} defined \emph{product} operation $\circ:\mathcal{G}^{(2)}\subseteq\mathcal{G}\times \mathcal{G}\rightarrow \mathcal{G}$ and a globally defined \emph{inverse} $^{-1}:\mathcal{G}\rightarrow \mathcal{G}$.  When $(g,h) \in \mathcal{G}^{(2)}$ we say that $g\circ h$ is \emph{defined}.
      \begin{description}
    \item[\textbf{(Inverse)}]  for any $g \in \mathcal{G}$, $g^{-1} \circ g$ and $g\circ g^{-1}$ are always defined.
    \item[\textbf{(Associativity)}] If $g_1 \circ g_2$ and $g_2 \circ g_3$ are defined, then $(g_1 \circ g_2) \circ g_3 = g_1 \circ (g_2 \circ g_3)$.
    \item[\textbf{(Identity)}] if $g_1 \circ g_2$ is defined, $g_1^{-1} \circ g_1 \circ g_2 = g_2$ and $g_1 \circ g_2 \circ g_2^{-1}=g_1$.    
 \end{description}

We often call the elements of the groupoid by \emph{arrows}. There is a distinguished subset of $\mathcal{G}$ called \emph{unit space} that is defined as
\[
\mathcal{G}^{(0)}\coloneqq \{g\circ g^{-1}: g \in \mathcal{G}\}.
\]
This is exactly the space of all the identities of the groupoid. For every groupoid $\mathcal{G}$, there are distinguished maps $r,s:\mathcal{G}\rightarrow\mathcal{G}^{(0)}$ called, respectively, the \emph{range} and \emph{source} maps, and are given by
\[
r(g) = g \circ g^{-1}, \quad \text{ and } \quad s(g) = g^{-1} \circ g.
\]
Other important subsets of 
a groupoid $\mathcal{G}$ are, for each $g \in \mathcal{G}$, 
\[
\mathcal{G}^{r(g)} = r^{-1}(\{r(g)\}) \quad \text{ and } \quad \mathcal{G}_{s(g)} = s^{-1}(\{s(g)\}).
\]
\begin{figure}[h]
	\centering
	
	\tikzset{every picture/.style={line width=0.7pt}} 
	\resizebox{1\textwidth}{!}{
		
		\begin{tikzpicture}[x=0.75pt,y=0.75pt,yscale=-1,xscale=1]
			
			
			\draw  [line width=0.75]  (470.5,75) .. controls (604.5,70) and (849.5,89) .. (904.5,157) .. controls (959.5,225) and (789.5,273) .. (800.5,335) .. controls (811.5,397) and (995.5,524) .. (901.5,601) .. controls (807.5,678) and (558.5,669) .. (475.5,667) .. controls (392.5,665) and (170.5,656) .. (86.5,580) .. controls (2.5,504) and (186.5,465) .. (157.5,398) .. controls (128.5,331) and (-24.5,294) .. (30.5,216) .. controls (85.5,138) and (336.5,80) .. (470.5,75) -- cycle ;
			\draw [line width=0.75]    (329,288.75) .. controls (348.3,235.67) and (323.83,224.97) .. (301.43,215.3) ;
			\draw [shift={(299,214.25)}, rotate = 383.5] [fill={rgb, 255:red, 0; green, 0; blue, 0 }  ][line width=0.08]  [draw opacity=0] (10.72,-5.15) -- (0,0) -- (10.72,5.15) -- (7.12,0) -- cycle    ;
			
			\draw [line width=0.75]    (409,200.25) .. controls (398.22,230.14) and (402.81,246.1) .. (431.07,280.22) ;
			\draw [shift={(432.83,282.33)}, rotate = 230.02] [fill={rgb, 255:red, 0; green, 0; blue, 0 }  ][line width=0.08]  [draw opacity=0] (10.72,-5.15) -- (0,0) -- (10.72,5.15) -- (7.12,0) -- cycle    ;
			
			\draw [line width=0.75]    (262.3,522.67) .. controls (279.83,463.89) and (326.89,457.68) .. (356.97,504.65) ;
			\draw [shift={(358.33,506.83)}, rotate = 238.7] [fill={rgb, 255:red, 0; green, 0; blue, 0 }  ][line width=0.08]  [draw opacity=0] (10.72,-5.15) -- (0,0) -- (10.72,5.15) -- (7.12,0) -- cycle    ;
			
			\draw [line width=0.75]    (367,507.5) .. controls (395.91,466.32) and (442.69,456.85) .. (459.79,525.88) ;
			\draw [shift={(460.3,528)}, rotate = 256.85] [fill={rgb, 255:red, 0; green, 0; blue, 0 }  ][line width=0.08]  [draw opacity=0] (10.72,-5.15) -- (0,0) -- (10.72,5.15) -- (7.12,0) -- cycle    ;
			
			\draw [line width=0.75]    (261.63,537.33) .. controls (288.83,616.93) and (406.58,633.7) .. (458.19,542.06) ;
			\draw [shift={(458.97,540.67)}, rotate = 478.78] [fill={rgb, 255:red, 0; green, 0; blue, 0 }  ][line width=0.08]  [draw opacity=0] (10.72,-5.15) -- (0,0) -- (10.72,5.15) -- (7.12,0) -- cycle    ;
			
			\draw [line width=0.75]    (609.5,532.5) .. controls (636.7,612.1) and (752.6,628.7) .. (804.19,537.06) ;
			\draw [shift={(804.97,535.67)}, rotate = 478.78] [fill={rgb, 255:red, 0; green, 0; blue, 0 }  ][line width=0.08]  [draw opacity=0] (10.72,-5.15) -- (0,0) -- (10.72,5.15) -- (7.12,0) -- cycle    ;
			
			\draw [line width=0.75]    (804.5,523) .. controls (774.28,444.5) and (658.27,424.94) .. (610.22,518.42) ;
			\draw [shift={(609.5,519.83)}, rotate = 296.59000000000003] [fill={rgb, 255:red, 0; green, 0; blue, 0 }  ][line width=0.08]  [draw opacity=0] (10.72,-5.15) -- (0,0) -- (10.72,5.15) -- (7.12,0) -- cycle    ;
			
			\draw  [fill={rgb, 255:red, 0; green, 0; blue, 0 }  ,fill opacity=1 ] (803,529.5) .. controls (803,527.57) and (804.57,526) .. (806.5,526) .. controls (808.43,526) and (810,527.57) .. (810,529.5) .. controls (810,531.43) and (808.43,533) .. (806.5,533) .. controls (804.57,533) and (803,531.43) .. (803,529.5) -- cycle ;
			\draw  [fill={rgb, 255:red, 0; green, 0; blue, 0 }  ,fill opacity=1 ] (290.67,212.5) .. controls (290.67,210.57) and (292.23,209) .. (294.17,209) .. controls (296.1,209) and (297.67,210.57) .. (297.67,212.5) .. controls (297.67,214.43) and (296.1,216) .. (294.17,216) .. controls (292.23,216) and (290.67,214.43) .. (290.67,212.5) -- cycle ;
			\draw  [fill={rgb, 255:red, 0; green, 0; blue, 0 }  ,fill opacity=1 ] (323,294.83) .. controls (323,292.9) and (324.57,291.33) .. (326.5,291.33) .. controls (328.43,291.33) and (330,292.9) .. (330,294.83) .. controls (330,296.77) and (328.43,298.33) .. (326.5,298.33) .. controls (324.57,298.33) and (323,296.77) .. (323,294.83) -- cycle ;
			\draw  [fill={rgb, 255:red, 0; green, 0; blue, 0 }  ,fill opacity=1 ] (433.67,286.5) .. controls (433.67,284.57) and (435.23,283) .. (437.17,283) .. controls (439.1,283) and (440.67,284.57) .. (440.67,286.5) .. controls (440.67,288.43) and (439.1,290) .. (437.17,290) .. controls (435.23,290) and (433.67,288.43) .. (433.67,286.5) -- cycle ;
			\draw  [fill={rgb, 255:red, 0; green, 0; blue, 0 }  ,fill opacity=1 ] (408.67,193.83) .. controls (408.67,191.9) and (410.23,190.33) .. (412.17,190.33) .. controls (414.1,190.33) and (415.67,191.9) .. (415.67,193.83) .. controls (415.67,195.77) and (414.1,197.33) .. (412.17,197.33) .. controls (410.23,197.33) and (408.67,195.77) .. (408.67,193.83) -- cycle ;
			\draw  [fill={rgb, 255:red, 0; green, 0; blue, 0 }  ,fill opacity=1 ] (258,529.5) .. controls (258,527.57) and (259.57,526) .. (261.5,526) .. controls (263.43,526) and (265,527.57) .. (265,529.5) .. controls (265,531.43) and (263.43,533) .. (261.5,533) .. controls (259.57,533) and (258,531.43) .. (258,529.5) -- cycle ;
			\draw  [fill={rgb, 255:red, 0; green, 0; blue, 0 }  ,fill opacity=1 ] (358,512.5) .. controls (358,510.57) and (359.57,509) .. (361.5,509) .. controls (363.43,509) and (365,510.57) .. (365,512.5) .. controls (365,514.43) and (363.43,516) .. (361.5,516) .. controls (359.57,516) and (358,514.43) .. (358,512.5) -- cycle ;
			\draw  [fill={rgb, 255:red, 0; green, 0; blue, 0 }  ,fill opacity=1 ] (457.33,534.83) .. controls (457.33,532.9) and (458.9,531.33) .. (460.83,531.33) .. controls (462.77,531.33) and (464.33,532.9) .. (464.33,534.83) .. controls (464.33,536.77) and (462.77,538.33) .. (460.83,538.33) .. controls (458.9,538.33) and (457.33,536.77) .. (457.33,534.83) -- cycle ;
			\draw  [fill={rgb, 255:red, 0; green, 0; blue, 0 }  ,fill opacity=1 ] (605.67,525.5) .. controls (605.67,523.57) and (607.23,522) .. (609.17,522) .. controls (611.1,522) and (612.67,523.57) .. (612.67,525.5) .. controls (612.67,527.43) and (611.1,529) .. (609.17,529) .. controls (607.23,529) and (605.67,527.43) .. (605.67,525.5) -- cycle ;
			
			\draw (391.83,237) node  [font=\Large]  {$g$};
			\draw (415,172) node  [font=\Large]  {$s( g)$};
			\draw (435,301) node  [font=\Large]  {$r( g)$};
			\draw (350,241) node  [font=\Large]  {$h$};
			\draw (324,312) node  [font=\Large]  {$s( h)$};
			\draw (294,191) node  [font=\Large]  {$r( h)$};
			\draw (113,232) node  [font=\Huge]  {$\mathcal{G}^{( 0)}$};
			\draw (422,457) node  [font=\Large]  {$g_{2}$};
			\draw (361,623) node  [font=\Large]  {$g_{2} g_{1}$};
			\draw (362,526) node  [font=\large]  {$r( g_{1}) =s( g_{2})$};
			\draw (312,456) node  [font=\Large]  {$g_{1}$};
			\draw (232,528) node  [font=\Large]  {$s( g_{1})$};
			\draw (491,534) node  [font=\Large]  {$r( g_{2})$};
			\draw (247,188) node  [font=\LARGE] [align=left] {I.};
			\draw (245,441) node  [font=\LARGE] [align=left] {II.};
			\draw (600,443) node  [font=\LARGE] [align=left] {III.};
			\draw (838,527) node  [font=\Large]  {$r(g)$};
			\draw (578,524) node  [font=\Large]  {$s(g)$};
			\draw (699,618) node  [font=\Large]  {$g$};
			\draw (702,437) node  [font=\Large]  {$g^{-1}$};

		\end{tikzpicture}
		
	}
	
\caption{Modified from Thiago Raszeja PhD Thesis \cite{Raszeja}. From left to right: I. Is the picture of seen elements of the groupoid as arrows. II. The product is similar to the composition of functions. III. The definitions of the range as source maps, as being, respectively, the end and beginning of the arrow $g$.}
\end{figure}
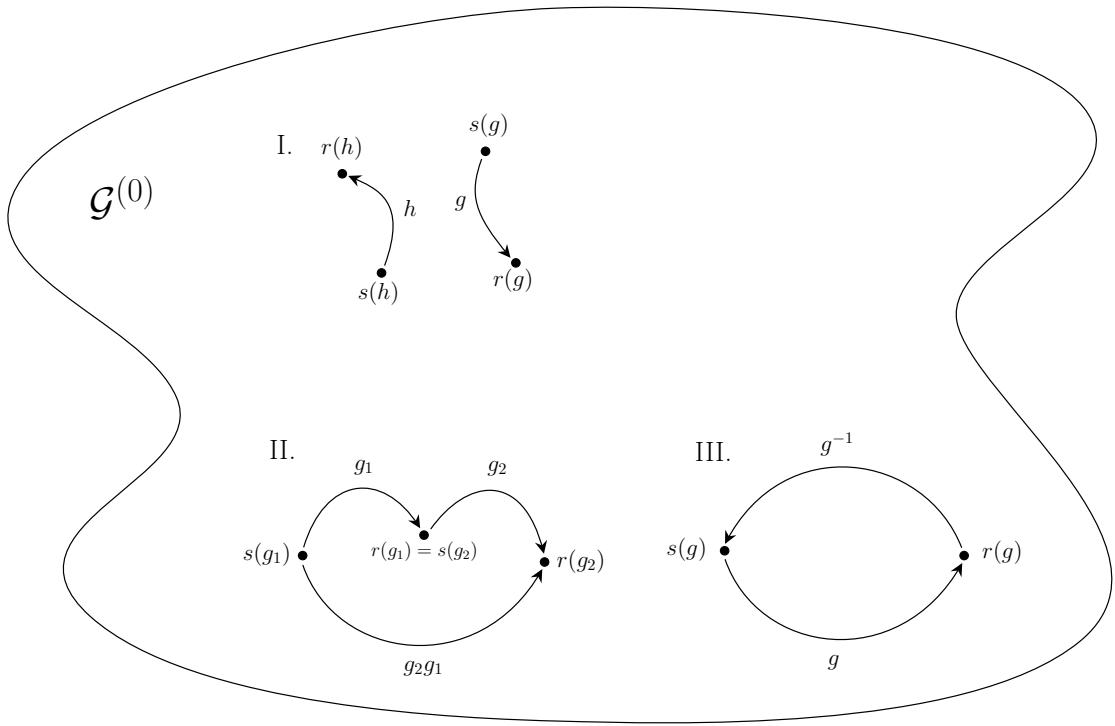

\pagebreak

We give some elementary examples of groupoids. The groupoid that will be our object of study will be presented with further details later.  
\begin{example}
	\begin{itemize}
		\item[(i)] Every group $G$ is a groupoid with their product and inverse structure, where $\mathcal{G}^{(0)} = \{e\}$, the identity element.
		\item[(ii)] For a set $X$, an equivalence relation $R\subset X \times X$ can be used to give a groupoid structure on $X\times X$. The inverse of $(x,y)^{-1} = (y,x)$ and the product is $(x,y)\circ(y,z) = (x,z)$. 
		\item[(iii)] Let $\mathcal{G}=\bigcup_{n\geq 1}\mathrm{GL}_n(\mathbb{C})$, the union of all invertible $n\times n$ complex matrices. The product is only defined when two matrices have the same size. The inverse is the usual inverse of matrices. Notice that the unit space is the union of all identity matrices.  
	\end{itemize}
\end{example}
A \emph{topological groupoid} is a groupoid $\mathcal{G}$ endowed with a topology such that the product and the inverse are continuous operations. This is sufficient for the range and source maps also to be continuous, where $\mathcal{G}^{(0)}$ is endowed with the subspace topology since the range and source maps can be written in terms of the product and inverse operations. Inspired by \cite{Gruber}, we will consider the state space as having a group structure. Let $G_q \subset \mathbb{S}^1$ be the subgroup of all $q$-th roots of unity, i.e.,
\[
G_q = \{e^{2\pi i\frac{k}{q}}: k=0,1, \dots, q-1\}.
\]
Define the action $\alpha: G_q \times G_q \rightarrow G_q$ given by multiplication
\[
\alpha(e^{2\pi i \frac{k}{q}},e^{2\pi i \frac{m}{q}}) = e^{2\pi i \frac{k+m}{q}}.
\]
The \emph{transformation groupoid} $G_q \rtimes_\alpha G_q$ is the set $G_q\times G_q$ with the following groupoid structure
\begin{align*}
	&(e^{2\pi i \frac{k_2}{q}},e^{2\pi i \frac{m_2}{q}})\circ (e^{2\pi i \frac{k_1}{q}},e^{2\pi i \frac{m_1}{q}}) = (e^{2\pi i \frac{k_1}{q}},e^{2\pi i \frac{m_1+m_2}{q}}), \text{ if } k_1+m_1=k_2 \mod q \\
	&(e^{2\pi i \frac{k_1}{q}},e^{2\pi i \frac{m_1}{q}})^{-1} = (e^{2\pi i \frac{k_1+m_1}{q}},e^{2\pi i \frac{(q-m_1)}{q}}).
\end{align*}
We can endow $G_q\rtimes_\alpha G_q$ with the discrete topology, making it a topological groupoid. For each subset $\Lambda \subset \Z^d$, consider
\[
\Omega_\Lambda = \prod_{x\in \Lambda}G_q \quad \text{and} \quad G_\Lambda = \bigoplus_{x \in \Lambda} G_q.
\]
When $\Lambda = \Z^d$, we will write $\Omega_{\Z^d} \coloneqq \Omega$ and $G_{\Z^d}\coloneqq G$. The identity element of the group $G_\Lambda$ as $1$. Let action $\alpha_\Lambda : \Omega_\Lambda \times G_\Lambda \rightarrow \Omega_\Lambda$ be the product action given by
\[
\alpha_\Lambda(\sigma,g) = (\alpha_x(\sigma_x,g_x))_{x\in \Lambda}, 
\] 
where $g = \oplus_{x\in \Lambda}g_x$ and $\alpha_x$ is the corresponding copy of the action defined previously. We will omit $\alpha$ to denote the action anymore and instead adopt the more convenient notation $\alpha(\sigma_\Lambda,g_\Lambda) = g_\Lambda \sigma_\Lambda$.  For $q=2$, this action is known as \emph{spin-flip}. The transformation groupoid $\mathcal{G}_\Lambda = \Omega_\Lambda \rtimes_{\alpha_\Lambda} G_\Lambda$ then is the product space $\Omega_\Lambda \times G_\Lambda$ with the groupoid structure given by
\be\label{prod_inv}
(\sigma_\Lambda, g_\Lambda)\circ(\omega_\Lambda,h_\Lambda) = (\sigma_\Lambda,h_\Lambda g_\Lambda) \text{  if  } g_\Lambda\sigma_\Lambda = \omega_\Lambda \quad \text{and}
\quad (\sigma_\Lambda, g_\Lambda)^{-1} = (g_\Lambda\sigma_\Lambda,g_\Lambda^{-1}).
\ee
when $\Lambda = \Z^d$ we will once more suppress it from the notation, i.e.,  $\mathcal{G}\coloneqq \mathcal{G}_{\Z^d}$. Let $\supp{g_\Lambda} = \{x\in \Lambda: g_x \neq 1\}$. Since $G_\Lambda$ is the direct sum, $\supp g_\Lambda$ is a finite subset of $\Lambda$. The unit space of this groupoid then is the set
\[
\mathcal{G}^{(0)}_\Lambda= \{(\sigma_\Lambda,1): \sigma_\Lambda \in \Omega_\Lambda\}.
\]
The unit space $\mathcal{G}^{(0)}_\Lambda$ can be identified with the configuration space $\Omega_\Lambda$, and for now on we will make this identification and also write $\Omega_\Lambda$ for the unit space of $\mathcal{G}_\Lambda$. The range and source maps are
\[
r(\sigma_\Lambda,g_\Lambda) = g_\Lambda \sigma_\Lambda
\quad \text{ and } \quad s(\sigma_\Lambda,g_\Lambda) = \sigma_\Lambda  
\]
We can give $\Omega_\Lambda$ the product topology, making it a compact metrizable space. For the group $G_\Lambda$ when $\Lambda$ is finite, the direct sum of groups is isomorphic to the product of the groups, thus using the product topology is enough. There is a distinction between the direct sum and the product of groups when $\Lambda$ has infinitely many points and the question of what topology to use may be nontrivial. Despite that, we will always consider the group with the discrete topology. Then, we can endow the groupoid $\mathcal{G}_\Lambda$ with the product topology of $\Omega_\Lambda$ and $G_\Lambda$.
 
 The following proposition is a straightforward consequence of us choosing the product topology. 
\begin{proposition}
	The product and inverse operations defined in \eqref{prod_inv} are continuous functions in $\mathcal{G}_\Lambda$ equipped with the product topology.
\end{proposition}
\begin{proof}
   Notice that the product topology on $\Omega_\Lambda$ is metrizable. We can choose a metric on $\Omega_\Lambda$ defined by
    \[
    d(\omega_\Lambda, \sigma_\Lambda) = \sum_{x\in\Lambda}\frac{|\omega_x-\sigma_x|}{2^{|x|}},
    \]
    where $|\omega_x-\sigma_x|$ is the absolute value of the complex number. With this choice, the action of $G_\Lambda$ is isometric and, therefore continuous. The continuity of the inverse is due to the continuity of the inverse operation of the group $G_\Lambda$ together with the continuity of the action on $\Omega_\Lambda$. The product is the restriction of the map on $\mathcal{G}_\Lambda \times \mathcal{G}_\Lambda$ that sends two arrows $(\sigma_\Lambda,g_\Lambda)$ and $(\omega_\Lambda,h_\Lambda)$ to an arrow $(\sigma_\Lambda,h_\Lambda g_\Lambda)$. This map is continuous in the product topology of $\mathcal{G}_\Lambda \times \mathcal{G}_\Lambda$ and to get that the actual groupoid product is continuous one just needs to notice that the product is the restriction of the map just cited in the set of composable arrows. 
\end{proof}
\begin{remark}
	We want to stress that the product topology is not always the obvious choice in order to have a topological groupoid. For instance, an important class of groupoids in the theory of dynamical systems and thermodynamic formalism is known as Renault-Deaconu groupoids \cite{Renault1980, Raszeja}. Although they have a product structure as a set, the topology used to transform it into a topological groupoid is not the product topology. 
\end{remark}

We finish this section with an important proposition about the topology of $\mathcal{G}_\Lambda$.
\begin{proposition}\label{prop_clopen}
	For each $g_\Lambda \in G_\Lambda$, we have that the set $\Omega_\Lambda\times \{g_\Lambda\}$ is a clopen subset of the groupoid $\mathcal{G}_\Lambda$. In particular, the unit space is a clopen subset of the groupoid. 
\end{proposition}
\begin{proof}
	We prove the Proposition only for the unit space, the general case being a straightforward generalization. That the configuration space $\Omega_\Lambda$ is closed is obvious since the topology of the groupoid is the product one. Since the topology on the group $G_\Lambda$ is the discrete one, the only converging sequences are the eventually constant ones, yielding that the set
	\[
	\mathcal{G}_\Lambda \setminus \Omega_\Lambda=\{(\sigma_\Lambda,g_\Lambda): g_\Lambda \neq 1\},
	\]
	is also closed, finishing the proof. 
\end{proof}

This proposition will have important consequences in the next section when we construct a C*-algebra with the groupoid. 

\section{The Groupoid C*-algebra}

We introduced in the last subsection the definition of the transformation groupoids that we will use, and now we will sketch how to construct C*-algebras with them. The last result in this subsection will be an isomorphism between the usual spin-algebra of Bratteli-Robinson's book \cite{Bra2} and the C*-algebra of the transformation groupoid. The results are standard and can be consulted in \cite{putnam, Raszeja, Renault1980,  SimsSzaboWilliams2020}. 
Consider 
\[
C_c(\mathcal{G}_\Lambda)\coloneqq \{ f:\mathcal{G}_\Lambda \rightarrow \mathbb{C}: f \text{ is continuous and} \supp f \text{ is compact}\}, 
\]
the continuous complex-valued functions defined on the groupoid with compact support. Important examples are the functions defined as follows. Let $A$ be a clopen compact subset of $\mathcal{G}_\Lambda$. Then we can define the delta function $\delta_A$ by 
\[
\delta_A(\sigma_\Lambda,g_\Lambda) = \begin{cases}
	1,& (\sigma_\Lambda,g_\Lambda) \in A \\
	0,& (\sigma_\Lambda,g_\Lambda) \not\in A.
\end{cases}
\]
If $\Lambda$ is finite, the groupoid $\mathcal{G}_\Lambda$ have 
a discrete topology, the unitary sets $\{(\sigma_\Lambda,g_\Lambda)\}$ are clopen. In order to lighten the notation, we will write $\delta_{(\sigma_\Lambda,g_\Lambda)}$ for the delta function on $\{(\sigma_\Lambda,g_\Lambda)\}$ and, when $g_\Lambda=1$ we will only write $\delta_{\sigma_\Lambda}$. Using these delta functions together with Proposition \ref{prop_clopen} we can extend every continuous function $f:\Omega_\Lambda\rightarrow \mathbb{C}$ on the configuration space to a continuous function with compact support in the groupoid by the formula 
\[
f(\sigma_\Lambda,g_\Lambda) =
f(\sigma_\Lambda) \delta_{\Omega_\Lambda}(\sigma_\Lambda,g_\Lambda).
\]
 thus recovering the classical observable algebra as a subspace of $C_c(\mathcal{G}_\Lambda)$. We will describe briefly how to use $C_c(\mathcal{G}_\Lambda)$ to construct a C*-algebra. Since it has an obvious complex vector space structure, we will start by describing the product of two functions and the adjoint operation, giving us a $*$-algebra. The convolution product of two functions $f_1,f_2 \in C_c(\mathcal{G}_\Lambda)$ 
\be\label{prod_groupoid}
f_1 \cdot f_2(\sigma_\Lambda,g_\Lambda) = \sum_{(\omega_\Lambda,h_\Lambda)\in \mathcal{G}^{g_\Lambda\sigma_\Lambda}_\Lambda}\hspace{-0.5cm}f_1(\omega_\Lambda,h_\Lambda)f_2(\sigma_\Lambda, h_\Lambda^{-1}g_\Lambda),
\ee
and the adjoint is
\be\label{adj_groupoid}
f_1^*(\sigma_\Lambda,g_\Lambda) = \overline{f_1(g_\Lambda\sigma_\Lambda,g_\Lambda^{-1})}.
\ee
The product \eqref{prod_groupoid} can be written in a different way, depending on the set of all arrows with the same source $\mathcal{G}_{\Lambda,\sigma_{\Lambda}}=\{(\sigma,g_\Lambda): g_\Lambda \in G_\Lambda\}$ by using a change of variable
\be\label{prod_groupoid_source}
f_1\cdot f_2 (\sigma_\Lambda,g_\Lambda) = \sum_{(\sigma_\Lambda,k_\Lambda)\in \mathcal{G}_{\Lambda,\sigma_\Lambda}}\hspace{-0.5cm}f_1(k_\Lambda\sigma_\Lambda,g_\Lambda k_\Lambda^{-1})f_2(\sigma_\Lambda,k_\Lambda).
\ee
\begin{proposition}
    The operations \eqref{prod_groupoid} and \eqref{adj_groupoid} make $C_c(\mathcal{G}_\Lambda)$ into a $*$-algebra.
\end{proposition}
\begin{proof}

In our case, since the group $\mathcal{G}_\Lambda$ is discrete and countable, in order for a continuous function $f$ to have compact support it must satisfy 
\be\label{compactsupport}
|G_{\Lambda,f}|\coloneqq |\{g_\Lambda \in G_\Lambda: \exists \sigma_\Lambda \in \Omega_\Lambda \text{ s.t. } f(\sigma_\Lambda,g_\Lambda) \neq 0\}|<\infty,
\ee
since by Proposition \ref{prop_clopen}, we can decompose $\mathcal{G}_\Lambda$ in a countable union of the open sets $\Omega_\Lambda\times\{g_\Lambda\}$. Actually, the condition \eqref{compactsupport} is sufficient to have compact support, since the configuration space $\Omega_\Lambda$ is compact with the product topology. It holds, 
\begin{equation*}
    \begin{split}
    |f_1\cdot f_2(\sigma_\Lambda,g_\Lambda)|\leq \sum_{h_\Lambda \in G_\Lambda}|f_1(h_\Lambda^{-1} g_\Lambda\sigma_\Lambda,h_\Lambda)||f_2(\sigma_\Lambda,h_\Lambda^{-1}g_\Lambda)|\leq \|f_2\|_\infty\|f_1\|_\infty |G_{\Lambda,f_1}|,
    \end{split}
\end{equation*}
because $(\omega_\Lambda,h_\Lambda)\in \mathcal{G}_\Lambda^{g_\Lambda\sigma_\Lambda}$ implies $\omega_\Lambda = h^{-1}_\Lambda g_\Lambda \sigma_\Lambda$ and where $\|f\|_\infty$ is the supremum norm. Therefore, the product \eqref{prod_groupoid} is well defined since it can only contain a finite amount of nonzero terms. Notice that the convolution product can be written as
\begin{equation*}
    \begin{split}
    f_1\cdot f_2(\sigma_\Lambda,g_\Lambda)= \sum_{h_\Lambda \in G_\Lambda}f_1(h_\Lambda^{-1} g_\Lambda\sigma_\Lambda,h_\Lambda) f_2(\sigma_\Lambda,h_\Lambda^{-1}g_\Lambda) \\
    =\sum_{h_\Lambda \in G_{\Lambda,f_1}}f_1(h_\Lambda^{-1} g_\Lambda\sigma_\Lambda,h_\Lambda) f_2(\sigma_\Lambda,h_\Lambda^{-1}g_\Lambda),
    \end{split}
\end{equation*}
and it will not be zero only when also $h_\Lambda^{-1}g_\Lambda \in G_{\Lambda,f_2}$, thus if $g_\Lambda \in h_\Lambda G_{\Lambda,f_2}$ for some $h_\Lambda \in G_{\Lambda,f_1}$. Therefore $f_1\cdot f_2(\sigma_\Lambda,g_\Lambda)\neq 0$ only if $g_\Lambda \in G_{\Lambda,f_1}G_{\Lambda,f_2}$, thus $G_{\Lambda,f_1\cdot f_2}\subset G_{\Lambda, f_1} G_{\Lambda,f_2}$, yielding that the support of $f_1\cdot f_2$ is compact. The continuity of the convolution product follows from the continuity of the action. \footnote{This holds in much more generality see Proposition 1.1 in \cite{Renault1980} and also Lemmas 3.21 and 3.23 in \cite{Raszeja}.} 
The continuity of adjoint \ref{adj_groupoid} follows from the continuity of the inverse operation of the groupoid. Notice that $G_{\Lambda,f}^{-1}= G_{\Lambda,f^*}$, thus the support of the adjoint is still compact. 
\end{proof}

The important feature of the product formula \eqref{prod_groupoid} is that the convolution product of two continuous functions on the configuration space is just the pointwise product of the classical algebra $C(\Omega_\Lambda)$, therefore we can see $C(\Omega_\Lambda)$ as an abelian subalgebra of $C_c(\mathcal{G}_\Lambda)$. The product of a function $f_1\in C(\Omega_\Lambda)$ and a general $f_2 \in C_c(\mathcal{G}_\Lambda)$ can readily be calculated to be
\be\label{prod_classical_quantum}
\begin{split}
	f_1\cdot f_2(\sigma_\Lambda,g_\Lambda) &= \sum_{(\omega_\Lambda,h_\Lambda)\in \mathcal{G}_\Lambda^{g_\Lambda\sigma_\Lambda}} f_1(\omega_\Lambda,h_\Lambda) f_2(\sigma_\Lambda,h_\Lambda^{-1}g_\Lambda) \\
	&= \sum_{(\omega_\Lambda,h_\Lambda)\in \mathcal{G}_\Lambda^{g_\Lambda\sigma_\Lambda}} f_1(\omega_\Lambda)\delta_{\{h_\Lambda=1\}}(\omega_\Lambda,h_\Lambda) f_2(\sigma_\Lambda,h_\Lambda^{-1}g_\Lambda) \\
	&=f_1(g_\Lambda\sigma_\Lambda)f_2(\sigma_\Lambda,g_\Lambda),
\end{split}
\ee 
thus showing that the product by classical functions is almost like a pointwise product. The next lemma will be useful and give us a nice formula of how the delta functions introduced earlier behave with respect to the convolution product.

\begin{lemma}\label{matrix_elements}
	 For every $\Lambda \Subset \Z^d$ and function $f\in C_c(\mathcal{G}_\Lambda)$ we have
	\be\label{eq3:poirep}
	\delta_{(\omega_\Lambda,h_\Lambda)} \cdot f \cdot \delta_{(\eta_\Lambda,k_\Lambda)} = f(k_\Lambda \eta_\Lambda, h_\Lambda^{-1}\widetilde{g}_\Lambda k_\Lambda^{-1})\delta_{(\eta_\Lambda, \widetilde{g}_\Lambda)},
	\ee
	where $\widetilde{g}_\Lambda\in G_\Lambda$ is the unique element such that $ \widetilde{g}_\Lambda\eta_\Lambda=h_\Lambda\omega_\Lambda$.
\end{lemma}
\begin{proof}
	Using the product formula \eqref{prod_groupoid_source} we have
	\[
	f\cdot \delta_{(\eta_\Lambda,k_\Lambda)}(\sigma_\Lambda,g_\Lambda) = \sum_{(\sigma_\Lambda,k'_\Lambda) \in \mathcal{G}_{\Lambda,\sigma_\Lambda}}\hspace{-0.5cm}f(k'_\Lambda\sigma_\Lambda,g_\Lambda (k'_\Lambda)^{-1})\delta_{(\eta_\Lambda,k_\Lambda)}(\sigma_\Lambda, k'_\Lambda) = f(k_\Lambda\eta_\Lambda,g_\Lambda k_\Lambda^{-1})\delta_{\mathcal{G}_{\Lambda,\eta_\Lambda}}(\sigma_\Lambda,g_\Lambda).
	\]
	Plug the formula above into the product with $\delta_{(\omega_\Lambda,h_\Lambda)}$ to get
	\begin{align*}
		\delta_{(\omega_\Lambda,h_\Lambda)}\cdot f \cdot \delta_{(\eta_\Lambda,k_\Lambda)}(\sigma_\Lambda,g_\Lambda) &= \sum_{(\theta_\Lambda,k'_\Lambda) \in \mathcal{G}_\Lambda^{g_\Lambda\sigma_\Lambda}}\delta_{(\omega_\Lambda,h_\Lambda)}(\theta_\Lambda,k'_\Lambda) (f\cdot\delta_{(\eta_\Lambda,k_\Lambda)})(\sigma_\Lambda,(k_\Lambda')^{-1}g_\Lambda) \\
		&= \delta_{\mathcal{G}_\Lambda^{h_\Lambda \omega_\Lambda}}(\sigma_\Lambda,g_\Lambda)(f\cdot \delta_{(\eta_\Lambda,k_\Lambda)})(\sigma_\Lambda,h_\Lambda^{-1}g_\Lambda)\\
		&= f(k_\Lambda\eta_\Lambda,h_\Lambda^{-1}g_\Lambda k_\Lambda^{-1})\delta_{\mathcal{G}_\Lambda^{h_\Lambda \omega_\Lambda}}(\sigma_\Lambda,g_\Lambda)\delta_{\mathcal{G}_{\Lambda,\eta_\Lambda}}(\sigma_\Lambda,g_\Lambda).
	\end{align*}
	Since the action is free, we know that $\delta_{\mathcal{G}_\Lambda^{h_\Lambda \omega_\Lambda}}(\sigma_\Lambda,g_\Lambda)\delta_{\mathcal{G}_{\Lambda,\eta_\Lambda}}(\sigma_\Lambda,g_\Lambda)$ is not zero only at one arrow, namely, the arrow $(\eta_\Lambda,\widetilde{g}_\Lambda)$, where $\widetilde{g}_\Lambda$ is the unique group element such that $\widetilde{g}_\Lambda\eta_\Lambda =h_\Lambda \omega_\Lambda$.
\end{proof}
 To finish the construction of the C*-algebra through the groupoid, we proceed to construct a C*-norm. Just to remember, $\mathcal{G}_{\Lambda,\sigma_\Lambda}=s^{-1}(\{\sigma_\Lambda\})$ is the set where all arrows have $\sigma_\Lambda$ as its source. For each $\sigma_\Lambda\in \Omega_\Lambda$, we can consider the Hilbert space 
\[ \ell^2(\mathcal{G}_{\Lambda,\sigma_\Lambda}) = \left\{\xi:\mathcal{G}_{\Lambda,\sigma_\Lambda}\rightarrow \mathbb{C}: \sum_{(\sigma_\Lambda,g_\Lambda) \in \mathcal{G}_{\Lambda,\sigma_\Lambda}} |\xi(\sigma_\Lambda,g_\Lambda)|^2 < \infty\right\}
\]
Where the inner product is given by
\[
\langle \xi_1,\xi_2 \rangle = \sum_{(\sigma_\Lambda,g_\Lambda) \in \mathcal{G}_{\Lambda,\sigma_\Lambda}}\overline{\xi_1(\sigma_\Lambda,g_\Lambda)}\xi_2(\sigma_\Lambda,g_\Lambda).
\]
With the inner product above the delta functions $\delta_{(\sigma_\Lambda,g_\Lambda)}$ for $g_\Lambda \in G_\Lambda$ can be shown to be an orthonormal basis for $\ell^2(\mathcal{G}_{\Lambda,\sigma_\Lambda})$. Just for concreteness, let us further analyze the Hilbert space $\ell^2(\mathcal{G}_{\Lambda,\sigma_\Lambda})$. Notice that $\mathcal{G}_{\Lambda,\sigma_\Lambda}$ is isomorphic to $G_\Lambda$, yielding $\ell^2(\mathcal{G}_{\Lambda,\sigma_\Lambda})\simeq \ell^2(G_\Lambda)$. When $\Lambda$ is finite, this is a finite-dimensional Hilbert space. For each configuration $\sigma_\Lambda \in \Omega_\Lambda$ we can construct a $*$-representation $\pi^{\sigma_\Lambda}:C_c(\mathcal{G}_\Lambda)\rightarrow B(\ell^2(\mathcal{G}_{\Lambda,\sigma_\Lambda}))$ defined as
\[
\pi^{\sigma_\Lambda}(f)\xi(\sigma_\Lambda,g_\Lambda) = \hspace{-0.5cm}\sum_{(\sigma_\Lambda,h_\Lambda)\in \mathcal{G}_{\Lambda, \sigma_\Lambda}} \hspace{-0.5cm}f(h_\Lambda\sigma_\Lambda,g_\Lambda h_\Lambda^{-1}) \xi(\sigma_\Lambda,h_\Lambda),
\]
\begin{proposition}
    The map $\pi^{\sigma_\Lambda}$ is a well-defined *-representation of $C_c(\mathcal{G}_\Lambda)$ into the Hilbert space $\ell^2(\mathcal{G}_{\Lambda,\sigma_\Lambda})$.
\end{proposition}
\begin{proof}
Notice that $(h_\Lambda\sigma_\Lambda,g_\Lambda h_\Lambda^{-1}) = (\sigma_\Lambda,h_\Lambda)^{-1}\circ (\sigma_\Lambda, g_\Lambda)$, hence by writing $$f_{g_\Lambda}(\sigma_\Lambda,h_\Lambda) = \overline{f((\sigma_\Lambda,h_\Lambda)^{-1}\circ (\sigma_\Lambda, g_\Lambda))},$$ the following form for the representation $\pi^{\sigma_\Lambda}$ follows
\[
\pi^{\sigma_\Lambda}(f)\xi(\sigma_\Lambda,g_\Lambda) = \langle f_{g_\Lambda},\xi\rangle.
\]
It always holds that $f_{g_\Lambda}\in \ell^2(\mathcal{G}_{\Lambda,\sigma_\Lambda})$ since the function $f$ has compact support, therefore $\pi^{\sigma_\Lambda}(f)$ is well-defined as a linear operator in $\ell^2(\mathcal{G}_{\Lambda,g_\Lambda})$. Indeed, for every $\xi \in \ell^2(\mathcal{G}_{\Lambda,\sigma_\Lambda})$ we have
\begin{equation*}
\begin{split}
    \|\pi^{\sigma_\Lambda}(f)\xi\|^2 &= \hspace{-0.5cm}\sum_{(\sigma_\Lambda,g_\Lambda)\in \mathcal{G}_{\Lambda,\sigma_\Lambda}}\hspace{-0.5cm}\overline{\langle f_{g_\Lambda},\xi\rangle}\langle f_{g_\Lambda},\xi\rangle \leq \|\xi\|^2\hspace{-0.5cm}\sum_{(\sigma_\Lambda,g_\Lambda)\in \mathcal{G}_{\Lambda,\sigma_\Lambda}}\hspace{-0.5cm}\|f_{g_\Lambda}\|^2,
\end{split}    
\end{equation*}
where we used Cauchy-Schwarz inequality in the last passage above. Notice that this implies that $\pi^{\sigma_\Lambda}(f)$ is a bounded linear operator, since $\|f_{g_\Lambda}\|$ is not zero only in at most a finite number of values for $g_\Lambda$ since $f$ has compact support. It is easy to see also that $\pi^{\sigma_\Lambda}$ is linear as a map from $C_c(\mathcal{G}_\Lambda)$ to $B(\ell^2(\mathcal{G}_{\Lambda,\sigma_\Lambda})$. To finish the proof, we must show that $\pi^{\sigma_\Lambda}$ respects the product and the adjoint structure. Fix $f_1,f_2 \in C_c(\mathcal{G}_\Lambda)$. Then it holds
\begin{equation*}
    \begin{split}
        \pi^{\sigma_\Lambda}(f_1\cdot f_2)\xi(\sigma_\Lambda,g_\Lambda) &= \hspace{-0.5cm}\sum_{(\sigma_\Lambda,h_\Lambda)\in \mathcal{G}_{\Lambda,\sigma_\Lambda}}\hspace{-0.5cm}f_1\cdot f_2(h_\Lambda\sigma_\Lambda, g_\Lambda h_\Lambda^{-1})\xi(\sigma_\Lambda,h_\Lambda) \\
        &=\hspace{-0.5cm}\sum_{\substack{(\sigma_\Lambda,h_\Lambda)\in \mathcal{G}_{\Lambda,\sigma_\Lambda} \\(h_\Lambda\sigma_\Lambda,k_\Lambda)\in \mathcal{G}_{\Lambda,h_\Lambda\sigma_\Lambda}}}\hspace{-0.5cm}f_1(k_\Lambda h_\Lambda \sigma_\Lambda, g_\Lambda h_\Lambda^{-1}k_\Lambda^{-1}) f_2(h_\Lambda\sigma_\Lambda, k_\Lambda)\xi(\sigma_\Lambda,h_\Lambda) \\
        &=\hspace{-0.5cm}\sum_{\substack{(\sigma_\Lambda,h_\Lambda)\in \mathcal{G}_{\Lambda,\sigma_\Lambda} \\(h_\Lambda\sigma_\Lambda,k_\Lambda)\in \mathcal{G}_{\Lambda,h_\Lambda\sigma_\Lambda}}}\hspace{-0.5cm}f_1((\sigma_\Lambda,k_\Lambda h_\Lambda)^{-1}(\sigma_\Lambda,g_\Lambda)) f_2((\sigma_\Lambda, h_\Lambda)^{-1}(\sigma_\Lambda, k_\Lambda h_\Lambda))\xi(\sigma_\Lambda,h_\Lambda).
    \end{split}
\end{equation*}
The innermost sum above sum dependent on $\mathcal{G}_{\Lambda,h_\Lambda\sigma_\Lambda}$ is, essentially, a sum over $k_\Lambda \in G_\Lambda$. By making a change of variables $k_\Lambda h_\Lambda = l_\Lambda$, this sum then can written as a sum over $(\sigma_\Lambda,l_\Lambda)\in \mathcal{G}_{\Lambda,\sigma_\Lambda}$, thus
\begin{equation*}
\begin{split}
\pi^{\sigma_\Lambda}(f_1\cdot f_2)\xi(\sigma_\Lambda,g_\Lambda) &= \hspace{-0.5cm}\sum_{\substack{(\sigma_\Lambda,h_\Lambda), (\sigma_\Lambda,l_\Lambda) \in \mathcal{G}_{\Lambda,\sigma_\Lambda}}}\hspace{-1cm}f_1((\sigma_\Lambda,l_\Lambda)^{-1}(\sigma_\Lambda,g_\Lambda)) f_2((\sigma_\Lambda, h_\Lambda)^{-1}(\sigma_\Lambda, l_\Lambda))\xi(\sigma_\Lambda,h_\Lambda) \\
&= \hspace{-0.5cm}\sum_{ (\sigma_\Lambda,l_\Lambda) \in \mathcal{G}_{\Lambda,\sigma_\Lambda}}\hspace{-0.5 cm}f_1((\sigma_\Lambda,l_\Lambda)^{-1}(\sigma_\Lambda,g_\Lambda)) \pi^{\sigma_\Lambda}(f_2)\xi(\sigma_\Lambda,l_\Lambda) \\
&= \pi^{\sigma_\Lambda}(f_1)\pi^{\sigma_\Lambda}(f_2)\xi(\sigma_\Lambda,g_\Lambda).
\end{split}
\end{equation*}
Since the equality above holds for all $\xi$, we get that $\pi^{\sigma_\Lambda}$ is multiplicative. The last part is to show that the adjoint operation is respected. Let $\xi_1,\xi_2 \in \ell^2(\mathcal{G}_{\Lambda,\sigma_\Lambda})$ we have
\begin{equation*}
    \begin{split}
        \langle\xi_1,\pi^{\sigma_\Lambda}(f^*)\xi_2\rangle &= \sum_{(\sigma_\Lambda,g_\Lambda)\in \mathcal{G}_{\Lambda,\sigma_\Lambda}} \overline{\xi_1(\sigma_\Lambda,g_\Lambda)}\pi^{\sigma_\Lambda}(f^*)\xi_2(\sigma_\Lambda,g_\Lambda) \\
        &=\sum_{(\sigma_\Lambda,g_\Lambda)\in \mathcal{G}_{\Lambda,\sigma_\Lambda}} \overline{\xi_1(\sigma_\Lambda,g_\Lambda)}\pi^{\sigma_\Lambda}(f^*)\xi_2(\sigma_\Lambda,g_\Lambda) \\
        &=\sum_{(\sigma_\Lambda,g_\Lambda)\in \mathcal{G}_{\Lambda,\sigma_\Lambda}} \sum_{(\sigma_\Lambda,h_\Lambda)\in \mathcal{G}_{\Lambda,\sigma_\Lambda}}\overline{\xi_1(\sigma_\Lambda,g_\Lambda)f(g_\Lambda\sigma_\Lambda, h_\Lambda g_\Lambda^{-1})}\xi_2(\sigma_\Lambda,h_\Lambda) \\
        &= \sum_{(\sigma_\Lambda,h_\Lambda) \in \mathcal{G}_{\Lambda,h_\Lambda}}\overline{\pi^{\sigma_\Lambda}(f)\xi_1(\sigma_\Lambda,h_\Lambda)}\xi_2(\sigma_\Lambda,h_\Lambda) \\
        &=\langle \pi^{\sigma_\Lambda}(f)\xi_1,\xi_2\rangle = \langle \xi_1,\pi^{\sigma_\Lambda}(f)^*\xi_2\rangle.
    \end{split}
\end{equation*}
Since the calculation above holds for any $\xi_1,\xi_2$, we have $\pi^{\sigma_\Lambda}(f^*) = \pi^{\sigma_\Lambda}(f)^*$ finishing our proof.
\end{proof}
Again for concreteness, let us calculate the action of $\pi^{\sigma_\Lambda}(f)$ on a delta function
$\delta_{(\sigma_\Lambda,h_\Lambda)}$ when $\Lambda$ is finite. We have
\begin{equation}\label{rep_acting_delta}
\begin{split}
	\pi^{\sigma_\Lambda}(f)\delta_{(\sigma_\Lambda,g_\Lambda)}(\sigma_\Lambda,h_\Lambda) &= \hspace{-0.5cm}\sum_{(\sigma_\Lambda,k_\Lambda) \in \mathcal{G}_{\Lambda,\sigma_\Lambda}}\hspace{-0.5cm}f(k_\Lambda\sigma_\Lambda,h_\Lambda k_\Lambda^{-1})\delta_{(\sigma_\Lambda,g_\Lambda)}(\sigma_\Lambda,k_\Lambda) \\ 
 &= f(g_\Lambda \sigma_\Lambda, h_\Lambda g_\Lambda^{-1})\delta_{(\sigma_\Lambda,g_\Lambda)}(\sigma_\Lambda,h_\Lambda).
 \end{split}
\end{equation}   
  
The \emph{left regular representation} $\pi:C_c(\mathcal{G}_\Lambda)\rightarrow B(\ell^2(\mathcal{G}_\Lambda))$  is defined by 
\[
\pi(f)\Bigg(\sum_{(\sigma_\Lambda,g_\Lambda)\in \mathcal{G}_\Lambda}\lambda_{(\sigma_\Lambda,g_\Lambda)}\delta_{(\sigma_\Lambda,g_\Lambda)}\Bigg) = \sum_{(\sigma_\Lambda,g_\Lambda)\in \mathcal{G}_\Lambda}\lambda_{(\sigma_\Lambda,g_\Lambda)}\pi^{\sigma_\Lambda}(f)\delta_{(\sigma_\Lambda,g_\Lambda)},
\]
where
\[
\ell^2(\mathcal{G}_\Lambda) \coloneqq \bigoplus_{\sigma_\Lambda \in \Omega_\Lambda}\ell^2(\mathcal{G}_{\Lambda,\sigma_\Lambda}).
\]

\begin{proposition}
A continuous function $f\in C_c(\mathcal{G}_\Lambda)$ is positive as an operator if and only if for every $n\geq 1$ and $(g_{\Lambda,i}\sigma_\Lambda,g_{\Lambda,i})$, $i=1,\dots,n$ it holds
	\be\label{positivity_eq1}
	\sum_{i,j=1}^n \overline{\lambda_i}\lambda_j f(g_{\Lambda,i}\sigma_\Lambda,g_{\Lambda,j}g_{\Lambda,i}^{-1}) \geq 0,
	\ee
	for $\lambda_i \in \mathbb{C}$. Moreover, the left-regular representation $\pi$ is faithful.
\end{proposition}
\begin{proof}
	To show that $\pi(f)$ is positive we only need to check Inequality \eqref{positivity_eq1} for finite combinations of the form
	\[
	\psi_{\sigma_\Lambda} = \sum_{(\sigma_\Lambda,g_\Lambda) \in \mathcal{G}_{\Lambda,\sigma_\Lambda}}\lambda_{(\sigma_\Lambda,g_\Lambda)}\delta_{(\sigma_\Lambda,g_\Lambda)}.
	\]
	Then,
	\begin{align*}
		\langle \psi_{\sigma_\Lambda}, \pi(f)\psi_{\sigma_\Lambda}\rangle &= \sum_{\substack{(\sigma_\Lambda,g_\Lambda)\in \mathcal{G}_{\Lambda,\sigma_\Lambda} \\ (\sigma_\Lambda,h_\Lambda)\in \mathcal{G}_{\Lambda,\sigma_\Lambda}}}\overline{\lambda_{(\sigma_\Lambda,g_\Lambda)}}\lambda_{(\sigma_\Lambda,h_\Lambda)}\langle \delta_{(\sigma_\Lambda,g_\Lambda)}, \pi^{\sigma_\Lambda}(f) \delta_{(\sigma_\Lambda,h_\Lambda)} 
		\rangle \\
		&= \sum_{\substack{(\sigma_\Lambda,g_\Lambda)\in \mathcal{G}_{\Lambda,\sigma_\Lambda} \\ (\sigma_\Lambda,h_\Lambda)\in \mathcal{G}_{\Lambda,\sigma_\Lambda}}}\sum_{(\sigma_\Lambda,k_\Lambda) \in \mathcal{G}_{\Lambda,\sigma_\Lambda}}\overline{\lambda_{(\sigma_\Lambda,g_\Lambda)}}\lambda_{(\sigma_\Lambda,h_\Lambda)}f(h_\Lambda\sigma_\Lambda,k_\Lambda h_\Lambda^{-1})\langle \delta_{(\sigma_\Lambda,g_\Lambda)}, \delta_{(\sigma_\Lambda,k_\Lambda)} 
		\rangle \\
		&=\sum_{\substack{(\sigma_\Lambda,k_\Lambda)\in \mathcal{G}_{\Lambda,\sigma_\Lambda} \\ (\sigma_\Lambda,h_\Lambda)\in \mathcal{G}_{\Lambda,\sigma_\Lambda}}}\overline{\lambda_{(\sigma_\Lambda,k_\Lambda)}}\lambda_{(\sigma_\Lambda,h_\Lambda)}f(h_\Lambda\sigma_\Lambda,k_\Lambda h_\Lambda^{-1}),
	\end{align*}
	Hence if Inequality \eqref{positivity_eq1} holds the function $f$ must be positive.  Let $f \in C_c(\mathcal{G}_\Lambda)$. Then if $f$ is positive but $\pi(f)=0$, we know that for any $\psi \in \ell^2(\mathcal{G}_\Lambda)$ we know that $\pi(f)\psi=0$. Then, taking $\psi=\delta_{(\sigma_\Lambda,g_\Lambda)}$ we get
	\[
	\langle \delta_{(\sigma_\Lambda,g_\Lambda)},\pi^{\sigma_\Lambda}(f)\delta_{(\sigma_\Lambda, g_\Lambda)}\rangle = \sum_{(\sigma_\Lambda,h_\Lambda)\in \mathcal{G}_{\Lambda,\sigma_\Lambda}}f(g_\Lambda\sigma_\Lambda, h_\Lambda g_\Lambda^{-1}) \langle \delta_{(\sigma_\Lambda,g_\Lambda)},\delta_{(\sigma_\Lambda,h_\Lambda)}\rangle = f(g_\Lambda \sigma_\Lambda,1)=0,
	\]
	concluding that $f(\sigma_\Lambda,1) = 0$ for any $\sigma_\Lambda \in \Omega_\Lambda$. Choosing $\psi = \delta_{(\sigma_\Lambda,1)}+\delta_{(\sigma_\Lambda,h_\Lambda)}$ we get
	\begin{align*}
		\langle \psi, \pi^{\sigma_\Lambda}(f)\psi\rangle &= 2 \langle \delta_{(\sigma_\Lambda,1)},\pi^{\sigma_\Lambda}(f)\delta_{(\sigma_\Lambda,h_\Lambda)}\rangle \\
		& = 2\sum_{(\sigma_\Lambda,k_\Lambda)\in \mathcal{G}_{\Lambda,\sigma_\Lambda}}f(h_\Lambda \sigma_\Lambda, k_\Lambda h_\Lambda^{-1})\langle\delta_{(\sigma_\Lambda,1)},\delta_{(\sigma_\Lambda,k_\Lambda)}\rangle \\
		&=2f(h_\Lambda\sigma_\Lambda, h_\Lambda^{-1}) = 0,
	\end{align*}
	hence $f = 0$. 
\end{proof}
\begin{remark}
	The notion of positivity differs from simply being positive in every element of the groupoid. Indeed, being nonnegative in every element of the groupoid does not guarantee that $f$ is positive as an operator. Nonetheless, for functions $f \in C(\Omega_\Lambda)$, for it to be positive as an operator it only needs to be nonnegative in each point of the configuration space.
\end{remark}
The C*-algebra norm that we will choose to give to $C_c(\mathcal{G}_\Lambda)$ is $\|f\|_r = \|\pi(f)\|$ and $C^*(\mathcal{G}_\Lambda)$ is the completion of $C_c(\mathcal{G}_\Lambda)$ with respect to the norm $\|\cdot\|_r$. The C*-algebra $C^*(\mathcal{G}_\Lambda)$ is called the \emph{reduced C* algebra} of the groupoid $\mathcal{G}_\Lambda$. Since the groupoids $\mathcal{G}_\Lambda$ are amenable (see \cite{Renault1980}), this is isomorphic to other construction called the full C* algebra.
More details on this construction can be found in for groupoid $\text{C}^*$-algebras  are \cite{Paterson, putnam, Renault1980, Renault2,  SimsSzaboWilliams2020}.

\subsection{The quasilocal algebra}

Let us introduce some important functions. To get acquainted with the groupoid picture, we first define and study properties of these functions for the case of one lattice point $x$ then we define the general picture later. For the groupoid $\mathcal{G}_x$ let $u_x,v_x: \mathcal{G}_x\rightarrow \mathbb{C}$ be the functions
\be\label{algebra_generator}
u_x(\sigma_x,g_x) = 
\sigma_x \delta_{\{g_x=1\}}  
\quad
\text{ and }
\quad 
v_x(\sigma_x,g_x) = \delta_{\{g_x=z_q\}}, 
\ee
where $z_q \coloneqq e^{\frac{2\pi i}{q}}$. These functions are continuous since the groupoid $\mathcal{G}_x$ has the discrete topology. They also satisfy the following relations
\[
u_x^q(\sigma_x,g_x) = \sum_{(\omega_x,h_x)\in \mathcal{G}_x^{g_x\sigma_x}}u_x(\omega_x,h_x)u^{q-1}_x(\sigma_x,h_x^{-1}g_x) = g_x\sigma_x u_x^{q-1}(\sigma_x,g_x)
\]
By iterating the above procedure we get $u_x^q(\sigma_x,g_x) = (\sigma_x)^q \delta_{\{g_x=1\}} = \delta_{\{g_x=1\}}$, 
where the last equality follows from the fact that $\sigma_x \in G_q$. A similar argument holds for $v_x$, 
\[
v_x^q(\sigma_x,g_x) = \sum_{(\omega_x,h_x)\in \mathcal{G}_x^{g_x\sigma_x}}v_x(\omega_x,h_x)v^{q-1}_x(\sigma_x,h_x^{-1}g_x) = v_x^{q-1}(\sigma_x,z_q^{-1}g_x).
\]
Iterating the formula above yields $v_x^q(\sigma_x,g_x)= v_x(\sigma_x,z_q^{-(q-1)}g_x) $ which will not be zero only when $g_x = z_q^q = 1$.  Also, 
\[
u_x \cdot v_x (\sigma_x,g_x)  = \sum_{(\omega_x,h_x)\in \mathcal{G}_x^{g_x\sigma_x}} u_x(\omega_x,h_x)v_x(\sigma_x,h_x^{-1}g_x) = g_x\sigma_x v_x(\sigma_x,g_x) 
\]
If we change the order of multiplication we get
\[
v_x \cdot u_x (\sigma_x,g_x)  = \sum_{(\omega_x,h_x)\in \mathcal{G}_x^{g_x\sigma_x}} v_x(\omega_x,h_x)u_x(\sigma_x,h_x^{-1}g_x) = \sigma_x v_x(\sigma_x,g_x) 
\]
since the products are only different from zero when $g_x=1$, we know that $g_x\sigma_x = z_q \sigma_x$, thus $u_x \cdot v_x = z_q v_x u_x$. This shows that the convolution product in groupoids is not commutative in general.

The \emph{observables algebra}\footnote{Some authors reserve the term observable only for the self-adjoint operators in a C*-algebra. Here, we follow the terminology in Bratteli and Robinson book \cite{Bra2}.} in quantum statistical mechanics as in Bratteli-Robinson's books \cite{Bra1,Bra2}, is the inductive limit C$^*$-algebra, constructed in section 6.2.1 of \cite{Bra2}. For completeness, we proceed to briefly explain this construction and show that it is a groupoid C*-algebra. For each finite set $\Lambda \in \mathcal{F}(\Z^d)$, consider the local algebras as $\mathfrak{A}_\Lambda = \otimes_{x\in \Lambda}M_q(\mathbb{C}) \simeq M_{q^{|\Lambda|}}(\mathbb{C})$ the tensor product of $|\Lambda|$ copies of the $q\times q$ matrices. Then, when $\Lambda \subset \Lambda'$, define the inclusion maps $\varphi_{\Lambda',\Lambda}: \mathfrak{A}_{\Lambda}\rightarrow \mathfrak{A}_{\Lambda'}$
\[
\varphi_{\Lambda',\Lambda}(a) = a \otimes \mathbbm{1}.
\]
These maps have some important properties such as
\begin{itemize}
	\item[(i)] \textbf{(Composition)} If $\Lambda \subset \Lambda' \subset \Lambda''$ then $\varphi_{\Lambda'',\Lambda'}\circ \varphi_{\Lambda',\Lambda}= \varphi_{\Lambda'',\Lambda}$.
	\item[(ii)] \textbf{(Norm-preserving)} For every $\mathfrak{A}_{\Lambda}$, if $\|\cdot\|_\Lambda$ is its C*-norm, we have $\|a\|_\Lambda = \|\varphi_{\Lambda',\Lambda}(a)\|_{\Lambda'}$.
	\item[(iii)]\textbf{(Locality)} For $a \in \mathfrak{A}_\Lambda$ and $b\in \mathfrak{A}_{\Lambda'}$ if $\Lambda'\cap \Lambda = \emptyset$ then, for any $\Lambda\cup\Lambda' \subset\Lambda''$ we have
	\[
	[\varphi_{\Lambda'',\Lambda}(a),\varphi_{\Lambda'',\Lambda'}(b)]=0.
	\]
\end{itemize}

The composition property is essential to construct an inductive limit of the algebras and the norm-preserving property is paramount in order to complete the algebraic direct limit to a C* algebra (for a detailed construction, check Appendix L in \cite{Olsen}). According to the construction of inductive limits, there exists a C*-algebra $\mathfrak{A}_q$ and *-homomorphisms $\varphi_\Lambda:\mathfrak{A}_\Lambda\rightarrow \mathfrak{A}_q$ such that, for any $\Lambda \subset \Lambda'$ we have $\varphi_{\Lambda'} \circ \varphi_{\Lambda',\Lambda} = \varphi_\Lambda$. The \emph{spin algebra} is 
\[
\mathfrak{A}_q = \overline{\bigcup_{\Lambda \in \mathcal{F}(\Z^d)}\mathfrak{A}_\Lambda},
\]
where we tacitly identified the algebras $\mathfrak{A}_\Lambda$ with $\varphi_\Lambda(\mathfrak{A}_\Lambda)$. The algebras $\mathfrak{A}_q$ and $\mathfrak{A}_\Lambda$ share the same identity element $\mathfrak{1}$ by construction. The separation property manifests the physical principle that observables localized in disjoint systems should be independent of each other. This algebra is also known as the UHF-algebra of type $q^\infty$. 
\pagebreak
\begin{theorem}
	The C$^*$ algebra $C^*(\mathcal{G})$ is isomorphic to the spin algebra $\mathfrak{A}_q$.
\end{theorem}
\begin{proof}
	By the universal property of the inductive limit, we need only to find morphisms $\psi_\Lambda:\mathfrak{A}_\Lambda \rightarrow C^*(\mathcal{G})$ such that $\psi_{\Lambda'} \circ \varphi_{\Lambda,\Lambda'} = \psi_\Lambda$, for any $\Lambda \subset \Lambda'$. Since each matrix algebra $M_q(\mathbb{C})$ is generated by the matrices $U,V$ given by 
	\begin{equation*}
		U = 
		\begin{pmatrix}
			1 & 0 & \dots & \dots & 0\\
			0 & z_q & 0 & \dots & 0 \\
			0 & 0 & z_q^2 & \dots & 0 \\
			\vdots & \ddots &\ddots& \ddots & \vdots \\
			0& \dots & \dots& 0 & z_q^{q-1}
		\end{pmatrix}
		\;\;\;\;
		\text{and}
		\;\;\;\;
		V = 
		\begin{pmatrix}
			0 & 1 & \dots & \dots & 0\\
			0 & 0 & 1 & \dots & 0 \\
			\vdots & \ddots &\ddots& \ddots & \vdots  \\
			0 & 0 &\dots& \dots & 1 \\
			1 & \dots & \dots& 0 & 0
		\end{pmatrix}.
	\end{equation*}
	These matrices satisfy $UV = z_q VU$, $U^q = V^q = \mathbbm{1}$. The matrices $U,V$ generate the algebra of $q \times q $ matrices, i.e., 
	\[
	\{U^k V^\ell: k,\ell = 0, \dots, q-1\},
	\]
	is a basis for $M_q(\mathbb{C})$ as a complex vector space. Indeed, using the inner product $\langle A,B\rangle \coloneqq \text{Tr}(A^*B)$, we get that 
	\[
	\text{Tr}((U^{k_1}V^{\ell_1})^*U^{k_2}V^{\ell_2}) = \text{Tr}((V^{\ell_1})^*(U^{k_1})^*U^{k_2}V^{\ell_2}) = \text{Tr}((U^{k_1})^*U^{k_2}V^{\ell_2}(V^{\ell_1})^*) 
	\]
	But $(U^k)^* = U^{q-k}$ and $(V^k)^*=V^{q-k}$
	\[
	\text{Tr}((U^{k_1})^*U^{k_2}V^{\ell_2}(V^{\ell_1})^*) = \text{Tr}((U^{k_2-k_1}V^{\ell_2-\ell_1}) = \delta_{k_1,k_2}\delta_{\ell_1,\ell_2},
	\]
	yielding us that the elements $U^kV^\ell$ are linearly independent. Since we have $q^2$ matrices, their span must be all the matrices. This fact implies that every diagonal matrix is a polynomial in $U$. The groupoid $\mathcal{G}_x$ is discrete with $q^2$ points. Thus $C(\mathcal{G}_x)$ is a complex vector space of dimension $q^2$. Thus defining the map $\psi_x:C(\mathcal{G}_x)\rightarrow M_q(\mathbb{C})$ by
	\[
	\psi_x(U^kV^\ell) = u_x^k \cdot v_x^\ell, \quad \text{for } k,\ell=0,\dots, q-1, 
	\]
	where $u_x$ and $v_x$ are the functions defined in \eqref{algebra_generator}. Extending $\psi_x$ linearly we get an isomorphism of $C_c(\mathcal{G}_x)$ and $M_q(\mathbb{C})$ as vector spaces. It is easy to verify that $\psi_x$ is actually an $*$-isomorphism between the C*-algebras, since $u_x,v_x$ and $U,V$ satisfy the same relations. For every finite $\Lambda \subset \Z^d$, we have $C_c(\mathcal{G}_\Lambda) \simeq \bigotimes_{x\in \Lambda}C_c(\mathcal{G}_x)$, since $C_c(X\times Y) = C_c(X)\otimes C_c(Y)$ for locally compact Hausdorff spaces $X$ and $Y$. Thus we can define $*$-isomorphisms of $C^*(\mathcal{G}_\Lambda)$ with the local algebras $\mathfrak{A}_\Lambda$ with the maps $\psi_\Lambda =\otimes_{x\in\Lambda}\psi_x$. For every $\Delta \subset \Lambda$, we can embed $C_c(\mathcal{G}_\Delta)$ in $C_c(\mathcal{G}_\Lambda)$ as a subalgebra using the following extension
	\[
	f(\sigma_\Lambda,g_\Lambda) = f(\sigma_\Delta,g_\Delta)\delta_{\{g_{\Lambda\setminus\Delta}=1\}}.
	\]
	This extension defines a family of inclusion maps $i_{\Delta,\Lambda}:C^*(\mathcal{G}_\Delta) \rightarrow C^*(\mathcal{G}_\Lambda)$ satisfying the properties of the composition, norm-preservation, and spatial separation introduced earlier. Thus, together with the isomorphisms $\psi_\Lambda$, the universal property of C*-inductive limit guarantees that $C^*(\mathcal{G})\simeq \mathfrak{A}_q$.
\end{proof}

The theorem above has some interesting consequences for continuous functions on the groupoid. 

\begin{corollary}\label{local_functions}
	Every continuous function $f\in C(\mathcal{G}_\Lambda)$ is given by
	\[
	f = \sum_{k_A, \ell_B} c_{k_A,\ell_B} u_A^{k_A} \cdot v_B^{\ell_B}, 
	\]
	where we used a multi-index notation,
	\[
	u_A^{k_A} = \Conv_{x \in A} u_x^{k_x} \quad \text{ and } \quad v_B^{\ell_B} = \Conv_{x \in B} v_x^{\ell_x},
	\]
	with $\ell_x, k_x = 1,\dots,q-1$, and $\Conv$ is the convolution product. 
\end{corollary} 

We also make the assumption that if $A$ or $B$ is the empty set then $u_\emptyset = v_\emptyset =\mathbbm{1}$. For the important abelian subalgebra $C(\Omega)$, we get that every local function is a polynomial on the variables $u_x$, for $x\in \Lambda$ finite. We can identify a subset of the C*-algebra $C^*(\mathcal{G})$ with 
\[
\bigcup_{\Lambda \in \mathcal{F}(\Z^d)}C(\mathcal{G}_\Lambda).
\]
These will be called the \emph{local operators}, using the standard nomenclature of quantum statistical mechanics.  

\subsection{The Jordan-Wigner Transformation}

The quasilocal algebra of operators has a very peculiar structure regarding the local observables $u_x$ and $v_x$, namely, the commutation of observables whenever they are in algebras localized in disjoint subsets of the lattice. However, other important systems can be described by the same algebra of observables, only changing the generators of the algebra. One important example of this kind of system is interacting fermions (see \cite{ArakiMoriya, BK, Datta1, Pedra, putnam, Raszeja, Renault1980,  SimsSzaboWilliams2020} and references therein) on the lattice, whose algebra of observables is the UHF algebra $2^\infty$ or CAR algebra, that is described by anticommutation relations between the annihilation and creation operators in different sites of the lattice. The automorphism of the algebra $\mathfrak{A}_2$ is called the Jordan-Wigner transform in the physics literature and we will describe it in this section in the groupoid language. The discussion that follows highly benefited from \cite{kochmanski1998jordanwigner}.

The generators of the algebra are $\sigma_x^{(3)}$ and $\sigma^{(1)}_y$, for $x,y\in \Z^d$. These are 
\[
\sigma_x^{(3)}(\sigma,g) =
\sigma_x\delta_{\Omega}(\sigma,g)
\quad
\text{ and }
\quad 
\sigma_x^{(1)}(\sigma,g) = \delta_{\{g_x=-1, g_{\{x\}^c} = 1\}}(\sigma,g)
\]

These functions have compact support, respectively $\supp(\sigma_x^{(3)}) = \Omega$ and $\supp(\sigma_x^{(1)})=\Omega\times\{\widehat{g}_x\}$, where $\widehat{g}_x = -1$ while $\widehat{g}_{\{x\}^c}=1$, and satisfy the following commutation relations
\[
\sigma_x^{(3)}\cdot\sigma_y^{(1)}= (-1)^{\delta_{x,y}}\sigma_y^{(1)}\cdot \sigma_x^{(3)},
\]
where $\delta_{x,y}$ is the Kronecker delta function. Furthermore, it is easy to see that $(\sigma_x^{(3)})^2 = (\sigma_x^{(1)})^2 = \mathbbm{1}$. The CAR algebra is the C*-algebra generated by the creation and annihilation operators $a^*_x,a_x$, for $x\in \Z^d$, respectively, that satisfy the following relations
\[
\{a^*_x,a_y\}=\delta_{x,y}\mathbbm{1}, \quad \text{ and } \quad \{a^*_x,a^*_y\}=0
\]
where $\{A,B\} = AB+BA$. As one may see, the spin observables have a very different spatial structure, so we will define the map in two steps. First, we make the local operators anticommute and then after that we modify the structure for the whole lattice. Define the raising and lowering operators $\sigma^+_x$ and $\sigma_x^-$ by
\[
\sigma_x^+ \coloneqq \mathbbm{1}_{\{\sigma_x=-1\}}\cdot\sigma_x^{(1)} 
\quad 
\text{ and }
\quad
\sigma_x^- = \mathbbm{1}_{\{\sigma_x=+1\}}\cdot \sigma_x^{(1)}
\]
Notice that $\sigma_x^+ = (\sigma_x^-)^*$ and $\{\sigma_x^+,\sigma_x^-\}=\mathbbm{1}$, but still $[\sigma_x^s,\sigma_y^{s'}]=0$, for any choice of $s\in \{+,-\}$, yielding the wrong commutation relation between sites. To fix this problem, let $f_x:\Omega_{\Lambda\setminus \{x\}} \rightarrow \mathbb{R}$ be a continuous function and let us multiply by a phase
\[
a_x = e^{i\pi f_x}\sigma_x^-,\quad \text{ and } \quad a_x^* = e^{-i\pi f_x}\sigma_x^+.
\]
Notice that $a_x \cdot a_x^* = \sigma_x^-\cdot \sigma_x^+$ and $a_x^*\cdot a_x =\sigma_x^+\cdot \sigma_x^-$ since $f_x$ does not depend on the value of the spin located in $x$ thus
\[
\{a_x^*,a_x\} = \{\sigma_x^+,\sigma_x^-\}=\mathbbm{1}\quad \text{and}\quad \{a_x,a_x\} = \{\sigma_x^-,\sigma_x^-\}=0.
\]
Notice that for any function $f\in C(\Omega)$ it holds
\[
f \cdot \sigma_x^+ = \mathbbm{1}_{\{\sigma_x=-1\}} \cdot f\cdot \sigma_x^{(1)} = \sigma_x^+\cdot f^{g_x},
\]
where $f^g(\sigma) = f(g\sigma)$. For different sites, using the formula above we can calculate the anticommutator 
\begin{align*}
	\{a_x^*,a_y\}= e^{-i\pi f_x}\sigma_x^+\cdot e^{i\pi f_y}\sigma_y^- +\cdot e^{i\pi f_y}\sigma_y^-e^{-i\pi f_x}\sigma_x^+ = (e^{i\pi (f_y^{g_x}-f_x)}+e^{i\pi (f_y- f_x^{g_y})})\sigma_x^+\sigma_y^-,
\end{align*}
and
\[
\{a_x,a_y\} = (e^{i\pi (f_y^{g_x}-f_x)}+e^{i\pi (f_y- f_x^{g_y})})\sigma_x^+\sigma_y^+.
\]
Thus, to satisfy the anticommutation relations, on the other hand, we need that the functions $f_x$ satisfy the relations
\be\label{jw_relation}
e^{i\pi (f_y^{g_x}-f_x)}+e^{i\pi (f_y- f_x^{g_y})} = 0 \Rightarrow f_y^{g_x}-f_x - f_y + f_x^{g_y} \in 2\Z+1.
\ee
Hence, it is enough for the family of functions $f_x$ to satisfy the relations above to define a working Jordan-Wigner transformation. These considerations motivate us to introduce the following definition
\begin{definition}[\textbf{Jordan-Wigner transformation}] Let $\Lambda \in \mathcal{F}(\Z^d)$. Take a family of functions $F=\{f_x\}_{x\in \Lambda}$, with $f_x \in C(\Omega_{\Lambda \setminus \{x\}})$ such that  for every pair $x,y$ it holds
    \[
    e^{i\pi (f_y^{g_x}-f_x)}+e^{i\pi (f_y- f_x^{g_y})} = 0 \Rightarrow f_y^{g_x}-f_x - f_y + f_x^{g_y} \in 2\Z+1.
    \]
    Then an \textbf{F-Jordan-Wigner transformation} is defined as the following definition for the creation and annihilation operators
    \[
    a_x = e^{i\pi f_x}\sigma_x^-,\quad \text{ and } \quad a_x^* = e^{-i\pi f_x}\sigma_x^+.
    \]
\end{definition}

We will proceed to give an example of a family of functions where these relations are satisfied, to show that they are not empty. Consider the number operators $n_x = \mathbbm{1}_{\{\sigma_x = +1\}}$. They are also known as the occupation number, since

\[
n_x = \sigma_x^+\cdot \sigma_x^-,
\]
thus giving the value $1$ if there is a fermion occupying the site $x$ or $0$ otherwise. Let $w:\Z^d\times\Z^d \rightarrow \mathbb{R}$ a function such that $w(x,x)=0$. We can define
\[
f_x = \sum_{z \in \Lambda}w(x,z) n_z
\]
In order to satisfy the relation \eqref{jw_relation}, we must have
\[
f_y^{g_x}-f_x - f_y + f_x^{g_y} = w(y,x)(n_x^{g_x}-n_x) + w(x,y)(n_y^{g_y}-n_y) = -w(y,x)\sigma_x - w(x,y)\sigma_y  \in 2\Z+1.
\]
Hence the function $w$ should satisfy the constraints of both $w(x,y)+w(y,x)$ and $w(x,y)-w(y,x)$ to be an odd integer. To illustrate that such a choice is realizable, notice that the usual Jordan-Wigner transformation in $d=1$ has the following choice
\[
w(x,y) = \delta_{\{y<x\}}.
\]
It is easy to verify that the sum $w(x,y) + w(y,x) = 1$ and $w(x,y)-w(y,x) = \pm 1$, and also $w(x,x) = 0$. For general $d\geq 2$, one can use the product order of the lattice $\Z$ for instance, this is known as the lexicographic order.  

\section{Interactions}

Once the observable algebras are constructed, we need to introduce the models to study quantum statistical mechanics systems properly. 

\begin{definition}
	A function $\phi:\mathcal{F}(\Z^d)\rightarrow \mathfrak{A}_q$ is called an interaction if 
	\begin{itemize}
		\item[i)] It is \textbf{self-adjoint}, i.e., $\phi(\Lambda)\coloneqq \phi_\Lambda = \phi_\Lambda^*$,
		\item[ii)] It is \textbf{local}, i.e., $\phi_\Lambda \in C^*(\mathcal{G}_\Lambda)$.
	\end{itemize} 
	An interaction is said to have \textbf{short-range} if there is $R>0$ such that if $\diam(X)>R$ implies $\phi_X = 0$. Otherwise, the interaction will be said to have \textbf{long-range}. 
\end{definition}

It is not hard to see that the space of all interactions $\phi$ has a structure of complex vector space. One can consider many different norms in this space and consider a Banach space of the interactions where the chosen norm is finite. For instance, one can consider, for each $\lambda>0$, the norm 
\[
\|\phi\|_\lambda \coloneqq \sum_{n\geq 0}e^{\lambda n}\sup_{x \in \Z^d} \sum_{\substack{X\ni x \\ |X|=n+1}}\|\phi_X\|.
\]
All these Banach spaces contain the short-range interactions as a dense subspace. The study of these Banach spaces and the convex functionals defined on them form a rich chapter of rigorous statistical mechanics. The interested reader can see more of these results in \cite{Is} for general spin systems and \cite{Pedra} for fermionic systems. The \emph{Hamiltonian operator} is defined as
\[
H_\Lambda(\phi) \coloneqq \sum_{X \subset \Lambda }\phi_X.
\]
Often we will refer to the Hamiltonian above as the \emph{empty boundary condition Hamiltonian}. By definition of interaction, the Hamiltonian is a well-defined self-adjoint element of the algebra $C^*(\mathcal{G}_\Lambda)$. Some important examples of interactions are presented below. 

\begin{example}[\textbf{$q=2$ interactions}]\label{ex.heisenberg}
	Let $J^{(1)}_{x,y}, J^{(2)}_{x,y}, J^{(3)}_{x,y}, h_x, \varepsilon_x, \varrho_x$ be real numbers. Then define
	\[
	\phi_\Lambda = \begin{cases}
		-J^{(1)}_{x,y}\sigma_x^{(1)}\sigma_y^{(1)}-J^{(2)}_{x,y}\sigma_x^{(2)}\sigma_y^{(2)}-J^{(3)}_{x,y}\sigma_x^{(3)}\sigma_y^{(3)} & \Lambda = \{x,y\} \\
		-\varepsilon_x \sigma_x^{(1)}-\varrho_x \sigma_x^{(2)}-h_x \sigma_x^{(3)} & \Lambda=\{x\} \\
		0 & o.w.
	\end{cases}
	\]
	where $\sigma^{(2)}_x = i \sigma_x^{(3)}\sigma_x^{(1)}$. These models have quite different behaviors, we will use the following classification
	\begin{itemize}
		\item[(i)] If $J^{(k)}_{x,y}\neq 0$ for all $x,y$, for $k=1,2,3$ then we call it a \textbf{Heisenberg-type interaction}.
		\item[(ii)] If $J^{(k)}_{x,y}\neq 0$, for only two different indices $k$ then we call it a \textbf{XY-type interaction}.
		\item[(iii)] If $J^{(k)}_{x,y}\neq 0$ for only one $k$ then we call it a \textbf{Ising-type interaction}.
	\end{itemize}
\end{example}
\pagebreak
\begin{example}[\textbf{Toric Code}]
Another important example is the Toric Code model. We will give the definition found in  \cite{PN} for the lattice $\Z^2$. A set $X \in \mathcal{F}(\Z^d)$ is called a \textbf{star} if there exists some $x \in \Z^2$ such that $X = B_1(x)$. A \textbf{plaquette} $X$ is a set such that there is $x \in \Z^2$ such that $X = \{x, x+(1,0),x+(0,1), x+(1,1)\}$. Then
\[
\phi_X = \begin{cases}
    \sigma_X^{(3)} & \text{X is a star} \\
    \sigma_X^{(1)} & \text{X is a plaquette}\\
    0 & o.w.
    \end{cases}
\]
\end{example}
\begin{example}[\textbf{$q>2$ interactions}]
	For $q>2$, there is an important class of two-body interactions described by
	\[
	\phi_X = \begin{cases}
		V_{x,y}(u_x,u_y)  & if X = \{x,y\} \\
		f_{1,x}(u_x)(v_x+v_x^*)+f_{2,x}(u_x)& if X = \{x\} \\
		0 & o.w.
	\end{cases}
	\]
	where $V_{x,y}(u_x,u_y), f_{1,x}(u_x)$ and $f_{2,x}(u_x)$ are self-adjoint polynomials of the operators $u_x, u_y$. 
 \begin{itemize}
     \item[(i)] The \textbf{clock models} corresponds to the choice $V_{x,y}(u_x,u_y) = J_{x,y}(u_x\cdot u^*_y+u_x^* \cdot u_y)$. Notice that for every configuration $\sigma \in \Omega$ one has
     \[
     (u_x\cdot u_y^*+u_y\cdot u_x^*)(\sigma) = 2 \cos(\theta_x-\theta_y).
     \]
     where $\theta_x$ is such that $\sigma_x = e^{2\pi i \theta_x}$.
		\item[(ii)] The \textbf{Potts models} corresponds to the choice $V_{x,y}(u_x,u_y) = J_{x,y} \chi_{\{u_xu^*_y+u_x^* u_y = 1\}}$,
		where $\chi$ is the characteristic function. It is easy to see that this characteristic function is an element of the algebra since
  \[
  \lim_{n\rightarrow \infty}\Bigg(\frac{u_x \cdot u_y^*-u_y\cdot u_x^*}{2}\Bigg)^n = \chi_{\{u_xu^*_y+u_x^* u_y = 1\}},
  \]
 \end{itemize}
  in the norm. The terms $f_{1,x}(u_x)$ and $f_{2,x}(u_x)$ represent, respectively, a family of transverse and longitudinal magnetic fields. 
\end{example}

 For fermions, an important class of interactions is collected in the following example (see \cite{Datta2, BK, Pedra} for more details).

\begin{example}[\textbf{Fermi-Hubbard-type interactions}]
	First, we need two copies of fermionic creation and annihilation operators depending on the spin of the particle $s=\{\uparrow,\downarrow\}$ satisfying the relations $$\{a_{x,s},a_{y,s'}\}=0 \quad \text{ and } \quad \{a_{x,s}^*,a_{y,s'}\}=\delta_{x,y}\delta_{s,s'}.$$ Let
    \[
    h_{x,y}^s = a_{x,s}^*\cdot a_{y,s}
    \]
    where $T$ is a self-adjoint matrix called the \textbf{hopping matrix}.
	\[
	\phi_X =\begin{cases}
				T_{x,y}(h_{x,y}^\uparrow+h_{x,y}^\downarrow) + T_{y,x}(h_{y,x}^\uparrow + h_{y,x}^\downarrow)
		  & X=\{x,y\} \\
		U_x(n_{x,\uparrow}-1/2)(n_{x,\downarrow}-1/2) & X=\{x\} \\
		0  & o.w.
	\end{cases}
	\]
	where $T$ is a self-adjoint matrix called the \textbf{hopping matrix} and $U_x$ is a real number. 
\end{example}

\begin{example}
The following model is known as the spinless Kitaev p-wave wire model, see \cite{Kitaev2001, Lapa2023} for details.
\[
\phi_X = \begin{cases}
    T_{x,y}h_{x,y} + T_{y,x}h_{y,x} + i\lambda_{x,y} a_x \cdot a_y - i\lambda_{x,y} a_x^* \cdot a_y^*,& X = \{x,y\} \\
    0 & o.w,
\end{cases}
\]
where $\lambda_{x,y}$ are real numbers.
\end{example}

These fermionic models have the following representation as spin models, through the Jordan-Wigner transformation introduced earlier. The hopping terms $h_{x,y}$ can be written as
\[
h_{x,y} (\sigma,g) = e^{i\pi w(y,x) \sigma_x}\sigma_x^+ \cdot \sigma_y^-(\sigma,g).
\]
A direct computation shows that the term $\sigma_x^+ \cdot \sigma_y^-$ is
\begin{equation*}
    \begin{split}
        4\sigma_x^+ \cdot \sigma_y^- &= (\mathbbm{1}-\sigma_x^{(3)})\cdot (\mathbbm{1}+\sigma_y^{(3)})\cdot \sigma_x^{(1)}\cdot \sigma_y^{(1)}
    \end{split}
\end{equation*}
Thus, using the self-adjointness of the matrix $T_{x,y}$ and supposing that $x<y$ in the lexicographic order, the other case being analogous, we get 
\begin{align*}
4(T_{x,y}h_{x,y}+T_{y,x}h_{y,x}) &=\Bigg(T_{x,y}(\mathbbm{1}-\sigma_{x}^{(3)})\cdot (\mathbbm{1}+\sigma_y^{(3)})+T_{y,x}(\mathbbm{1}-\sigma_y^{(3)})\cdot (\mathbbm{1}+\sigma_x^{(3)})\Bigg)\cdot \sigma_x^{(1)}\cdot \sigma_y^{(1)} \\
&=2\mathrm{Re}(T_{x,y})\Bigg(\mathbbm{1} - \sigma_x^{(3)}\cdot\sigma_y^{(3)}\Bigg)\cdot \sigma_x^{(1)}\cdot \sigma_y^{(1)}  - 2i \mathrm{Im}(T_{xy})(\sigma_x^{(3)}-\sigma_y^{(3)})\cdot \sigma_x^{(1)}\cdot \sigma_y^{(1)} \\
&= 2 \mathrm{Re}(T_{x,y}) (\sigma_x^{(1)}\sigma_y^{(1)}+\sigma_x^{(2)}\sigma_y^{(2)}) + 2i\mathrm{Im}(T_{x,y})(\sigma_x^{(2)}\cdot\sigma_y^{(1)}-\sigma_y^{(2)}\cdot\sigma_x^{(1)}).
\end{align*}
The novelty is that, since the matrix $T_{x,y}$ is only self-adjoint, the Jordan-Wigner map transforms the hopping term of the Hamiltonian into a $XY$-model with coupling constant equal to $\mathrm{Re}(T_{x,y})$. The novelty is a new term depending on $\mathrm{Im}(T_{x,y})$ that can be identified with a \emph{Dzyaloshinskii-Moriya} interaction  \cite{RicardodeSousa1995}. For the pairing terms $a_x\cdot a_y$ similar calculations apply and we get
\begin{equation*}
\begin{split}
a_x \cdot a_y - a_x^*\cdot a_y^*&=\Bigg((\mathbbm{1}+\sigma_{x}^{(3)})\cdot (\mathbbm{1}+\sigma_y^{(3)})-(\mathbbm{1}-\sigma_y^{(3)})\cdot (\mathbbm{1}-\sigma_x^{(3)})\Bigg)\cdot \sigma_x^{(1)}\cdot \sigma_y^{(1)} \\
&=2(\sigma_x^{(3)}+\sigma_y^{(3)})\cdot \sigma_x^{(1)}\cdot \sigma_y^{(1)} = \sigma_x^{(2)}\cdot\sigma_y^{(1)}+\sigma_y^{(2)}\cdot\sigma_x^{(1)}.     
\end{split}    
\end{equation*}

Following Bratteli-Robinson \cite{Bra2}, we define the \emph{surface energy} term corresponding to the interaction between the region $\Lambda$ and the exterior region $\Lambda^c$
\[
W_\Lambda(\phi) \coloneqq \sum_{\substack{X \cap \Lambda \neq \emptyset \\ X \cap \Lambda^c \neq \emptyset}} \phi_X.
\]
Different from the empty boundary condition Hamiltonian, the surface energy term is not well defined for all interactions since in its definition we need to sum potentially infinitely many terms, thus bringing convergence issues. But if we assume that $\phi$ is an interaction with finite $\|\cdot\|_\lambda$-norm for some $\lambda>0$, then we can show that $W_\Lambda(\phi)$ is a well-defined element of the algebra. Indeed, in Banach spaces absolute convergence implies convergence of the series, hence
\[
\sum_{\substack{X\cap \Lambda \neq \emptyset \\ X\cap \Lambda^c \neq \emptyset}}\|\phi_X\| =\sum_{n\geq 0}\sum_{\substack{X\cap \Lambda \neq \emptyset \\ X\cap \Lambda^c \neq \emptyset \\ |X|=n+1}}\|\phi_X\|. 
\]
But since $\|\phi\|_\lambda <\infty$, we have
\[
\sum_{\substack{X\cap \Lambda \neq \emptyset \\ X\cap \Lambda^c \neq \emptyset \\ |X|=n+1}}\|\phi_X\| \leq \sum_{x \in \Lambda}\sum_{\substack{X \ni x , X\cap \Lambda \neq \emptyset \\ X\cap \Lambda^c \neq \emptyset \\ |X|=n+1}}\|\phi_X\| \leq |\Lambda|\sup_{x\in \Z^d}\sum_{\substack{X\ni x \\ |X|=n+1}}\|\phi_X\| \leq |\Lambda|e^{-\lambda n}\|\phi\|_\lambda.
\]
Plugging again in what we had before we get 
\[
\sum_{\substack{X\cap \Lambda \neq \emptyset \\ X\cap \Lambda^c \neq \emptyset}}\|\phi_X\| \leq |\Lambda| \frac{\|\phi\|_\lambda}{1-e^{-\lambda}},
\]
yielding us that $W_\Lambda(\phi)$ is a well defined element of the C*-algebra $\mathfrak{A}_q$ for every $\Lambda \in \mathcal{F}(\Z^d)$.
\section{The KMS Condition and the Gibbs-Araki-Ion Condition}

	Motivated by considerations in condensed matter and quantum field theory Kubo \cite{Kubo} and Martin and Schwinger \cite{MS} noticed an important relation that the thermodynamic Green functions (see for example Chapter 11 of \cite{Bech}) satisfies with respect to the imaginary time. This periodicity property depends on the temperature. In a seminal work, Haag, Hugenholtz, and Winnink \cite{HHW} studied this relation in a C*-algebraic language and called it for the first time the \emph{Kubo-Martin-Schwinger boundary condition}, or KMS condition for short. Since the Hamiltonians for infinite systems are not defined, the Gibbs state cannot be defined as usual for these systems. However, they showed that, under some assumptions on the convergence of the Heisenberg dynamics and the states with respect to the thermodynamic limit procedure, the KMS condition can still have a meaningful sense. Let $\mu_{\beta,\phi,\Lambda}$ be the finite volume Gibbs state
    \be\label{freegibbsstate}
    \mu_{\beta,\phi,\Lambda}(f) = \frac{1}{Z_{\beta,\phi,\Lambda}}\sum_{\sigma_\Lambda \in \Omega_\Lambda} f \cdot e^{-\beta H_\Lambda(\phi)}(\sigma_\Lambda,1),
    \ee
    where $Z_{\beta,\phi,\Lambda} = \sum_{\sigma_\Lambda \in \Omega_\Lambda}e^{-\beta H_\Lambda(\phi)}(\sigma_\Lambda,1)$ is the partition function. The sum over the configuration space is the trace on the groupoid C*-algebra, since for any $f_1,f_2 \in C_c(\mathcal{G}_\Lambda)$ it holds,
    \begin{equation*}
    \begin{split}
    \sum_{\sigma_\Lambda \in \Omega_\Lambda} f_1 \cdot f_2 (\sigma_\Lambda,1) &= \sum_{\sigma_\Lambda \in \Omega_\Lambda}\sum_{(\omega_\Lambda,g_\Lambda)\in \mathcal{G}_\Lambda^{\sigma_\Lambda}}f_1(\omega_\Lambda,g_\Lambda)f_2(\sigma_\Lambda,g_\Lambda^{-1}) \\
     &= \sum_{\sigma_\Lambda \in \Omega_\Lambda}\sum_{(\omega_\Lambda,g_\Lambda)\in \mathcal{G}_\Lambda^{\sigma_\Lambda}}f_1(\omega_\Lambda,g_\Lambda)f_2((\omega_\Lambda,g_\Lambda)^{-1}) \\
     &=\sum_{(\omega_\Lambda,g_\Lambda)\in \mathcal{G}_\Lambda} f_1(\omega_\Lambda,g_\Lambda)f_2((\omega_\Lambda,g_\Lambda)^{-1}) \\
     &=\sum_{(\omega_\Lambda,g_\Lambda)\in \mathcal{G}_\Lambda} f_1((\omega_\Lambda,g_\Lambda)^{-1})f_2(\omega_\Lambda,g_\Lambda) \\
     & =  \sum_{\sigma_\Lambda \in \Omega_\Lambda}\sum_{(\omega_\Lambda,g_\Lambda)\in \mathcal{G}_\Lambda^{\sigma_\Lambda}}f_2(\omega_\Lambda,g_\Lambda)f_1(\sigma_\Lambda,g_\Lambda^{-1}) = \sum_{\sigma_\Lambda \in \Omega_\Lambda}f_2\cdot f_1(\sigma_\Lambda,1).
    \end{split}
    \end{equation*}
    
    For each box $\Lambda \subset \Z^d$, the Hamiltonian operators $H_\Lambda(\phi)$ can be used to define a dynamics on $C^*(\mathcal{G}_\Lambda)$,
	\[
	\tau_t^\Lambda(f) = e^{itH_\Lambda (\phi)}\cdot f \cdot e^{-itH_\Lambda(\phi)},
	\]
	for every $f \in C^*(\mathcal{G}_\Lambda)$, and $t\in \mathbb{R}$. Since, in this case, the function $t\mapsto e^{-it H_\Lambda(\phi)}$ makes sense for every complex argument, using the cyclicity of the trace, one can observe the following property for the finite volume Gibbs states
	\begin{align*}
		\mu_{\beta,\Lambda}(f\cdot \tau_{i\beta}^\Lambda(g)) = \frac{\text{tr}( f \cdot e^{-\beta H_\Lambda(\phi)}\cdot g \cdot e^{\beta H_\Lambda(\phi)}\cdot e^{-\beta H_\Lambda(\phi)} )}{\text{tr}(e^{-\beta H_\Lambda(\phi)})} = \frac{\text{tr}(g\cdot f\cdot e^{-\beta H_\Lambda(\phi)})}{\text{tr}(e^{-\beta H_\Lambda(\phi)})}= \mu_{\beta,\Lambda}(g\cdot f)
	\end{align*} 
In the thermodynamic limit, the Hamiltonian function $H_{\Z^d}(\phi)$ can be defined only formally. But not everything is lost. Indeed, for each $f\in C^*(\mathcal{G}_\Lambda)$, we can see that the dynamics $\tau_t^\Lambda$ satisfy the equation
\be\label{cauchy}
\frac{d}{dt}\tau_t^\Lambda(f) = i[H_\Lambda(\phi),f], \quad \text{and} \quad \tau_0(f) = f.
\ee
where $[f,g]\coloneqq f\cdot g - g\cdot f$. Formally, the commutator of the infinite volume Hamiltonian with any local operator $f$ is given by
\[
i[H_{\Z^d}(\phi),f] = i\sum_{X\cap \Lambda \neq \emptyset}[\phi_X,f].
\]
The right-hand side of the equation above is a linear operator in $C^*(\mathcal{G})$ that can potentially make sense for some classes of interactions. Therefore, we may try to solve the Cauchy problem stated in \eqref{cauchy} for it. We will show that this procedure can be carried out for interactions with $\|\phi\|_\lambda<\infty$ for some $\lambda>0$, following \cite{Bra2}. 

\begin{definition}
    For a C*-algebra $\mathfrak{A}$, a map $\mathcal{D}:\mathrm{dom}(\D)\rightarrow \mathfrak{A}$ is called a \textbf{*-derivation} if
    \begin{itemize}
        \item [(i)] $\D$ is linear;
        \item [(ii)] $\D(f\cdot g) = \D(f)\cdot g + f \cdot \D(g)$;
        \item [(iii)] $\D(f^*) = \D(f)^*$.
    \end{itemize}
\end{definition}
The set $\mathrm{dom}(\D)$ is called the \emph{domain} of the derivation $\D$, and it is only a proper subset of the C*-algebra $\mathfrak{A}$ in general.
    \begin{lemma}\label{l3}
Let $\phi$ be an interaction such that $\|\phi\|_\lambda <\infty$ in the spin algebra $C^*(\mathcal{G})$ for some $\lambda>0$. Let $\D^\phi$ be the *-derivation defined by
\[
\D^\phi(f) = i \sum_{X \cap \Lambda \neq \emptyset}[\phi_X, f], \;\; \text{ for  } f \in C^*(\mathcal{G}_\Lambda).
\]
Let $\D^\phi(f)^{(n)} = \D^\phi(\D^\phi(f))^{(n-1)}$. Then it holds that,
\[
\|\D^\phi(f)^{(n)}\| \leq \left(\frac{2\|\phi\|_\lambda}{\lambda}\right)^n  \|f\|n!e^{\lambda|\Lambda|}
\]
\end{lemma}

\begin{proof}
We must show that this $\D^\phi$ is a well-defined *-derivation in $C^*(\mathcal{G})$. Let $f \in C^*(\mathcal{G}_\Lambda)$. We will show that $\D^\phi(f) \in C^*(\mathcal{G})$. The definition of the commutant implies
$
\|[\phi_X, f]\|\leq 2\|\phi_X\|\|f\|.
$
Also, we have that
$$
\underset{X\cap \Lambda\neq \emptyset}{\sum}\|\phi_X\| = \sum_{n\geq 0}\sum_{\stackrel{X\cap\Lambda \neq \emptyset}{|X|=n+1}}\|\phi_X\| \leq \sum_{n\geq 0}\sum_{x \in \Lambda}\sum_{\stackrel{x \in X}{|X|=n+1}}\|\phi_X\|. 
$$
Then,
\begin{equation}\label{1}
\underset{X\cap \Lambda\neq \emptyset}{\sum}\|[\phi_X,f]\| \leq 2|\Lambda|\|f\| \|\phi\|_\lambda 
\end{equation}
Our hypothesis on the interaction implies the series defined on the left-hand side of \eqref{1} is absolutely convergent, so it is convergent in $C^*(\mathcal{G})$. Hence $\D^\phi(f)$ is a well-defined element of the C*-algebra. It is obviously linear and preserves the involution, so all it remains to show is the Leibniz property. Consider $f, g \in C^*(\mathcal{G}_\Lambda)$. 
\[
\D^\phi(f\cdot g) = i \underset{X \cap \Lambda \neq \emptyset}{\sum}[\phi_X, f\cdot g] = i \underset{X \cap \Lambda \neq \emptyset}{\sum}f\cdot [\phi_X, g]+[\phi_X,f]\cdot g = f\cdot \D^\phi(g)+\D^\phi(f)\cdot g.
\]
The operator $\D^\phi(f)^{(n)}$ can be written as 
$$
\D^\phi(f)^{(n)}= i^n\sum_{X_1\cap\Lambda \neq \emptyset}\dots\sum_{X_{n-1}\cap{S_{n-1}} \neq \emptyset}\sum_{X_n\cap S_n \neq \emptyset}[\phi_{X_n},[\phi_{X_{n-1}},\dot[\phi_{X_1},f]\dot]].
$$
Where $S_j = (\cup_{k=1}^{j-1} X_k)\cup \Lambda$. The following bound holds,
\be\label{lema_derivation_eq1}
\begin{split}
\|\D^\phi(f)^{(n)}\| &\leq 2^n \|f\|\sum_{X_1\cap\Lambda \neq \emptyset} \dots \sum_{X_{n-1}\cap S_{n-1} \neq \emptyset} \sum_{X_n\cap S_n \neq \emptyset}\prod_{j=1}^n \|\phi_{X_j}\| \\
&\leq 2^n\|f\|\Bigg[\sum_{X_1\cap\Lambda\neq \emptyset} \|\phi\|_{X_1}\Bigg[\dots \sum_{X_{n-1}\cap S_{n-1} \neq \emptyset}\|\phi_{X_{n-1}}\|\Bigg[ \sum_{X_n\cap S_n \neq \emptyset} \|\phi_{X_n}\|\Bigg]\dots\Bigg]
\end{split}
\ee
For each $j$, we can bound the sum over all $X_j\cap S_j\neq \emptyset$ for $S_j$ fixed as
\be\label{lema_derivation_eq2}
\sum_{X_j\cap S_j \neq \emptyset} \|\phi_{X_j}\| =\sum_{m_j\geq 0}\sum_{\substack{X_j\cap S_j \neq \emptyset \\ |X_j|=m_j+1}} \|\phi_{X_j}\| \leq \sum_{m_j\geq 0}\sum_{ x \in S_j} \sum_{ \substack{x \in X_j\\|X_j|=m_j+1}} \|\phi_{X_j}\|.
\ee
Let us call for a moment
$$|\phi|_{m_j} \coloneqq \underset{x \in \mathbb{Z}^d}{\sup}\sum_{\stackrel{x \in X}{|X|=m_j+1}} \|\phi_X\|.
$$
When the sizes of $X_j$ are fixed to be $m_j$ respectively, we have
\[
|S_j|\leq m_1+m_2+\dots+m_{j-1}+|\Lambda|.
\]
Iterating the bound \eqref{lema_derivation_eq2} and using the upper bound in the number of points of $S_j$ above the right-hand side of Inequality \eqref{lema_derivation_eq1} becomes
\begin{equation*}
\begin{split}
    \frac{\|\D^\phi(f)^{(n)}\|}{2^n\|f\|}&\leq \Bigg(\sum_{X_1\cap\Lambda\neq \emptyset} \|\phi\|_{X_1}\Bigg(\dots \sum_{X_{n-1}\cap S_{n-1} \neq \emptyset}\hspace{-0.5cm}\|\phi_{X_{n-1}}\|\Bigg( \sum_{X_n\cap S_n \neq \emptyset} \hspace{-0.3cm} \|\phi_{X_n}\|\Bigg)\dots\Bigg) \\
    &=\sum_{\substack{m_j\geq 0\\ j=1,\dots,n}}\Bigg(\sum_{\substack{X_1\cap\Lambda\neq \emptyset \\ |X_1|=m_1}} \|\phi\|_{X_1}\Bigg(\dots \sum_{\substack{X_{n-1}\cap S_{n-1} \neq \emptyset\\ |X_{n-1}|=m_{n-1}}}\hspace{-0.5cm}\|\phi_{X_{n-1}}\|\Bigg( \sum_{\substack{X_n\cap S_n \neq \emptyset \\ |X_n|=m_n}}  \hspace{-0.3cm}\|\phi_{X_n}\|\Bigg)\dots\Bigg) \\
    &\leq\sum_{\substack{m_j\geq 0\\ j=1,\dots,n}}\Bigg(\sum_{\substack{X_1\cap\Lambda\neq \emptyset \\ |X_1|=m_1}} \|\phi\|_{X_1}\Bigg(\dots \sum_{\substack{X_{n-2}\cap S_{n-2} \neq \emptyset\\ |X_{n-2}|=m_{n-1}}}\hspace{-0.5cm}\|\phi_{X_{n-2}}\|\Bigg( \sum_{\substack{X_{n-1}\cap S_{n-1} \neq \emptyset \\ |X_{n-1}|=m_{n-1}}} \hspace{-0.5cm} \|\phi_{X_{n-1}}\||\phi|_{m_n}|S_n|\Bigg)\dots\Bigg) \\
    &\vdots \\
    &\leq \sum_{\substack{m_j\geq 0\\ j=1,\dots,n}} \prod_{j=1}^n(m_1+\dots+m_{j-1}+|\Lambda|)|\phi|_{m_j}.
\end{split}
\end{equation*}
Since each $m_j$ is positive we have 
$$
\prod_{j=1}^n (m_1+...+m_{j-1} + |\Lambda|) \leq (m_1+...+m_n + |\Lambda|)^n \leq \frac{n!e^{\lambda |\Lambda|}}{\lambda^n}\prod_{j=1}^n e^{\lambda m_j}.
$$
Thus,
\begin{equation*}
    \begin{split}
        \|\D^\phi(f)^{(n)}\| \leq 2^n\|f\|\frac{n!e^{\lambda|\Lambda|}}{\lambda^n}\prod_{j=1}^n\Bigg(\sum_{m_j\geq 0}e^{\lambda m_j}|\phi|_{m_j}\Bigg),
    \end{split}
\end{equation*}
yielding us the desired result.
\end{proof}

\pagebreak 

\begin{theorem}[\textbf{Existence of the dynamics}]\label{dynamicsexistence}
Under the hypothesis of the previous lemma, we get that the *-derivation $\D^\phi$ has a dense domain, and its closure defines a strongly continuous one-parameter group of $*$-automorphisms $(\tau_t)_{t \in \mathbb{R}}$ in $C^*(\mathcal{G})$. Besides that, we have the convergence
$$
\underset{\Lambda \nearrow \mathbb{Z}^d}{\lim}\| \tau_t(f) - \tau_t^\Lambda (f)\| = 0
$$
is uniform in compacts.
\end{theorem}
\begin{proof}
We will prove that both $\pm\D^\phi$ are dissipative, i.e, that
\be\label{existence_dynamics_eq1}
\| (\lambda\mathbbm{1}\pm \D^\phi)(f)\| \geq \lambda \|f\|, 
\ee
for all $\lambda >0, f\in C^*(\mathcal{G}_\Lambda)$. It is a consequence of the Hahn-Banach theorem that for all $f \in C^*(\mathcal{G})$ there exists a state $\mu$ in $C^*(\mathcal{G})$ such that $\mu(f^*\cdot f) = \|f\|^2$. If we can show that $\mathrm{Re}(\mu(f^*\D^\phi(f))) = 0$ then Inequality \eqref{existence_dynamics_eq1} will follow because
\begin{equation*}
\begin{split}
 \|\lambda f\|^2 &=  \mu(\lambda^ 2 f^*\cdot f) =  \mathrm{Re}(\mu(\lambda^2 f^* \cdot f)) \\
 &= \mathrm{Re}(\mu(\lambda f^*\cdot(\lambda f-\D^\phi(f))))\leq \|\lambda f\|\|(\lambda \mathbbm{1}-\D^\phi)(f)\|.
\end{split}
\end{equation*}
Since $\mu$ is a state, it holds $$2\mathrm{Re}(\mu(f^*\cdot\D^\phi(f))) = \mu(f^*\cdot\D^\phi(f))+\overline{\mu(f^*\cdot \D^\phi(f))}=\mu(\D^\phi(f^*\cdot f)).
$$
Define $f'= (\|f^*\cdot f\|-f^*\cdot f)^{1/2}$. Note that $f'$ is a positive operator and $(f')^2 = \|f^*\cdot f\|-f^*\cdot f$. Lemma \ref{l3} implies that $f'$ is in the domain of $\D^\phi$ since $f'$ is also a local operator. We find 
$$
-\mu(\D^\phi(f)) = \mu(\D^\phi((f')^2)) = \mu(f'\cdot \D^\phi(f'))+\mu(\D^\phi(f')\cdot f').
$$
The Cauchy-Schwarz inequality implies
$$
|\mu(f'\cdot \D^\phi(f'))|^2 \leq \mu((f')^2)\mu(\D^\phi(f')^2) = (\|f^*\cdot f\|-\mu(f^*\cdot f))\mu(\D^\phi(f')^2) = 0,
$$
Showing that $\D^\phi$ is dissipative.  In order to apply the Lummer-Phillips Theorem (see Theorem \ref{t3}), we must find some positive $\lambda$ such that $(\lambda - \overline{\D^\phi})(\overline{\mathrm{dom}(\D^\phi)})$, is dense in $C^*(\mathcal{G})$, where $\overline{\D^\phi}$ is the closure of $\D^\phi$. By Lemma \ref{l3} the series
\be\label{analytic_element_series}
\T_t(f)=\sum_{n\geq0} \frac{t^n\D^\phi(f)^{(n)}}{n!},
\ee
has a radius of convergence $r = \frac{\lambda}{2\|\phi\|_\lambda}$. Each $\T_t(f)$ is in the domain of the closure of $\D^\phi$ for all local observables $f$. Indeed,
\[
\D^\phi\left(\sum_{n=0}^m \frac{t^n \D^\phi(f)^{(n)}}{n!}\right) = \sum_{n=0}^m \frac{t^n \D^\phi(\D^\phi(f))^{(n)}}{n!},
\]
and taking the limit $m\rightarrow \infty$ on the equation above we conclude that $\overline{\D^\phi}(\T_t(f)) = \T_t(\D^\phi(f))$. Define the function $\tau_t(f) = \T_{r/2}^k(\T_s(f))$ for all $t\in \mathbb{R}$, where $k$ is the largest integer such that $t = kr/2+s$. For some $\lambda>0$, let the resolvent function be
\[
R(\lambda)f = \int_0^\infty e^{-\lambda t}\tau_t(f) dt.
\]
Note that $(\lambda - \overline{\D^\phi})R(\lambda)(f) = f$ for all local operators. Since the local operators are dense in the C*-algebra, we conclude that the domain $\text{dom}(\overline{\D^\phi})$ is dense and we can apply Theorem \ref{t3}. The same reasoning can be applied to $-\D^\phi$ and we conclude the derivation can be exponentiated to a strongly continuous one-parameter group $\tau_t$. Finally, notice that the generator of $\tau_t^\Lambda$ is the derivation given by  $\D^\phi_\Lambda= i[H_\Lambda(\phi),\cdot]$. It is easy to verify that derivations approximate $\D^\phi$ for each $f$ a local observable
\[
\lim_{\Lambda \nearrow \mathbb{Z}^d} \|\D^\phi_\Lambda(f) - \D^\phi(f)\| = 0
\] 
Then the Trotter-Kato theorem (see \ref{TrotterKato} for the statement and Theorem 4.8 in \cite{En} for the proof) implies that $\tau_t^\Lambda$ converges strongly to $\tau_t$.
\end{proof}
\begin{remark}
The theorem has a converse, in a certain sense. One can ask what kind of derivations we have in the spin algebra. In \cite{Bra2} it is proved that all the derivations with $\underset{\Lambda \in \mathcal{F}(\mathbb{Z}^d)}{\bigcup}C^*(\mathcal{G}_\Lambda)$ as the dense domain comes from interaction, i.e., the derivation has the form of $\D^\phi$ for some interaction $\phi$ satisfying $\sum_{x \in X} \|\phi_X\| < \infty$ for all $x \in \mathbb{Z}^d$.
\end{remark}

In the previous discussion, the infinite volume dynamics $\tau_t$ was defined only for real $t$, and different from the finite-dimensional case discussed at the beginning of this section it is not easy to see how to extend it to a function on $\mathbb{C}$. This motivates us to introduce the following definition
	\begin{definition}
		An element $f \in C^*(\mathcal{G})$ is said to be analytic if there is a strip in the complex plane $I_\lambda = \{z\in \mathbb{C}: |\mathrm{Im}(z)|<\lambda\}$ and $F:I_\lambda \rightarrow C^*(\mathcal{G})$ such that
		\begin{itemize}
			\item[(i)] $F(t) = \tau_t(f)$, for any $t\in \mathbb{R}$.
			\item[(ii)]  the function $F(z)$ is analytic.
		\end{itemize}
  When $\lambda = \infty$ we call the elements entire analytic.
	\end{definition}
	
	Theorem 6.2.4 in \cite{Bra2} (whose content is stated and proved in Lemma \ref{l3}) shows that any local function is analytic for the dynamics $\tau_t$. We are ready to define the KMS condition. 
	\begin{definition}
		Let $\mathfrak{A}$ be a C*-algebra, a strongly continuous one-parameter group $\tau_t$, and $\mu$ a state. We say that $\mu$ satisfies the KMS condition if, and only if, for all entire analytic elements $f, g \in \mathfrak{A}$ the following holds
		\[
		\mu(f\cdot \tau_{i\beta}(g)) = \mu(g\cdot f),
		\] 
    and we denote the set of all $(\tau,\beta)$-KMS states by $K_\beta(\phi)$. If the states are $\Z^d$-invariant, then we denote it by $K_\beta^\theta(\phi)$. 
	\end{definition}

         Limits preserve the KMS condition. This was one of the motivations for Haag, Hugenholtz, and Winnink, in the seminal paper \cite{HHW}, to propose the KMS condition as a characterization for equilibrium since, as we already stated, the infinite volume Hamiltonian is often only formal. But the limit state and dynamics can be rigorously constructed, as shown in Theorem \ref{dynamicsexistence}.
        
        \begin{theorem}\label{limitkms}
            Let $\tau_{n,t}$ be a sequence of one-parameter groups of *-automorphisms converging in the strong operator topology to another one-parameter group of *-automorphisms $\tau_t$. Let $\beta_n$ a sequence of positive real numbers and $\mu_n$ a sequence of $(\tau_n,\beta_n)$-KMS states. Suppose that, $\beta_n$ and $\mu_n$ converge, respectively, to $\beta,$ and $\mu$. Then $\mu$ is a $(\tau,\beta)$-KMS state.
        \end{theorem}
        \begin{proof}
            The proof can be found in Theorem 5.3.25 of \cite{Bra2}.
        \end{proof}
	
 The KMS condition is known to be equivalent to different characterizations of equilibrium such as the variational principle \cite{Ara1,Ara2}. A thorough analysis of these results can be found, for instance, in \cite{Bra2, Is, Simon}). Araki and Ion \cite{Ara2} proved their result by defining a new condition for equilibrium, which we call here the \emph{Gibbs-Araki-Ion condition}.  To define this condition we must introduce first the notion of \emph{perturbation of a state}. Let $P \in C^*(\mathcal{G})$ be a self-adjoint element and define the perturbed derivation $\D_P^\phi: C^*(\mathcal{G}) \rightarrow C^*(\mathcal{G})$ by
	\[
	\D_P^\phi(f) = \D^\phi(f) +i [P,f].
	\]
        That the perturbed derivation above generates another strongly continuous one-parameter group of *-automorphisms is the subject of the next proposition.
	\begin{proposition}
	  Let $\tau_t$ a strongly continuous one-parameter group of *-automorphisms in $C^*(\mathcal{G})$, $\D^\phi$ its generator. For every self-adjoint $P=P^*$ the perturbed derivation $\D_P^\phi$ generates a strongly continuous one-parameter group of *-automorphisms $\tau^P_t$ and it relates to $\tau_t$ by the expression 
	\[
	\tau_t^P(f) = \tau_t(f) + \sum_{n \geq 1} i^n \int_{S_{n,t}} dt^n [\tau_{t_n}(P),[\dots[\tau_{t_1}(P), \tau_t(f)]]],
	\]
	for all $f \in C^*(\mathcal{G})$, where $S_{n,a,b} \coloneqq \{(t_1,\dots,t_n): a<t_n <t_{n-1}<\dots < t_n<b\}$. When $a=0$, we write $S_{n,b}$ instead. Furthermore, we can define a unitary operator
	\[
	E_{P,t} = \sum_{n \geq0} i^n \int_{S_{n,t}} dt^n \tau_{t_n}(P)\dots\tau_{t_1}(P).
	\]
	know as \textbf{Araki-Dyson series}. They satisfy the following properties
	\be\label{cocycle_properties}
	E_{P,t} = E_{P,t} \tau_t(E_{P,s}) \quad \text{ and } \quad (E_{P,t})^{-1} = E_{P,t}^*,
	\ee
	 the perturbed and the original dynamics are related by
	\[
	\tau^P_t(f) = E_{P,t} \tau_t (f) E_{P,t}^*.
	\] 
	\end{proposition}
    \begin{proof}
      See Proposition 5.4.1 of \cite{Bra2} for a more general result. 
    \end{proof}
	The first relation in \eqref{cocycle_properties} is known as the cocycle relation, and the second is just the definition of a unitary operator. Often we need to perturb the state $\mu$ from which $\tau_t$ is a KMS state. We can define the analytically continued Araki-Dyson series by the expression
 \be\label{araki_dyson}
 E_{P,it} = \sum_{n\geq 0}(-1)^n\int_{S_{n,t}}dt^n \tau_{it_n}(P)...\tau_{it_1}(P),
 \ee
 since the series above converges uniformly. All the relations in Equation \eqref{cocycle_properties} hold for the analytic continuation of the Araki-Dyson series. We are ready to introduce the perturbed state.
	\begin{definition}
		Let $\mathfrak{A}_q$ be the spin algebra, $\phi$ a short-range interaction, and $\tau$ the strongly continuous group it generates. Let $P \in \mathfrak{A}_q$ be a self-adjoint element. Let $\mu$ be a state, then we define the perturbed state $\mu^P$ by
		\be\label{perturbed_state}
		\mu^P(f) \coloneqq \frac{(\xi_\mu^P,\pi_\mu(f) \xi_\mu^P)}{(\xi_\mu^P,\xi_\mu^P)} = \frac{\mu((E_{P,i\beta/2})^*\cdot f\cdot E_{P,i\beta/2})}{\mu((E_{P,i\beta/2})^*\cdot E_{P,i\beta/2})} 
		\ee
		Where $\xi^P_\mu = \pi_\mu(E_{P,i\beta/2})\xi_\mu$, $\xi_\mu$, and $\pi_\mu$ are, respectively, the cyclic vector and the representation of the GNS representation associated with $\mu$. 
	\end{definition}
	Before we introduce the following definition, we need to introduce an important notion of \emph{product state}. In the case of the spin algebra, one can always decompose it with respect to a subset $\Lambda \subset \Z^d$ as a tensor product $C^*(\mathcal{G})\simeq C^*(\mathcal{G}_\Lambda)\otimes C^*(\mathcal{G}_{\Lambda^c})$. Then if one has two states $\mu_1:C^*(\mathcal{G}_\Lambda)\rightarrow \mathbb{C}$ and $\mu_2: C^*(\mathcal{G}_{\Lambda^c})\rightarrow \mathbb{C}$ one can define a product state, denoted by $\mu_1\otimes \mu_2: C^*(\mathcal{G})\rightarrow \mathbb{C}$ by
 \[
 \mu_1\otimes \mu_2(f_1\otimes f_2) = \mu_1(f_1)\mu_2(f_2).
 \]
 Of course, not all the elements of $C^*(\mathcal{G})$ are of the form $f_1 \otimes f_2$, but their span is dense on the original C*-algebra, so the definition above is sufficient to have a bounded linear functional on the whole algebra.
	\begin{definition}[\textbf{Gibbs-Araki-Ion Condition}]\label{def_GibbsAraki}
		Let $\mu:C^*(\mathcal{G})\rightarrow \mathbb{C}$ be a state. We say that it satisfies the Gibbs-Araki condition for $\beta$ and interaction $\phi$ if, and only if
		\begin{enumerate}
			\item $\mu$ is faithful;
			\item $\mu^P_\beta = \mu_{\beta,\phi,\Lambda} \otimes \bar{\mu}_\beta$, where $\mu^P_\beta$ is the perturbed state defined in Equation \eqref{perturbed_state} for $P = \beta W_\Lambda(\phi)$ and $\bar{\mu}$ is a state in $C^*(\mathcal{G}_{\Lambda^c})$.
		\end{enumerate}  
	\end{definition} 
	For any state $\mu:C^*(\mathcal{G}_\Lambda)\rightarrow \mathbb{C}$ we can define a linear map $\text{Id}\otimes \mu: C^*(\mathcal{G})\rightarrow C^*(\mathcal{G}_\Lambda)$ by
 \[
 \text{Id}\otimes \mu(f_1\otimes f_2) = \mu(f_1)f_2,
 \]
 where $f_1 \in C^*(\mathcal{G}_\Lambda)$ and $f_2 \in C^*(\mathcal{G}_{\Lambda^c})$. This kind of map will appear later on when we discuss the notion of boundary conditions for KMS states. Note that the Gibbs-Araki-Ion condition is equivalent to 
	\be\label{gai_equivalence}
	\mu^P_\beta = \bar{\mu}_\beta(\text{Id} \otimes \mu_{\beta,\phi,\Lambda}).
	\ee
 for every $\Lambda \Subset \Z^d$.
	\begin{theorem}
		Let $\mathfrak{A}$ be the spin algebra and $\phi$ an interaction with $\|\phi\|_\lambda<\infty$ for some $\lambda>0$. Then, the following assertions are equivalent for a state $\mu$
		\begin{enumerate}
			\item $\mu$ satisfies the Gibbs-Araki-Ion condition for the interaction $\phi$ at inverse temperature $\beta$;
			\item $\mu$ is a $(\tau,\beta)$-KMS state.
		\end{enumerate}
	\end{theorem}
    \begin{proof}
    We follow the exposition of \cite{Bra2}. Assume first that $(i)$ holds. By rescaling $\tau_t \rightarrow \tau_{\beta t}$, we can consider $\beta= 1$. Because $\mu$ is faithful by hypothesis, we can construct the modular automorphism group for $\mu$, and call it $\T_t$. Now, if $\xi_\mu$ is the cyclic and separating vector associated with the GNS representation of $\mu$, then it is easy to see that $\xi_\mu^P$  also cyclic and separating. 

 We know that $\mu$ satisfies the KMS for only one dynamics (Theorem III.3.4 of \cite{Is}) and when $\mu$ is a KMS state for $\tau_t$, then $\omega^P$ is a KMS state for $\tau^P$(Theorem 5.4.4 \cite{Bra2}). These assertions together imply that the modular automorphism group of $\mu^P$ is exactly the perturbation of the modular group $\T$.

Since $\mu$ satisfies the Gibbs-Araki-Ion condition for $\phi$ at $1$ we know that when $P=  W_\Lambda$ the perturbed state is the product state $\mu_{1,\phi,\Lambda} \otimes \bar{\mu}$. Because $\mu^P$ is faithful its restrictions are also faithful so we can construct the modular group for both states $\mu_{1,\phi,\Lambda}$ and $\bar{\mu}$. Let us call the modular groups of the restrictions $\tau^\Lambda$ and $\bar{\tau}$. 

It is easy to see that the product state satisfies the KMS condition at $\beta=1$ for the product dynamics and, again invoking the uniqueness of the automorphism group for which the state is KMS, we come to the conclusion that $\T^P_t = \tau^\Lambda_t \otimes \bar{\tau}_t$. Moreover, we know that the generator of $\T$ and $\T^P$, let us call them $\D$ and $\D^P$, are related by
\[
\D(f) = \D^P(f) + i[P,f],
\]
for all $f$ in the domain of $\D$. Consider $f \in C^*(\mathcal{G}_\Lambda)$, for $\Lambda$ the same finite subset of the perturbation. By our previous considerations, we know that $\T^P_t(f) = \tau^\Lambda_t(f)$ for all $t \in \mathbb{R}$. Because $\tau^\Lambda$ is the modular group of the local free Gibbs state we know 
 \[
 \D^P(f) = i[H_\Lambda(\phi),f] \Rightarrow \D(f) = i \sum_{X\cap \Lambda \neq \emptyset}[\phi_X,f]. 
 \]
 Since this is valid for every $\Lambda$ we conclude that $\D = \D^\phi$ and, consequently, that $\tau_t = \T_t$. Assume that $\mu$ is a $(\tau, \beta)$ KMS state. Again, by rescaling we only need to consider the case where $\beta = 1$. All local observables $f$ satisfy
 \[
 \D^\phi_P(f) = \D^\phi(f) - i[W_\Lambda(\phi),f] =  i\sum_{X \subset \Lambda}[\phi_X,f] + i\sum_{X \subset \Lambda^c}[\phi_X,f]
 \]
 Let us denote the second derivation in the right-hand side of the equation $\D_{\Lambda^c}^\phi$. It is easy to see that this derivation can be obtained by perturbing $\D^\phi$ with the operator $Q = -W_\Lambda(\phi) - H_\Lambda(\phi)$, so it generates a strongly continuous one-parameter group that acts trivially in $C^*(\mathcal{G}_\Lambda)$ we have $\D^\phi_{\Lambda^c}(f) = 0$. 
 Let us denote the group generated by $\D_{\Lambda^c}^\phi$ by $\tau^{\Lambda^c}_t$. One finds $\tau^P_t = \tau^\Lambda_t \otimes \tau^{\Lambda^c}_t$. Consider a positive operator $f_3 \in C^*(\mathcal{G}_{\Lambda^c})$ and consider the linear functional in $C^*(\mathcal{G}_\Lambda)$
 \[
 \mu^P_3(f) = \frac{\mu^P(f\cdot f_3)}{\mu^P(f_3)}
 \]
 Clearly, $\mu^P_3$ is positive, continuous, and $\mu_3^P(\mathbbm{1}) = 1$. Since $\mu^P$ is $\tau^P$ KMS and this dynamics is equal the local dynamics $\tau^\Lambda$ restricted to $C^*(\mathcal{G}_\Lambda)$ one finds
 \[
 \mu_3^P(f_1\cdot f_2) = \frac{\mu^P(f_1\cdot f_2 \cdot f_3)}{\mu^P(f_3)} = \frac{\mu^P(f_1\cdot f_3 \cdot f_2)}{\omega^P(f_3)} = \frac{\omega^P(\tau_i^\Lambda(f_2)\cdot f_1 \cdot f_3)}{\omega^P(g)} = \omega_3^P(\tau_i^\Lambda(f_2)\cdot f_1),
 \]
 for any $f_1,f_2 \in C^*(\mathcal{G}_\Lambda)$. But the only KMS state for the local dynamics is the local free Gibbs state and since we can extend this argument to all $f_3\in C^*(\mathcal{G}_{\Lambda^c})$ we arrive at the desired result. 
    \end{proof}

 \section{KMS for classical interactions}
	Define the function $\mathbb{E}:C^*(\mathcal{G}) \rightarrow C(\Omega)$ first on local elements $f \in C_c(\mathcal{G})$, 
	\[
	\mathbb{E}(f)(\sigma)= f(\sigma,1).
	\]
	One can readily see that $\|\mathbb{E}(f)\| \leq \|f\|$, and $\mathbb{E}$ is linear, so there is a bounded extension to all $C^*(\mathcal{G})$. There is also an important property with respect to the product
	\begin{lemma}\label{conditional_expectation}
		Let $f,g \in C_c(\mathcal{G})$. If $f\in C(\Omega)$, then 
		\[
		\mathbb{E}(f\cdot g) = f \cdot \mathbb{E}(g) \text{ and } \mathbb{E}(g \cdot f) = \mathbb{E}(g) \cdot f.
		\]
	\end{lemma}
	\begin{proof}
		We prove just the first property, the second one being analogous. Let $\Lambda \subset \Z^d$ be a finite subset such that $\supp(f) \cup \supp g \subset \Lambda$. By definition, we have
		\[
		\mathbb{E}(f\cdot g)(\sigma,1) = f \cdot g (\sigma,1) = \sum_{(\eta_\Lambda,h_\Lambda) \in \mathcal{G}_\Lambda^{\sigma_\Lambda}}f(\eta_\Lambda,h_\Lambda) g(\sigma,h_\Lambda^{-1}). 
		\]
		Since $f\in C(\Omega)$, $f(\eta_\Lambda,h_\Lambda) \neq 0$ only when $h_\Lambda =1$, therefore the sum above reads
		\[
		\mathbb{E}(f\cdot g)(\sigma) = f(\sigma,1) g(\sigma,1) = f \cdot \mathbb{E}(g) (\sigma).
		\]
	\end{proof}
	The map above is what is called a \emph{conditional expectation} in operator algebras and has important applications (see \cite{Kad}).
	\begin{definition}
		Let $\mu$ be a state. We call it a \textbf{classical state} if and only if for all $a \in \mathfrak{A}_q$ 
		\[
		\mu(a) = \mu(\mathbb{E}(a))
		\]
	\end{definition}
	In other words, classical states can only see the classical part of the observables. The Riesz-Markov theorem, $\omega$ can be identified with a probability measure $\mu$ on the configuration space $\Omega$. In spirit, the next theorem is already proved in Israel \cite{Is}.

\pagebreak
 
	\begin{theorem}\label{t2}
		Let $\phi$ be a classical interaction. Then a state $\mu$ satisfies the Gibbs-Araki-Ion condition if, and only if, the state is classical and satisfies the DLR equations.
	\end{theorem}
	\begin{proof}
		Assume, first, that our state satisfies the Gibbs-Araki-Ion condition. Notice that, since the interaction is classical, we have $\alpha_z(W_\Lambda(\phi)) = W_\Lambda(\phi)$ for all $z\in \mathbb{C}$, hence the operators $E_{P,i\beta}$ simplifies to
		\[
		E_{P,i\beta} = e^{\beta W_\Lambda(\phi)}.
		\]
		Let us show that if $\mu(a)=\mu(\mathbb{E}(a))$. Let $f\in C_c(\mathcal{G})$ be a local operator. Take $\Lambda \subset \Z^d$ a finite set such that $\supp f \subset \Lambda$, then
		\begin{align}\label{classical_eq1}
			\mu_\beta(f) &= \mu_\beta(e^{\frac{\beta}{2} W_\Lambda(\phi)}\cdot e^{-\frac{\beta}{2}W_\Lambda(\phi)}\cdot f \cdot e^{-\frac{\beta}{2}W_\Lambda(\phi)} \cdot e^{\frac{\beta}{2} W_\Lambda(\phi)}) \nonumber \\
			&= \mu^P_\beta(e^{-\frac{\beta}{2}W_\Lambda(\phi)}\cdot f \cdot e^{-\frac{\beta}{2}W_\Lambda(\phi)})\mu_\beta(e^{\beta W_\Lambda(\phi)})
		\end{align}
		due to the Gibbs-Araki-Ion condition, we have
		\[
		\mu^P_\beta(e^{-\frac{\beta}{2}W_\Lambda(\phi)}\cdot f \cdot e^{-\frac{\beta}{2}W_\Lambda(\phi)}) = \bar{\mu}_\beta(\text{Id}\otimes\mu_{\beta,\Lambda}(e^{-\frac{\beta}{2}W_\Lambda(\phi)}\cdot f \cdot e^{-\frac{\beta}{2}W_\Lambda(\phi)})).
		\]
		Calculating the innermost term in a point $(\omega,g) \in \mathcal{G}$, we have
		\begin{align*}
			\text{Id}\otimes\mu_{\beta,\Lambda}(e^{-\frac{\beta}{2}W_\Lambda(\phi)}\cdot f \cdot e^{-\frac{\beta}{2}W_\Lambda(\phi)})(\omega,g) = \frac{1}{Z_{\beta,\Lambda}}\sum_{\sigma_\Lambda \in \Omega_\Lambda} e^{-\frac{\beta}{2}W_\Lambda(\phi)}\cdot f \cdot e^{-\frac{\beta}{2}W_\Lambda(\phi)}(\sigma_\Lambda \omega_{\Lambda^c}, g_{\Lambda^c}),
		\end{align*}
		evaluating the product above we have
		\begin{align*}
			e^{-\frac{\beta}{2}W_\Lambda(\phi)}\cdot f \cdot e^{-\frac{\beta}{2}W_\Lambda(\phi)}(\sigma_\Lambda \omega_{\Lambda^c}, g_{\Lambda^c}) &= \hspace{1.0cm}\sum_{\mathclap{(\eta_{\Lambda_R},h_{\Lambda_R})\in \mathcal{G}_{\Lambda_R}^{\sigma_\Lambda g_{\Lambda^c}\omega_{\Lambda^c}}}}\hspace{0.5cm}e^{-\frac{\beta}{2}W_\Lambda(\phi)}\cdot f(\eta_{\Lambda_R},h_{\Lambda_R})  e^{-\frac{\beta}{2}W_\Lambda(\phi)}(\sigma_\Lambda\omega_{\Lambda^c},h_{\Lambda_R}^{-1}g_{\Lambda^c}) \\
			&= e^{-\frac{\beta}{2}W_\Lambda(\phi)}\cdot f(\eta_{\Lambda_R},g_{\Lambda_R}^{-1})  e^{-\frac{\beta}{2}W_\Lambda(\phi)}(\sigma_\Lambda\omega_{\Lambda^c},1)
		\end{align*}
		But,
		\[
		e^{-\frac{\beta}{2}W_\Lambda(\phi)}(\sigma_\Lambda\omega_{\Lambda^c},1)f(\sigma_\Lambda\omega_{\Lambda^c},1)e^{-\frac{\beta}{2}W_\Lambda(\phi)}(\sigma_\Lambda\omega_{\Lambda^c},1)\\
		= \mathbbm{E}(e^{-\frac{\beta}{2}W_\Lambda(\phi)}\cdot f \cdot e^{-\frac{\beta}{2}W_\Lambda(\phi)})(\sigma_\Lambda\omega_{\Lambda^c},g_{\Lambda^c}).
		\]
		By Lemma \ref{conditional_expectation}, we have
		\[
		\mathbb{E}(e^{-\frac{\beta}{2}W_\Lambda(\phi)}\cdot f \cdot e^{-\frac{\beta}{2}W_\Lambda(\phi)}) = e^{-\frac{\beta}{2}W_\Lambda(\phi)}\cdot \mathbb{E}(f) \cdot e^{-\frac{\beta}{2}W_\Lambda(\phi)},
		\]
		thus plugging again into Equation \eqref{classical_eq1} we get the desired result. The Riesz-Markov theorem tells us that there are $\mu, \mu_P$ probability measures on $\Omega$ such that for all $f \in C(\Omega)$
		\[
		\mu_\beta(f) = \int_\Omega f d\mu_\beta \quad \text{ and } \quad \mu^P_\beta(f) = \int_\Omega f d\mu^P_\beta.
		\]
		Both measures are related by the Radon-Nykodim derivative
		\[
		\frac{d \mu_P}{d\mu} = \frac{e^{\beta W_\Lambda(\phi)}}{\int_\Omega e^{\beta W_\Lambda(\phi)}d\mu},
		\]
		due to the Gibbs-Araki-Ion condition. Equation \ref{gai_equivalence} tells us that the conditional expectation of $\mu^P_\beta$ relative to the $\sigma$-algebra $\mathcal{F}_{\Lambda^c}$ is, for all $f \in C(\Omega)$
		\[
		\mathbb{E}_{\mu^P_\beta}(f)(\eta_{\Lambda^c}) = \mu_{\beta,\Lambda}(f)(\eta_{\Lambda^c}) = \frac{\sum_{\sigma_\Lambda \in \Omega_\Lambda}f(\sigma_\Lambda \eta_{\Lambda^c})e^{-\beta H_\Lambda(\phi)(\sigma_\Lambda)}}{\sum_{\sigma_\Lambda \in \Omega_\Lambda}e^{-\beta H_\Lambda(\phi)(\sigma_\Lambda)}}
		\]
		
		Denote $\left. \nu \right|_{\Lambda^c}$ the restriction of a measure $\nu$ to the $\sigma$-algebra $\mathcal{F}_{\Lambda^c}$. We want to calculate the conditional expectation for $\mu_\beta$. Note that
		\begin{equation}\label{classical_eq2}
			\int_\Omega \mathbb{E}_{\mu_\beta}(f) d\left.\mu\right|_{\Lambda^c} = \int_\Omega f d\mu_\beta = \int_\Omega f \frac{d\mu_\beta}{d\mu_\beta^P} d\mu_\beta^P = \int_\Omega \mathbb{E}_{\mu_\beta^P}\left(f \frac{d\mu_\beta}{d\mu_\beta^P}\right) d\left.\mu_\beta^P\right|_{\Lambda^c}.
		\end{equation}
		The restriction of absolute continuous measures are absolutely continuous. The relation between the Radon-Nykodim derivatives is
		\begin{equation}\label{classical_eq3}
			\int_\Omega \frac{d\left.\mu_\beta\right|_{\Lambda^c}}{d\left.\mu_\beta^P\right|_{\Lambda^c}}d\left.\mu_\beta^P\right|_{\Lambda^c} = \int_\Omega d\left.\mu_\beta\right|_{\Lambda^c} = \int_\Omega d\mu_\beta = \int_\Omega \frac{d\mu_\beta}{d\mu_\beta^P}d\mu_\beta^P = \int_\Omega \mathbb{E}_{\mu_\beta^P}\left(\frac{d\mu_\beta}{d\mu_\beta^P}\right)d\left.\mu_\beta^P\right|_{\Lambda^c}.
		\end{equation}
		
		The Equations \ref{classical_eq2} and \ref{classical_eq3} together give us
		\be\label{classical_eq4}
		\mathbb{E}_\mu(f)\mathbb{E}_{\mu_\beta^P}\left(\frac{d\mu_\beta}{d\mu_\beta^P}\right)= \mathbb{E}_{\mu_\beta^P}\left(f\frac{d\mu_\beta}{d\mu_\beta^P}\right). 
		\ee
		For all $f \in C(\Omega)$. Now, consider $\eta_{\Lambda^c} \in \Omega_{\Lambda^c}$ and let us calculate both sides of the Equation \eqref{classical_eq4}. The right-hand side is
		\[
		\mathbb{E}_{\mu_\beta^P}\left(f\frac{d\mu_\beta}{d\mu_\beta^P}\right)(\eta_{\Lambda^c}) = \frac{\sum_{\sigma_\Lambda \in \Omega_\Lambda}f(\sigma_\Lambda \eta_{\Lambda^c})e^{-\beta (H_\Lambda(\phi)(\sigma_\Lambda)+W_\Lambda(\phi)(\sigma_\Lambda \eta_{\Lambda^c}))}}{\sum_{\sigma_\Lambda \in \Omega_\Lambda}e^{-\beta H_\Lambda(\phi)(\sigma_\Lambda)}},
		\]
		then we get
		\[
		\sum_{i,j=1}^n \langle \psi, \pi^{\eta_\Lambda\omega_\Lambda}(f) \psi\rangle \geq 0
		\]
		since $f$ is positive. and the left-hand side is
		\[
		\mathbb{E}_{\mu_\beta^P}\left(\frac{d\mu_\beta}{d\mu_\beta^P}\right)(\eta_{\Lambda^c}) = \frac{\sum_{\sigma_\Lambda \in \Omega_\Lambda} e^{-\beta (H_\Lambda(\phi)(\sigma_\Lambda)+W_\Lambda(\phi)(\sigma_\Lambda \eta_{\Lambda^c}))}}{\sum_{\sigma_\Lambda \in \Omega_\Lambda}e^{-\beta H_\Lambda(\phi)(\sigma_\Lambda)}} 
		\]
		Since $H_\Lambda^\eta(\phi)(\sigma_{\Lambda}) = H_\Lambda(\phi)(\sigma_{\Lambda}) + W_\Lambda(\phi)(\sigma_\Lambda \eta_{\Lambda^c})$ we conclude that
		\[
		\mathbb{E}_{\mu_\beta}(f) = \mu_{\beta,\phi,\Lambda}^\eta(f),
		\]
		exactly the expression for the finite volume Gibbs measure \eqref{finiteclassical}. Assume that the state $\mu_\beta$ is classical and satisfies the DLR equations. Our previous calculations can be reversed and we conclude that all perturbation states $\mu^P_\beta$ satisfy \eqref{gai_equivalence}. We just need to show that the state $\mu_\beta$ is faithful. First, we will prove a relation between the left regular representation $\pi$ of a local element $f$ and $E(f)$. Let $\sigma$ be a configuration and $\Lambda$ a finite set containing $\supp f$. Then, 
		\[
		\pi^\sigma(f)\delta_{(\sigma,g)} = \sum_{k_\Lambda \in G_\Lambda}f(g_\Lambda\sigma_\Lambda, k_\Lambda g^{-1}_\Lambda) \delta_{(g_{\Lambda}\sigma_\Lambda,k_\Lambda g_\Lambda^{-1})}.
		\]
		
		\begin{align*}
			\|\pi^\sigma(f)\delta_{(\sigma, g)}\|^2 &= \langle\pi^\sigma(f)\delta_{(\sigma, g)},\pi^\sigma(f)\delta_{(\sigma, g)} \rangle \\
			&=\sum_{k_\Lambda \in G_\Lambda} \overline{f(g_\Lambda\sigma_\Lambda, k_\Lambda g^{-1}_\Lambda)}f(g_\Lambda\sigma_\Lambda, k_\Lambda g^{-1}_\Lambda) \\
			&= \sum_{h_\Lambda \in G_\Lambda} f^*(h_\Lambda^{-1}g_\Lambda\sigma_\Lambda, h_\Lambda)f(g_\Lambda\sigma_\Lambda,h_\Lambda) \\
			&=\sum_{(\omega_\Lambda,h_\Lambda)\in\mathcal{G}_\Lambda^{g_\Lambda\sigma_\Lambda}}f^*(\omega_\Lambda,h_\Lambda) f(g_\Lambda \sigma_\Lambda,h_\Lambda^{-1}) = E(f^*\cdot f) (g_\Lambda\sigma_\Lambda,1)
		\end{align*}
		Thus, if $f \in C^*(\mathcal{G})$ is such that $\mathbb{E}(f^*\cdot f) = 0$, we just need to take a sequence $f_n$ of local elements converging to $f$, then
		\[
		\|\pi^\sigma(f)\delta_{(\sigma, g)}\|^2 = \lim_{n\rightarrow \infty}\|\pi^\sigma(f_n)\delta_{(\sigma, g)}\|^2 = \lim_{n\rightarrow \infty} \mathbb{E}(f_n^*\cdot f_n) = \mathbb{E}(f^*\cdot f) = 0
		\]
		Since $\delta_{(\sigma,g)}$ is a basis for $\ell^2(\mathcal{G})$, we conclude that $\pi(f) = 0$ if $\mathbb{E}(f^*\cdot f) = 0$. Going back to the state $\mu_\beta$, if $f\in C^*(\mathcal{G})$ is such that $\mu_\beta(f^* \cdot f) =0$, since it is classical, we must have
		$\mathbb{E}(f^* \cdot f) =0$ because the integral of a positive function is zero if and only if the function is zero. By our previous reasoning, $f=0$ and we are done.  
	\end{proof}
	
One can argue that our definition of classical interaction may be too restrictive. Indeed, interactions $\phi$ that are polynomials only depending on the variables $\sigma^{(1)}_x$ can be seen as classical also although they most certainly do not satisfy the condition $\phi_X \in C(\Omega_X)$. Certainly a variant of the Theorem \ref{t2} must hold when one changes $C(\Omega)$ by the abelian sub-C*-algebra generated by the interactions $\phi_X$.

\chapter{Boundary Conditions and Quantum DLR equations}
\label{ch:groupoids}

This chapter contains an expanded version of some of the results announced in the preliminary version of \cite{Aff2}. Here we construct the random representation for the Gibbs density for all short-range interactions and introduce the quantum DLR equations for a subclass of interactions that we called \emph{admissible}. Then, we discuss the relation between the states satisfying the DLR equations and the states coming from the usual thermodynamic limit.

\section{The random representation}

In this chapter, the interactions will be considered to be short-range.  

\begin{proposition}\label{poirep}
	The Gibbs density operator
	$e^{-\beta H_{\Lambda}(\phi)}$ has a  Poisson point process representation. 
\end{proposition}
\begin{proof}
	We will construct the random representation for systems with $q=2$, the general case, see Remark \ref{generalcase} at the end of this proof. The Hamiltonian can be written as 
	\be\label{hamiltonian}
	H_\Lambda(\phi)= H_\Lambda^{(0)}(\phi)-\sum_{\emptyset\neq B\subset \Lambda} f_B \cdot \sigma_B^{(1)},
	\ee
	The Hamiltonian $H_{\Lambda}^{(0)}(\phi)$ is the classical part of the initial quantum Hamiltonian $H_{\Lambda}(\phi)$, defined as
	\[
	H_\Lambda^{(0)}(\phi) = -\sum_{A\subset  \Lambda} J_A\sigma_A^{(3)},
	\]
	and, for each $B\neq \emptyset$, the functions $f_B\in C(\Omega_\Lambda)$, that can be written as
\[
f_B = \sum_{A\subset X}c_{A,B}\sigma_A^{(3)}.
\] 
The interactions should be self-adjoint so there are more restrictions on the coefficients $c_{A,B}$
\[
(c_{A,B}\sigma_A^{(3)}\sigma_B^{(1)})^* = \overline{c_{A,B}}\sigma_B^{(1)}\sigma_A^{(3)} = \overline{c_{A,B}}(-1)^{|A\cap B|}\sigma_A^{(3)}\sigma_B^{(1)},
\]
using the polar decomposition for the constants $c_{A,B} = r_{A,B}e^{i\pi\theta_{A,B}}$ we discover that the angles must satisfy $\theta_{A,B}\in \{0, \pm 1/2, 1\}$. The constants $J_A = r_A e^{i\pi \theta_A}$ are real numbers since the self-adjointeness of the interaction $\phi$ implies that $\theta_A=0$ or $1$.  Let us call $V_\Lambda(\phi) = -\sum_{\emptyset \neq B\subset \Lambda}f_B\sigma_B^{(1)}$. The Lie-Trotter formula yields
	\be\label{eq1:poirep}
	e^{-\beta H_\Lambda(\phi)}= e^{\beta }\lim_{n\rightarrow \infty} \left[e^{-\frac{\beta}{n+1} H_\Lambda^{(0)}(\phi)}\cdot \left(\left(1-\frac{\beta}{n}\right)\mathbbm{1}+\frac{\beta}{n}V_\Lambda(\phi)\right)\right]^n \cdot e^{-\frac{\beta}{n+1} H_\Lambda^{(0)}(\phi)}.
	\ee
	For each $n\geq 1$, the right-hand-side of Equation \eqref{eq1:poirep} can be expanded as
	
	\begin{align}\label{eq2:poirep}
		\left[e^{-\frac{\beta}{n+1} H_\Lambda^{(0)}(\phi)}\left(\left(1-\frac{\beta}{n}\right)\mathbbm{1}+\frac{\beta}{n}V_\Lambda(\phi)\right)\right]^n& = \nonumber\\
		\sum_{j\in \{0,1\}^n} \Conv_{m=1}^n e^{-\frac{\beta}{n+1} H_\Lambda^{(0)}(\phi)}&\cdot V_\Lambda(\phi)^{j(m)}\left(1-\frac{\beta}{n}\right)^{n-\sum_{m=1}^n j(m)}\left(\frac{\beta}{n}\right)^{\sum_{m=1}^n j(m)},
	\end{align}
	where we use the convention $V_\Lambda(\phi)^0 = \mathbbm{1}$. By breaking the sum depending on the functions $j$ depending on how many $m$ the function $j(m)=1$ in the right-hand side of Equation \eqref{eq2:poirep}, we get that the right-hand side of Equation \eqref{eq1:poirep} is equal to
	\be\label{eq10:poirep}
	\sum_{\ell =0}^n\sum_{\substack{j\in \{0,1\}^n \\ |j(m)= 1|=\ell}} \Bigg[\Conv_{m=1}^n e^{-\frac{\beta}{n+1} H_\Lambda^{(0)}(\phi)}\cdot V_\Lambda(\phi)^{j(m)}\Bigg]\cdot e^{-\frac{\beta}{n+1} H_\Lambda^{(0)}(\phi)}\left(1-\frac{\beta}{n}\right)^{n-\ell}\left(\frac{\beta}{n}\right)^{\ell}.
	\ee
	Enumerate the points where the function $j$ is not zero in an increasing order $m_1 <\dots < m_\ell$, where given $j \in \{0,1\}^n$, we have $|j(m)=1|=\ell$. Hence,
	\begin{align*}
		\Bigg[\Conv_{m=1}^n e^{-\frac{\beta}{n+1} H_\Lambda^{(0)}(\phi)}\cdot V_\Lambda(\phi)^{j(m)}\Bigg]\cdot e^{-\frac{\beta}{n+1} H_\Lambda^{(0)}(\phi)} = \Bigg[\Conv_{j=1}^\ell e^{-\frac{\beta(m_{j}-m_{j-1})}{n+1} H_\Lambda^{(0)}(\phi)}\cdot V_\Lambda(\phi)\Bigg] \cdot e^{-\frac{\beta(n+1-m_{\ell})}{n+1} H_\Lambda^{(0)}(\phi)}.
	\end{align*}
	
	We will calculate the value of the operator on the left-hand side \eqref{eq1:poirep} in a point $(\sigma_\Lambda,g_\Lambda)$ and expand the right-hand side using the  partition of the identity of the algebra
	\[
	\mathbbm{1} = \sum_{\sigma_\Lambda \in \Omega_\Lambda}\delta_{\sigma_\Lambda},
	\]
	where $\delta_{\sigma_\Lambda}$ are delta functions on the unit $\sigma_\Lambda\in \mathcal{G}^{(0)}_\Lambda$. Hence, using Equations  \eqref{eq10:poirep} and \eqref{prod_classical_quantum},
	\begin{align*}
		&\delta_{g_\Lambda\sigma_\Lambda}\cdot \left(\Conv_{m=1}^n e^{-\frac{\beta}{n+1} H_\Lambda^{(0)}(\phi)}\cdot V_\Lambda(\phi)^{j(m)}\right)\cdot e^{-\frac{\beta}{n+1} H_\Lambda^{(0)}(\phi)}\cdot \delta_{\sigma_\Lambda} =\\
		&\sum_{\substack{\sigma_{\Lambda,k} \in \Omega_\Lambda \\ 1\leq k \leq \ell}}\delta_{g_\Lambda\sigma_\Lambda}\cdot \left(\Conv_{k=1}^\ell e^{-\frac{\beta(m_k -m_{k-1})}{n+1}H_\Lambda^{(0)}(\phi)(\sigma_{\Lambda,k})}\delta_{\sigma_{\Lambda,k}}\cdot V_\Lambda(\phi)\right)\cdot e^{-\frac{\beta(n+1-m_{\ell})}{n+1}H_\Lambda^{(0)}(\phi)}\cdot\delta_{\sigma_\Lambda} =\\
		&\sum_{\substack{\sigma_{\Lambda,k} \in \Omega_\Lambda \\ 1\leq k \leq \ell}}e^{-\beta E_\Lambda(\phi)((\sigma_{\Lambda,1},m_1/(n+1)), \dots, (\sigma_{\Lambda,\ell},m_\ell/(n+1)))}\delta_{g_\Lambda\sigma_\Lambda}\cdot\left( \Conv_{k=1}^\ell\delta_{\sigma_{\Lambda,k}}\cdot V_\Lambda(\phi)\right)\cdot\delta_{\sigma_\Lambda} ,
	\end{align*}
	where  the function $E_\Lambda(\phi)$ is
	\[
	E_\Lambda(\phi)((\sigma_{\Lambda,1},m_1/(n+1)), \dots, (\sigma_{\Lambda,\ell},m_\ell/(n+1)) = \frac{1}{n+1}\sum_{k=1}^{\ell+1} (m_k-m_{k-1})H_\Lambda^{(0)}(\phi)(\sigma_{\Lambda,k}),
	\]
	with $\sigma_{\Lambda,\ell+1}=\sigma_\Lambda$.  We have, by the iterated use of Lemma \ref{matrix_elements},
	\begin{align*}
		&\delta_{g_\Lambda\sigma_\Lambda}\cdot\left( \Conv_{k=1}^\ell\delta_{\sigma_{\Lambda,k}}\cdot V_\Lambda(\phi)\right)\cdot\delta_{\sigma_\Lambda}= \delta_{g_\Lambda\sigma_\Lambda}\cdot\delta_{\sigma_{\Lambda,1}}\cdot V_\Lambda(\phi)\cdot \delta_{\sigma_{\Lambda,2}}\cdot \left( \Conv_{k=2}^{\ell} V_\Lambda(\phi)\cdot\delta_{\sigma_{\Lambda,k+1}} \right)\\
		&= \delta_{\{g_\Lambda\sigma_\Lambda = \sigma_{\Lambda,1}\}}\cdot \delta_{g_\Lambda\sigma_\Lambda}\cdot V_\Lambda(\phi)\cdot \delta_{\sigma_{\Lambda,2}}\cdot \left( \Conv_{k=2}^{\ell} V_\Lambda(\phi)\cdot\delta_{\sigma_{\Lambda,k+1}} \right)\\
		&= \delta_{\{g_\Lambda\sigma_\Lambda = \sigma_{\Lambda,1}\}}V_\Lambda(\phi)(\sigma_{\Lambda,2},g_2)\delta_{g_\Lambda\sigma_\Lambda}\cdot \delta_{(\sigma_{\Lambda,2},g_2)}\cdot \left( \Conv_{k=2}^{\ell} V_\Lambda(\phi)\cdot\delta_{\sigma_{\Lambda,k+1}} \right) \\
		&=\delta_{\{g_\Lambda\sigma_\Lambda = \sigma_{\Lambda,1}\}}\Bigg(\prod_{k=2}^3V_\Lambda(\phi)(\sigma_{\Lambda,k},\widetilde{g}_k)\Bigg)\delta_{g_\Lambda\sigma_\Lambda}\cdot \delta_{(\sigma_{\Lambda,3},\widetilde{g}_3)}\cdot \left( \Conv_{k=3}^{\ell} V_\Lambda(\phi)\cdot\delta_{\sigma_{\Lambda,k+1}} \right)\\
		&\vdots \\
		&= \delta_{\{g_\Lambda\sigma_\Lambda = \sigma_{\Lambda,1}\}}\left(\prod_{k=2}^{\ell+1} V_\Lambda(\phi)(\sigma_{\Lambda,k},\widetilde{g}_k)\right)\delta_{g_\Lambda\sigma_\Lambda}\cdot \delta_{(\sigma_{\Lambda,\ell-1},\widetilde{g}_{\ell-1})}\cdot V_\Lambda(\phi)\cdot\delta_{\sigma_{\Lambda}} \\
		&= \delta_{\{g_\Lambda\sigma_\Lambda = \sigma_{\Lambda,1}\}}\left(\prod_{k=2}^{\ell+1} V_\Lambda(\phi)(\sigma_{\Lambda,k},\widetilde{g}_k)\right)\delta_{g_\Lambda\sigma_\Lambda}\cdot \delta_{(\sigma_{\Lambda,\ell},\widetilde{g}_\ell)}\cdot\delta_{\sigma_\Lambda} = \delta_{\{g_\Lambda\sigma_\Lambda = \sigma_{\Lambda,1}\}}\left(\prod_{k=2}^{\ell+1} V_\Lambda(\phi)(\sigma_{\Lambda,k},\widetilde{g}_k)\right) \delta_{(\sigma_\Lambda,\widetilde{g}_\Lambda)}
	\end{align*}
	where $\widetilde{g}_k$ come from Lemma \ref{matrix_elements}.
 Define the function
	\[
	V_\Lambda(\phi)(\sigma_{\Lambda,1},\dots,\sigma_{\Lambda,\ell};(\sigma_\Lambda,g_\Lambda)) = \prod_{k=2}^{\ell+1} V_\Lambda(\phi)(\sigma_{\Lambda,k},\widetilde{g}_k)\delta_{\{g_\Lambda\sigma_\Lambda=\sigma_{\Lambda,1}\}}.
	\]
\begin{figure}[H]
\centering

\scalebox{1}{

\tikzset{every picture/.style={line width=0.75pt}} 

\begin{tikzpicture}[x=0.75pt,y=0.75pt,yscale=-1,xscale=1]

\draw  [fill={rgb, 255:red, 0; green, 0; blue, 0 }  ,fill opacity=1 ] (561.95,91.78) .. controls (561.95,88.36) and (564.65,85.59) .. (567.97,85.59) .. controls (571.3,85.59) and (574,88.36) .. (574,91.78) .. controls (574,95.2) and (571.3,97.98) .. (567.97,97.98) .. controls (564.65,97.98) and (561.95,95.2) .. (561.95,91.78) -- cycle ;
\draw  [fill={rgb, 255:red, 0; green, 0; blue, 0 }  ,fill opacity=1 ] (461.09,91.78) .. controls (461.09,88.36) and (463.79,85.59) .. (467.12,85.59) .. controls (470.45,85.59) and (473.14,88.36) .. (473.14,91.78) .. controls (473.14,95.2) and (470.45,97.98) .. (467.12,97.98) .. controls (463.79,97.98) and (461.09,95.2) .. (461.09,91.78) -- cycle ;
\draw [line width=1.5]    (468.46,81.66) .. controls (487.07,31.83) and (547.87,27.09) .. (567.97,85.59) ;
\draw [shift={(467.12,85.59)}, rotate = 287.45] [fill={rgb, 255:red, 0; green, 0; blue, 0 }  ][line width=0.08]  [draw opacity=0] (13.4,-6.43) -- (0,0) -- (13.4,6.44) -- (8.9,0) -- cycle    ;

\draw  [fill={rgb, 255:red, 0; green, 0; blue, 0 }  ,fill opacity=1 ] (384.34,91.78) .. controls (384.34,88.36) and (387.04,85.59) .. (390.36,85.59) .. controls (393.69,85.59) and (396.39,88.36) .. (396.39,91.78) .. controls (396.39,95.2) and (393.69,97.98) .. (390.36,97.98) .. controls (387.04,97.98) and (384.34,95.2) .. (384.34,91.78) -- cycle ;
\draw  [fill={rgb, 255:red, 0; green, 0; blue, 0 }  ,fill opacity=1 ] (283.48,91.78) .. controls (283.48,88.36) and (286.18,85.59) .. (289.51,85.59) .. controls (292.84,85.59) and (295.53,88.36) .. (295.53,91.78) .. controls (295.53,95.2) and (292.84,97.98) .. (289.51,97.98) .. controls (286.18,97.98) and (283.48,95.2) .. (283.48,91.78) -- cycle ;
\draw [line width=1.5]    (290.85,81.66) .. controls (309.46,31.83) and (370.26,27.09) .. (390.36,85.59) ;
\draw [shift={(289.51,85.59)}, rotate = 287.45] [fill={rgb, 255:red, 0; green, 0; blue, 0 }  ][line width=0.08]  [draw opacity=0] (13.4,-6.43) -- (0,0) -- (13.4,6.44) -- (8.9,0) -- cycle    ;
\draw  [fill={rgb, 255:red, 0; green, 0; blue, 0 }  ,fill opacity=1 ] (181.99,91.78) .. controls (181.99,88.36) and (184.69,85.59) .. (188.02,85.59) .. controls (191.35,85.59) and (194.04,88.36) .. (194.04,91.78) .. controls (194.04,95.2) and (191.35,97.98) .. (188.02,97.98) .. controls (184.69,97.98) and (181.99,95.2) .. (181.99,91.78) -- cycle ;
\draw [line width=1.5]    (189.36,81.66) .. controls (207.97,31.83) and (268.77,27.09) .. (288.87,85.59) ;
\draw [shift={(188.02,85.59)}, rotate = 287.45] [fill={rgb, 255:red, 0; green, 0; blue, 0 }  ][line width=0.08]  [draw opacity=0] (13.4,-6.43) -- (0,0) -- (13.4,6.44) -- (8.9,0) -- cycle    ;

\draw  [fill={rgb, 255:red, 0; green, 0; blue, 0 }  ,fill opacity=1 ] (80.5,91.78) .. controls (80.5,88.36) and (83.2,85.59) .. (86.53,85.59) .. controls (89.85,85.59) and (92.55,88.36) .. (92.55,91.78) .. controls (92.55,95.2) and (89.85,97.98) .. (86.53,97.98) .. controls (83.2,97.98) and (80.5,95.2) .. (80.5,91.78) -- cycle ;
\draw [line width=1.5]    (87.87,81.66) .. controls (106.48,31.83) and (167.28,27.09) .. (187.38,85.59) ;
\draw [shift={(86.53,85.59)}, rotate = 287.45] [fill={rgb, 255:red, 0; green, 0; blue, 0 }  ][line width=0.08]  [draw opacity=0] (13.4,-6.43) -- (0,0) -- (13.4,6.44) -- (8.9,0) -- cycle    ;

\draw (410.88,84.5) node [anchor=north west][inner sep=0.75pt]  [font=\LARGE] [align=left] {$\displaystyle \cdots $};
\draw (549,105) node [anchor=north west][inner sep=0.75pt]   [align=left] {$\displaystyle \sigma _{\Lambda ,\ell+1} =\sigma _{\Lambda }$};
\draw (108,14) node [anchor=north west][inner sep=0.75pt]   [align=left] {$\displaystyle ( \sigma _{\Lambda ,2} ,\widetilde{g}_{2})$};
\draw (208,13) node [anchor=north west][inner sep=0.75pt]   [align=left] {$\displaystyle ( \sigma _{\Lambda ,3} ,\widetilde{g}_{3})$};
\draw (309,13) node [anchor=north west][inner sep=0.75pt]   [align=left] {$\displaystyle ( \sigma _{\Lambda ,4} ,\widetilde{g}_{4})$};
\draw (470,13) node [anchor=north west][inner sep=0.75pt]   [align=left] {$\displaystyle ( \sigma _{\Lambda ,\ell+1} ,\widetilde{g}_{\ell+1})$};
\draw (14,104) node [anchor=north west][inner sep=0.75pt]   [align=left] {$\displaystyle \widetilde{g}_{\Lambda } \sigma _{\Lambda } =\sigma _{\Lambda ,1}$};
\draw (173,104) node [anchor=north west][inner sep=0.75pt]   [align=left] {$\displaystyle \sigma _{\Lambda ,2}$};
\draw (274,104) node [anchor=north west][inner sep=0.75pt]   [align=left] {$\displaystyle \sigma _{\Lambda ,3}$};
\draw (377,105) node [anchor=north west][inner sep=0.75pt]   [align=left] {$\displaystyle \sigma _{\Lambda ,4}$};
\draw (453,105) node [anchor=north west][inner sep=0.75pt]   [align=left] {$\displaystyle \sigma _{\Lambda ,n}$};

\end{tikzpicture}

}

\caption{The path of arrows.}
\end{figure}
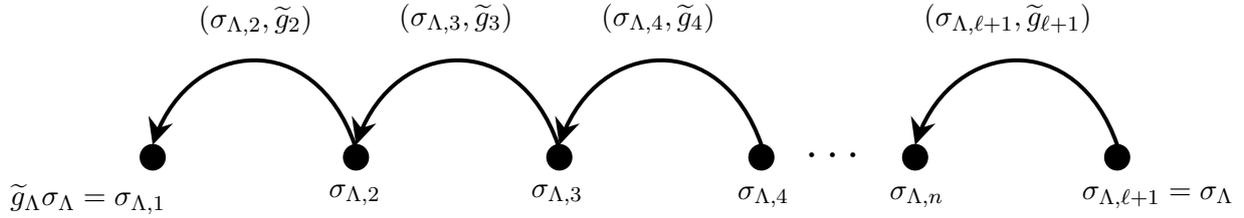	
 
 Equation \eqref{eq10:poirep} becomes
	\be\label{eq4:poirep}
	\begin{split}			
		\sum_{\ell =0}^n\sum_{\substack{j\in \{0,1\}^n \\ |j(m)= 1|=\ell}}\sum_{\substack{\sigma_{\Lambda,k} \in \Omega_\Lambda \\ 1\leq k \leq \ell}}&e^{-\beta E_\Lambda(\phi)((\sigma_{\Lambda,1},m_1/n), \dots, (\sigma_{\Lambda,\ell},m_\ell/n))} \times \\
		&V_\Lambda(\phi)(\sigma_{\Lambda,1},\dots,\sigma_{\Lambda,\ell};(\sigma_\Lambda,g_\Lambda))\left(1-\frac{\beta}{n}\right)^{n-\ell}\left(\frac{\beta}{n}\right)^{\ell}.
	\end{split}
	\ee
	Define the time ordering functions $T_\ell:[0,1]^\ell \rightarrow [0,1]^\ell$
	\[
	T_\ell(t_1,\dots,t_\ell) = (t'_1,\dots, t'_\ell),
	\]
	where the right-hand side is a permutation of $(t_1,\dots,t_\ell)$ satisfying $t'_k \leq t'_{k+1}$, $k=1,\dots, \ell-1$. One can write the function $T_k$ more explicitly using characteristics functions 
	\[
	T_\ell(t_1,\dots,t_\ell) = \sum_{p\in\mathfrak{S}_\ell}\left(\prod_{k=1}^{\ell-1} \mathbbm{1}_{\{x>0\}}(t_{p(k)}-t_{p(k+1)})\right)(t_{p(1)},\dots,t_{p(\ell)}),
	\] 
	where $\mathfrak{S}_\ell$ is the permutation group of $n$ points. Notice that this function is zero whenever we have two coordinates $t_k$ and $t_j$ that are equal. Since the set of points of $[0,1]^\ell$ where at least two coordinates are equal have measure zero with respect to the Lebesgue measure, we can define it as we want on this set. Using the functions $T_n$ we can construct a time ordering function $T$ on the coproduct (see Proposition \ref{prop_measurable_coproduct}). Also,  let $F^\phi_{\beta,\Lambda}$ be the function defined on the coproduct of $\widetilde{\Omega}_\Lambda^\ell$, for $n\geq 0$ using also Proposition \ref{prop_measurable_coproduct} and the functions $F_{\beta,\ell}^\phi:\widetilde{\Omega}_\Lambda^\ell \rightarrow \mathbb{C}$ given by
	\[
	F_{\beta,\Lambda,\ell}^\phi((\sigma_{\Lambda,k},t_k)_\ell;(\sigma_\Lambda,g_\Lambda)) = e^{-\beta E_\Lambda(\phi)((\sigma_{\Lambda,k},t_k)_\ell;(\sigma_\Lambda,g_\Lambda))}V_\Lambda(\phi)((\sigma_{\Lambda,k},t_k)_\ell;(\sigma_\Lambda,g_\Lambda)),
	\]
	where $(\sigma_{\Lambda,k},t_k)_\ell = ((\sigma_{\Lambda,1},t_1),\dots,(\sigma_{\Lambda,n},t_\ell))$ with
	\begin{align*}
		&E_\Lambda(\phi)((\sigma_{\Lambda,k},t_k)_\ell;(\sigma_\Lambda,g_\Lambda)) = \sum_{k=1}^{\ell+1} (t'_k-t'_{k-1})H_\Lambda^{(0)}(\phi)(\sigma_{\Lambda,k}) \quad \text{and}\quad \\
		&V_\Lambda(\phi)((\sigma_{\Lambda,k},t_k)_\ell;(\sigma_\Lambda,g_\Lambda)) = \prod_{k=2}^{\ell+1} V_\Lambda(\phi)(\sigma_{\Lambda,k},g_k)\delta_{\{g_\Lambda\sigma_\Lambda=\sigma_{\Lambda,1}\}},
	\end{align*}
	
	where $t_0 = 0$, $t_{\ell+1}=1$, $\sigma_{\Lambda,1} = g_\Lambda\sigma_\Lambda$, and $\sigma_{\Lambda,\ell+1} = \sigma_\Lambda$. Consider the Bernoulli point process
	\[
	N_n(s,C) = \sum_{j=1}^n \chi_{n,j}(s)\delta_{(\zeta_j(s),\frac{j}{n+1})}(C), 
	\]
	where $C$ is a Borel subset of $\widetilde{\Omega}_\Lambda = \Omega_\Lambda\times [0,1]$ and $\{\chi_{i,j}\}_{i \in \mathbb{N}, 1\leq j \leq n}$ and $\{\zeta_j\}_{j\in\mathbb{N}}$ are  families of i.i.d random variables with distribution
	\[
	\mathbb{P}(\chi_{n,j}=1) = 1- \mathbb{P}(\chi_{n,j}=0)= \frac{\beta}{n}\quad \text{  and  } \quad \mathbb{P}(\zeta_j=\sigma_\Lambda)= 1.
	\] 
	The expression \eqref{eq4:poirep} is then, an integral of the function $F_{\beta,\Lambda}^\phi$ with respect to the Binomial point process (see Proposition \ref{integformpoirep:app1}). Since the Binomial point process converges in distribution to a Poisson point process we have 
	\[
	e^{-\beta H_\Lambda(\phi)}(\sigma_\Lambda,g_\Lambda) = \lim_{n\rightarrow \infty}\int_{\mathbb{N}(\widetilde{\Omega}_\Lambda)} F_{\beta,\Lambda}^\phi(\nu) dN_n(\nu) = \int_{\mathbb{N}(\widetilde{\Omega}_\Lambda)} F_{\beta,\Lambda}^\phi(\nu) dN(\nu).
	\]
	We know that the function $V_\Lambda(\phi)$ can be expanded by using the fact that $V_\Lambda(\phi)= \sum_{\emptyset\neq B\subset X} f_B \cdot \sigma_B^{(1)}$, in the following way
	\begin{align*}
		V_\Lambda(\phi)((\sigma_{\Lambda_R,k},t_k)_\ell;(\sigma_\Lambda,g_\Lambda)) &= \prod_{k=2}^{\ell+1} V_\Lambda(\phi)(\sigma_{\Lambda,k},g_k)\delta_{\{g_\Lambda\sigma_\Lambda=\sigma_{\Lambda,1}\}} \\
		& = \prod_{k=2}^{\ell+1}\left(\sum_{\emptyset \neq B \subset X}f_B\cdot\sigma_B^{(1)}(\sigma_{\Lambda,k},g_k)\right)\delta_{\{g_\Lambda\sigma_\Lambda=\sigma_{\Lambda,1}\}} \\
		& = \sum_{\substack{\emptyset \neq B_k \subset  X \\ 1\leq k \leq n}} \prod_{k=2}^{\ell+1} f_{B_{k-1}}\cdot\sigma_{B_{k-1}}^{(1)}(\sigma_{\Lambda,k-1},g_{k-1})\delta_{\{g_\Lambda\sigma_\Lambda=\sigma_{\Lambda,1}\}}.
	\end{align*}
	The Poisson point process representation presented here can be understood as a rigorous path integral for quantum spin systems. We will proceed to write it in a way that the analogy with the integration with respect to paths becomes more evident. The integration formula for Poisson point processes in Proposition \ref{Poisson_int} yields
	
	\begin{equation}\label{eq_randomrep}
		\begin{split}
			\int_{\mathbb{N}(\widetilde{\Omega}_\Lambda)} F_{\beta,\Lambda}^\phi(\nu) dN(\nu) = \sum_{n\geq 0} \frac{\beta^n}{n!}\int_{[0,1]^n}\sum_{\substack{\sigma_{\Lambda,k} \in \Omega_\Lambda \\ 1\leq k \leq n}}e^{-\beta E_\Lambda(\phi)((\sigma_{\Lambda,k},t_k)_n;(\sigma_\Lambda,g_\Lambda))}V_\Lambda(\phi)((\sigma_{\Lambda,k},t_k)_n;(\sigma_\Lambda,g_\Lambda))dt^n& \\
			= \sum_{n\geq 0} \frac{\beta^n}{n!}\int_{[0,1]^n}\sum_{\substack{\sigma_{\Lambda,k} \in \Omega_\Lambda \\ 1\leq k \leq n}}\sum_{\substack{\emptyset\neq B_k \subset \Lambda \\ 1\leq k \leq n}}e^{-\beta E_\Lambda(\phi)((\sigma_{\Lambda,k},t_k)_n;(\sigma_\Lambda,g_\Lambda))} \prod_{k=1}^n f_{B_k}\cdot\sigma_{B_k}^{(1)}(\sigma_{\Lambda,k},g_k)\delta_{\{g_\Lambda\sigma_\Lambda=\sigma_{\Lambda,1}\}}dt^n&.
		\end{split}
	\end{equation}
	By Corollary \ref{Corol_ppp}, the formula above is just the integration of the function
	\[
	(\sigma_{\Lambda,k},t_k, B_k)_n \mapsto  e^{-\beta E_\Lambda(\phi)((\sigma_{\Lambda,k},t_k)_n;(\sigma_\Lambda,g_\Lambda))} \prod_{k=1}^n  f_{B_k}\cdot\sigma_{B_k}^{(1)}(\sigma_{\Lambda,k},g_k),
	\]
	with respect to the Poisson Point process $N_\Lambda\coloneqq \sum_{\emptyset \neq B\subset X} N_B$, where each $N_B$ is a Poisson point process with intensity measure $\beta dt$.
\end{proof}
\begin{remark}
	For Ising spin quantum systems, Ioffe \cite{Iof} constructed a random representation using Poisson point processes in a Hilbert space language. We proved it using the groupoid structure introduced in Chapter \ref{ch:quantstatmech}.
\end{remark}	

\begin{remark}\label{generalcase}
    By the definition of interaction $\phi$, we know that every $\phi_X$ is a local operator. Therefore, we can write the Hamiltonian $H_\Lambda(\phi)$ separating the classical and quantum part as 
    \[
    H_\Lambda(\phi) = \sum_{\substack{A \subset \Lambda \\ k_A \in \{1,\dots,q-1\}^A}} J_A (u_A^{k_A}+u_A^{q-k_A}) + \sum_{\substack{B \subset \Lambda \\ \ell_B \in \{1,\dots,q-1\}^B}} (f_{B,\ell_B}\cdot v^{\ell_B}_B+f_{B,q-\ell_B}\cdot v^{q-\ell_B}_B),
    \]
    where we used the multi-index notation for $k_A$ and $\ell_B$. Also, the operators $$f_{B,\ell_B} = \sum_{\substack{A \subset \Lambda \\ k_A \in \{1,\dots,q-1\}^A}} c_{A,B,\ell_B} u_A^{k_A}.$$
   The selfadjointness of $H_\Lambda(\phi)$ implies that the coefficients $J_A$ must be real, as in the case $q=2$, but a more involved relation to the polynomials $f_{B,\ell_B}$ are needed. Also, the point processes $N_B$ will depend not only on the set $B$ but also on a collection of multi indices $\ell_B$.
\end{remark}

\begin{remark}
	The random representation for the Gibbs density operator has the defining feature of always coming with a preferred permutation for the arrival times: it is always in increasing order, realized by the composition with the time ordering function $T$ introduced earlier. Hence a good way of interpreting this random representation is through the notion of what is known as a \emph{path integral}. For the Poisson point process, the paths are usually functions with jumps, known as cadlag functions, and an integral on the space of cadlag functions can be constructed as in \cite{Bil}. 
\end{remark}
Let $\phi$ be a short-ranged interaction and $R$ its range. For each $\Lambda \subset \Z^d$, we define the set
\[
\Lambda_R\coloneqq \{x\in \Z^d: \exists y \in \Lambda \text{ 
	s.t.  }|x-y|\leq R\}.
\]	 
The point process $N_{\Lambda_R}$ can be decomposed as the sum of two independent Poisson point processes
\be\label{ppp_lambda}
N_{\Lambda_R}  = \sum_{\emptyset \neq B' \subset \Lambda} N_B + \sum_{B' \cap(\Lambda_R\setminus \Lambda) \neq \emptyset} N_B = N_\Lambda + N_{\Lambda_R,\Lambda},
\ee
where $B' = B \cup \Lambda_{f_B}$, where $\Lambda_{f_B}$ is the smallest set $\Lambda$ in the partial order given by the inclusion, where $f_B\in C(\Omega_{\Lambda})$. Remember from Equation \eqref{prod_classical_quantum} that the functions $f_B\cdot\sigma_B^{(1)}$ satisfy $$f_B\cdot\sigma_B^{(1)}(\sigma_{\Lambda_R},g_{\Lambda_R}) = f_B(g_{\Lambda_R}\sigma_{\Lambda_R})\delta_{\supp g_{\Lambda_R},B}.$$ Not every counting measure in $\mathbb{N}(\widetilde{\Omega}_{\Lambda_R}\times (\mathcal{P}(\Lambda_R)\setminus\{\emptyset\}))$ will contribute to the integral representation, only those that after the time ordering operation are coherent with respect to the sets appearing in the jumps for the interaction $\phi$. 
Hence, the idea of an integral over paths motivates us to introduce the following definitions to make this analogy more explicit. Let $\Lambda \subset \Z^d$ be a finite set and $(\sigma_\Lambda,g_\Lambda) \in \mathcal{G}_\Lambda$. Then,
\be\label{path_def}
\mathfrak{P}^{\sigma_{\Lambda},g_{\Lambda}} \coloneqq \Bigg\{\begin{array}{@{}c|c@{}}
	\mathfrak{p}=\sum_{k=1}^n\delta_{(\sigma_{\Lambda,k},t_k, B_k)}, n\geq 0 
	&
	\begin{matrix}
g_k\sigma_{\Lambda, k} = \sigma_{\Lambda,k+1} , 1\leq k\leq n \\[2ex]
  \sigma_{\Lambda,1} = \sigma_{\Lambda}, \sigma_{\Lambda,n+1}= g_{\Lambda}\sigma_{\Lambda}
  \end{matrix}
\end{array}
\Bigg\},
\ee
where $g_k \in G_{\Lambda}$ is the unique element such that $\supp g_k = B_k$. We are assuming that the times $t_k$ are already ordered, i.e., $t_k\leq t_{k+1}$. By convention, the case $n=0$ corresponds to the null measure. 
\begin{lemma}\label{lemma_path_measurable}
	The set $\mathfrak{P}^{\sigma_\Lambda,g_\Lambda}$ is a measurable subset of $\mathbb{N}(\tilde{\Omega}_\Lambda \times (\mathcal{P}(\Lambda) \setminus \{\emptyset\}))$ 
\end{lemma}
\begin{proof}
	Let 
	\be \label{2.15}	\mathfrak{P}^{\sigma_{\Lambda},g_{\Lambda},n} \coloneqq \Bigg\{\begin{array}{@{}c|c@{}}
		\mathfrak{p}=\sum_{k=1}^n\delta_{(\sigma_{\Lambda,k},t_k, B_k)} 
		&
		\begin{matrix}
g_k\sigma_{\Lambda, k} = \sigma_{\Lambda,k+1} , 1\leq k\leq n \\[2ex]
  \sigma_{\Lambda,1} = \sigma_{\Lambda}, \sigma_{\Lambda,n+1}= g_{\Lambda}\sigma_{\Lambda}		\end{matrix}
	\end{array}
	\Bigg\},
	\ee
	It is easy to see that $\mathfrak{P}^{\sigma_{\Lambda},g_{\Lambda}} = \cup_{n\geq 0}\mathfrak{P}^{\sigma_{\Lambda},g_{\Lambda},n}$. Let $S_n\subset [0,1]^n$ be the $n$-simplex. Using the notation of Lemma \ref{lema_measurable_countmeasure}, we can see that
	\[
	\mathfrak{P}^{\sigma_{\Lambda},g_{\Lambda},n} = \widetilde{S_n \times A},
	\]
	where $A \subset (\Omega_\Lambda \times (\mathcal{P}(\Lambda)\setminus \{\emptyset\}))^n$ is the set all $(\sigma_{\Lambda,1},B_1,\dots,\sigma_{\Lambda,n},B_n)$ satisfying the conditions in the right-hand-side of Equation \eqref{2.15}.
	
\end{proof}
Each counting measure in $\mathfrak{P}^{\sigma_{\Lambda_R},g_{\Lambda_R}}$ can be viewed as a \emph{path} by rearranging the jumps following the increasing order of time, 
\[
\mathfrak{p}(t) = \sigma_{\Lambda_R,k}, \quad \text{ if} \;\; t_{k-1}\leq t <t_k,
\]
for $k=1,\dots,n+1$ and $\sigma_{\Lambda_R,n+1}=g_{\Lambda_R}\sigma_{\Lambda_R}$, see Figure \ref{path_figure}. 
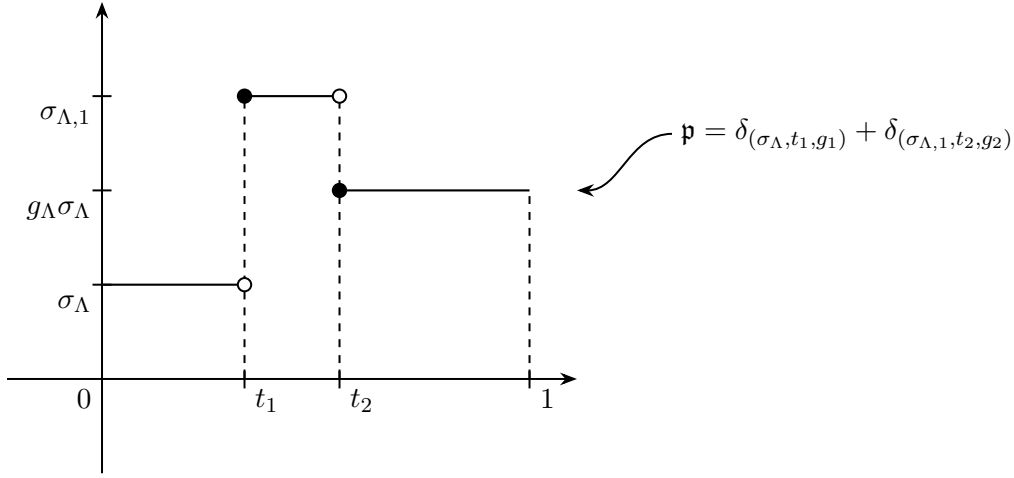
\begin{figure}[H]
	\centering
	
	\tikzset{every picture/.style={line width=0.75pt}} 
	
	\begin{tikzpicture}[scale=1.25]
		
		
		\draw[-{Stealth}, black] (-1,0)--(5,0);
		\draw[-{Stealth}, black] (0,-1)--(0,4);
		\draw[-{Stealth}, black] (6,2.6)to[out=180, in=0](5,2);
		
		\draw[-, dashed] (1.5,0)--(1.5,3);
		\draw[-, dashed] (2.5,0)--(2.5,3);
		\draw[-, dashed] (4.5,0)--(4.5,2);
		\draw[-, thick] (0,1)--(1.5,1);
		\draw[-, thick] (1.5,3)--(2.5,3);
		\draw[-, thick] (2.5,2)--(4.5,2);
		
		
		\draw (0,0) node[anchor=north east, scale=1]         {$0$};
		\draw (4.5,0) node[anchor=north west, scale=1]       {$1$};
		\draw (2.5,0) node[anchor=north west, scale=1] {$t_2$};
		\draw (1.5,0) node[anchor=north west, scale=1] {$t_1$};
		\draw (0,3) node[anchor=north east, scale=1] {$\sigma_{\Lambda,1}$};
		\draw (0,2) node[anchor=north east, scale=1] {$g_\Lambda\sigma_\Lambda$};
		\draw (0,1) node[anchor=north east, scale=1] {$\sigma_\Lambda$};
		\draw (0.1,2)--(-0.1,2);
		\draw (0.1,3)--(-0.1,3);
		\draw (0.1,1)--(-0.1,1);
		\draw (2.5,0.1)--(2.5,-0.1);
		\draw (1.5,0.1)--(1.5,-0.1);
		\draw (4.5,0.1)--(4.5,-0.1);
		\draw [fill=white] (1.5,1) circle (2pt);
		\draw [fill=black] (1.5,3) circle (2pt);
		\draw [fill=white] (2.5,3) circle (2pt);
		\draw [fill=black] (2.5,2) circle (2pt);
		\draw(5.9,2.3) node[anchor=south west, scale=1] {\;$\mathfrak{p} = \delta_{(\sigma_\Lambda,t_1,g_1)}+\delta_{(\sigma_{\Lambda,1},t_2,g_2)}$};

	\end{tikzpicture}
	\label{figpt}
	\caption{The path generated by a counting measure $\mathfrak{p}_\Lambda$ consisting of two points.}
 \label{path_figure}
\end{figure}

We will use the same notation $\mathfrak{p}$ for the path constructed above or the counting measure. In what follows we will denote the functions $E_{\Lambda_R}(\phi)$ and $V_{\Lambda_R}(\phi)$ by
\begin{equation}\label{integralforms}
\begin{split}	
 &E_{\Lambda_R}(\phi)((\sigma_{\Lambda_R,k},t_k)_n;(\sigma_{\Lambda_R},g_{\Lambda_R})) = \sum_{k=1}^{n+1} (t'_k-t'_{k-1})H_{\Lambda_R}^{(0)}(\phi)(\sigma_{{\Lambda_R},k}) =\int_{\mathfrak{p}} H_{\Lambda_R}^{(0)}(\phi) \quad \text{and}\quad \\
	&\prod_{k=1}^n f_{B_k}\cdot\sigma_{B_k}^{(1)}(\sigma_{\Lambda_R,k},g_{\Lambda_R,k})= \prod_{k=1}^n f_{B_k}(g_{\Lambda_R,k}\sigma_{\Lambda_R,k})\delta_{B_k,\supp g_{\Lambda_R,k}} = V_{\Lambda_R}(\phi)(\mathfrak{p}).
 \end{split}
\end{equation}

Since the set of counting measures where $V_{\Lambda_R}(\phi)(\mathfrak{p}) \neq 0$ is contained in the set $\mathfrak{P}_{\Lambda_R}^{\sigma_{\Lambda_R},g_{\Lambda_R}}$ we can write left-hand side of Equation \eqref{eq_randomrep} in the following way
\[
e^{-\beta H_{\Lambda_R}(\phi)}(\sigma_{\Lambda_R},g_{\Lambda_R})= e^{\beta}\int_{\mathfrak{P}^{\sigma_{\Lambda_R},g_{\Lambda_R}}} e^{-\beta \int_{\mathfrak{p}} H_{\Lambda_R}^{(0)}(\phi)} V_{\Lambda_R}(\phi)(\mathfrak{p})d\mathfrak{p}.
\]

\subsubsection{A Decomposition for the Gibbs densities using the Poisson Point Process}

The classical Hamiltonian $H_{\Lambda_R}^{(0)}(\phi)$ can be written as
\begin{align*}
	H_{\Lambda_R}^{(0)}(\phi)= \sum_{A \cap \Lambda \neq \emptyset} J_A \sigma_{A}^{(3)} +\sum_{A \subset \Lambda_R\setminus\Lambda}J_A \sigma_A^{(3)} = H_\Lambda^{(0)}(\phi)+W_\Lambda^{(0)}(\phi) + H_{\Lambda_R\setminus\Lambda}^{(0)}(\phi),
\end{align*}

hence

\begin{align*}
	\int_{\mathfrak{p}} H_{\Lambda_R}^{(0)}(\phi) = \int_{\mathfrak{p}}(H_\Lambda^{(0)}(\phi)+W_\Lambda^{(0)}(\phi))+\int_{\mathfrak{p}} H_{\Lambda_R\setminus\Lambda}^{(0)}(\phi).
\end{align*}

The path $\mathfrak{p}_{\Lambda_R}$ as a counting measure can be broken depending on the sets where the jumps occur, 
\begin{equation}\label{decomp_counting_measure}
\begin{split}
\mathfrak{p} = \sum_{k=1}^n \delta_{(\sigma_{\Lambda_R,k},t_k,B_k)} &= \sum_{k: B'_k \subset  \Lambda}\delta_{(\sigma_{\Lambda_R,k},t_k,B_k)}+\sum_{k: B'_k \cap (\Lambda_R\setminus \Lambda)\neq \emptyset} \delta_{(\sigma_{\Lambda_R,k},t_k,B_k)} \\
&\coloneqq \mathfrak{p}_\Lambda+\mathfrak{p}_{\Lambda_R,\Lambda}.
\end{split}
\end{equation}
Remember that $B' = B \cup \Lambda_{f_B}$. If there is no jump in $B_k$ intersecting the set $\Lambda_R\setminus \Lambda$ then $\mathfrak{p}_{\Lambda_R,\Lambda}=\mu_\emptyset$, the null measure. Let $0<t_{\ell_1}<\dots <t_{\ell_m}<1$ the times when the jumps of $\mathfrak{p}_{\Lambda_R,\Lambda}$ occur, also let $\ell_0=0$ and $\ell_{m+1}=1$. Then, 
\be\label{eq1_decomp_integral_classical_hamiltonian}
\int_{\mathfrak{p}} H_{\Lambda_R\setminus\Lambda}^{(0)}(\phi) = \sum_{k=0}^n (t_{k+1}-t_k) H_{\Lambda_R\setminus \Lambda}^{(0)}(\phi)(\sigma_{\Lambda_R,k}) = \sum_{j=0}^m\Bigg(\sum_{k=\ell_j}^{\ell_{j+1}-1}(t_{k+1}-t_k)H_{\Lambda_R\setminus\Lambda}(\phi)(\sigma_{\Lambda_R,k})\Bigg).
\ee
Since the Hamiltonian $H_{\Lambda_R\setminus \Lambda}^{(0)}(\phi)$ is a local function, it does not depend on the values of the configuration inside $\Lambda$, therefore is not sensitive to any jumps occurring inside $\Lambda$. Hence, 
\be\label{eq2_decomp_integral_classical_hamiltonian}
\sum_{k=\ell_j}^{\ell_{j+1}-1}(t_{k+1}-t_k)H_{\Lambda_R\setminus\Lambda}(\phi)(\sigma_{\Lambda_R,k}) = (t_{\ell_{j+1}}-t_{\ell_j})H_{\Lambda_R\setminus\Lambda}^{(0)}(\phi)(\sigma_{\Lambda_R,\ell_j}) = \int_{\mathfrak{p}_{\Lambda_R,\Lambda}}H_{\Lambda_R\setminus\Lambda}^{(0)}(\phi).
\ee
In contrast to the Hamiltonian $H_{\Lambda_R\setminus\Lambda}^{(0)}(\phi)$,
we cannot eliminate the dependence on the jumps occurring at $\Lambda_R\setminus \Lambda$ from the integral of Hamiltonian $H_\Lambda^{(0)}(\phi)+W_\Lambda^{(0)}(\phi)$ with respect to the path $\mathfrak{p}$, since the term $W_\Lambda^{(0)}(\phi)$ depends on the configurations outside $\Lambda$. We write
\[
\int_{\mathfrak{p}}(H_\Lambda^{(0)}(\phi)+W_\Lambda^{(0)}(\phi)) \coloneqq \int_{\mathfrak{p}_\Lambda}H_{\Lambda}^{\mathfrak{p}_{\Lambda_R,\Lambda}}(\phi).
\]
A similar decomposition is possible for the operator $V_{\Lambda_R}(\phi)$. First, we need to separate the product in \eqref{integralforms} depending not only on the jumps but also depending on the support of the functions $f_{B_k}$. Remember that $B_k' = B_k \cup \Lambda_{f_B}$, where $\Lambda_{f_B}$ is the smallest set $\Lambda$ in the partial order given by the inclusion, where $f_B\in C(\Omega_{\Lambda})$. 
\begin{align*}
	V_{\Lambda_R}(\phi)(\mathfrak{p}) &= \left(\prod_{k:B'_k\cap\Lambda\neq \emptyset}f_{B_k}(g_k\sigma_{\Lambda_R,k})\right)\left(\prod_{k:B'_k\subset \Lambda_R\setminus \Lambda}f_{B_k}(g_k\sigma_{\Lambda_R,k})\right) \coloneqq V_\Lambda^{\phi,\mathfrak{p}_{\Lambda_R,\Lambda}}(\mathfrak{p}_{\Lambda})V_{\Lambda_R\setminus\Lambda}^{\phi}(\mathfrak{p}_{\Lambda_R,\Lambda}).
\end{align*}
In light of the decomposition \eqref{decomp_counting_measure}, notice that the counting measures $\mathfrak{p}_\Lambda$ only have jumps inside $\Lambda$, therefore the initial and final configuration in $\Lambda_R\setminus \Lambda$, although they can differ from one another, must be the same as long as $\mathfrak{p}_{\Lambda_R,\Lambda}$ is fixed. 
We want to consider the counting measure $\mathfrak{p}_{\Lambda_R,\Lambda}$ as a boundary condition. However, it contains the information of the whole configuration in $\Lambda_R$, despite only having jumps in sets $B$ where $B'\cap(\Lambda_R\setminus\Lambda) \neq \emptyset$. Thus, in what follows, we will tacitly assume an identification between counting measures. 
\begin{definition}
    We say that two counting measures $\mathfrak{p}= \sum_{k=1}^n \delta_{(\sigma_k,t_k,B_k)},\mathfrak{q}=\sum_{j=1}^m \delta_{(\omega_j,s_j,C_j)}$ are $\Gamma$-boundary equivalent if it holds that
    \begin{itemize}
        \item[(i)] $|k: B'_k \cap \Gamma \neq \emptyset|=|j: C'_j \cap \Gamma \neq \emptyset|=n'$.
        \item[(ii)] $(t_{k_\ell},B_{k_\ell})=(s_{j_\ell},C_{j_\ell})$, where $\ell=1,\dots, n'$ is an enumeration of the set above preserving the time ordering,
    \end{itemize}
    we denote the equivalence between two counting measures $\mathfrak{p},\mathfrak{q}$ by $\mathfrak{p}\sim_\Gamma \mathfrak{q}$. Then a \textbf{path boundary condition} $\mathfrak{p}$ for $\Lambda$ is an equivalence class of $\Lambda^c$-boundary equivalent counting measures. 
\end{definition}
\begin{remark}
    The equivalence relation above says that the only part that is important in the definition of the path boundary conditions are the jumps (the times and the places occurring spin-flips) on the boundary of the set. 
\end{remark}
This motivates us to introduce the following set
\be\label{path_bc_def}
\mathfrak{P}^{\sigma_\Lambda,g_\Lambda,\mathfrak{p}_{\Lambda^c}} \coloneqq 	
\{\mathfrak{q}\in \mathfrak{P}^{\sigma_{\Lambda_R},g_{\Lambda_R}}:\mathfrak{q}\sim_{\Lambda^c}\mathfrak{p}\}.
\ee
\pagebreak

\begin{lemma}
    For any path boundary condition $\mathfrak{p}$, the set $\mathfrak{P}^{\sigma_\Lambda,g_\Lambda,\mathfrak{p}_{\Lambda^c}}$ is measurable.
\end{lemma}
\begin{proof}
    The proof follows similar lines as the one in Lemma \ref{lemma_path_measurable}. Let $s_1,...,s_m$ and $B_1,..., B_m$ be the times and jumps associated with the path boundary condition $\mathfrak{p}$. The set $\mathfrak{P}^{\sigma_\Lambda,g_\Lambda,\mathfrak{p}_{\Lambda^c}}$ can be described as the subset of counting measures $\mathfrak{q} = \sum_{k=1}^{n+m}\delta_{(\sigma_{\Lambda_R,k},t_k,C_k})$, with variable $n\geq 0$, satisfying the following conditions
    \begin{enumerate}
        \item[$(i)$] $g_k\sigma_{\Lambda_R, k} = \sigma_{\Lambda_R,k+1} , 1\leq k\leq n+m,
  \sigma_{\Lambda_R,1} = \sigma_{\Lambda_R}, \sigma_{\Lambda_R,n+m+1}= g_{\Lambda_R}\sigma_{\Lambda_R}$;
  \item[$(ii)$]  $\exists t_{k_1}<t_{k_2}<\dots<t_{k_m}$ such that $t_{k_j} = s_j$ for $j=1,\dots,m$.
    \end{enumerate}
        Let us call $\mathfrak{P}^{\sigma_{\Lambda},g_{\Lambda},\mathfrak{p}_{\Lambda^c},n}$ the set of counting measures $\mathfrak{q}$ that satisfies condition $(i),(ii)$ above and have exactly $n+m$ Dirac measures in the sum. It is easy to see that $\mathfrak{P}^{\sigma_{\Lambda},g_{\Lambda},\mathfrak{p}_{\Lambda^c}} = \cup_{n\geq 0}\mathfrak{P}^{\sigma_{\Lambda},g_{\Lambda},\mathfrak{p}_{\Lambda^c},n}$. Call $\mathbf{s} = (s_1,\dots, s_m)$ Let $S_{n,\mathbf{s}}\subset S_{n+m}$ be the subset of the $n+m$-simplex given by
        \[
        S_{n,\mathbf{s}} = \bigcup_{\substack{1\leq k_j\leq n+m \\ k_j\leq k_{j+1}}}\{(t_1,\dots,t_{n+m}\in S_{n+m}: t_{k_j}=s_j\}.
        \]
        The sets above are clearly measurable subsets of $[0,1]^{n+m}$. Using the notation of Lemma \ref{lema_measurable_countmeasure}, we can see that
	\[
	\mathfrak{P}^{\sigma_{\Lambda},g_{\Lambda},\mathfrak{p}_{\Lambda^c},n} = \widetilde{S_{n,\mathbf{s}} \times A},
	\]
	where $A \subset (\Omega_{\Lambda_R} \times (\mathcal{P}(\Lambda_R)\setminus \{\emptyset\}))^{n+m}$ is the set all $(\sigma_{\Lambda_R,1},B_1,\dots,\sigma_{\Lambda_R,n+m},B_{n+m})$ satisfying the conditions in item $(i)$ above.
\end{proof}
\begin{remark}\label{groupelement}
    Once the path boundary condition is fixed, since the counting measures $\mathfrak{q} \in \mathfrak{P}^{\sigma_\Lambda,g_\Lambda,\mathfrak{p}_{\Lambda^c}}$ can only differ on jumps inside $\Lambda$, once we know that $\sigma_{\Lambda_R\setminus\Lambda}$ the endpoint  $g_{\Lambda_R\setminus\Lambda}\sigma_{\Lambda_R\setminus\Lambda}$ is fixed for all paths. We will denote, then, by $g_{\mathfrak{p}}$ the unique group element that takes $\mathfrak{p}(0)=\sigma_{\Lambda_R}$ to $\mathfrak{p}(1)=g_{\mathfrak{p}}\sigma_{\Lambda_R}$.
\end{remark}
\begin{remark}
    In order to simplify the notation, we will write $\mathfrak{p}_{\Lambda^c}$ instead of $\mathfrak{p}_{\Lambda_R,\Lambda}$.
\end{remark}
We use the same symbol to denote the path boundary condition as the usual counting measures in order to simplify the notation. Notice that it is not true that for every $\mathfrak{P}^{\sigma_{\Lambda_R},g_{\Lambda_R}}$ there is a representative of $\mathfrak{p}$. Indeed, since the flips can have an effect at the boundary, there may be a disagreement between the paths only acting on the spins in $\Lambda$ and the path boundary condition. 
\begin{figure}[H]
\centering

\scalebox{1}{

\tikzset{every picture/.style={line width=0.75pt}} 

\begin{tikzpicture}[x=0.75pt,y=0.75pt,yscale=-1,xscale=1]

\draw  [line width=3.75]  (24,23) -- (287,23) -- (287,286) -- (24,286) -- cycle ;
\draw  [line width=3.75]  (94.5,93.5) -- (218.5,93.5) -- (218.5,217.5) -- (94.5,217.5) -- cycle ;
\draw [color={rgb, 255:red, 13; green, 25; blue, 231 }  ,draw opacity=1 ][line width=3]    (218.5,161) .. controls (258.5,131) and (250.5,174) .. (285.5,154) ;

\draw (183,176.4) node [anchor=north west][inner sep=0.75pt]  [font=\LARGE]  {$\Lambda $};
\draw (239,160.4) node [anchor=north west][inner sep=0.75pt]  [font=\Large,color={rgb, 255:red, 13; green, 25; blue, 231 }  ,opacity=1 ]  {$\mathfrak{p}_{\Lambda^c}$};

\end{tikzpicture}

}

\caption{The path boundary condition.}
\end{figure}
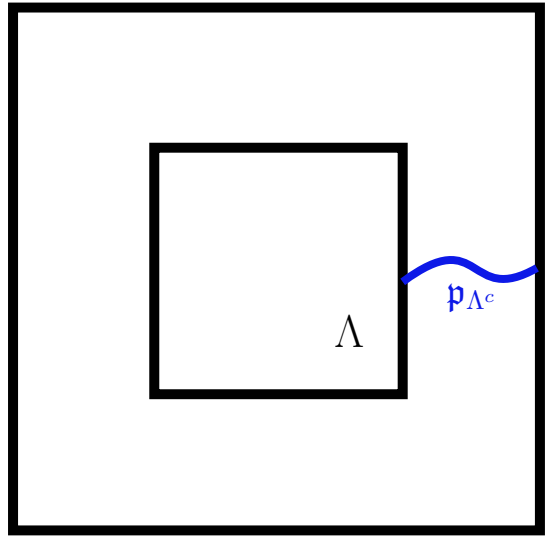	
 
For a bounded and measurable function $f:\mathbb{N}(\widetilde{\Omega}_{\Lambda_R}\times \mathcal{P}(\Lambda_R)\setminus\{\emptyset\})\rightarrow \mathbb{C}$ , we define
\[
\int_{\mathfrak{P}^{\sigma_\Lambda,g_\Lambda,\mathfrak{p}_{\Lambda^c}}} f(\mathfrak{p}_\Lambda)d\mathfrak{p}_\Lambda = \int_{\mathbb{N}(\widetilde{\Omega}_{\Lambda_R}\times \mathcal{P}(\Lambda_R)\setminus\{\emptyset\})}\mathbbm{1}_{\mathfrak{P}^{\sigma_\Lambda,g_\Lambda,\mathfrak{p}_{\Lambda^c}}}(\nu)f(\mathfrak{\nu}) dN_\Lambda(\nu),
\]
where $N_\Lambda$ is the Poisson point process defined in \eqref{ppp_lambda}. The next is the definition of the density that depends on the path boundary condition.

\begin{definition}[\textbf{Path boundary density}]
	Given $\phi$ a short-range interaction, a path boundary condition $\mathfrak{p}$, and a configuration $\omega_{\Lambda^c}$. The density operator $D_{\beta,\Lambda}^{\mathfrak{p},\omega}: \mathcal{G}_\Lambda \rightarrow \mathbb{C}$ be
	\be\label{def_density_path}
	D_{\beta,\Lambda}^{\mathfrak{p},\omega}(\sigma_\Lambda,g_\Lambda) \coloneqq \int_{\mathfrak{P}^{\sigma_\Lambda,g_\Lambda,\mathfrak{p}_{\Lambda^c}}}e^{-\beta \int_{\mathfrak{p}_\Lambda}H_\Lambda^{\mathfrak{p}_{\Lambda^c}}(\phi)}V_\Lambda^{\mathfrak{p}_{\Lambda^c}}(\phi)(\mathfrak{p}_\Lambda)d\mathfrak{p}_\Lambda,
	\ee
 where the paths being integrated are
 \[
 \mathfrak{P}^{\sigma_\Lambda,g_\Lambda, \mathfrak{p}_{\Lambda^c}} = \{\mathfrak{q}\in\mathfrak{P}^{\sigma_\Lambda\omega_{\Lambda_R\setminus\Lambda},g_{\Lambda_R}}:\mathfrak{q}\sim_{\Lambda^c}\mathfrak{p}\}.
 \]
\end{definition}
Note that the density operator $D_{\beta,\Lambda}^{\mathfrak{p},\omega}$ is a continuous function on $\mathcal{G}_\Lambda$ for every $\mathfrak{p}$ since the groupoid has the discrete topology. Before we further present the properties of the densities $D_{\beta,\Lambda}^{\mathfrak{p}}$, let us introduce a few definitions and that $\mathfrak{q}(0)=\sigma_\Lambda\omega_{\Lambda_R\setminus\Lambda}$, for every $\mathfrak{q}\in \mathfrak{P}^{\sigma_\Lambda,g_\Lambda,\mathfrak{p}_{\Lambda^c}}$. The \emph{concatenation of paths} can only be defined on paths $\mathfrak{p}$, $\mathfrak{q}$ where $\mathfrak{p}(1)=\mathfrak{q}(0)$, by the composition rules depending on a real number $0<c<1$ 
\be\label{concatenationofpaths}
\mathfrak{p} \circ_c \mathfrak{q}(t) =  \begin{cases}
	\mathfrak{p}(t/c) & 0\leq t\leq c \\
	\mathfrak{q}(\frac{c-t}{c-1}) & c \leq t \leq 1.
\end{cases}
\ee

\begin{lemma}\label{lemma_property_path}
	Let $\omega$ be a configuration and $\mathfrak{p}$ and $\mathfrak{q}$ be two path boundary condition such that $\mathfrak{p}_{\Lambda^c}(1)=\mathfrak{q}_{\Lambda^c}(0)$ and $\beta,\beta'$ it holds that 
	\be\label{lemma_path_density}
	D_{\beta,\Lambda}^{\mathfrak{p},\omega}\cdot D_{\beta',\Lambda}^{\mathfrak{q},\omega'} = D_{\beta+\beta',\Lambda}^{\mathfrak{p}\circ_c \mathfrak{q},\omega},
	\ee
 where $c = \beta/(\beta+\beta')$ and $\omega'_{\Lambda^c}=\mathfrak{p}_{\Lambda^c}(1)$.
\end{lemma}
\begin{proof}
	Given an element $(\sigma_\Lambda, g_\Lambda)$ of the groupoid $\mathcal{G}_\Lambda$, the product is given by
	\be\label{property_path_eq1}	
	D_{\beta,\Lambda}^{\mathfrak{p},\omega}\cdot D_{\beta',\Lambda}^{\mathfrak{q},\omega'}(\sigma_\Lambda, g_\Lambda) = \sum_{(\omega_\Lambda,k_\Lambda)\in \mathcal{G}_\Lambda^{g_\Lambda\sigma_\Lambda}}D_{\beta,\Lambda}^{\mathfrak{p},\omega}(\omega_\Lambda, k_\Lambda)D_{\beta',\Lambda}^{\mathfrak{q},\omega'}(\sigma_\Lambda,k_\Lambda^{-1}g_\Lambda).	
	\ee
	For each $(\omega_\Lambda,k_\Lambda)$, using the definition for the densities \eqref{def_density_path}, we have
	\begin{equation}\label{property_path_eq2}
		\begin{split}
			&D_{\beta,\Lambda}^{\mathfrak{p},\omega}(\omega_\Lambda,k_\Lambda)D_{\beta',\Lambda}^{\mathfrak{q},\omega'}(\sigma_\Lambda, k_\Lambda^{-1}g_\Lambda) \\[2ex]
        &=\int_{\mathfrak{P}^{\omega_\Lambda,k_\Lambda,\mathfrak{p}_{\Lambda^c}}}\hspace{-0.8cm}e^{-\beta \int_{\mathfrak{p}_\Lambda}H_\Lambda^{\mathfrak{p}_{\Lambda^c}}(\phi)}V_\Lambda^{\mathfrak{p}_{\Lambda^c}}(\phi)(\mathfrak{p}_\Lambda)d\mathfrak{p}_\Lambda\int_{\mathfrak{P}^{\sigma_\Lambda,k_\Lambda^{-1}g_\Lambda,\mathfrak{q}_{\Lambda^c}}}\hspace{-0.8cm}e^{-\beta' \int_{\mathfrak{p}_\Lambda}H_\Lambda^{\mathfrak{q}_{\Lambda^c}}(\phi)}V_\Lambda^{\phi,\mathfrak{q}_{\Lambda^c}}(\mathfrak{p}_\Lambda)d\mathfrak{p}_\Lambda \\[2ex]
        &=\int_{\mathfrak{P}^{\omega_\Lambda,k_\Lambda,\mathfrak{p}_{\Lambda^c}}}\int_{\mathfrak{P}^{\sigma_\Lambda,k_\Lambda^{-1}g_\Lambda,\mathfrak{q}_{\Lambda^c}}}\hspace{-0.8cm}e^{-\beta \int_{\mathfrak{p}_\Lambda}H_\Lambda^{\mathfrak{p}_{\Lambda^c}}(\phi)-\beta' \int_{\mathfrak{q}_\Lambda}H_\Lambda^{\mathfrak{q}_{\Lambda^c}}(\phi)}V_\Lambda^{\mathfrak{p}_{\Lambda^c}}(\phi)(\mathfrak{p}_\Lambda)V_\Lambda^{\phi,\mathfrak{q}_{\Lambda^c}}(\mathfrak{q}_\Lambda)d\mathfrak{q}_\Lambda d\mathfrak{p}_\Lambda.
		\end{split}	
	\end{equation}
	We know that the paths $\mathfrak{p}_\Lambda \in \mathfrak{P}_{\Lambda}^{\omega_\Lambda,k_\Lambda,\mathfrak{p}_{\Lambda^c}}$ and $\mathfrak{q}_\Lambda \in \mathfrak{P}_\Lambda^{\sigma_\Lambda,k_\Lambda^{-1}g_\Lambda,\mathfrak{q}_{\Lambda^c}}$ satisfy $\mathfrak{p}_\Lambda(1) = \mathfrak{q}_\Lambda(0)$. Together with a trivial change of variables, this observation allows us to show that
	\begin{equation}\label{property_path_eq3}
		\begin{split}
			\beta \int_{\mathfrak{p}_\Lambda}H_\Lambda^{\mathfrak{p}_{\Lambda^c}}(\phi)+\beta' \int_{\mathfrak{q}_\Lambda}&H_\Lambda^{\mathfrak{q}_{\Lambda^c}}(\phi) = \beta\int_{\mathfrak{p}}(H_\Lambda^{(0)}(\phi)+W_\Lambda^{(0)}(\phi))+\beta'\int_{\mathfrak{q}}(H_\Lambda^{(0)}(\phi)+W_\Lambda^{(0)}(\phi))\\[2ex] 
			&=\beta \int_0^1 (H_\Lambda^{(0)}(\phi)+W_\Lambda^{(0)}(\phi))(\mathfrak{p}(t))dt+\beta'\int_0^1(H_\Lambda^{(0)}(\phi)+W_\Lambda^{(0)}(\phi))(\mathfrak{q}(t))dt \\[2ex]
			&=(\beta+\beta')\int_0^1 (H_\Lambda^{(0)}(\phi)+W_\Lambda^{(0)}(\phi))(\mathfrak{p}\circ\mathfrak{q}(t))dt \\[2ex]
			&=(\beta+\beta')\int_{\mathfrak{p}_\Lambda\circ_c \mathfrak{q}_\Lambda} H_\Lambda^{\mathfrak{p}_{\Lambda^c}\circ_c \mathfrak{q}_{\Lambda^c}}(\phi),
		\end{split}
	\end{equation}
	where $c=\beta/(\beta+\beta')$. A similar procedure can be carried for the operators $V^{\phi,\mathfrak{p}_{\Lambda^c}}_\Lambda$. Let us write the paths $\mathfrak{p}$ and $\mathfrak{q}$ in counting measure form as $\mathfrak{p}=\sum_{k=1}^n \delta_{(\omega_k, t_k, B_k)}$ and $\mathfrak{q} = \sum_{j=1}^m\delta_{(\omega_j,t_j,B_j)}$, then
	\begin{equation}\label{property_path_eq4}
		\begin{split}
			V^{\mathfrak{p}_{\Lambda^c}}_\Lambda(\phi)(\mathfrak{p}_\Lambda)V^{\mathfrak{q}_{\Lambda^c}}_\Lambda(\phi)(\mathfrak{q}_\Lambda) &= \prod_{k:B'_k\cap\Lambda\neq \emptyset}f_{B_k}(g_k\sigma_k)\prod_{j:B'_j\cap\Lambda\neq \emptyset}f_{B_j}(g_j\sigma_j) \\[2ex]
			&=V^{\mathfrak{p}_{\Lambda^c}\circ\mathfrak{q}_{\Lambda^c}}_\Lambda(\phi)(\mathfrak{p}_\Lambda\circ_c\mathfrak{q}_\Lambda).
		\end{split}
	\end{equation}
	Using Lemma \ref{lemma_ppp_decomp} and summing over all $(\omega_\Lambda,k_\Lambda) \in \mathcal{G}_{\Lambda}^{g_\Lambda \sigma_\Lambda}$ we get
	\[
	\sum_{(\omega_\Lambda,k_\Lambda)\in \mathcal{G}_\Lambda^{g_\Lambda\sigma_\Lambda}} \int_{\mathfrak{P}^{\omega_\Lambda,k_\Lambda,\mathfrak{p}_{\Lambda^c}}}\int_{\mathfrak{P}^{\sigma_\Lambda,k_\Lambda^{-1}g_\Lambda,\mathfrak{q}_{\Lambda^c}}}d\mathfrak{q}_\Lambda d\mathfrak{p}_\Lambda = \int_{\mathfrak{P}^{\sigma_\Lambda,g_\Lambda, \mathfrak{p}_{\Lambda^c}\circ \mathfrak{q}_{\Lambda^c}}} d\mathfrak{p}_\Lambda.
	\]
	Plugging Equations \eqref{property_path_eq3}, \eqref{property_path_eq4} into \eqref{property_path_eq2} and using the identity above we get the result.
\end{proof}

Given a path $\mathfrak{p}_\Lambda$, we can also define an involutive operation that reverses the path 
\be\label{involutionpath}
\mathfrak{p}^{-1}(t)= \mathfrak{p}(1-t).
\ee
\begin{lemma}\label{lemma_adjoint_density}
	Let $\omega$ be a configuration and $\mathfrak{p}$ be a path boundary condition. Let also $\omega'$ be a configuration such that $\mathfrak{p}_{\Lambda^c}(1)=\omega'$. Then,
	\be\label{adjoint_density}
	(D_{\beta,\Lambda}^{\mathfrak{p},\omega})^* = D_{\beta,\Lambda}^{\mathfrak{p}^{-1},\omega'}.
	\ee
\end{lemma}
\begin{proof}
	We just need to check what happens to the functions $\int_{\mathfrak{p}_\Lambda}H_{\Lambda}^{\mathfrak{p}_{\Lambda^c}}$ and $V_{\Lambda}^{\phi,\mathfrak{p}_{\Lambda^c}}(\mathfrak{p}_\Lambda)$. For the first one, we have,
	\begin{align*}
		&\int_{\mathfrak{p}} (H_\Lambda^{(0)}(\phi)+W_\Lambda^{(0)}(\phi)) = \sum_{k=0}^n (t'_{k+1}-t'_k)(H_\Lambda^{(0)}(\phi)+W_\Lambda^{(0)}(\phi))(\sigma_k) \\
		&=\sum_{k=0}^n ((1-t'_k)-(1-t'_{k+1}))(H_\Lambda^{(0)}(\phi)+W_\Lambda^{(0)}(\phi))(\sigma_k) = \int_{\mathfrak{p}^{-1}}(H_\Lambda^{(0)}(\phi)+W_\Lambda^{(0)}(\phi)),
	\end{align*}     
	and for the second,
	\begin{align*}
		V_\Lambda(\phi)(\mathfrak{p})=\prod_{k=1}^n f_{B_k}\cdot\sigma_{B_k}^{(1)}(\sigma_{\Lambda,k},g_k) = \prod_{k=1}^n \overline{f_{B_k}\cdot\sigma_{B_k}^{(1)}(g_k\sigma_{\Lambda,k},g_k^{-1})} = \overline{V_\Lambda(\phi)(\mathfrak{p}^{-1})}.
	\end{align*}
	Notice that $H_\Lambda^{(0)}(\phi)+W_\Lambda^{(0)}(\phi)$ is a real-valued function, therefore we get the desired result. 
\end{proof}

\subsubsection{On previous definitions of boundary conditions}

A special class of path boundary conditions $\mathfrak{p}_{\Lambda^c}$ is the one related to constant paths. We will refer to these boundary conditions as \emph{classical boundary conditions}. These are important path boundary conditions since they can be obtained directly by a Poisson point process representation of the Gibbs density of a specific Hamiltonian, which we will proceed to describe now. Let $\omega \in \Omega_{\Lambda^c}$ be a configuration and define the evaluation functional $\text{ev}_\omega$ on the dense subalgebra $C_c(\mathcal{G}_{\Lambda^c})$. By the definition of the norm in the regular representation, the evaluation functionals are actually states. Thus, we can form the conditional expectation $\text{Id} \otimes \text{ev}_{\omega_{\Lambda^c}}:C_c(\mathcal{G})\rightarrow C_c(\mathcal{G}_\Lambda)$ by 
\[
\text{Id}\otimes \text{ev}_{\omega_{\Lambda^c}}(f)(\sigma_\Lambda,g_\Lambda) = f(\sigma_\Lambda\omega_{\Lambda^c},g_\Lambda),
\]
and define the \emph{Hamiltonian with boundary condition $\omega$} by the expression 
\be\label{classicalbc}
H^\omega_\Lambda(\phi) \coloneqq \text{Id}\otimes \text{ev}_{\omega_{\Lambda^c}}(H_\Lambda(\phi) + W_\Lambda(\phi)).
\ee
The classical boundary conditions in \eqref{classicalbc} are the base for extensions of Pirogov-Sinai theory for quantum spin and fermionic systems that appeared previously in Datta, Fernand\'ez, and Fr\"ohlich \cite{Datta1} and also in Borgs, Koteck\'y, and Ueltschi \cite{BKU} for quantum spin systems, and Borgs and Koteck\'y \cite{BK} for fermionic systems. A proposal for boundary conditions for quantum spin systems appeared before in Israel's book \cite{Is} in much greater generality. See also Simon's book \cite{Simon}. What Simon proposes is to use a general state $\nu$ in $C^*(\mathcal{G}_{\Lambda^c})$ and use the product structure of $C^*(\mathcal{G}) \simeq C^*(\mathcal{G}_\Lambda) \otimes C^*(\mathcal{G}_{\Lambda^c})$ to define $\text{Id}\otimes \nu (f \otimes g) = \nu(g) f$ and extend it by linearity. In Israel, the proposal is to consider $\nu$ to be a pure state. The Hamiltonian is
\begin{equation}\label{israelbc}
H_\Lambda^\nu(\phi) \coloneqq  H_\Lambda(\phi) + \text{Id}\otimes \nu (W_\Lambda(\phi)).
\end{equation}
The map $\text{Id}\otimes \nu$ preserves adjoints since $\nu$ is a state, hence $H_\Lambda^\nu(\phi)$ is self-adjoint. We can define Gibbs states depending on $\nu$ similarly to Equation \eqref{freegibbsstate}. Notice that for every local operator $f$ we have $\lim_{\Lambda \nearrow \Z^d} i[H_\Lambda^\nu(\phi),f] = \D^\phi(f)$ and Theorem \eqref{dynamicsexistence} implies then that the dynamics $\tau^\nu_{\Lambda,t}$ converges strongly to $\tau_t$. By Theorem \ref{limitkms}, every accumulation point of the sequences $\{\mu_{\beta,\phi,\Lambda}^\nu\}_{\Lambda \in \mathcal{P}_f(\Z^d)}$ is a $(\tau,\beta)$-KMS state. 

It is instructive to think about the case of classical interactions. Since $W_\Lambda(\phi)$ will be a continuous function in $C(\Omega)$, by the Riesz-Markov theorem the boundary condition by Israel is a probability measure $\nu$ in $\Omega_{\Lambda^c}$, therefore Simon and Israel proposals for boundary conditions have the form
\[
H_\Lambda^\nu(\phi) = H_\Lambda(\phi) + \int_{\Omega_{\Lambda^c}}W_\Lambda(\phi)d\nu.
\]
This procedure generates any probability measure in $\Omega_{\Lambda^c}$, even when one restricts $\nu$ to be a pure state. 

\section{Path Gibbs functionals}

We readily see that the densities $D_{\beta,\Lambda}^{\mathfrak{p},\omega}$ are not always self-adjoint (they are self-adjoint, for example, when the path is symmetric, i.e., $\mathfrak{p} = \mathfrak{q} \circ_{1/2} \mathfrak{q}^{-1}$, for some path $\mathfrak{q}$). The densities $D_{\beta,\Lambda}^{\mathfrak{p},\omega}$ can be zero sometimes. Indeed, we must recall that a classical interaction $\phi$ is zero whenever $g\neq 1$, therefore there is no quantum part $V_\Lambda(\phi)$, yielding us a trivial random representation. Another way of having the density $D_{\beta,\Lambda}^{\mathfrak{p},\omega}$ identically zero is when the jumps in the given boundary path $\mathfrak{p}_{\Lambda^c}$ cannot be realized by the interaction, i.e., there are no sets $B_1,\dots, B_n$ such that $f_{B_i} \neq 0$ in the Hamiltonian \eqref{hamiltonian} and 
\[
B_1\Delta B_2\Delta \dots \Delta B_n = \{x\in \Lambda_R\setminus \Lambda: \mathfrak{p}_{\Lambda_R}(0)_x\neq \mathfrak{p}_{\Lambda_R}(1)_x\}=\supp g_{\mathfrak{p},\Lambda^c},
\] 
where $g_{\mathfrak{p},\Lambda^c}$ is the restriction to $\Lambda^c$ of the group element defined in Remark \ref{groupelement}. But different phenomena can happen since the only restriction on the quantum part $V_\Lambda(\phi)$ is to be self-adjoint, meaning that it can take potentially complex values when the arrow $(\sigma,g)$ have $g\neq 1$. Hence, even if the densities are equal to $0$, this does not guarantee that the partition functions
\be\label{path_partition_function}
Z_{\beta,\phi,\Lambda}^{\mathfrak{p},\omega}\coloneqq\sum_{\sigma_\Lambda \in \Omega_\Lambda}  D_{\beta,\Lambda}^{\mathfrak{p},\omega}(\sigma_\Lambda,g_{\mathfrak{p},\partial_R\Lambda}),  
\ee
are different from zero, where $\partial_R\Lambda = \{x\in \Lambda:\exists y\in \Lambda^c \text{ s.t. } |x-y|\leq R\}$. Then Gibbs functionals defined using the densities $D_{\beta,\Lambda}^{\mathfrak{p},\omega}$ may not make sense, even if we allow the functionals to not be states. We introduce the following definition inspired by this problem:
\begin{definition}
	An interaction $\phi$ is said to be \textbf{admissible} when for all $\Lambda \in \mathcal{P}_f(\Z^d)$ we have	
 \be\label{decoherence}
	D_{\beta,\Lambda}^{\mathfrak{p},\omega}(\sigma_\Lambda,g_\Lambda) \neq 0 \text{ for some } (\sigma_\Lambda,g_\Lambda) \in \mathcal{G}_\Lambda \Longrightarrow Z_{\beta,\phi,\Lambda}^{\mathfrak{p},\omega}\neq 0.
	\ee
 \end{definition}
 Another important class of interactions is the subject of the next definition. 
\begin{definition}
	An interaction is said to be \textbf{stoquastic} if and only if for any $\Lambda \in \mathcal{P}_f(\mathbb{Z}^d)$ it holds that
	\[
	g \neq 1 \Rightarrow \phi_\Lambda(\sigma_\Lambda,g_\Lambda) \leq 0
	\]
\end{definition}

Stoquastic interactions, and consequently stoquastic Hamiltonians \cite{Klassen}, are ubiquitous in quantum statistical mechanics due to the absence of what is known as the \emph{sign problem} \cite{Lohr}. Although the definition is basis dependent, we use the same terminology since with our definition a stoquastic interaction will be stoquastic in the usual definition for the basis in $\ell^2(\mathcal{G}_\Lambda)$ given by the delta functions on the arrows (See the construction of the regular representation in Chapter 4). For $q=2$, a fairly large class of stoquastic Hamiltonians is given by \eqref{hamiltonian} when the polynomials $f_B$ are nonnegative functions in $C(\Omega_\Lambda)$. When the interaction is stoquastic, we have the property 
\[
D_{\beta,\Lambda}^{\mathfrak{p},\omega}(\sigma_\Lambda,g_\Lambda) \geq 0,
\]
for every $(\sigma_\Lambda,g_\Lambda)\in \mathcal{G}_\Lambda$. 
\begin{proposition}
    Every stoquastic interaction is admissible.
\end{proposition}
\begin{proof}
    Suppose that $Z_{\beta, \phi,\Lambda}^{\mathfrak{p},\omega}=0$. Since the interaction is stoquastic, every term in the series calculating $D_{\beta,\Lambda}^{\mathfrak{p},\omega}$ is nonnegative and that exponentials of real numbers are positive numbers, this implies that
    \[
    V_{\Lambda}^{\mathfrak{p}_{\Lambda^c}}(\phi)(\mathfrak{p}_\Lambda)=0,
    \]
    for every path $\mathfrak{p}_\Lambda$ in the set $\mathfrak{P}^{\sigma_\Lambda,g_{\mathfrak{p}_{\partial_R\Lambda}},\mathfrak{p}}$. If one considers the case where $\mathfrak{p}_\Lambda = \emptyset$ and let $B_k$ be the jumps associated with $\mathfrak{p}_{\Lambda^c}$. Our conclusion is that 
    $$\prod_{k=1}^n f_{B_k}(g_k\sigma_{\Lambda,k})=0,$$ for any starting configuration. Therefore, this polynomial is identically zero for every configuration. Since this polynomial always appears as a factor during the calculation of $D_{\beta,\Lambda}^{\mathfrak{p},\omega}$, we conclude that the latter must be identically 0. 
\end{proof}
Important examples of two-body interactions described in Example \ref{ex.heisenberg} are stoquastic, 
\begin{example}
	Let $J^{(1)}_{x,y}, J^{(2)}_{x,y}, J^{(3)}_{x,y}, h_x, \varepsilon_x, \varrho_x$ be real numbers defining a interaction $\phi_\Lambda$ in \ref{ex.heisenberg}. Then, the interaction $\phi_\Lambda$ is stoquastic when $J_{x,y}^{(1)}\geq 0$, $|J_{x,y}^{(2)}|\leq J_{x,y}^{(1)}$, $\varepsilon_x \geq 0$, and $|\varrho_x|\leq \varepsilon_x$. Any choices of $J_{x,y}^{(3)}$ and $h_x$ are possible. 
\end{example}

In particular, the ferromagnetic nearest neighbor Heisenberg model and the XY model are stoquastic, and the classical part can have any sign in the coupling constants. There are many examples of non-stoquastic systems that can in principle be treated by our methods. For instance, the constraint that the fields $\varepsilon_x$ are nonnegative seems artificial. In these cases, we can transform the interaction, through an *-automorphism, that preserves the structure of equilibrium states into a stoquastic interaction, thus eliminating this apparent sign problem \cite{Lohr}. The sign problem appears in the physics literature when one tries to use the Monte Carlo method to estimate means of relevant quantities for your system. To be able to do this, one usually relies on a path-integral representation. However, since the exponential of a self-adjoint operator may have complex off-diagonal terms, the behavior of the sign of the weights coming from the path-integral representation may pose some complications to this approach. 

The sign problem is not exclusively a theoretical/numerical problem in statistical mechanics, the sign problem can also appear when trying to study the models through rigorous methods. Ignoring the sign may lead to diverging expansions \cite{Datta2} or phase transition results that get increasingly worse as the temperature gets lower \cite{LM}. However, this sign does not prevent the system from having nice spatial decomposition properties, as shown in \cite{BK, Datta1},  motivating us to possibly include these systems in the quantum specification theory developed here.

\begin{definition}[\textbf{Finite Volume Path Gibbs functional}]\label{def_gibbs_path} Given a path $\mathfrak{p}_{\Lambda^c}$ we can define a linear functional for each $f\in C_c(\mathcal{G})$ by the following expression
	\be\label{f_gibbs_path}
	\mu^{\mathfrak{p},\omega}_{\beta,\phi,\Lambda}(f) = \frac{1}{Z_{\beta,\phi,\Lambda}^{\mathfrak{p},\omega}}\sum_{\substack{\sigma_\Lambda \in \Omega_\Lambda \\ (\eta_\Lambda,g_\Lambda)\in \mathcal{G}_\Lambda^{\sigma_\Lambda}}}\hspace{-0.5cm}f(\eta_\Lambda\mathfrak{p}_{\Lambda^c}(1), g_\Lambda g_{\mathfrak{p},\Lambda^c}) D_{\beta,\Lambda}^{\mathfrak{p},\omega}(\sigma_\Lambda,g_\Lambda^{-1}),
	\ee
	where $g_{\mathfrak{p},\Lambda^c}$ is the unique element of the group $G$ such that $\mathfrak{p}_{\Lambda^c}(0) =  g_{\mathfrak{p},\Lambda^c}\mathfrak{p}_{\Lambda^c}(1)$, and the normalization is given by the partition function defined in Equation \eqref{path_partition_function}. When the path $\mathfrak{p}$ is empty, there are no jumps, and one recovers the classical boundary condition $\omega$ in \eqref{classicalbc}.  
\end{definition}

The finite volume path Gibbs functionals are obviously linear and continuous on $C^*(\mathcal{G})$. There is also a version of the \emph{consistency condition} known to classical systems.

\begin{proposition}[\textbf{Consistency Condition}]\label{Consistency}
	Let $\phi$ be a short-range admissible interaction and $\Lambda \in \mathcal{F}(\Z^d)$. Then for any $\Delta \subset \Lambda$ and $f \in C_c(\mathcal{G})$ we have
	\be
	\mu_{\beta,\phi,\Lambda}^{\mathfrak{p},\omega}(f) = \mu_{\beta,\phi,\Lambda}^{\mathfrak{p},\omega}(\mu_{\beta,\phi,\Delta}^{(\cdot)}(f))
	\ee
\end{proposition}
\begin{proof}
	The finite volume path Gibbs functional is         \begin{equation}\label{cons_eq1}
		\begin{split} \mu_{\beta,\phi,\Lambda}^{\mathfrak{p},\omega}(f) &= \frac{1}{Z_{\beta,\phi,\Lambda}^{\mathfrak{p},\omega}}\sum_{\substack{\sigma_\Lambda \in \Omega_\Lambda \\ (\eta_\Lambda,g_\Lambda) \in \mathcal{G}_\Lambda^{\sigma_\Lambda}}} f(\eta_\Lambda \mathfrak{p}_{\Lambda^c}(1), g_\Lambda g_{\mathfrak{p},\Lambda^c}) D_{\beta,\Lambda}^{\mathfrak{p},\omega}(\sigma_\Lambda,g_\Lambda^{-1})\\
			&=\frac{1}{Z_{\beta,\phi,\Lambda}^{\mathfrak{p},\omega}}\sum_{\substack{\sigma_\Lambda \in \Omega_\Lambda \\ (\eta_\Lambda,g_\Lambda) \in \mathcal{G}_\Lambda^{\sigma_\Lambda}}} \int_{\mathfrak{P}^{\sigma_\Lambda,g^{-1}_\Lambda,\mathfrak{p}_{\Lambda^c}}}f(\eta_\Lambda \mathfrak{p}_{\Lambda^c}(1), g_\Lambda g_{\mathfrak{p},\Lambda^c}) e^{-\beta \int_{\mathfrak{p}_\Lambda}H_\Lambda^{\mathfrak{p}_{\Lambda^c}}(\phi)}V_\Lambda^{\mathfrak{p}_{\Lambda^c}}(\phi)(\mathfrak{p}_\Lambda)d\mathfrak{p}_\Lambda.
		\end{split}
	\end{equation}
	Then we can break the first sum depending on the configurations into
	\be\label{cons_eq2}
	\sum_{\substack{\sigma_\Lambda \in \Omega_\Lambda \\ (\eta_\Lambda,g_\Lambda) \in \mathcal{G}_\Lambda^{\sigma_\Lambda}}} = \sum_{\substack{\sigma_{\Lambda\setminus\Delta}\in \Omega_{\Lambda\setminus\Delta} \\ (\eta_{\Lambda\setminus\Delta},g_{\Lambda\setminus\Delta}) \in \mathcal{G}_{\Lambda\setminus\Delta}^{\sigma_{\Lambda\setminus\Delta}}}}\sum_{\substack{\sigma_{\Delta}\in \Omega_\Delta \\ (\eta_\Delta,g_\Delta) \in \mathcal{G}_\Delta^{\sigma_\Delta}}}.
	\ee
  
Since $N_\Lambda = N_\Delta + N_{\Lambda,\Delta}$ as in Equation \eqref{ppp_lambda}. By Lemma \ref{lemma_ppp_decomp} we have
\begin{align}\label{cons_eq11}
&\int_{\mathfrak{P}^{\sigma_\Lambda,g_\Lambda^{-1},\mathfrak{p}_{\Lambda^c}}}f(\mathfrak{p}_\Lambda)d\mathfrak{p}_\Lambda = \int_{\mathbb{N}(\widetilde{\Omega}_{\Lambda_R}\times \mathcal{P}(\Lambda_R)\setminus\{\emptyset\})}\mathbbm{1}_{\mathfrak{P}^{\sigma_\Lambda,g_\Lambda^{-1},\mathfrak{p}_{\Lambda^c}}}(\nu)f(\nu)dN_\Lambda(\nu) \\
&= \int_{\mathbb{N}(\widetilde{\Omega}_{\Lambda_R}\times \mathcal{P}(\Lambda_R)\setminus\{\emptyset\})}\int_{\mathbb{N}(\widetilde{\Omega}_{\Lambda_R}\times \mathcal{P}(\Lambda_R)\setminus\{\emptyset\})}\mathbbm{1}_{\mathfrak{P}^{\sigma_\Lambda,g_\Lambda^{-1},\mathfrak{p}_{\Lambda^c}}}(\nu_1+\nu_2)f(\nu_1+\nu_2)dN_\Delta(\nu_1)dN_{\Lambda\setminus \Delta}(\nu_2). \nonumber
\end{align}
Notice that for the innermost integral, the counting measure $\nu_2$ is fixed and $\nu_1$ is a counting measure drawn accordingly to the point process $N_\Delta$. Thus it holds that if $\nu_1+\nu_2 +\mathfrak{p}_{\Lambda^c}\sim_{\Lambda^c} \mathfrak{p}$ then $\nu_1+\nu_2 +\mathfrak{p}_{\Lambda^c}\sim_{\Delta^c} \nu_2+\mathfrak{p}$. Since $\nu_2$ is a counting measure draw accordingly with $N_{\Lambda,\Delta}$, it only has jumps such that $B'\cap (\Lambda\setminus\Delta) \neq \emptyset$, where $B'=\Lambda_{f_B}\cup B$. Then, define the set
\[
\mathfrak{B}_{\Lambda\setminus\Delta}^{\sigma_{\Lambda\setminus\Delta},g_{\Lambda\setminus\Delta}}\coloneqq \Bigg\{\begin{array}{@{}c|c@{}}
	\nu_2=\sum_{k=1}^n\delta_{(\sigma_{\Lambda_R,k},t_k, B_k)}, n\geq 0 
	&
	\begin{matrix}
\sigma_{\Lambda_R,1} = \sigma_{\Lambda_R}, \sigma_{\Lambda,n+1}= g_{\Lambda_R}\sigma_{\Lambda_R}
  \\[2ex]
 g_k\sigma_{\Gamma, k} = \sigma_{\Gamma,k+1} , B'_k\cap (\Lambda \setminus \Delta)^c \neq \emptyset, 1\leq k\leq n
  \end{matrix}
\end{array}
\Bigg\},
\]
where $g_k\in G_{\Lambda_R}$ such that $\supp g_k = B_k$ and $\Gamma = \Lambda_R\setminus\Delta$. It is measurable by Lemma \ref{sec:appendix:ppp}. Although the counting measures in $\mathfrak{B}_{\Lambda\setminus\Delta}^{\sigma_{\Lambda\setminus\Delta},g_{\Lambda\setminus\Delta}}$ do not form paths as we defined previously since the configurations inside $\Delta$ do not need to be connected by the flips $B_k$. But they can always be completed into a path by adding the relevant transitions inside $\Delta$. Hence it holds for all bounded and measurable functions $f:\mathbb{N}(\widetilde{\Omega}_{\Lambda_R}\times \mathcal{P}(\Lambda_R)\setminus\{\emptyset\})\rightarrow \mathbb{C}$
\be\label{cons_eq10}
\mathbbm{1}_{\mathfrak{P}^{\sigma_\Lambda,g_\Lambda^{-1},\mathfrak{p}_{\Lambda^c}}}(\nu_1+\nu_2) = \mathbbm{1}_{\mathfrak{B}_{\Lambda\setminus\Delta}^{\sigma_{\Lambda\setminus\Delta},g_{\Lambda\setminus\Delta}^{-1}}}(\nu_2)\mathbbm{1}_{\mathfrak{P}^{\sigma_\Delta,g_\Delta^{-1},\nu_2+\mathfrak{p}_{\Lambda^c}}}(\nu_1).
\ee
In order to simplify the notation, we will denote the elements of $\mathfrak{B}_{\Lambda\setminus\Delta}^{\sigma_{\Lambda\setminus\Delta},g_{\Lambda\setminus\Delta}^{-1}}$ by $\mathfrak{p}_{\Lambda,\Delta}$ and write $\mathfrak{p}_{\Lambda,\Delta}+\mathfrak{p}_{\Lambda^c} = \mathfrak{p}_{\Delta^c}$. Hence using Equation \eqref{cons_eq11} and \eqref{cons_eq10} we get 
	\be\label{cons_eq3}
	\int_{\mathfrak{P}^{\sigma_\Lambda,g^{-1}_\Lambda,\mathfrak{p}_{\Lambda^c}}} f(\mathfrak{p}_\Lambda) d\mathfrak{p}_\Lambda
	=\int_{\mathfrak{B}_{\Lambda\setminus\Delta}^{\sigma_{\Lambda\setminus\Delta},g_{\Lambda\setminus\Delta}^{-1}}}\int_{\mathfrak{P}^{\sigma_{\Delta},g^{-1}_{\Delta},\mathfrak{p}_{\Delta^c}}} f(\mathfrak{p}_\Lambda)d \mathfrak{p}_\Delta d\mathfrak{p}_{\Lambda,\Delta}.
	\ee
Let $\varrho$ be any configuration such that $\varrho_{\Delta^c}=\eta_{\Lambda\setminus\Delta}\omega_{\Lambda^c}$. Using Equations \eqref{cons_eq2} and \eqref{cons_eq3} into Equation \eqref{cons_eq1} gives us to the innermost sum is equal to
	\be\label{cons_eq4}
	\sum_{\substack{\sigma_{\Delta}\in \Omega_\Delta \\ (\eta_\Delta,g_\Delta) \in \mathcal{G}_\Delta^{\sigma_\Delta}}} f(\eta_\Delta \varrho_{\Delta^c}, g_\Delta g_{\mathfrak{p},\Delta^c})\int_{\mathfrak{P}^{\sigma_{\Delta},g^{-1}_{\Delta},\mathfrak{p}_{\Delta^c}}}  F_\Lambda(\mathfrak{p}_\Lambda)d \mathfrak{p}_\Delta,
	\ee
	where $$F_\Lambda(\mathfrak{p}_\Lambda)=e^{-\beta \int_{\mathfrak{p}_\Lambda}H_\Lambda^{\mathfrak{p}_{\Lambda^c}}(\phi)}V_\Lambda^{\mathfrak{p}_{\Lambda^c}}(\phi)(\mathfrak{p}_\Lambda).$$  We can decompose the classical Hamiltonian as 
	\[
	H_\Lambda^{(0)}(\phi)+W_\Lambda^{(0)}(\phi)= H_\Delta^{(0)}(\phi)+W_\Delta^{(0)}(\phi) +W_{\Lambda,\Delta}^{(0)}(\phi),
	\]
	where
	\[
	W_{\Lambda,\Delta}^{(0)}(\phi)\coloneqq-\sum_{\substack{A\cap \Lambda \neq \emptyset \\ A\cap \Delta = \emptyset}}J_A\sigma_A,
	\]
	yielding
	\[
	\int_{\mathfrak{p}}(H_\Lambda^{(0)}(\phi)+W_\Lambda^{(0)}(\phi)) = \int_{\mathfrak{p}}(H_\Delta^{(0)}(\phi)+W_\Delta^{(0)}(\phi))+\int_{\mathfrak{p}_{\Delta^c}}W_{\Lambda,\Delta}^{(0)}(\phi),
	\]
	by arguments similar to the ones that gave us Equations \eqref{eq1_decomp_integral_classical_hamiltonian} and \eqref{eq2_decomp_integral_classical_hamiltonian}. A similar decomposition holds for $V_\Lambda(\phi)$,
	\begin{equation*}
		\begin{split}
			V^{\mathfrak{p}_{\Lambda^c}}_\Lambda(\phi)(\mathfrak{p}) = \prod_{k:B'_k\cap\Lambda\neq \emptyset}f_{B_k}(g_k\sigma_{\Lambda_R,k}) &= \prod_{k:B'_k\cap\Delta\neq \emptyset}f_{B_k}(g_k\sigma_{\Lambda_R,k})\prod_{\substack{k:B'_k\cap\Lambda\neq \emptyset \\ B'_k\cap\Delta = \emptyset}}f_{B_k}(g_k\sigma_{\Lambda_R,k}) 
			\\
			&\coloneqq V_\Delta^{\mathfrak{p}_{\Delta^c}}(\phi)(\mathfrak{p}_\Delta)V_{\Lambda,\Delta}(\phi)(\mathfrak{p}_{\Delta^c}).
		\end{split}
	\end{equation*}
	Since $V_{\Lambda,\Delta}(\phi)$ and $W_{\Lambda,\Delta}(\phi)$ do not depend on the configuration inside $\Delta$ we can write
	\begin{equation}\label{cons_eq5}
		\begin{split}
			\int_{\mathfrak{P}^{\sigma_\Delta,g_\Delta^{-1},\mathfrak{p}_{\Delta^c}}}F_\Lambda(\mathfrak{p}_\Lambda)d\mathfrak{p}_\Delta=e^{-\beta \int_{\mathfrak{p}_{\Delta^c}}W_{\Lambda,\Delta}(\phi)}V_{\Lambda,\Delta}(\phi)(\mathfrak{p}_{\Delta^c})\int_{\mathfrak{P}^{\sigma_\Delta,g_\Delta^{-1},\mathfrak{p}_{\Delta^c}}}F_\Delta(\mathfrak{p}_\Delta)d\mathfrak{p}_\Delta.
		\end{split}
	\end{equation}
	Plugging Equation \eqref{cons_eq5} into \eqref{cons_eq4}, we get
	\be\label{cons_eq6}
	\Bigg(\sum_{\substack{\sigma_{\Delta}\in \Omega_\Delta \\ (\eta_\Delta,g_\Delta) \in \mathcal{G}_\Delta^{\sigma_\Delta}}} f(\eta_\Delta \varrho_{\Delta^c}, g_\Delta g_{\mathfrak{p},\Delta^c})\int_{\mathfrak{P}^{\sigma_\Delta,g_\Delta^{-1},\mathfrak{p}_{\Delta^c}}}F_\Delta(\mathfrak{p}_\Delta)d\mathfrak{p}_\Delta\Bigg) e^{-\beta \int_{\mathfrak{p}_{\Delta^c}}W_{\Lambda,\Delta}(\phi)}V_{\Lambda,\Delta}(\phi)(\mathfrak{p}_{\Delta^c}),
	\ee
	We can multiply and divide Equation \eqref{cons_eq6} by the partition function $Z_{\beta,\phi,\Delta}^{\mathfrak{p},\sigma_{\Lambda\setminus\Delta}\omega_{\Delta^c}}$ and writing it as
 \[
 Z_{\beta,\phi,\Delta}^{\mathfrak{p},\sigma_{\Lambda\setminus\Delta}\omega_{\Delta^c}} = \sum_{\sigma_\Delta \in \Omega_\Delta}  D_{\beta,\Delta}^{\mathfrak{p},\sigma_{\Lambda\setminus\Delta}\omega_{\Delta^c}}(\sigma_\Delta,g_{\mathfrak{p},\partial_R\Delta}) = \sum_{\sigma_\Delta\in \Omega_\Delta}\int_{\mathfrak{P}^{\sigma_\Delta,g_{\mathfrak{p},\partial_R\Delta},\mathfrak{p}_{\Delta^c}}}F_{\Delta}(\mathfrak{p}_\Delta)d\mathfrak{p}_\Delta.
 \]
 substituting again the variables we get
	\be\label{cons_eq7}
	\begin{split}
		\Bigg(\sum_{\sigma_{\Delta}\in \Omega_\Delta } \mu_{\beta,\phi,\Delta}^{\mathfrak{p},\sigma_{\Lambda\setminus\Delta}\omega_{\Lambda^c}}(f)\int_{\mathfrak{P}^{\sigma_\Delta,g_{\mathfrak{p},\partial_R\Delta},\mathfrak{p}_{\Delta^c}}}F_\Delta(\mathfrak{p}_\Delta)d\mathfrak{p}_\Delta\Bigg)e^{-\beta \int_{\mathfrak{p}_{\Delta^c}}W_{\Lambda,\Delta}(\phi)}V_{\Lambda,\Delta}(\phi)(\mathfrak{p}_{\Delta^c}) =\\
		\sum_{\substack{\sigma_{\Delta}\in \Omega_\Delta \\ (\eta_\Delta,g_\Delta) \in \mathcal{G}_\Delta^{\sigma_\Delta}}} \mathbbm{1}(\eta_{\Delta},g_\Delta)\mu_{\beta,\phi,\Delta}^{\mathfrak{p},\sigma_{\Lambda\setminus\Delta}\omega_{\Lambda^c}}(f)\int_{\mathfrak{P}^{\sigma_\Delta,g_\Delta^{-1},\mathfrak{p}_{\Delta^c}}}F_\Lambda(\mathfrak{p}_\Lambda)d\mathfrak{p}_\Delta.
	\end{split}
	\ee
	
	Using Equations \eqref{cons_eq2} and \eqref{cons_eq3} together with \eqref{cons_eq7} finally gives us the desired result.
\end{proof}
\begin{remark}
    The proof is also valid if we change the boundary condition for the evaluation with respect to a state as in \eqref{israelbc}, since we can also use the random representation for the modified Hamiltonian.
\end{remark}

The finite volume path Gibbs measures depends on the paths, so it is not straightforward to define a function on the groupoid, depending on the boundary conditions. Given a groupoid $\mathcal{G}_\Lambda$ with $\Lambda$ finite and let
\[
\mathfrak{F}(\mathcal{G}_\Lambda) = \bigcup_{(\sigma_\Lambda,g_\Lambda) \in \mathcal{G}_\Lambda} \mathfrak{P}^{\sigma_\Lambda,g_\Lambda},
\]
be the set of all paths $\mathfrak{p}$ in $\mathcal{G}_\Lambda$. We can define a metric 
\[
 d_2(\mathfrak{p},\mathfrak{q}) = \begin{cases}
    d_1(\mathfrak{p}(0),\mathfrak{q}(0))& j_{\mathfrak{p}}=0.\\
     \frac{1}{2(j_{\mathfrak{p}}+1)}\sum_{k=1}^{j_{\mathfrak{p}}}(|s_k-t_k|+d_1(\mathfrak{p}(t_k),\mathfrak{q}(s_k)))  & j_{\mathfrak{p}}= j_{\mathfrak{q}}\neq 0 \\
     1 & j_{\mathfrak{p}}\neq  j_{\mathfrak{q}}
 \end{cases}
 \]
where $j_{\mathfrak{p}}$ is the number of jumps occuring in $\mathfrak{p}$ and $d_1$ is the metric on $\Omega_\Lambda$ given by
\[
d_1(\sigma_\Lambda,\sigma'_\Lambda)=\sum_{n\geq 0}\frac{1}{|B_n(0)\cap \Lambda|}\sum_{x\in B_n(0)\cap\Lambda}\frac{|\sigma_x-\sigma'_x|}{2^{|x|+1}},
\]
where $B_n(0)\cap \Lambda = \{x\in \Lambda: |x|=n\}$. Every continuous function $f\in C_c(\mathcal{G}_\Lambda)$ can be lifted to a continuous function $\mathfrak{F}(\mathcal{G}_\Lambda)$ by the map $f'(\mathfrak{p}) = f(\sigma_\Lambda,g_\Lambda)$, when $\mathfrak{p}\in \mathfrak{P}^{\sigma_\Lambda,g_\Lambda}$. Notice that every finite volume Gibbs state can be lifted to a measure on $\mathfrak{F}(\mathcal{G}_\Lambda)$ using the random representation
\be
\begin{split}
\mu_{\beta,\phi,\Lambda}(f) &= \frac{1}{Z_{\beta,\phi,\Lambda}}\sum_{\substack{\sigma_\Lambda \in \Omega_\Lambda \\ (\omega_\Lambda,g_\Lambda)\in \mathcal{G}_\Lambda^{\sigma_\Lambda}}}f(\omega_\Lambda,g_\Lambda)e^{-\beta H_\Lambda(\phi)}(\sigma_\Lambda,g_\Lambda^{-1}) \\
&=\frac{1}{e^{-\beta}Z_{\beta,\phi,\Lambda}} \sum_{\substack{\sigma_\Lambda \in \Omega_\Lambda \\ (\omega_\Lambda,g_\Lambda)\in \mathcal{G}_\Lambda^{\sigma_\Lambda}}}f(\omega_\Lambda,g_\Lambda)\int_{\mathfrak{P}^{\sigma_\Lambda,g^{-1}_\Lambda}} e^{-\beta \int_{\mathfrak{p}} H_\Lambda^{(0)}(\phi)} V_\Lambda(\phi)(\mathfrak{p})d\mathfrak{p} 
\\
&=\frac{1}{e^{-\beta}Z_{\beta,\phi,\Lambda}} \sum_{ (\sigma_\Lambda,g_\Lambda)\in \mathcal{G}_\Lambda}\int_{\mathfrak{P}^{\sigma_\Lambda,g_\Lambda}} f'(\mathfrak{p}^{-1})e^{-\beta \int_{\mathfrak{p}} H_\Lambda^{(0)}(\phi)} V_\Lambda(\phi)(\mathfrak{p})d\mathfrak{p}
\end{split}
\ee
In a similar fashion, every continuous linear function on $C_c(\mathfrak{F}(\mathcal{G}_\Lambda))$ can be projected to be a linear continuous functional in $C^*(\mathcal{G}_\Lambda)$.

\begin{proposition}[\textbf{Feller Continuity}]\label{Feller}
	Let $\phi$ be a short-range admissible interaction and $\Lambda \subset \Z^d$ be a finite set. Then, for every $f\in C_c(\mathcal{G})$, the function $\mathfrak{p}\mapsto \mu_{\beta,\phi,\Lambda}^{\mathfrak{p},\omega}(f)$ is continuous on $\mathfrak{F}(\mathcal{G}_{\Lambda_R})$. Moreover, it only depends on paths $\mathfrak{p}$ such that $\mathfrak{p}_{\Lambda}=\mu_{\emptyset}$.
\end{proposition}
\begin{proof}
It is a straightforward consequence of the Lebesgue-dominated convergence theorem. 
\end{proof}

\begin{proposition}[\textbf{Proper}]\label{Proper}
	Let $\phi$ be an admissible short-range interaction and $\Lambda \subset \Z^d$. If $f$ is in $C_c(\mathcal{G}_{\Lambda_R^c})$, then $\mu_{\beta,\phi,\Lambda}^{\mathfrak{p},\omega}(f) = f$.
\end{proposition}
\begin{proof}
	Notice that when a local operator $f$ is in $C_c(\mathcal{G}_{\Lambda_R^c})$ there exists a $\Delta \subset \Lambda_R^c$ such that $f(\sigma,g) = \mathbbm{1}(\sigma_{\Delta^c},g_{\Lambda_R^c \setminus \Delta}) f(\sigma_\Delta,g_\Delta)$. Hence
	\begin{align*}
		\sum_{\substack{\sigma_\Lambda \in \Omega_\Lambda \\ (\eta_\Lambda,g_\Lambda)\in \mathcal{G}_\Lambda^{\sigma_\Lambda}}}\hspace{-0.5cm}f(\eta_\Lambda\mathfrak{p}_{\Lambda^c}(1), g_\Lambda g_{\mathfrak{p},\Lambda^c}) D_{\beta,\Lambda}^{\mathfrak{p}}(\sigma_\Lambda,g_\Lambda^{-1})=f(\mathfrak{p}_\Delta(1),g_{\mathfrak{p},\Delta})\hspace{-0.5cm}\sum_{\substack{\sigma_\Lambda \in \Omega_\Lambda \\ (\eta_\Lambda,g_\Lambda)\in \mathcal{G}_\Lambda^{\sigma_\Lambda}}}\hspace{-0.5cm} \mathbbm{1}(\eta_{\Lambda}\mathfrak{p}_{\Lambda^c}(1),g_\Lambda g_{\mathfrak{p},\Lambda_R^c})  D_{\beta,\Lambda}^{\mathfrak{p}}(\sigma_\Lambda,g_\Lambda^{-1}), 
	\end{align*}
	and the last line is equal to $f(\mathfrak{p}_\Delta(1),g_{\mathfrak{p},\Delta})Z_{\beta,\phi,\Lambda}^{\mathfrak{p}}$, yielding us the desired result. 
\end{proof}

\begin{definition}\label{perron-positive}
    An operator $f\in C^*(\mathcal{G})$ is said to be \textbf{Perron-positive} if it is the limit of $f_n \in C_c(\mathcal{G})$ satisfying $f_n(\sigma,g) \geq 0$.
\end{definition}

If the path boundary condition $\mathfrak{p}_{\Lambda^c}$ is symmetric, then they can be shown to be states, even if the interaction is only admissible. 

\begin{proposition}
	For $\mathfrak{p}_{\Lambda^c}$ symmetric and $\Lambda \in \mathcal{F}(\Z^d)$ we have that $\mu_{\beta,\phi,\Lambda}^{\mathfrak{p},\omega}$ is a state.
\end{proposition}
\begin{proof}
	Let $\mathfrak{p}_{\Lambda^c} = \mathfrak{q}_{\Lambda^c}^{-1}\circ_{1/2} \mathfrak{q}_{\Lambda^c}$ be the symmetric path. Due to Lemmas \ref{lemma_path_density} and \ref{lemma_adjoint_density} it holds that
	\[
	D_{\beta,\Lambda}^{\mathfrak{p},\omega} = D_{\beta/2,\Lambda}^{\mathfrak{q}^{-1},\omega}\cdot D_{\beta/2,\Lambda}^{\mathfrak{q},\omega'} = (D_{\beta/2,\Lambda}^{\mathfrak{q}},\omega')^*\cdot D_{\beta/2,\Lambda}^{\mathfrak{q},\omega}. 
	\]
 Thus using the cyclic property of the trace, we conclude the proof.
\end{proof}

 Inspired by classical statistical mechanics, we introduce the following definition for the infinite volume Gibbs functionals
\be
\mathscr{G}_\beta(\phi) \coloneqq \overline{\text{co}}\Bigg\{\begin{array}{@{}c|c@{}}
	\mu_{\beta,\phi} 
	&
	\begin{matrix}
		\exists \{\Lambda_m\}_{m \geq 1}, \Lambda_m \nearrow \Z^d, \text{   and   } \{\mathfrak{p}_m,\omega_m\}_{m\geq 1},\\[2ex]
		 \text{ s.t. } \mu_{\beta,\phi}=w^*-\underset{m\rightarrow \infty}{\lim} \mu_{\beta,\phi,\Lambda_m}^{\mathfrak{p}_m,\omega_m}
	\end{matrix}
\end{array}
\Bigg\}
\ee

There is an important consequence for the Gibbs functionals to be positive and normalized on the cone of Perron-positive operators. It implies that they are positive normalized functionals in $C_c(\mathcal{G})$, therefore by the Riesz-Markov theorem, the elements of $\mathscr{G}_\beta(\phi)$ are \emph{probability measures} on the groupoid $\mathcal{G}$. This is in contrast to states that can only be guaranteed to be complex measures using the Riesz-Markov theorem. We arrived at the following proposal for a specification for quantum spin systems. 

\begin{definition}
	Let $\mathcal{G}$ be the transformation groupoid introduced in Chapter 4. Then, a family of functions $\mu_\Lambda:C_c(\mathcal{G}_\Lambda)\times
    \mathfrak{F}(\mathcal{G}_{\Lambda^c})\rightarrow \mathbb{C}$ indexed by the finite subsets $\Lambda$ of $\Z^d$ is called a \textbf{proper quantum specification} if and only if
	\begin{itemize}
		\item[(i)]\textbf{(Linear Functionals)} For every $\mathfrak{p} \in \mathfrak{F}(\mathcal{G}_{\Lambda^c})$, and $\omega\in \Omega_{\Lambda^c}$ $\mu_\Lambda^{\mathfrak{p},\omega}$ is a continuous linear functional; if $\mathfrak{p}$ is symmetric, then it is a state.
		\item[(ii)]\textbf{(Consistency)} For every $\Delta\subset \Lambda$, it holds $\mu_\Lambda(\mu_{\Lambda'}(f)) = \mu_\Lambda(f)$.
        \item [(iii)] \textbf{(Feller Continuity)} For every $f \in C_c(\mathcal{G}_\Lambda)$, we know that $\mu_{\Lambda}(f)$ is a continuous function in $\mathfrak{F}(\mathcal{G}_{\Lambda^c})$.
		\item[(iv)]\textbf{(Proper)} For each $\Lambda$ there is $\Delta$ a subset containing $\Lambda$ such that $f \in C_c(\mathcal{G}_{\Delta^c})$, then $\mu_\Lambda(f)=f$.
        \end{itemize}
\end{definition}

As we showed in Propositions \ref{Proper}, \ref{Consistency}, and \ref{Feller}, the family of finite volume path Gibbs functionals we introduced in \ref{def_gibbs_path} is a quantum specification for the transformation groupoid $\mathcal{G}$. When a specification is constructed in this way, using an interaction, we call this a \emph{quantum Gibbs specification}. We are motivated to introduce the following definition for quantum DLR states:

\begin{definition}[\textbf{Quantum DLR equations}]
	A state $\mu$ of $C^*(\mathcal{G})$ is said to be a quantum DLR state when, for every $\Lambda$, 
	\[
	\mu_\beta(f)= \mu_\beta(\mu^{(\cdot)}_{\beta,\Lambda}(f)).
	\]
\end{definition}
The set of all quantum $DLR$ states is
\be
\mathscr{G}_{\beta, DLR}(\phi) = \{\mu_\beta: \mu_\beta= \mu_\beta(\mu^{(\cdot)}_{\beta,\Lambda}), \forall \Lambda \in \mathcal{F}(\Z^d)\},
\ee

and it is a compact convex set. 

\begin{theorem}
Let $\phi$ be a short-range stoquastic interaction. Then $\mathscr{G}_\beta(\phi)=\mathscr{G}_{\beta, DLR}(\phi)$	
\end{theorem}
\begin{proof}
	The fact that $\mathscr{G}_\beta(\phi)\subset\mathscr{G}_{\beta, DLR}(\phi)$	follows from the consistency condition and the definition of $w^*$-convergence. For the other inclusion, the proof is almost the same as the classical case (see Theorem \ref{classical_gibbs=dlr}). Suppose that there exists $\mu_\beta \in \mathscr{G}_{\beta,DLR}(\phi)\setminus\mathscr{G}_{\beta}(\phi)$. Since $\phi$ is stoquastic, the functionals $\mu_\beta$ are probability measures on $\mathcal{G}$, therefore we have that $\mu_\beta$ is in the closed convex hull of the finite volume path Gibbs functionals. The rest of the argument follows in the same way.
\end{proof}

\chapter{Conclusions and Further Research in Quantum Statistical Mechanics}
\label{ch:conclusion}
\epigraph{We escape into dream [fantasy] to avoid a deadlock in our real life. But then, what we encounter in the dream is even more horrible, so that at the end we literally escape back into reality. It starts with, "dreams are for those who cannot endure, who are not strong enough for reality", and ends with "reality is for those who are not strong enough to endure, to confront their dreams".}{\textbf{Slavoj \v{Z}i\v{z}ek} \\ \textit{The Pervert's Guide to Cinema}}

In Chapter \ref{ch:quantstatmech} we made a review of quantum statistical mechanics using the language of transformation groupoids, introducing important examples of quantum spin systems, for an arbitrary discrete spin, and also fermionic models. In particular, we introduced the $d$-dimensional Jordan-Wigner transformation in the groupoid language, based on ideas from \cite{kochmanski1998jordanwigner}. We also collected some well-known results on the existence of the dynamics for the thermodynamic limit and discussed the equivalence between the KMS condition and the Gibbs-Araki-Ion condition, reproving in modern language the equivalence between the Gibbs-Araki-Ion condition and the DLR equations discussed in Chapter \ref{ch:classtatmech} for arbitrary finite spin systems, see Theorem \ref{t2}. In view of the discussions of Chapter \ref{ch:groupoids}, every KMS state is a complex measure on the groupoid $\mathcal{G}$ and one could ask if the support of the measure can characterize the interaction, i.e., the KMS state from a nonclassical interaction has support outside the unit space $\Omega$. 

In Chapter \ref{ch:groupoids} we obtained the Poisson point process representation for an arbitrary short-range interaction in the case $q=2$. The extension of the calculations to arbitrary finite spin is straightforward using the general form of the Hamiltonian operators discussed in Remark \ref{generalcase}. The extension of the random representation to long-range interactions seems feasible with some modification; one can try to construct the random representation with a fixed boundary condition outside a large box. The random representation for long-range quantum spin systems could be used to study phase transitions in some models such as the one in \cite{Lapa2023} using the techniques of Chapter \ref{ch:longrange}. We proposed a theory of Quantum Gibbsian specification in the spirit of classical statistical mechanics and proved the extension of the classical result of the convex hull of the thermodynamic limit of finite volume path Gibbs functionals is exactly the set of all the states satisfying the Quantum DLR states for stoquastic interactions. Since we can recover the classical boundary condition states present in \cite{BKU, Datta1} we expect that more states can be constructed using our procedure. Moreover, since their boundary conditions are imposed in the Hamiltonians, the finite volume dynamics can be readily defined by exponentiation of the corresponding derivation (see the discussion below Equation \eqref{israelbc}). To investigate if such a picture can be obtained for the path boundary condition to describe their dynamics seems an interesting question to investigate. 

In \cite{Klein}, Klein and Landau constructed a stochastic process for stochastically positive KMS states and related properties from the constructed stochastic process to the Tomita-Takesaki theory for the state. One interesting question is to study the same problem for quantum spin and fermionic systems using the Poisson point process representation, i.e., relate the Tomita-Takesaki theory for the state to the properties of a Poisson point process. One can investigate if this approach can have any impact on the perturbation theory for these systems.

Also, differently from \cite{Datta1}, Datta, Fernand\'ez, and Fr\"ohlich and Borgs, Koteck\'y, and Ueltschi, our methods apply to any temperature regime. But \cite{Datta1} extended the Pirogov-Sinai theory for bosonic lattice systems also, and one interesting question is to extend the Quantum DLR approach to these systems. An approach using point processes for the continuum already exists and can be consulted in \cite{Fich1, Fich2}. The theory of specifications for classical statistical mechanics is so robust it can handle very general state spaces, so one can explore the possibility of extending the quantum specification theory to more general systems.

About extensions of the Quantum DLR formalism developed here, the Poisson point process representation is used in many situations to study ground states of the system (see \cite{Kennedy1991} for instance). But a DLR theory for ground states of \emph{classical interactions} already exists, although it is not Gibbsian, and can be found in the paper by van Enter, Fernand\'ez, and Sokal \cite{vanEnter1993}. In light of the results of Cha, Naaijkens, and Nachtergaele \cite{Pieter} and the fact that the Toric Code is stoquastic, one can wonder if the quantum DLR theory can be extended to ground states. 

\renewcommand{\chaptermark}[1]{\markboth{\MakeUppercase{\appendixname\ \thechapter}} {\MakeUppercase{#1}} }
\fancyhead[RE,LO]{}
\appendix

%
\chapter{Strongly Continuous (semi)groups}
\label{sec:appendix:dynamics}

In this chapter, we will collect standard theorems in the theory of strongly continuous (semi)groups, as can be found in \cite{En}. 

\section{Strongly Continuous (semi)groups}

 \begin{definition}
 Let $X$ be a complex Banach space and $L(X)$ the space of all bounded linear operators $A:X\rightarrow X$. A function $\T:[0,\infty) \rightarrow L(X)$ is called a \textbf{one-parameter semigroup} if 
\begin{itemize}
 \item[(i)] $\T_0 = \mathbbm{1}$. 
 \item[(ii)] $\T_{t+s} = \T_t \T_s, \; \forall t,s\geq 0$. 
\end{itemize}
 When $\T_t$ is defined in the whole real line and still satisfies $(i)$ and $(ii)$ it is called a \textbf{one-parameter group}. In the case of the function is continuous in the strong operator topology\footnote{The strong operator topology in $L(X)$ is the topology defined by the seminorms $A \rightarrow \|Ax\|$, for all $x \in X$} of $L(X)$, the semigroup is called  \textbf{strongly continuous}.
\end{definition}

\begin{example} \label{ex1}
Let $X = M_n(\mathbb{C})$ and $A \in M_n(\mathbb{C})$. We can define $\T_t = e^{tA}$, for all $t \in \mathbb{R}$.
\end{example}

\begin{definition}\label{def_generator}
Let $X$ be a Banach space and $\{\T_t\}_{t\geq 0}$ a strongly continuous one-parameter semigroup. The operator $A:\mathrm{dom}(A) \subset X \rightarrow X$ defined by
\[
Ax = \lim_{t \rightarrow 0} \frac{\T_tx - x}{t},
\]
for all $x \in \mathrm{dom}(A) = \{ x \in X: \lim_{t \rightarrow 0} \frac{\T_tx - x}{h} \; \text{exists}\}$ is called the \textbf{generator} of $\T$.
\end{definition}
It is easy to see that the generator is always a linear operator.
\begin{lemma}\label{l1}
Let $X$ be a Banach space and $\T:[0,\infty) \rightarrow L(X)$ a strongly continuous one-parameter semigroup. Then, the generator $A$ of the semigroup satisfies
\begin{itemize}
\item[(i)] $\frac{d}{dt}\T_tx = A\T_tx = \T_tAx, x \in \mathrm{dom}(A)$.
\item[(ii)] $\int_0^t \T_sx ds \in \mathrm{dom}(A)$, $\forall t\geq 0 $, $\forall x \in X$.
\item[(iii)] $\T_tx - x = A\int_0^t \T_sx ds$ $\forall x \in X$. Moreover, $\T_tx - x = \int_0^t \T_sAx ds$ whenever $x \in \mathrm{dom}(A)$.
\end{itemize}
\end{lemma}
\begin{proof}
We can calculate the limit explicitly, obtaining
\[
\lim_{h \rightarrow 0} \frac{\T_{t+h}x-\T_{t}x}{h}=\T_{t}\lim_{h \rightarrow 0} \frac{\T_hx-x}{h}=\T_tAx.
\]
The other equality in $(i)$ can be obtained in a similar fashion. The last two items will be proved at the same time. We have
\[
\begin{aligned}
&\lim_{h \rightarrow 0} \frac{(\T_h-\mathbbm{1})}{h}\int_0^t\T_sx ds=\lim_{h \rightarrow 0} \frac{1}{h}\Bigg(\int_0^t\T_{s+h}x ds-\int_0^t\T_s x ds\Bigg) \\ 
&=\lim_{h \rightarrow 0}\frac{1}{h}\Bigg(\int_h^{t+h}\T_s x ds-\int_0^t\T_s x ds\Bigg) =\lim_{h \rightarrow 0}\frac{1}{h}\Bigg(\int_t^{t+h}\T_s x ds-\int_0^h\T_s x ds\Bigg)
\end{aligned}
\] 
The Fundamental Theorem of Calculus implies that both limits exist and are, respectively, $\T_tx$ and $x$. The second equality of the item $(iii)$ follows easily from the Fundamental Theorem of Calculus together with $(i)$
\[
\T_tx -x = \int_0^t \frac{d}{ds}\T_s x ds = \int_0^t\T_sAx ds.
\]
\end{proof}

Remember that a linear operator $A:\mathrm{dom}(A)\subset X\rightarrow X$ is said to be \emph{closed} when for each sequence $x_n\in X$ converging to $x$ such that $Ax_n$ converges to $y$, then we have that $x \in \mathrm{dom}(A)$ and $Ax = y$.

\begin{proposition}
Let $X$ be a Banach space and $\T$ a strongly continuous one-parameter semigroup. Then, its generator $A$ is closed and $\mathrm{dom}(A)$ is dense. Moreover, every generator determines its semigroup uniquely.
\end{proposition}
\begin{proof}
Let $\{x_n\}_{n \in \mathbb{N}} \in \mathrm{dom}(A)$ such that $\lim_{n\rightarrow \infty }Ax_n = y$ and $\lim_{n\rightarrow \infty }x_n = x$. Lemma \ref{l1} implies $\T_tx_n - x_n= \int_0^t \T_sAx_n ds$. We claim that
\be\label{generator_closed_eq1}
\lim_{n\rightarrow \infty}  \int_0^t \T_sAx_n ds= \int_0^t \T_sy ds.
\ee
Indeed, let $f_n:[0,t] \rightarrow X$ be a sequence of functions defined as $f_n(s) = \T_sAx_n$ and $f:[0,t] \rightarrow X$ as the function $f(s) = \T_sy$.
These functions are all continuous because $\T$ is strongly continuous. The strong continuity also implies for fixed $x \in X$ that the set $B_x = \{\T_sx: 0\leq s \leq t\}$ is compact, therefore
\[
\underset{0 \leq s \leq t}{\sup}\|\T_sx\| \leq \infty.
\]
The Banach-Steinhauss Theorem implies $M = \underset{0 \leq s \leq t}{\sup}\|\T_s\| \leq \infty$. Thus, the sequence $f_n$ converges uniformly to $f$, since
\[
\|f_n - f\| = \underset{0\leq s \leq t}{\sup}\|\T_sAx_n - \T_sy\| \leq M \|Ax_n -y\|.
\]
Taking the limit in both sides in Equation \eqref{generator_closed_eq1} yields $\T_tx - x = \int_0^t \T_sy ds$. We can divide both sides by $t$ and, taking the limit, we conclude that $x \in \mathrm{dom}(A)$, because the right-hand side limit exists, and $Ax = y$. Item (ii) of the Lemma \ref{l1} tells us that $\int_0^t \T_sx ds \in \mathrm{dom}(A)$ for all $t>0$. Choosing $t_n$ a sequence of positive numbers converging to zero yields
\[
\lim_{n\rightarrow \infty} \int_0^{t_n} \T_sx ds = x, \;\; \forall x \in X,
\]
concluding that $\mathrm{dom}(A)$ is dense. Suppose there are $\T_t$ and $\mathcal{S}_t$ two different strongly continuous one-parameter semigroups with the same generator $A$. Fix $t>0$ and define, for $0\leq s \leq t$, the following function $\mathcal{U}_sx = \T_{t-s}\mathcal{S}_sx$. The Fundamental Theorem of Calculus yields
\[
\mathcal{S}_tx-\T_tx = \int_0^t \frac{d}{ds}\mathcal{U}_s ds =  \int_0^t \mathcal{U}_s(Ax-Ax) ds =0. 
\]
\end{proof}

The situation we usually encounter in applications, such as the one in Chapter \ref{ch:quantstatmech}, is  that we already have a linear operator in a certain Banach space of interest and we want to say whether or not it generates a strongly continuous semigroup. For this, we will use the notion of the resolvent of an operator.

\begin{lemma}
Let $A$ be the generator of a strongly continuous semigroup $\T$. Then, for all $\lambda \in \mathbb{C}$ and $t>0$, it holds
\[
e^{-\lambda t}\T_tx-x = (A-\lambda) \int_0^t e^{-\lambda s} \T_s x ds.
\]
\end{lemma}
\begin{proof}
Define the semigroup $\mathcal{S}_t = e^{-\lambda t} \T_t$. Item $(iii)$ of Lemma \ref{l2} implies
\[
\mathcal{S}_tx -x = B \int_0^t \mathcal{S}_s x ds,
\]
where $B$ is the generator of $\mathcal{S}_t$. A direct computation yields
\[
Bx = \lim_{h\rightarrow 0} \frac{\mathcal{S}_hx-x}{h} =\lim_{h\rightarrow 0} \frac{e^{-\lambda h} \T_hx-e^{-\lambda h}x+ e^{-\lambda h}x-x}{h} = Ax - \lambda x.
\]
\end{proof}

\begin{lemma}\label{l2}
Let $X$ be a Banach space and $T$ a strongly continuous semigroup. Then, there is $M\geq 0$ and $\lambda \in \mathbb{R}$ such that
\[
\|\T_t\| \leq M e^{\lambda^* t},
\]
where $\lambda^* = \log \|T(1)\|$.
\end{lemma}
\begin{proof}
The property of semigroup gives
\[
\|\T_t\| = \|\T_k\T_s\|\leq \|\T_1\|^k \|\T_k\| \leq M e^{k \log\|\T_1\|} \leq M e^{\lambda t} 
\]
where $M = \underset{0\leq s <1}{\sup}\|\T_s\|$, $k$ is the integer part of $t$. Since the semigroup is strongly continuous, the Banach-Steinhauss theorem implies that $M$ is finite.
\end{proof}

\begin{proposition}
Let $\T$ be a strongly continuous semigroup and $A$ its generator. Then,
\begin{itemize}
\item[(i)] If $\lambda \in \mathbb{C}$ and $R(\lambda)x = \int_0^\infty e^{-\lambda s} \T_s x ds$ exists for all $x \in X$ then $\lambda \in \rho(A)$ and it is equal to the resolvent of $A$.
\item[(ii)] If $\mathrm{Re}(\lambda) > \lambda^*$, then $\lambda \in \rho(A)$.
\item[(iii)] $\|R(\lambda, A)\| \leq \frac{M}{\mathrm{Re}(\lambda) - \lambda^*}$, for all  $\mathrm{Re}(\lambda) > \lambda^*$.
\end{itemize}
\end{proposition}

\begin{proof}
To show $(i)$ we can calculate $(\lambda - A)R(\lambda)$ and show that this is the identity. 
\begin{align*}
\frac{e^{-\lambda h}\T_h- \mathbbm{1}}{h} \int_0^\infty e^{-\lambda s} \T_sx ds = \frac{1}{h}\left( \int_0^\infty e^{-\lambda (s+h)} \T_{s+h}x ds -\int_0^\infty e^{-\lambda s} \T_sx ds \right) = \frac{-1}{h}\int_0^h e^{\lambda s} \T_s x ds.
\end{align*}
Making $h \rightarrow 0$, the Fundamental Theorem of Calculus implies that $(A - \lambda)R(\lambda)x = -x$. Our hypothesis is that the integral $R(\lambda)x$ exists for all $x \in X$. This fact together with Lemma \ref{l1} yields
\[
(A-\lambda)\int_0^\infty e^{-\lambda s} \T_sx ds = \int_0^\infty e^{-\lambda s} \T_s(A-\lambda) x ds,
\]
since $A$ is closed. For $(ii)$, the previous item says that is sufficient to show that the integral $\int_0^\infty e^{-\lambda s} \T_s x ds$ exists for all $x \in X$. Lemma \ref{l2} together with the inequality
\[
\Bigg\|\int_0^t e^{-\lambda s} \T_s ds \Bigg\| \leq M  \int_0^t e^{\lambda^* - \mathrm{Re}(\lambda)} ds,
\]
yields the desired result. Item $(iii)$ follows easily from the last inequality.
\end{proof}

\begin{theorem}[\textbf{Hille-Yosida}]
Let $X$ be a Banach space and $A:\mathrm{dom}(A) \subset X \rightarrow X$ a linear operator. The following assertions are equivalent
\begin{itemize}
\item[(i)] $A$ generates a strongly continuous semigroup of contractions.
\item[(ii)] $A$ is closed, $\mathrm{dom}(A)$ is dense and, for all $\lambda>0$, we have that $\lambda \in \rho(A)$ and $\|\lambda R(\lambda , A)\| \leq 1$.
\end{itemize}
\end{theorem}
\begin{proof}
That $(i)$ implies $(ii)$ follows from the previous Lemmas and Propositions. Assume $(ii)$ and define the operators $A_n = nAR(n,A)$. These are well-defined bounded operators, since $n \in \rho(A)$ by hypothesis then $n(A - n\mathbbm{1} + n\mathbbm{1})R(n,A) = n^2R(n,A)-n\mathbbm{1}$.
We can define the semigroups $\T_{n,t} = e^{tA_n}$ by the series of the exponential since the boundedness of $A_n$ implies the convergence of it for all $x \in X$. We now will use these semigroups to define a new one, that will have $A$ as its generator. First, we claim that the limit $\lim_{n \rightarrow \infty} \T_{n,t}x $ exists. Indeed, for all $t>0$ and $x \in X$ we have
\[
\T_{n,t}x- \T_{m,t}x = \int_0^t \frac{d}{ds}\T_{n,s}\T_{m,t-s} ds =  \int_0^t \T_{n,s}\T_{m,t-s} (A_n x - A_m x) ds.
\]
But each $\T_{n,t}$ is a contraction semigroup, since $$\|T_n(t)\| \leq e^{-nt} \|e^{n^2R(n,A)t}\| \leq e^{-nt}e^{nt} = 1,$$ where the last inequality follows from our hypothesis on the norm of the resolvent $\|\lambda R(\lambda,A)\|<1$. Hence,
\[
\|\T_{n,t}x- T_{m,t}x\| \leq t\|A_n x - A_m x\|.
\]
We claim that the sequence $A_n x$ converges to $Ax$ if $x \in \mathrm{dom}(A)$. Indeed, for $x \in \mathrm{dom}(A)$ we know that $nR(n,A) x = R(n,A)Ax + x$ and the bound on the norm of the resolvent implies that $\lim_{n\rightarrow \infty} nR(n,A)x =x$. Because the domain is dense by $(ii)$, $nR(n,A)x \rightarrow x$ for all $x in X$, in particular, for $Ax$. We conclude that the sequence $\T_{n,s}x$ converges for all $x \in \mathrm{dom}(A)$ and for all $t\geq 0$. 
Define the semigroup $\T_tx = \lim_{n \rightarrow \infty} \T_{n,t}x$. This is clearly a semigroup of contractions. It remains to show that $\T_t$ is strongly continuous. For each $x \in X$ we have
\[
\|\T_tx -\T_sx\| \leq \|\T_s\| \|\T_{t-s}x - x\| \leq \|\T{t-s}x - x\| \leq \|\T_{t-s}x-\T_{n,t-s}x\|+\|\T_{n,t-s}x-x\|. 
\]
Choosing $n$ sufficiently large, the first part gets small and the uniform continuity of the groups $\T_{n,t}$ implies that the second part is small for small $t-s$.
\end{proof}

\begin{corollary}
Let $X$ be a Banach space and $A:\mathrm{dom}(A) \subset X \rightarrow X$ a linear operator. The following assertions are equivalent
\begin{enumerate}
\item $A$ generates a strongly continuous semigroup such that
\[
\|\T_t\|\leq e^{\lambda^* t}
\]
\item $A$ is closed, $\mathrm{dom}(A)$ is dense and, for all $\lambda > \lambda^*$ we have that $\lambda \in \rho(A)$ and
\[
\|(\lambda-\lambda^*)R(\lambda,A)\|\leq 1 
\]
\end{enumerate}
\end{corollary}

\begin{theorem}[\textbf{Hille Yosida - General case}]
Let $X$ be a Banach space and $A:\mathrm{dom}(A) \subset X \rightarrow X$ a linear operator. The following assertions are equivalent
\begin{itemize}
\item[(i)] $A$ generates a strongly continuous semigroup such that
\[
\|\T_t\|\leq M e^{\lambda^* |t|}
\]
\item[(ii)] $A$ is closed, $\mathrm{dom}(A)$ is dense and, for all $|\lambda| > \lambda^*$ we have that $\lambda \in \rho(A)$ and
\[
\|(|\lambda| - \lambda^*)^n R(\lambda,A)^n\|\leq M, \;\; \forall n \in \mathbb{N}
\]
\end{itemize}
\end{theorem}

\begin{proof}
Assume that $(i)$ holds. The upper bound on the powers of the resolvent operators is the only thing that we did not proven yet. We claim that the resolvent satisfies
\[
R(\lambda,A)^nx = \frac{(-1)^{n-1}}{(n-1)!}\frac{d^{n-1}}{d\lambda^{n-1}}R(\lambda,A).
\] 
The proof will follow by induction. That it holds for $n=2$ follows from the identity $R(\lambda,A) - R(\mu,A) =  (\mu-\lambda)R(\lambda,A)R(\mu,A),$ for any $\lambda,\mu \in \rho(A)$. Suppose that our claim is valid for $n$. Then, for $n+1$, it holds
\[
\frac{(-1)^n}{n!}\frac{d^n}{d\lambda^n}R(\lambda,A)x = \frac{-1}{n}\frac{d}{d\lambda}R(\lambda,A)^n x = R(\lambda,A)^{n+1} x. 
\]
Using the integral formula proved in Lemma \ref{l1} with the relation of the powers of the resolvent with its derivatives implies
\[
\|R(\lambda,A)^nx\| \leq \frac{M}{(n-1)!}\int_0^{\infty}s^{n-1} e^{(\lambda^* - \mathrm{Re}(\lambda))s} ds,
\]
calculating the integral on the right-hand side we get the bounds in $(ii)$. Assume that $(ii)$ holds. We will change the topology in $X$ in a way that will be possible to apply the Hille-Yosida theorem. Define the new norm on $X$ by
\[
\|x\|_\lambda = \underset{n \geq 0}{\sup}\|\lambda^n R(\mu,A)^nx\|.
\]
The norm $\|\cdot\|_\lambda$ is well defined for $\lambda>\lambda^*$. These norms are equivalent to the previous norm $\|\cdot\|$ on $X$, since $\|x\| \leq \|x\|_\mu \leq M\|x\|$, where $M$ is the constant in $(ii)$. The norm $\|\cdot\|_\lambda$ have some properties and we list them.
\begin{enumerate}
\item $\|\mu R(\mu,A) \|_\mu \leq 1$.
\item $\| \lambda' R(\lambda',A)\|_\lambda \leq 1$, $\forall \omega<\lambda'\leq \lambda$.
\item $\|(\lambda')^n R(\lambda',A)^n x\| \leq \|(\lambda')^n R(\lambda',A)^n x\|_\lambda \leq \|x\|_\lambda$, for all $\lambda^* < \lambda' \leq \lambda'$ and $n\geq 0$.
\end{enumerate}
The first one follows from
\[
\|\lambda R(\lambda,A)x\|_\lambda = \underset{n \geq 0}{\sup}\|\lambda^{n+1} R(\lambda,A)^{n+1}x\| \leq \underset{n \geq 0}{\sup}\|\mu^n R(\lambda,A)^nx\| = \|x\|_\lambda,
\]
for all $x \in X$. Using the identity $R(\lambda,A) = R(\lambda',A) + (\lambda-\lambda')R(\lambda,A)R(\lambda',A)$ for the resolvent we can show
\[ 
\|R(\lambda',A)x\|_\lambda \leq \|R(\lambda,A)x\|_\lambda +(\lambda-\lambda')\|R(\lambda,A)R(\lambda',A)x\|_\lambda.
\]
The first property with the inequality above yields
\[
\|R(\lambda',A)x\|_\lambda \leq \frac{1}{\lambda}\|x\|_\lambda +\frac{(\lambda-\lambda')}{\lambda}\|R(\lambda',A)x\|_\lambda.
\]
After some algebraic manipulations, we conclude the proof of item $2$. Item $3$ is proved by just iterating property $2$. We can relate the norms with $\lambda^*<\lambda' <\lambda$ by $\|x\|_{\lambda'} \leq \|x\|_\lambda$. Define the new norm $\vertiii{x} = \underset{\mu>\omega}{\sup}\|x\|_\mu $. This norm is equivalent to the original norm in the Banach space $X$ and we have $\vertiii{\lambda R(\lambda,A)} \leq 1$ and applying the Hille-Yosida Theorem for the contraction case yields the desired result.
\end{proof}

\begin{theorem}[\textbf{Hille Yosida for Groups}]
Let $X$ a Banach space and $A:\mathrm{dom}(A) \subset X \rightarrow X$ a linear operator. The following assertions are equivalent:
\begin{enumerate}
\item $A$ generates a strongly continuous group such that
\[
\|\T_t\|\leq M e^{\omega |t|}
\]
\item $A$ and $-A$ are generators of strongly continuous semigroups;
\item $A$ is closed, $\mathrm{dom}(A)$ is dense and, for all $\lambda > \lambda^*$ we have that $\lambda \in \rho(A)$ e
\[
\|(|\lambda|-\lambda^*)R(\lambda,A)^n\|\leq M, \;\; \forall n \in \mathbb{N}
\]
\end{enumerate}
\end{theorem}

\begin{proof}
Just apply the Hille-Yosida for $A$ and $-A$.
\end{proof}

\begin{definition}
Let $X$ be a Banach space. A linear operator $A: \mathrm{dom}(A) \subset X \rightarrow X$ is called \textbf{dissipative} if, for all $\lambda >0$, it holds that
\[
\|(\lambda-A)x\| \geq \lambda \|x\|
\]
\end{definition}

\begin{lemma}\label{l4}
Let $X$ be a Banach space and $A:\mathrm{dom}(A) \subset X \rightarrow X$ a dissipative operator. Then the following hold
\begin{itemize}
\item[(i)] If $\mathrm{dom}(A)$ is dense, then $A$ is closable and the closure is dissipative.
\item[(ii)] $\lambda-\overline{A}(\mathrm{dom}(\overline{A}))=\overline{\lambda-A(\mathrm{dom}(A))}$, for all $\lambda>0$.
\item[(iii)] If $\lambda-A$ is surjective for some $\lambda$ then it is surjective for all $\lambda>0$.
\end{itemize}
\end{lemma}
\begin{proof}
$(1)$ Let us define
\[
D = \{x \in X: \exists \{x_n\}_{n \in \mathbb{N}} \in \mathrm{dom}(A)\; \text{such that}\; x_n \rightarrow x \;\text{and}\; Ax_n \rightarrow y\}.
\]
Our goal is to show that the operator $\overline{A}:D \rightarrow X$ defined as
\[
\overline{A}x = \lim_{n\rightarrow \infty} Ax_n, \;\; \text{where} \; \lim_{n\rightarrow\infty} x_n = x,
\]
is a well-defined closed operator that extends $A$. First, let us show that $\overline{A}$ does not depend on the chosen sequence. So, consider $\{x_n\}_{n \in \mathbb{N}},\{z_n\}_{n \in \mathbb{N}}\in \mathrm{dom}(A)$ converging to $x$ such that $Ax_n \rightarrow y$ and $Az_n \rightarrow w$. For all $\lambda>0$, the dissipativity of $A$ implies, for $x' \in \mathrm{dom}(A)$
\[
\|\lambda(\lambda-A)(x_n-z_n) -(\lambda-A)x'\| \geq \lambda\|\lambda(x_n - y_n) - x'\|.
\]
Taking limits on both sides of the inequality above and using the continuity of the norm gives us
\[
\|\lambda( y - w) - \lambda x' + Ax'\| \geq \lambda \| x'\|.
\]
Diving both sides of the inequality above and taking the limit $\lambda \rightarrow \infty$ we get $\|y-w - x'\| \geq \|x'\|$. The domain $\mathrm{dom}(A)$ is dense so we can choose $x'$ arbitrarily close to $y-w$, concluding that $y=w$. We must show that $\overline{A}$ is closed. Consider $\{x_n\}_{n \in \mathbb{N}} \in D$ such that $x_n \rightarrow x$ and $\overline{A}x_n \rightarrow y$. By the definition of the set $D$, there is $\{x_{n,k}\}_{k \in \mathbb{N}} \in \mathrm{dom}(A)$ such that
\[
\lim_{k \rightarrow \infty} x_{n,k} = x_n\;\; \text{and} \;\; \lim_{k \rightarrow \infty} A x_{n,k} = \overline{A} x_n.
\]
Thus for every $n>0$ we can find $k_n$ such that $\|x_n - x_{n,k_n}\| +\|Ax_{n,k_n} - Ax_n\| \leq 1/n$. Define $z_n = x_{n,k_n}$ and we found a sequence in $\mathrm{dom}(A)$ that converges to $x$ and $Az_n$ converges to $y$. Indeed,
\[
\|x - z_n\| + \|y - Az_n\|\leq \|x-x_n\| + \|x_n - z_n\| +\|y - \overline{A}x_n\|+\|Az_n - \overline{A}x_n\| \leq \|x-x_n\| + \|y - \overline{A}x_n\| +\frac{1}{n},
\]
concluding that $\overline{A}$ is closed and extends $A$. Because $\mathrm{dom}(A)$ is dense in the graph norm in $D(\overline{A}) = D$ for all $x$ in $D$ we can find a sequence $x_n \in \mathrm{dom}(A)$ such that $Ax_n$ converges to $\overline{A}x$. This is sufficient to see that the closure will also be dissipative. For item $(ii)$ note that our construction of the closure shows us that the graph of $A$ is dense in the graph of$\overline{A}$. This implies that $\overline{G(A)} = G(\overline{A})$ and, consequently, our assertion. For item $(iii)$ let $\lambda_0$ be the constant such that $\lambda_0 - A$ is surjective. Because $A$ is dissipative, we know that $\lambda_0 \in \rho(A)$. For $\lambda \in \rho(A)$ it follows
\[
\lambda - A = \lambda_0 - A +(\lambda-\lambda_0) \Rightarrow  \lambda - A = (\mathbbm{1} - (\lambda_0-\lambda)R(\lambda_0,A))(\lambda_0 - A).
\]
This identity tells us that $\lambda-A$ is invertible if and only if $(\mathbbm{1} - (\lambda_0-\lambda)R(\lambda_0,A))$ is invertible. This will happen when $|\lambda - \lambda_0| \leq 1/\|R(\lambda_0,A)\|$ because we can write the Neumann series
\[
(\mathbbm{1} - (\lambda_0-\lambda)R(\lambda_0,A))^{-1} = \sum_{n=0}^\infty (\lambda_0-\lambda)^nR(\lambda_0,A)^{n+1}.
\]
Since the operator is dissipative, the resolvent in $\lambda_0$ satisfies $\|R(\lambda_0,A)\|\leq \frac{1}{\lambda_0}$, implying that the radius of convergence of the Neumann series is greater than $\lambda_0$, so $(0,2\lambda_0) \subset \rho(A)$. Proceeding in this way, we can show that the Neumann series converges for all $\lambda>0$.
\end{proof}
\begin{theorem}[\textbf{Lummer-Phillips}]\label{t3}
Let $X$ be a Banach space and $A$ a linear operator densely defined and dissipative. The following assertions are equivalent
\begin{itemize}
\item[(i)] The closure of  $A$ generates a semigroup of contractions;
\item[(ii)] $(\lambda - A)(\mathrm{dom}(A))$ is dense for some $\lambda >0$.
\end{itemize}
\end{theorem}
\begin{proof}
Assume that $(i)$ holds. By the Hille-Yosida Theorem, we know that all $\lambda >0$ are in the resolvent set $\rho (\overline{A})$. The image of $\lambda-\overline{A}$ is all $X$ because the operator is invertible. But Lemma \ref{l4} implies that $(\lambda - \overline{A})(\mathrm{dom}(\overline{A}))= \overline((\lambda - A)(\mathrm{dom}(A)))=X$ and we conclude that the image is dense. Assume now that $(ii)$ holds. Since the image of $\lambda-A$ is dense we have that the image of $\lambda-\overline{A}$ is all $X$, by Lemma \ref{l4}. Because dissipative operators are clearly injective, this implies that $\lambda-\overline{A}$ is surjective, and this being true, implies that $\lambda \in \rho(A)$. Again, Lemma \ref{l4} gives us that $(0,\infty) \subset \rho(A)$ and the dissipativity gives $\|R(\lambda,A)\| \leq 1/\lambda$. the hypothesis of the contraction case of the Hille-Yosida theorem is satisfied, so $\overline{A}$ generates a strongly continuous semigroup of contractions. 
\end{proof}

We finish this section by stating the Trotter-Kato theorem, an important result about approximating one-parameter semigroups by approximating their generators.
\pagebreak

\begin{theorem}[\textbf{Trotter-Kato}]\label{TrotterKato}
    Let $\T_t$ and $\T_{n,t}$ strongly continuous one-parameter semigroups on a Banach space $X$, with generators $A$ an $A_n$, respectively. Suppose there is a $\lambda^*>0$ such that
    \[
    \max_{n\geq 1}\{ \|\T_t\|,\|\T_{n,t}\|\}\leq e^{\lambda^* t},
    \]
    for all $t \geq 0$. Then consider the following assertions
    \begin{itemize}
        \item [(i)] There exists a nontrivial $D\subset \mathrm{dom}(A)$ such that $A^n D \subset \mathrm{dom}(A)$ for all $n\geq 1$ and  $D\subset \mathrm{dom}(A_n)$ for all $n\in \mathbb{N}$ and $\lim_{n\rightarrow \infty}A_nx = Ax$ for all $x \in D$.
        \item [(ii)] For each $x\in D$, there exists $x_n \in \mathrm{dom}(A_n)$ such that
        \[
        \lim_{n\rightarrow \infty} x_n = x \quad \text{ and } \quad \lim_{n\rightarrow \infty} A_nx_n = Ax.
        \]
        \item [(iii)] $\lim_{n\rightarrow \infty}R(\lambda,A_n)\rightarrow R(\lambda,A)$ for all $\lambda > \lambda^*$.
        \item [(iv)] $\lim_{n\rightarrow \infty}\T_{n,t}x = \T_tx$ for all $x \in X$ uniformily in compacts.
    \end{itemize}
    Then we have the following chain of implications $(i) \Rightarrow (ii) \Leftrightarrow (iii) \Leftrightarrow (iv).$
\end{theorem}
The proof can be found in Theorem 4.8 of \cite{En}.

\chapter{Poisson Point Processes}
\label{sec:appendix}

\section{Point Processes}
\label{sec:appendix:ppp}

  The main goal of this section is to derive some important properties that we use in Chapter 5 to prove the consistency property for quantum DLR functionals. It is largely based on \cite{Reiss, Last2017, Jan}, and we suggest the interested reader consult them for a more detailed account. For our exposition, we will need to introduce a measure-theoretical construction called a \emph{coproduct space}. 
		
		\begin{definition}
			The \textbf{coproduct} or the \textbf{disjoint union} of countably infinitely many measure spaces $(X_n,\mathscr{A}_n, \mu_n)$ is defined as
			\begin{align*}
				\bigsqcup_{n \in \mathbb{N}}X_n \coloneqq \bigcup_{n\in \mathbb{N}}\{(x,n): x \in X_n\}, \;\;\;\;
				\mathscr{A} \coloneqq \left\{\bigsqcup_{n\in \mathbb{N}} A_n : A_n \in \mathscr{A}_n\right\},
			\end{align*}
			and the measure of each set is given by
			\[
			\;\;\;\;
			\mu\left(\bigsqcup_{n\in \mathbb{N}} A_n\right) \coloneqq \sum_{n \in \mathbb{N}}\mu_n(A_n).
			\]
		\end{definition}
		
		It is easy to check that the set $\mathscr{A}$ is a $\sigma$-algebra and that $\mu$ is a measure. 
  
  \begin{remark}
      The fact that we have a countable collection of spaces in our definition does not mean that one needs to restrict to this case: we could, as well, consider an uncountable family of measurable spaces.
  \end{remark} 
  
  We can define injections $i_n: X_n \rightarrow \bigsqcup_{n \in \mathbb{N}} X_n$ by $i_n(x) = (x,n)$. The maps $i_n$ are measurable since 
		\[
		i_m^{-1}\left(\bigsqcup_{n \in \mathbb{N}}A_n\right) = A_m \in \mathscr{A}_m,
		\] 
		by definition of the disjoint union. Moreover, the image of measurable sets is measurable. The next question would be what kind of measurable functions one can have on the coproduct. If we have a family of measurable functions $f_n: X_n\rightarrow Y$, with $\mathscr{B}$ the $\sigma$-algebra of $Y$, we can define a unique measurable function on the coproduct $f:\sqcup_{n \in \mathbb{N}}X_n\rightarrow Y$ by $$f(i_n(x)) \coloneqq f_n(x).$$ Actually, all measurable functions on the coproduct are of this form. This is summarized in the next proposition.
  \pagebreak
    \begin{proposition}\label{prop_measurable_coproduct}
        Let $(X_n,\mathscr{A}_n)_{n\in \mathbb{N}}$, $(Y,\mathscr{B})$ be measurable spaces and $\bigsqcup_{n\in \mathbb{N}} X_n$ the coproduct. Then for every measurable function $f:\bigsqcup_{n\in \mathbb{N}} X_n\rightarrow Y$ there are $f_n:X_n\rightarrow Y$ such that $f\circ i_n = f_n$. Moreover, if $f_n: X_n\rightarrow Y$ is a family of measurable functions there is a unique function $f$ measurable in the coproduct such that $f\circ i_n = f_n$.
    \end{proposition}
    \begin{proof}
        The first part of the proposition is trivial since the injections $i_n$ are measurable. We prove only the second part. Let $f_n:X_n\rightarrow Y$ and define $f$ by $f(x,n) = f_n(x)$. Given $B \in \mathscr{B}$ we have 
		\[
		f^{-1}(B) = \bigcup_{n \in \mathbb{N}}f^{-1}(B)\cap i_n(X_n) = \bigcup_{n \in \mathbb{N}}i_n(f_n^{-1}(B)),
		\]
		hence $f$ is measurable. This shows that the measurable functions on the coproduct as in one-to-one correspondence with sequences $\{f_n\}_{n\geq 1}$ of measurable functions $f_n: X_n\rightarrow Y$. 
	
    \end{proof}
   Let $X$ be a Polish space and $\mathscr{B}(X)$ its Borel $\sigma$-algebra. The space of counting measures on $X$, that we will denote by $\mathbb{N}(X)$ henceforth, is defined as
		\be\label{def_counting_measures}
		\mathbb{N}(X)\coloneqq \Bigg\{\mu = \sum_{k=1}^n \delta_{x_k}, x_k \in X, n \in \overline{\mathbb{N}}_0\Bigg\},
		\ee
		where $\overline{\mathbb{N}}_0=\mathbb{N}\cup\{0,\infty\}$ and the case $n=0$ is the null measure, i.e. $\mu_\emptyset(A)=0$ for every $A \in \mathscr{B}(X)$. Define for each measurable set $A\subset X$ the projection map $\pi_A : \mathbb{N}(X)\rightarrow \overline{\mathbb{N}}_0$ by
		\[
		\pi_A(\mu) = \mu(A),
		\]
		and consider $\mathscr{N}(X)$ be the smallest $\sigma$-algebra on $\mathbb{N}(X)$ such that all projections $\pi_A$ are measurable. 
    
		\begin{definition}
			Let $(\Omega, \mathscr{A}, \mathbb{P})$ be a probability space. A point process is a measurable function $N: \Omega \rightarrow \mathbb{N}(X)$. 
		\end{definition}
  
		A point process is said to be \emph{simple} if for almost every $\omega \in \Omega$ it holds that $N(\omega)(\{x\})<2$, for any $x \in X$. Notice that, since the projections are measurable, we can define new measurable functions using the point process by
		\[
		N_A\coloneqq \pi_A \circ N.
		\]
		One can see $N(A)$ as a random choice of points inside $B$. 
		\begin{proposition}\label{prop_measurable_countmeasure}
			The following two assertions are equivalent
			\begin{enumerate}[label=(\roman*), series=l_after]
				\item $N:\Omega \rightarrow \mathbb{N}(X)$ is a point process;
				\item $N_A: \Omega \rightarrow \overline{\mathbb{N}}_0$ is measurable for every $A \in \mathscr{B}(X)$. 
			\end{enumerate}
		\end{proposition}
		\begin{proof}
			Since the composition of measurable functions is measurable if $\pi_A$ and $N$ then $N_A = \pi_A \circ N$ is measurable. Assume that $(ii)$ holds. For each $C \subset \overline{\mathbb{N}}_0$, we have that $\pi_B^{-1}(C)$ is measurable by definition of the $\sigma$-algebra $\mathscr{N}(X)$. To show that $N$ is a measurable function, it is sufficient to show that $N^{-1}(\pi_B^{-1}(C))$ is measurable. But $N^{-1}(\pi_B^{-1}(C)) = (\pi_B\circ N)^{-1}(C)$ hence it is measurable. 
		\end{proof}
		
		\begin{example}
			Let $(\Omega, \mathscr{A}, \mathbb{P})$ be a probability space and $\xi_i:\Omega \rightarrow X$, $i=1,\dots, n$ random variable. Then,
			\[
			N = \sum_{i=1}^n \delta_{\xi_i}.
			\]
		\end{example}

        Next Lemma shows that there is a way to lift measurable sets in $X^n$ to measurable sets in the $\sigma$-algebra $\mathscr{N}(X)$. There is an intimate relationship between the coproduct space of the product spaces $X^n$, including $X^0\coloneqq \{0\}$ and $X^\infty$ the infinity countable product of copies of $X$ given the cylinder $\sigma$-algebra, and the space of counting measures $\mathbb{N}(X)$, where the later is basically a permutation invariant version of the coproduct. This idea can be seen in a rigorous fashion by defining the function $\varphi:\bigsqcup_{n \in \overline{\mathbb{N}}_0}X^n\rightarrow \mathbb{N}(X)$ be defined by
			\be\label{permutation}
			\varphi_n(x_1,\dots, x_n)= \sum_{i=1}^n \delta_{x_i}.
			\ee
    By Proposition \ref{prop_measurable_coproduct} we just need to show that each $\varphi_n$ is measurable.
    \begin{lemma}
        The functions $\varphi_n$ defined by Equation \eqref{permutation} are measurable.
    \end{lemma}
    \begin{proof}
        Notice that $\varphi_n$ are point processes according to our definition, thus Proposition \ref{prop_measurable_countmeasure} says that we only need to show that for each $A\in \mathscr{B}(X)$ the map $\varphi_{n,A}:X^n \rightarrow \mathbb{N}_0$ is measurable. It is sufficient to show that the pre-images of the singletons are measurable. Thus
			\[
			\varphi^{-1}_{n,A}(\{m\}) = \begin{cases}
				\emptyset & n>m, \\
				\bigcup_{\sigma \in F} B_\sigma & n\leq m,
			\end{cases}
			\]
			where $F =\{\sigma:\{1,\dots,n\}\rightarrow \{0,1\}: \sum_{i=1}^n \sigma(i) = m\}$, $B_\sigma = \prod_{i=1}^n A_{\sigma(i)}$, where $A_1= A$ and $A_0 = A^c$. The sets $B_\sigma$ are clearly measurable in the product space $X^n$, thus we conclude the proof. 
    \end{proof}	
  
		\begin{lemma}\label{lema_measurable_countmeasure}
            Let $n\geq 1$ and consider $X^n$ be the product space equipped with the product topology. For each $A \in \mathscr{B}(X^n)$, the set 
            \[
            \widetilde{A} = \Bigg\{\mu=\sum_{k=1}^n \delta_{x_k}, (x_1,\dots,x_n) \in A\Bigg\}   
            \]
            is measurable in $\mathbb{N}(X)$.
        \end{lemma}
        \begin{proof}
            The case $n=1$ is easy, the sets $\widetilde{A}$ can be written as $\pi_{A}^{-1}(\{1\})\cap \pi_{A^c}^{-1}(\{0\})$. For $n\geq 2$, consider sets of the form $A = A_1 \times \dots \times A_n$. The difficulty in this case relies on the fact that the sets $A_k$ may have a nonempty intersection. For each $A_k$, we can use the remaining $A_j$ to create a partition of it. Define,
            \[
            A_k = \bigcup_{T \in \{0,1\}^n, T(k) = 1} A_T,
            \]
            where each $A_T$ is defined as $$A_T = \bigcap_{k=1}^n A_{k, T(k)},$$ and $A_{k,T(k)}=A_k$ if $T(k) =1$ or $A^c_k$ if $T(k)=0$. Notice that even if we have $k\neq j$ we may get $A_{k,T(k)} = A_{j,T(j)}$ for some $T$. The strategy will be to write all the possibilities that the points to land in $\cup_{k=1}^nA_k$. Thus, first, we select some $T\in \{0,1\}^n$. Then, $m$ be the number of equivalence classes of the set $\{1,\dots,n\}$ separated according to $T$, i.e., $k,j$ are in the same class if  $A_{k,T(k)} = A_{j,T(j)}$. let $B_1,\dots, B_m$ the distinct sets appearing in the sequence $A_{k,T(k)}$. Let $n_\ell=|\{k: A_{k,T(k)}=B_\ell\}|$. Then,
            \[
            \pi_{B_1}^{-1}(\{n_1\})\cap \dots \cap \pi_{B_m}^{-1}(\{n_m\})\cap \pi_{(\cup_{\ell=1}^mB_\ell)^c}^{-1}(\{0\}), 
            \]
            is the set of all counting measures with $n$ points such that $n_1$ points are in $B_1$, $n_2$ points are in $B_2$ etc. This is clearly measurable and by taking the union for all $T$, we get that $\widetilde{A}$ is measurable whenever it is a product set. Consider now the set 
            \[
           \mathscr{C}= \{A \in \mathscr{B}(X^n): \tilde {A} \in \mathscr{N}(X)\}.
            \]
            We claim that $\mathscr{C}$ is a $\sigma$-algebra. Indeed, that $\emptyset \in \mathscr{C}$ is clear and $X^n$, since it is a product set, our previous argument applies. Notice that
            \[
            \widetilde{A^c} =\Bigg\{\mu=\sum_{k=1}^n \delta_{x_k}, (x_1,\dots,x_n) \in A^c\Bigg\}  = (\widetilde{A})^c \cap \widetilde{X^n}.
            \]
            The set $\mathscr{C}$ is obviously closed by countable unions, showing that our claim holds. Since we proved that $\mathscr{C}$ must contain all the product sets $A_1\times \dots \times A_n$, then $\mathscr{C} = \mathscr{B}(X^n)$.
        \end{proof}
		
		A point process  $N:\Omega \rightarrow \mathbb{N}(X)$ is a \emph{Poisson Point Process} with \emph{intensity measure} $\mu$ if for every $A \in \mathscr{B}(X)$ the two following conditions are satisfied
		\begin{itemize}
			\item[(1)] $\mathbb{P}(N_A=k)= \frac{\mu(A)^k}{k!} e^{-\mu(A)}$, for any $k\geq 0$. 
			\item[(2)] For any $B_1, \dots, B_m \in \mathscr{B}(X)$ pairwise disjoint the random variables $N(B_1), \dots, N(B_m)$ are independent.
		\end{itemize}
		
		We will show that a Poisson point process exists and that it also has a representation as an empirical process. Let $\mu$ be a finite measure on $X$ and $N$ be the following point process
		\begin{equation}\label{Poisson_rep}
			N \coloneqq \sum_{i=1}^\tau \delta_{\xi_i},
		\end{equation}
		where $\tau, \xi_1,\xi_2,\dots$ are a countable family of independent random variables such that
		\begin{itemize}
			
			\item $\mathbb{P}(\tau = k) = \frac{\mu(X)^k}{k!}e^{-\mu(X)}$ 
			
			\item $\mathbb{P}(\xi_i \in B) = \frac{\mu(B)}{\mu(X)}$, for any $B\in \mathscr{B}(X)$ and $i=1,2, \dots$.
			
		\end{itemize}
		We will show first that such a family of random variables exists. Consider $\Omega = \mathbb{N}_0 \times \prod_{i\in \mathbb{N}} X$, with the product $\sigma$-algebra. We can define on it the product probability measure $\mathbb{P}$ on the cylinder sets by 
		\[
		\mathbb{P}\Bigg(C \times \prod_{i \in \mathbb{N}} B_i\Bigg) = P_\mu(X)(C)\times \prod_{i\in \mathbb{N}}\frac{\mu(B_i)}{\mu(X)},
		\]
		where $P_\mu(X)$ is the Poisson distribution with parameter $\mu(X)$. That this defines a probability measure on the product space $\Omega$ follows from \cite{Saeki}. Note that, by construction, the projections are independent random variables with the desired distribution.
		
		\begin{proposition}
			Let $X$ be a Polish space with $\mathscr{B}(X)$ its Borel $\sigma$-algebra and $\mu$ a finite measure. Then the point process $N$ defined above is a Poisson Point Process. 
		\end{proposition}
		\begin{proof}
			Let $B \in \mathscr{B}(X)$. Using the independence of the random variables, we have
			\begin{align*}
				\mathbb{P}(N_B = k) &= \sum_{m\geq k}\mathbb{P}\left(\sum_{i=1}^m \delta_{\xi_i}(B)=k\right)\mathbb{P}(\tau = m) \\
				&= e^{-\mu(X)}\sum_{m\geq k}\frac{\mu(X)^m}{m!}\mathbb{P}\left(\sum_{i=1}^m \delta_{\xi_i}(B)=k\right)
			\end{align*}
			
			Hence
			\begin{align*}
				\mathbb{P}\left(\sum_{i=1}^m \delta_{\xi_i}(B)=k\right) &= \sum_{\substack{a_1+\dots+a_m = k \\ a_i=0,1}}\mathbb{P}\left(\delta_{\xi_i}(B) = a_i, i=1,\dots, m\right) \\
				&= \sum_{\substack{a_1+\dots+a_m = k \\ a_i=0,1}} \prod_{i=1}^m \mathbb{P}(\delta_{\xi_i}(B) = a_i),
			\end{align*} 
			where the last equality is due to the independence of the random variables. By our hypothesis on the random variables $\xi_i$, we have
			\[
			\mathbb{P}(\delta_{\xi_i}(B)= a_i) = \begin{cases}
				\frac{\mu(B^c)}{\mu(X)}, & a_i=0  \\
				\frac{\mu(B)}{\mu(X)}, & a_i = 1.
			\end{cases}
			\] 
			A standard stars and bars argument gives us 
			\[
			\sum_{\substack{a_1+\dots+a_m = k \\ a_i=0,1}} \prod_{i=1}^m \mathbb{P}(\delta_{\xi_i}(B) = a_i) = \binom{m}{k} \frac{\mu(B)^k\mu(B^c)^{m-k}}{\mu(X)^m}.
			\]
			Hence,
			\[
			\mathbb{P}(N_B = k) = \frac{\mu(B)^k}{k!}e^{-\mu(X)}\sum_{m\geq k} \frac{\mu(B^c)^{m-k}}{(m-k)!} = \frac{\mu(B)^k}{k!}e^{-\mu(B)}.
			\]
			
			Consider $B_1, B_2, \dots, B_m$ disjoint measurable sets. In order to show that the random variables $N_{B_1}, \dots, N_{B_m}$ are independent it is sufficient to show that
			\[
			\mathbb{P}(N_{B_i}= n_i, i=1, \dots, m) = \prod_{i=1}^m \mathbb{P}(N_{B_i} = n_i).
			\]
			We will assume that $\cup_{i=1}^m B_i = X$ now and show how to prove the general case later. With this assumption, necessarily we must have $\tau(\omega)= n= \sum_{i=1}^m n_i$. Thus, using independence we get
			\begin{align*}
				\mathbb{P}(N_{B_i}= n_i, i=1, \dots, m) &= \frac{\mu(X)^n}{n!}e^{-\mu(X)}\mathbb{P}\left(\sum_{i=1}^n\delta_{\xi_i}(B)= n_j, j=1, \dots, m\right) \\
				&=\sum_{\substack{a_{1,j}+\dots+a_{n,j} = n_j \\ a_{i,j}=0,1}}\mathbb{P}\left(\delta_{\xi_i}(B)= a_{i,j}, j=1, \dots, m, i=1,\dots, n\right)	
			\end{align*}
			Since the sets $B_1, \dots, B_m$ are disjoint, if $a_{i,j}=1$ for some $j$, then $a_{i,k}=0$ for any other $k\neq j$ since the opposite would imply that there is a point in $B_j \cap B_k$. Thus, using the independence of the random variables $\xi_i$ we get
			\begin{align*}
				\sum_{\substack{a_{1,j}+\dots+a_{n,j} = n_j \\ a_{i,j}=0,1}}\mathbb{P} \Bigg(\begin{array}{@{}c c@{}}
   \delta_{\xi_i}(B)= a_{i,j},
  &
  \begin{matrix}
  j=1, \dots, m \\[2ex]
  i=1,\dots, n
  \end{matrix}
  \end{array}
\Bigg) &= \sum_{\substack{a_{1,j}+\dots+a_{n,j} = n_j \\ a_{i,1}+\dots + a_{i,m}= 1 \\ a_{i,j}=0,1}}\prod_{i=1}^n\mathbb{P}\left(\delta_{\xi_i}(B)= a_{i,j}, j=1, \dots, m\right) \\
				& \sum_{\substack{a_{1,j}+\dots+a_{n,j} = n_j \\ a_{i,1}+\dots + a_{i,m}= 1 \\ a_{i,j}=0,1}}\prod_{i=1}^n\left(\frac{\mu(B_i)}{\mu(X)}\right)^{n_i}.
			\end{align*}
			Consider $a_{i,j}=0,1$ a solution to $a_{1,j}+\dots+a_{n,j} = n_j$, for any $j=1,\dots, m$ such that $a_{i,1}+\dots + a_{i,m}= 1$. The second equation says that for each $i$ there must be only one $j$ with $a_{i,j}\neq 0$. So we proceed in the following way. To produce a solution  to $a_{1,j}+\dots+a_{n,j} = n_j$ satisfying this constraint, we first choose $n_1$ indices $i$ to put as equal to $1$ and the rest we put equals to zero. For $j=2$ we now have $n-n_1$ indices avaible, so we choose $n_2$ of those to put as equal to $1$. We can proceed inductively until we reach the case $j=m$. This reasoning implies that the number of possible solutions $a_{i,j}$ is exactly
			\[
			\binom{n}{n_1}\binom{n-n_1}{n_2}\dots\binom{n-n_1-\dots-n_{m-1}}{n_m} = \frac{n!}{n_1!n_2!\dots n_m!}
			\]
			Hence,
			\[
			\mathbb{P}(N_{B_i}= n_i, i=1, \dots, m) = \frac{e^{-\mu(X)}}{n_1!n_2!\dots n_m!}\prod_{i=1}^m\mu(B_i)^{n_i},
			\]
			rearranging the terms and using that $\mu(X)=\mu(B_1)+\dots+\mu(B_m)$ yields the desired result. Consider now the general case, i.e., any family of disjoint measurable sets $B_1,\dots, B_m$. Write $B=\cup_{i=1}^m B_i$ and using our previous calculations we get
			\begin{align*}
				\mathbb{P}(N_{B_i}= n_i, i=1, \dots, m) &= \sum_{k\geq 0}\mathbb{P}(N_{B_i}= n_i, i=1, \dots, m, N_{X\setminus B}=k) \\
				&=\prod_{i=1}^m \mathbb{P}(N_{B_i} = n_i) \sum_{k\geq 1} \mathbb{P}(N_{X\setminus B} = k) = \prod_{i=1}^m \mathbb{P}(N_{B_i} = n_i)
			\end{align*}

		\end{proof}
		
		Poisson point process also has a uniqueness property in the sense that for any Poisson process with a given intensity measure $\mu$ are equal in distribution. The proof of this fact can be found in Theorem 1.2.1 in \cite{Reiss}. 
  \begin{proposition}\label{Poisson_int}
			Let $f: \mathbb{N}(X)\rightarrow \mathbb{R}$ be a bounded measurable function and $N:\Omega \rightarrow \mathbb{N}(X)$ a Poisson point process. Then, the following holds,
			\[
			\int_{\mathbb{N}(X)} f(\nu) dN(\nu) = f(\mu_\emptyset)e^{-\mu(X)}+e^{-\mu(X)}\sum_{n\geq 1}\frac{1}{n!}\int_{X^n} f(\delta_{x_1}+\dots+\delta_{x_n})d\mu^{\otimes n}(x_1,\dots,x_n),
			\]
			where $\mu^{\otimes n}$ is the $n$-fold product measure. 
		\end{proposition}
		\begin{proof}
			First, let $\tau, \xi_1,\xi_2, \dots$ be the random variables given by the Poisson point process \eqref{Poisson_rep}. Let $X^n$ be the product space, with the product $\sigma$-algebra. For the case $X^0 = \{0\}$ and $X^\infty$ the countable infinity product space with the cylinder $\sigma$-algebra. We will define the function  $\tilde{N}:\Omega\rightarrow \bigsqcup_{n \in \mathbb{N}\cup\{0,\infty\}}X^n$ given by
			\[
			\tilde{N}(\omega) = \begin{cases}
				(\xi_1(\omega),\xi_2(\omega), \dots ,\xi_{\tau(\omega)}(\omega)), & \tau(\omega) \neq 0 \\
				\mu_\emptyset& \tau(\omega)=0.
			\end{cases} 
			\]
			This function is measurable and $\varphi \circ \tilde{N} = N$. Given an integrable function $f:\bigsqcup_{n \in \overline{\mathbb{N}}_0}X^n \rightarrow \mathbb{R}$ and $f_n$ its restrictions to the subspace $X^n$. We have
			\[
			\int_\Omega f \circ \tilde{N}(\omega) d\mathbb{P}(\omega) = f_0(0)e^{-\mu(X)}+\sum_{n\geq 1} \int_{\tau = n} f_n(\xi_1(\omega), \dots, \xi_n(\omega)) d\mathbb{P}(\omega).
			\]
			Suppose that $f_n = \mathbbm{1}_{B_1 \times \dots \times B_n}$, for measurable sets $B_i \in \mathscr{B}(X)$. Independence of the random variables $\xi_i$ yields
			\begin{align*}
				\int_{\tau = n} f_n(\xi_1(\omega), \dots, \xi_n(\omega)) d\mathbb{P}(\omega) &= \mathbb{P}(\tau=n)\prod_{i=1}^n \mathbb{P}(\xi_i \in B_i) \\
				&= \frac{e^{-\mu(X)}}{n!}\int_{X^n} \mathbbm{1}_{B_1 \times \dots \times B_n}(x_1,\dots,x_n) d\mu^{\otimes n}(x_1,\dots,x_n).
			\end{align*}
			Standard measure theoretic techniques allow us to extend the above result for general integral functions. Hence
			\[
			\int_\Omega f \circ \tilde{N}(\omega) d\mathbb{P}(\omega) = f_0(0)e^{-\mu(X)}+\sum_{n\geq 1}\frac{e^{-\mu(X)}}{n!}\int_{X^n}f_n(x_1,\dots,x_n) d\mu^{\otimes n}(x_1,\dots,x_n).
			\]
		 Thus, for any $f:\mathbb{N}(X)\rightarrow \mathbb{R}$, it holds
			\begin{align*}
				\int_{\mathbb{N}(X)} f(\nu) dN(\nu) &=\int_\Omega f \circ(\varphi \circ \tilde{N})(\omega) d\mathbb{P}(\omega) \\
				& = f_0(\mu_\emptyset)e^{-\mu(X)}+e^{-\mu(X)}\sum_{n\geq 1}\frac{1}{n!}\int_{X^n}f(\delta_{x_1}+ \dots + \delta_{x_n}) d\mu^{\otimes n}(x_1,\dots,x_n).
			\end{align*}
		\end{proof}
		In the case that both point processes are independent Poisson point processes, the sum is again a Poisson Point Process, as the following proposition shows.
		\begin{proposition}\label{Poisson_independent}
			Let $N,\tilde{N}:\Omega \rightarrow \mathbb{N}(X)$ be two independent Poisson point processes with intensity measures $\mu$ and $\nu$ respectively. Then the point process $N+\tilde{N}$ is again a Poisson point process with intensity measure $\mu+\nu$.
		\end{proposition}
		\begin{proof}
			Take $B\in \mathscr{B}(X)$, and consider $N_B+\tilde{N}_{B}$. Then, the independence of $N$ and $\tilde{N}$ imply the independence of $N_B$ and $N_B$, thus
			\begin{align*}
				\mathbb{P}(N_B+ \tilde{N}_{B}=k)&=\sum_{j=0}^k \mathbb{P}(N_B=k-j)\mathbb{P}(\tilde{N}_{B}=j) = \sum_{j=0}^k \frac{\mu(B)^{k-j}\nu(B)^j}{(k-j)!j!}e^{-(\mu(B)+\nu(B))} \\
				& = \frac{e^{-(\mu(B)+\nu(B))}}{k!}\sum_{j=0}^k \binom{k}{j}\mu(B)^{k-j}\nu(B)^j = \frac{(\mu(B)+\nu(B))^ke^{-(\mu(B)+\nu(B))}}{k!}. 
			\end{align*}
			
			Consider $B_1, \dots, B_m$ disjoint measurable sets. 
			\[
			\mathbb{P}(N_{B_i}+ \tilde{N}_{B_i}=k_i, 1\leq i \leq m) = \sum_{\substack{k_{i,1}+k_{i,2}=k_i \\ i=1, \dots, m}} \mathbb{P}(N_{B_i}=k_{i,1}, \tilde{N}_{B_i}=k_{i,2}, i=1, \dots, m)  
			\]
			
			Since $N$ and $\tilde{N}$ are independent and the following holds, we have
			\begin{align*}
				\mathbb{P}\Bigg(N_{B_i}=k_{i,1}, \tilde{N}_{B_i}=k_{i,2}, 1\leq i &\leq m\Bigg)
				=  \mathbb{P}(N \in \pi_{B_i}^{-1}(\{k_{i,1}\}), \tilde{N} \in \pi_{B_i}^{-1}(\{k_{i,2}\}) i=1, \dots, m) \\
				&=  \mathbb{P}(N \in \pi_{B_i}^{-1}(\{k_{i,1}\})i=1, \dots, m)  \mathbb{P}(\tilde{N} \in \pi_{B_i}^{-1}(\{k_{i,2}\}) i=1, \dots, m).
			\end{align*}
			Hence,
			\begin{align*}
				\mathbb{P}(N \in \pi_{B_i}^{-1}(\{k_{i,1}\}), \tilde{N} \in \pi_{B_i}^{-1}(\{k_{i,2}\}) i=1, \dots, m) &=\prod_{i=1}^m\left(\sum_{k_{i,1}+k_{i,2}=k_i} \mathbb{P}(N \in \pi_{B_i}^{-1}(\{k_{i,1}\}))  \mathbb{P}(\tilde{N} \in \pi_{B_i}^{-1}(\{k_{i,2}\}))\right) \\
				&=\prod_{i=1}^m\sum_{k_{i,1}+k_{i,2}=k_i}\mathbb{P}(N_{B_i}=k_{i,1}, \tilde{N}_{B_i}=k_{i,2}).
			\end{align*}
			This yields the desired result. 
		\end{proof}

  In what follows, we will derive some properties of the integration with respect to Poisson point processes.

\begin{lemma}\label{lemma_ppp_decomp}
		Let $f:\mathbb{N}(X)\rightarrow \mathbb{R}$ be a bounded measurable function and $N_1, N_2$ two Poisson point processes with intensity measures $\mu$ and $\nu$, respectively. It holds
		\be\label{pp_decomp}
		\int_{\mathbb{N}(X)} f(\lambda) d (N_1+N_2)(\lambda) = \int_{\mathbb{N}(X)}\int_{\mathbb{N}(X)}  f(\lambda_1+\lambda_2) dN_1(\lambda_1)d N_2(\lambda_2)
		\ee
	\end{lemma}
	\begin{proof}
    By Propositions \ref{Poisson_int} and \ref{Poisson_independent} we know that
    \begin{equation}\label{eq_fd4}
    \begin{split}
    &\int_{\mathbb{N}(X)} f(\lambda) d (N+\tilde{N})(\lambda) = \\
    &f_0(\mu_\emptyset)e^{-\mu(X)-\nu(X)}+e^{-\mu(X)-\nu(X)}\sum_{n\geq 1}\frac{1}{n!}\int_{X^n}f(\delta_{x_1}+\dots +\delta_{x_n})d(\mu+\nu)^{\otimes n}(x_1,\dots,x_n)
    \end{split}
    \end{equation}
  If we have two measures $\mu, \nu$ on $X$, there are many ways to combine them into a product measure in $X^n$. For instance, for each $C\subset \{1,\dots, n\}$, define $\mu \times_C \nu$ as the measure given by
  \[
  \mu \times_C \nu(A) =\prod_{k\in C}\mu(A_k)\prod_{k \not\in C}\nu(A_k)
  \]
  for every cylinder set $A=A_1\times \dots \times A_n$. They related to the $n$-fold product measure of $\mu+\nu$ by
  \[
  (\mu+\nu)^{\otimes n}(A) = \prod_{k=1}^n(\mu(A_k)+\nu(A_k)) = \sum_{C\subset \{1,\dots,n\}}\mu\times_C\nu(A).
  \]
   By standard measure-theoretic arguments, the Equation above implies that for every measurable function $f:X^n\rightarrow \mathbb{C}$ we must have then
  \be\label{eq_fd1}
    \int_{X^n}f(x_1,\dots,x_n)d(\mu+\nu)^{\otimes n}(x_1,\dots,x_n) = \sum_{C\subset \{1,\dots,n\}}\int_{X^n} f(x_1,\dots,x_n) d(\mu\times_C\nu)(x_1,\dots,x_n).
  \ee
    A special case is the measure $\mu\times_C \nu$ when $C=\{1,\dots,k\}$, because it is simply $\mu^{\otimes k}\times \nu^{\otimes (n-k)}$. Each $\mu\times_C \nu$ with $|C|=k$ can be transformed into this one by a permutation. Since permutations have unit determinant, the formula for change of variables when the function $f$ is permutation invariant gives us
    \begin{equation}\label{eq_fd2}
    \begin{split}
&\sum_{\substack{C\subset \{1,\dots,n\} \\ |C|=k}}\int_{X^n} f(x_1,\dots,x_n) d(\mu\times_C\nu)(x_1,\dots,x_n) = \\
&\binom{n}{k} \int_{X^k}\int_{X^{n-k}} f(x_1,\dots,x_n) d\mu^{\otimes n}(x_1,\dots,x_k)d\nu^{\otimes (n-k)}(x_{k+1},\dots,x_n).
\end{split}
    \end{equation}
    Equations \eqref{eq_fd1} and \eqref{eq_fd2} together yield 
    \begin{equation}\label{eq_fd3}
    \begin{split}
    &\frac{1}{n!}\int_{X^n}f(\delta_{x_1}+\dots +\delta_{x_n})d(\mu+\nu)^{\otimes n}(x_1,\dots,x_n) = \\
    &\sum_{k=1}^n\frac{1}{k!}\frac{1}{(n-k)!}
    \int_{X^k}\int_{X^{n-k}} f(\delta_{x_1}+\dots+\delta_{x_n}) d\mu^{\otimes n}(x_1,\dots,x_k)d\nu^{\otimes (n-k)}(x_{k+1},\dots,x_n)
    \end{split}
    \end{equation}
   by renaming the variables $n-k = m$ and using that $\sum_{n\geq 0}\sum_{k+m = n} = \sum_{k\geq 0}\sum_{m\geq 0}$, Equation \eqref{eq_fd3} together with Equation \eqref{eq_fd4} implies the desired result.
  \end{proof}
        The result above implies that when we have a finite number of independent Poisson point processes we can associate to each draw a definite label allowing us to integrate more general functions, that even depend on these labels. 
		\begin{corollary}\label{Corol_ppp}
			Let $N_i$, for $i=1,\dots,M$ be independent Poisson point processes on $[0,1]$ with intensity measures $\lambda_i dt$, for $\lambda_i>0$. Let $N = \sum_{i=1}^N N_i$ and $f:\mathbb{N}([0,1]\times\{1,\dots,M\})\rightarrow \mathbb{R}$ be a bounded measurable function. Then, it holds
			\[
			\int_{\mathbb{N}(X)} f(\nu) dN(\nu) = e^{-\sum_{i=1}^M  \lambda_i}\sum_{n\geq 0} \frac{1}{n!}\int_{[0,1]^n}  \sum_{i_1,\dots,i_n\in [M]} f\left(\sum_{j=1}^n\delta_{t_j,i_j}\right) \prod_{j=1}^n \lambda_j dt^n 
			\]
		\end{corollary}

		Subsequently, we consider $X=[0,1]$. Another important example that we will use to construct random representations for spin systems is the Bernoulli point process. Given a two point process $N, \tilde{N}$ it is straightforward to see that $N+\tilde{N}$ is again a point process. Consider $r \in [0,+\infty)$ and $\{\xi_{i,j}\}_{i\in \mathbb{N}, j =1,\dots, n}$ a sequence of i.i.d variables such that
		\[
		\mathbbm{P}(\xi_{n,j} = 0) = 1 - \mathbbm{P}(\xi_{n,j} = 1) = \frac{r}{n},
		\]
		for $1\leq j \leq n$. These are probabilities for $n$ large enough. Define the point process
		\be\label{BPP}
		N_n(x,B) = \sum_{j=1}^n \xi_{n,j}(x)\delta_{\frac{j}{n}}(B),
		\ee
		where $B\in \mathscr{B}([0,1])$. 
		
		\begin{proposition}\label{integformpoirep:app1}
			Let $N_n$ be the Bernoulli point process defined in \eqref{BPP}. Then, we have 
			\[
			\int_{\mathbb{N}(X)} f(\nu) dN_n(\nu) = \sum_{m\geq 0} \sum_{\substack{j_l \in \{0,\dots,n-1\} \\ 1 \leq l \leq k}} f\left(\delta_{\frac{j_1}{n}}+ \dots +\delta_{\frac{j_m}{n}}\right)\left(1- \frac{r}{n}\right)^{n-m}\left(\frac{r}{n}\right)^m 
			\]
		\end{proposition}  
		\begin{proof}
			The strategy of this proof will be the same as the one employed in Proposition \ref{Poisson_int}. Let $\tilde{N}: \Omega \rightarrow \bigsqcup_{n \in \mathbb{N}_0}[0,1]^n$ defined by
			\[
			\tilde{N}_n(\omega) = \begin{cases}
				(\frac{j_1}{n}, \dots, \frac{j_k}{n}), & \xi_{n,j_l}=1 \;\; \mathrm{ and } \;\;\xi_{n,j}=0, j\neq j_l, 1\leq l \leq k, \\
				0, & \sum_{j=1}^n \xi_{n,j}=0.
			\end{cases}
			\]
			It holds 
			\[
			\int_\Omega f\circ \tilde{N}_n(\omega) d\mathbb{P}(\omega) = \sum_{m\geq 0} \int_{\sum_j \xi_{n,j}=m} f_m \circ \tilde{N}_n(\omega)d\mathbb{P}(\omega). 
			\]
			It is straightforward to see that the r.h.s of the equation above is equal to
			\[
			\sum_{\substack{j_l \in \{0,\dots,n-1\}\\ 1 \leq l \leq k}}f\left(\frac{j_1}{n}, \dots, \frac{j_k}{n}\right) \left(1-\frac{r}{n}\right)^{n-m}\left(\frac{r}{n}\right)^m.
			\]
			Using the map $\varphi$ defined in Proposition \ref{Poisson_int} yields the desired result. 
		\end{proof}
		
		Suppose that we have a sequence of probability measures $\mu_n$ in the coproduct space $\bigsqcup_{m \in \mathbb{N}_0}[0,1]^m$. Then, one can define measures on $[0,1]^m$ by restriction. Let these restrictions be denoted by $\mu_{n,m}$. Then, if we have that each $\mu_{n,m}$ converges weakly to a $\mu_m$, then the monotone convergence theorem implies that $\mu_n$ converges to $\mu = \sum_m \mu_m$. Let $B_1,\dots, B_m$ be Borel sets in $[0,1]$. Then, for a continuous function $f:[0,1]^m \rightarrow \mathbb{R} $, we have 
		\[
		\sum_{\substack{j_l \in \{0,\dots,n-1\}\\ 1 \leq l \leq k}}f\left(\frac{j_1}{n}, \dots, \frac{j_k}{n}\right) \left(1-\frac{r}{n}\right)^{n-m}\left(\frac{r}{n}\right)^m =\left(1-\frac{r}{n}\right)^{n-m}\int_{[0,1]^m} g_n(x) d\lambda^{\otimes n},
		\]
		where the function $g_n:[0,1]^m \rightarrow \mathbb{R}$ is defined by 
		\[
		g_n(x) = f\left(\frac{j_1}{n}, \dots, \frac{j_l}{n}\right), \quad \textrm{if} \;\; x_l \in \left(\frac{j_l-1}{n}, \frac{j_l}{n}\right],
		\]
		and $dr = r dx$, where $dx$ is the Lebesgue measure. Notice that $\lim_{n\rightarrow \infty} g_n = f$ pointwise. The Lebesgue dominated convergence theorem gives us that,
		\[
		\lim_{n\rightarrow \infty} \left(1-\frac{\lambda}{n}\right)^{n-m}\int_{[0,1]^m}g_n(x) d\lambda^{\otimes m} = e^{-\lambda}\int_{[0,1]^m}f(x)d\lambda^{\otimes m}.
		\]
		Thus, we get that the Bernoulli point processes converge weakly to a Poisson point process with intensity measure $d\lambda$. 

\backmatter \singlespacing   
\bibliographystyle{hacm}
\bibliography{bib-refs}  



\end{document}